\documentclass[12pt]{article}
\usepackage{amsmath}
\usepackage{amsfonts}
\usepackage{makeidx}
\usepackage{amssymb}
\usepackage{times}

\newtheorem{satz}{Theorem}[section]
\newtheorem{definition}[satz]{Definition}
\newtheorem{lemma}[satz]{Lemma}
\newtheorem{koro}[satz]{Corollary}
\newtheorem{bemerkung}[satz]{Remark}

\newtheorem{proposition}[satz]{Proposition}
\newtheorem{notation}[satz]{Notation}
\newenvironment{proof}{\par\noindent {\it Proof:} \hspace{7pt}}{\hfill\hbox{\vrule width 7pt depth 0pt height 7pt}
\par\vspace{10pt}}

\makeindex
\input{tcilatex}
\begin{document}

\title{Lieb--Robinson Bounds for Multi--Commutators and Applications to
Response Theory}
\author{J.-B. Bru \and W. de Siqueira Pedra}
\date{\today }
\maketitle

\begin{abstract}
We generalize to multi--com%
\-%
mutators the usual Lieb--Robin%
\-%
son bounds for commutators. In the spirit of constructive QFT, this is done
so as to allow the use of combinatorics of minimally connected graphs (tree
expansions) in order to estimate time--dependent multi--com%
\-%
mutators for interacting fermions.\ Lieb--Robinson bounds for
multi--commutators are effective mathematical tools to handle analytic
aspects of the dynamics of quantum particles with interactions which are
non--vanishing in the whole space and possibly time--dependent. To
illustrate this, we prove that the bounds for multi--commutators of order
three yield existence of fundamental solutions for the corresponding
non--auto%
\-%
nomous initial value problems for observables of interacting fermions on
lattices. We further show how bounds for multi--commutators of an order
higher than two can be used to study linear and non--linear responses of
interacting fermions to external perturbations. All results also apply to
quantum spin systems, with obvious modifications. However, we only explain
the fermionic case in detail, in view of applications to microscopic quantum
theory of electrical conduction discussed here and because this case is
technically more involved.
\end{abstract}

\noindent\textbf{MSC2010:} (Primary) 82C10, 82C20, 82C22, 47D06, 58D25; (Secondary) 82C70, 82C44, 34G10
\newpage
\tableofcontents%

\section{Introduction\label{Sectino introduction}}

\emph{Lieb--Robinson bounds} are upper--bounds on time--dependent
commutators and were originally used to estimate propagation velocities of
information in quantum spin systems. They have first been derived in 1972 by
Lieb and Robinson \cite{liebrobinsonbounds}. Nowadays, they are widely used
in quantum information and condensed matter physics. Phenomenological
consequences of Lieb--Robinson bounds have been experimentally observed in
recent years, see \cite{expLR}. For an historical overview on
Lieb--Robinson bounds, see \cite{VL} and references therein.

For the reader's convenience and completeness, we start by deriving such
bounds for fermions on the lattice with (possibly non--autonomous)
interactions. As explained in \cite{NS} in the context of quantum spin
systems, Lieb--Robinson bounds are only expected to hold true for systems
with short--range interactions. We thus define Banach spaces $\mathcal{W}$
of short--range interactions and prove Lieb--Robinson bounds for the
corresponding fermion systems. The spaces $\mathcal{W}$ include
density--density interactions resulting from the second quantization of
two--body interactions defined via a real--valued and integrable interaction
kernel $v\left( r\right) :[0,\infty )\rightarrow \mathbb{R}$. Considering
fermions with spin $1/2$, our setting includes, for instance, the celebrated
Hubbard model (and any other system with finite--range interactions) or
models with Yukawa--type potentials. Two--body interactions decaying
polynomially fast in space with sufficiently large degree are also allowed,
but the Coulomb potential is excluded because it is not summable at large
distances. The method of proof we use to get Lieb--Robinson bounds for
non--autonomous $C^{\ast }$--dynamical systems related to lattice fermions
is, up to simple adaptations, the one used in \cite{NS} for (autonomous)
quantum spin systems. Compare Theorem \ref{Theorem Lieb-Robinson}, Lemma \ref%
{Theorem Lieb-Robinson copy(2)}, Theorem \ref{Theorem Lieb-Robinsonnew} and
Corollary \ref{Theorem Lieb-Robinson copy(4)} with \cite[Theorems 2.3.\ and
3.1.]{NS}. See also \cite{BMNS} where (usual) Lieb--Robinson bounds for
non--autonomous quantum spin systems have already been derived \cite[%
Theorems 4.6]{BMNS}.

Once the Lieb--Robinson bounds for commutators are established, we combine
them with results of the theory of strongly continuous semigroups to derive
properties of the infinite--volume dynamics. These allow us to extend
Lieb--Robinson bounds to time--dependent \emph{multi}--com%
\-%
mutators, see Theorems \ref{Theorem Lieb-Robinson copy(1)}--\ref{theorem exp
tree decay copy(1)} and \ref{Theorem Lieb-Robinson copy(1)-bis}. The new
bounds on multi--commutators make possible rigorous studies of dynamical
properties that are relevant for response theory of interacting fermion
systems. For instance, they yield tree--decay bounds in the sense of \cite[%
Section 4]{OhmI} if interactions decay sufficiently fast in space (typically
some polynomial decay with large enough degree is needed). In fact, by using
the Lieb--Robinson bounds for multi--commutators, we extend in \cite%
{OhmV,OhmVI} our results \cite{OhmI,OhmII,OhmIII,OhmIV} on free fermions to
interacting particles with short--range interactions. This is an important
application of such new bounds: The rigorous microscopic derivation of Ohm
and Joule's laws for \emph{interacting} fermions, in the AC--regime. See
Section \ref{Section Ohm law} and \cite{brupedrahistoire} for a historical
perspective on this subject.

Via Theorems \ref{Thm para regularite} and \ref{thm moment mu}, we show, for
example, how Lieb--Robinson bounds for multi--com%
\-%
mutators can be applied to derive decay properties of the so--called \emph{%
AC--con%
\-%
ductivity measure} at high frequencies. This result is new and is obtained
in Section \ref{Section Ohm law}. Cf. \cite{OhmV,OhmVI}. Lieb--Robinson
bounds for multi--commutators have, moreover, further applications which go
beyond the use on linear response theory presented in Section \ref{Section
Ohm law}. For instance, as explained in Sections \ref{section Energy
Increments as Power Series} and \ref{Section existence dynamics}, they also
make possible the study of \emph{non--linear} corrections to linear
responses to external perturbations.

The new bounds can also be applied to \emph{non--auto%
\-%
nomous systems}. Indeed, the existence of a fundamental solution for the
non--auto%
\-%
nomous initial value problem related to infinite systems of fermions with
time--dependent interactions is usually a non--trivial problem because the
corresponding generators are time--dependent unbounded operators. The
time--dependency cannot, in general, be isolated into a bounded perturbation
around some unbounded time--constant generator and usual perturbation theory
cannot be applied. In many important cases, the time--dependent part of the
generator is not even relatively bounded with respect to (w.r.t.) the
constant part. In fact, no unified theory of non--auto%
\-%
nomous evolution equations that gives a complete characterization of the
existence of fundamental solutions in terms of properties of generators,
analogously to the Hille--Yosida generation theorems for the autonomous
case, is available. See, e.g., \cite{Katobis,Caps,Schnaubelt1,Pazy,Bru-Bach}
and references therein. Note that the existence of a fundamental solution
implies the well--posedness of the initial value problem related to states
or observables of interacting lattice fermions, provided the corresponding
evolution equation has a unique solution for any initial condition.

The Lieb--Robinson bounds on multi--com%
\-%
mutators we derive here yield the existence of fundamental solutions as well
as other general results on non--auto%
\-%
nomous initial value problems related to fermion systems on lattices with
interactions which are non--vanishing in the whole space and
time--dependent. This is done in a rather constructive way, by considering
the large volume limit of finite--volume dynamics, without using standard
sufficient conditions for existence of fundamental solutions of
non--autonomous linear evolution equations. If interactions decay
exponentially fast in space, then we moreover show, also by using
Lieb--Robinson bounds on multi--com%
\-%
mutators, that the \emph{non--autonomous} dynamics is smooth w.r.t. its
generator on the dense set of local observables. See Theorem \ref{thm non
auto copy(1)}. Note that the generator of the (non--autonomous) dynamics
generally has, in our case, a time--dependent domain, and the existence of a
dense set of smooth vectors is a priori not at all clear.

Observe that the evolution equations for lattice fermions are not of
parabolic type, in the precise sense formulated in \cite{parabolic def},
because the corresponding generators do not generate analytic semigroups.
They seem to be rather related to Kato's hyperbolic case \cite%
{Kato,Kato1973,Katobis}. Indeed, by structural reasons -- more precisely,
the fact that the generators are derivations on a $C^{\ast }$--algebra --
the time--dependent generator defines a stable family of operators in the
sense of Kato. Moreover, this family always possesses a common core. In some
specific situations one can directly show that the completion of this core
w.r.t. a conveniently chosen norm defines a so--called admissible Banach
space $\mathcal{Y}$ of the generator at any time, which satisfies further
technical conditions leading to Kato's hyperbolic conditions \cite%
{Kato,Kato1973,Katobis}. See also \cite[Sect. 5.3.]{Pazy} and \cite[Sect.
VII.1]{Bru-Bach}. Nevertheless, the existence of such a Banach space $%
\mathcal{Y}$ is a priori unclear in the general case treated here (Theorem %
\ref{thm non auto}).

Our central results are Theorems \ref{Theorem Lieb-Robinson copy(1)}--\ref%
{theorem exp tree decay copy(1)} and \ref{Theorem Lieb-Robinson copy(1)-bis}%
. Other important assertions are Corollary \ref{theorem exp tree decay} and
Theorems \ref{thm non auto}--\ref{thm non auto copy(1)}, \ref{Thm Heat
production as power series copy(3)}--\ref{Thm Heat production as power
series copy(4)}, \ref{Thm para regularite}, \ref{thm moment mu}. The
manuscript is organized as follows:

\begin{itemize}
\item In order to make our results accessible to a wide audience, in
particular to students in Mathematics with little Physics background,
Section \ref{Algebraic Quantum Mechanics} presents basics of Quantum
Mechanics, keeping in mind its algebraic formulation.

\item Section \ref{sect 2.1 copy(1)} introduces the algebraic setting for
fermions, in particular the CAR $C^{\ast }$--algebra. Other standard objects
(like fermions, bosons, Fock space, CAR, etc.) of quantum theory are also
presented, for pedagogical reasons.

\item Section \ref{Generalized Lieb--Robinson Bounds} is devoted to
Lieb--Robinson bounds, which are generalized to multi--commutators. We also
give a proof of the existence of the infinite--volume dynamics as well as
some applications of such bounds. The tree--decay bounds on time--dependent
multi--commutators (Corollary \ref{theorem exp tree decay}) are proven here.
However, only the autonomous dynamics is considered in this section.

\item Section \ref{section LR non-auto} extends results of Section \ref%
{Generalized Lieb--Robinson Bounds} to the non--auto%
\-%
nomous case. We prove, in particular, the existence of a fundamental
solution for the non--auto%
\-%
nomous initial value problems related to infinite interacting systems of
fermions on lattices with time--dependent interactions (Theorem \ref{thm non
auto}). This implies well--posedness of the corresponding initial value
problems for states and observables, provided their solutions are unique for
any initial condition. Applications in (possibly non--linear) response
theory (Theorems \ref{Thm Heat production as power series copy(3)}--\ref{Thm
Heat production as power series copy(4)}) are discussed as well.

\item Finally, Section \ref{Section Ohm law} explains how Lieb--Robinson
bounds for multi--commu%
\-%
tators can be applied to study (quantum) charged transport properties within
the AC--regime. This analysis yields, in particular, the asymptotics at high
frequencies of the so--called AC--con%
\-%
ductivity measure. See Theorems \ref{Thm para regularite} and \ref{thm
moment mu}.
\end{itemize}

\begin{notation}
\label{remark constant}\mbox{
}\newline
\emph{(i)} We denote by $D$ any positive and finite generic constant. These
constants do not need to be the same from one statement to another. \newline
\emph{(ii)} A norm on the generic vector space $\mathcal{X}$ is denoted by $%
\Vert \cdot \Vert _{\mathcal{X}}$ and the identity map of $\mathcal{X}$ by $%
\mathbf{1}_{\mathcal{X}}$.
\index{Space of bounded operators}The $C^{\ast }$--algebra of all bounded
linear operators on $(\mathcal{X},\Vert \cdot \Vert _{\mathcal{X}}\mathcal{)}
$ is denoted by $\mathcal{B}(\mathcal{X})$. The scalar product on a Hilbert
space $\mathcal{X}$ is denoted by $\langle \cdot ,\cdot \rangle _{\mathcal{X}%
}$. \newline
\emph{(iii)} If $O$ is an operator, $\Vert \cdot \Vert _{O}$ stands for the
graph norm on its domain. \newline
\emph{(iv)} By a slight abuse of notation, we denote in the sequel elements $%
X_{i}\in Y$ depending on the index $i\in I$ by expressions of the form $%
\{X_{i}\}_{i\in I}\subset Y$ (instead of $(X_{i})_{i\in I}\subset I\times Y$%
).
\end{notation}

\section{Algebraic Quantum Mechanics\label{Algebraic Quantum Mechanics}}

\subsection{Emergence of Quantum Mechanics}

The main principles of physics were considered as well--founded by the end
of the nineteenth century, even with, for instance, no satisfactory
explanation of the phenomenon of thermal radiation, first discovered in 1860
by G. Kirchhoff. In contrast to classical physics, which deals with
continuous quantities, Planck's intuition was to introduce an intrinsic
discontinuity of energy and a unsual\footnote{%
in regards to Boltzmann's studies, which meanwhile have strongly influenced
Planck's work. In modern terms M.K.E.L. Planck used the celebrated
Bose--Einstein statistics.} statistics (without any conceptual foundation,
in a \textit{ad hoc} way) to explain thermal radiation in 1900. Assuming the
existence of a quantum of action $h$, the celebrated Planck's constant, and
this pivotal statistics he derived the well--known Planck's law of thermal
radiation. Inspired by Planck's ideas, Einstein presented his famous
discrete (corpuscular) theory of light to explain the photoelectric effect.

Emission spectra of chemical elements had also been known since the
nineteenth century and no theoretical explanation was available at that
time. It became clear that electrons play a key role in this phenomenon.
However, the classical solar system model of the atom failed to explain the
emitted or absorbed radiation. Following again Planck's ideas, N. Bohr
proposed in 1913 an atomic model based on discrete energies that
characterize electron orbits. It became clear that the main principles of
classical physics are unable to describe atomic physics.

Planck's quantum of action, Einstein's quanta of light (photons), and Bohr's
atomic model could not be a simple extension of classical physics, which, in
turn, could also not be questioned in its field of validity. N. Bohr tried
during almost a decade to conciliate the paradoxical--looking microscopic
phenomena by defining a radically different kind of logic. Bohr's concept of
complementarity gave in 1928 a conceptual solution to that problem and
revolutionized the usual vision of nature. See, e.g., \cite{chevalley}. For
more details on the emergence of quantum mechanics, see also \cite%
{shrodinger}. Classical logic should be replaced by quantum logic as claimed
\cite{BvonNeu} by G. Birkhoff and J. von Neumann in 1936. See also \cite%
{Flori}.

On the level of theoretical physics, until 1925, quantum corrections were
systematically included, in a rather \emph{ad hoc} manner, into classical
theories to allow explicit discontinuous properties. Then, two apparently
complementary directions were taken by W.K. Heisenberg and E. Shr\"{o}%
dinger, respectively, to establish basic principles of the new quantum
physics, in contrast with the \textquotedblleft old quantum
theory\textquotedblright\ starting in 1900. Indeed, even with the so--called
correspondence principle of N. Bohr, \textquotedblleft many problems, even
quite central ones like the spectrum of helium atom, proved inaccessible to
any solution, no matter how elaborate the conversion\textquotedblright , see
\cite[p. 18]{shrodinger}.

These parallel theories elaborated almost at the same time were in
competition to be the new quantum theory until their equivalence became
clear, thanks to J. von Neumann who strongly contributed to the mathematical
foundations of Quantum Mechanics in the years following 1926. They are
nowadays known in any textbook on Quantum Mechanics as the Schr\"{o}dinger
and Heisenberg pictures of Quantum Mechanics. Schr\"{o}dinger's view point
is generally the most known and refers to the approach we first explain.

\subsection{Schr\"{o}dinger Picture of Quantum Mechanics (S1)\label{Section
Schro}%
\index{Schr\"{o}dinger picture!S1}}

Following de Broglie's studies on (Rutherford--) Bohr's model and Einstein's
theory of gases, E. Schr\"{o}dinger took into account the wave theory of
matter in 1925. Indeed, by learning from wave optics in Classical Physics as
well as from de Broglie's hypothesis on the wave property of matter, he
derived the celebrated
\index{Schr\"{o}dinger equation}\emph{Schr\"{o}dinger equation}, which
describes the time evolution of the wave behavior of all quantum objects. In
mathematical words, this time--dependent behavior is described by some
family $\{\psi \left( t\right) \}_{t\in \mathbb{R}}$\ of wave functions%
\index{Wave functions} within some Hilbert space $\mathcal{H}$, which
depends on the quantum system under consideration. This evolution is fixed
by a (possibly unbounded) self--adjoint operator%
\index{Hamiltonian!S1-H2} $\mathrm{H}=\mathrm{H}^{\ast }$ acting on $%
\mathcal{H}$: Indeed, for any initial wave function $\psi \left( 0\right)
\in \mathcal{H}$ at $t=0$, the wave function at arbitrary time $t\in \mathbb{%
R}$ is uniquely determined by the Schr\"{o}dinger equation%
\index{Schr\"{o}dinger equation}%
\begin{equation}
i\partial _{t}\psi \left( t\right) =\mathrm{H}\psi \left( t\right) \ ,%
\text{\qquad }t\in \mathbb{R}\ .  \label{Schrodinger equation}
\end{equation}%
This implies in particular that the time evolution is \emph{unitary}:
\begin{equation}
\psi \left( t\right) =\mathrm{e}^{-it\mathrm{H}}\psi \left( 0\right) \ ,%
\text{\qquad }t\in \mathbb{R}\ .  \label{Schrodinger equation2}
\end{equation}%
A typical example is given by $\mathcal{H}=L^{2}(\mathbb{R}^{3})$ with $\psi
\left( 0\right) $ being taken to be a normalized vector of $\mathcal{H}$. In
this case, $|\psi \left( t,x\right) |^{2}$ is interpreted as the probability
density to detect the quantum particle at time $t\in \mathbb{R}$ and space
position $x\in \mathbb{R}^{3}$.

\subsection{Heisenberg Picture of Quantum Mechanics (H2)\label{Section Hei}%
\index{Heisenberg picture!H2}}

Quantities like position, momentum, energy, etc., are represented by \emph{%
self--adjoint} operators acting on $\mathcal{H}$ and are called
\index{Observable!S1-H2}\emph{observables}. They refer to all properties of
the physical system that can be measured. An important one is of course the
energy observable, also named \emph{Hamiltonian}, in reference to the
celebrated Hamiltonian mechanics.
\index{Hamiltonian!S1-H2}It is, by definition, the self--adjoint operator $%
\mathrm{H}$\ in the Schr\"{o}dinger equation (\ref{Schrodinger equation}).

In this context, the outcomes of measurements of the physical quantity
associated with an arbitrary observable $B$ have a random character, the
statistical distribution of which is completely described by the family $%
\{\psi \left( t\right) \}_{t\in \mathbb{R}}$ of wave functions solving (\ref%
{Schrodinger equation}). At time $t\in \mathbb{R}$, its expectation value is
given by the real number%
\index{Observable!S1-H2!expectation value}
\begin{equation}
\langle \psi \left( t\right) ,B\psi \left( t\right) \rangle _{\mathcal{H}%
}=\langle \psi \left( 0\right) ,\mathrm{e}^{it\mathrm{H}}B\mathrm{e}^{-it%
\mathrm{H}}\psi \left( 0\right) \rangle _{\mathcal{H}}\ .
\label{expectation value}
\end{equation}%
See (\ref{Schrodinger equation2}). Here, $\langle \cdot ,\cdot \rangle _{%
\mathcal{H}}$ denotes the scalar product of $\mathcal{H}$. Viewing the state
as time--dependent and the observable fixed, like in
\index{Schr\"{o}dinger picture!S1}Schr\"{o}dinger's picture of Quantum
Mechanics, is equivalent to viewing the state as being fixed and the
observable evolving as follows:
\begin{equation}
B\mapsto \tau _{t}\left( B\right) \doteq \mathrm{e}^{it\mathrm{H}}B\mathrm{e}%
^{-it\mathrm{H}}\ ,%
\text{\qquad }t\in \mathbb{R}\ .  \label{automorphism}
\end{equation}

The latter refers to Heisenberg's view point: For every bounded Hamiltonians
$\mathrm{H}\in \mathcal{B}(\mathcal{H})$, the map (\ref{automorphism})
defines a one--parameter continuous group $\{\tau _{t}\}_{t\in \mathbb{R}}$
acting on $\mathcal{B}(\mathcal{H})$, the Banach space $\mathcal{B}(\mathcal{%
H})$ of all bounded linear operators on $\mathcal{H}$, and satisfying the
(autonomous) evolution equation%
\index{Evolution equation!autonomous}
\begin{equation}
\forall t\in \mathbb{R}:\quad \partial _{t}\tau _{t}=\tau _{t}\circ \delta
=\delta \circ \tau _{t}\ ,%
\text{\qquad }\tau _{0}=\mathbf{1}_{\mathcal{B}(\mathcal{H})}\ ,
\label{automorphism2}
\end{equation}%
where $\delta $ is the generator defined by
\begin{equation}
\delta \left( B\right) \doteq i\left[ \mathrm{H},B\right] \doteq \mathrm{H}%
B-B\mathrm{H}\ ,\text{\qquad }B\in \mathcal{B}(\mathcal{H})\ .
\label{automorphism3}
\end{equation}%
Note that $\{\tau _{t}\}_{t\in \mathbb{R}}$ is a family of isomorphims of
the Banach space $\mathcal{B}(\mathcal{H})$ and, for all $B_{1},B_{2}\in
\mathcal{B}(\mathcal{H})$,
\begin{equation}
\delta (B_{1}^{\ast })=\delta (B_{1})^{\ast }\quad \text{and}\quad \delta
(B_{1}B_{2})=\delta (B_{1})B_{2}+B_{1}\delta (B_{2})\ .
\label{symmetric derivation}
\end{equation}%
A linear operator $\delta $ acting on any algebra with involution (like $%
\mathcal{B}(\mathcal{H})$, see Section \ref{sect Algebraic Formulation of
Quantum Mechanics}) that satisfies such properties is called
\index{Symmetric derivation}\emph{symmetric derivation}. (The symmetry
property refers to the first equality.) Indeed, generators of groups of
automorphims of $C^{\ast }$--algebras (Section \ref{sect Algebraic
Formulation of Quantum Mechanics}) are necessarily symmetric derivations.

In this approach the wave function is then fixed for all times. This view
point took its origin in Heisenberg's study of the dispersion relation done
in 1925. Schr\"{o}dinger's \emph{wave} mechanics dovetailed with
Heisenberg's \emph{matrix} mechanics.

\begin{bemerkung}[Unbounded Hamiltonians]
\mbox{
}\newline
If $\mathrm{H}=\mathrm{H}^{\ast }$ is unbounded%
\index{Hamiltonian!S1-H2} then it is not clear whether (\ref{automorphism})
defines a $C_{0}$--group (that is, a strongly continuous group) $\{\tau
_{t}\}_{t\in \mathbb{R}}$ of automorphisms of $\mathcal{B}(\mathcal{H})$ or
not. This fact is, however, not important here. Indeed, one starts
\index{Schr\"{o}dinger picture!S1}(S1) either with Schr\"{o}dinger's
equation and (\ref{automorphism}) is well--defined,
\index{Heisenberg picture!H1}(H1) or with a $C_{0}$--group $\{\tau
_{t}\}_{t\in \mathbb{R}}$ of automorphisms generated by a (possibly
unbounded) symmetric derivation $\delta $, see (\ref{automorphism2}) and (%
\ref{symmetric derivation}). The latter uses the semigroup theory \cite%
{BratteliRobinsonI,EngelNagel} and refers to the algebraic formulation of
Quantum Mechanics explained in\ Section \ref{sect Algebraic Formulation of
Quantum Mechanics}.
\end{bemerkung}

\subsection{Non--Autonomous Quantum Dynamics\label{sect non auto}}

If
\index{Hamiltonian!S1-H2}$\mathrm{H}_{t}=\mathrm{H}_{t}^{\ast }$ is now a
time--dependent self--adjoint operator acting on some Hilbert space $%
\mathcal{H}$ for any time $t\in \mathbb{R}$, the Schr\"{o}dinger equation%
\index{Schr\"{o}dinger equation!non--autonomous}%
\index{Wave functions}
\begin{equation*}
i\partial _{t}\psi \left( t\right) =\mathrm{H}_{t}\psi \left( t\right) \ ,%
\text{\qquad }t\in \mathbb{R}\ ,
\end{equation*}%
\emph{formally} leads to a solution
\begin{equation}
\psi \left( t\right) =\mathrm{U}_{t,0}\psi \left( 0\right) \ ,\text{\qquad }%
t\in \mathbb{R}\ ,  \label{wave functino non auto}
\end{equation}%
with $\{\mathrm{U}_{t,s}\}_{s,t\in \mathbb{R}}$ being, a priori, the
two--parameter group of unitary operators on $\mathcal{H}$ generated by the
(anti--self--adjoint) operator $-i\mathrm{H}_{t}$:%
\index{Evolution equation!non--autonomous}
\begin{equation}
\forall s,t\in {\mathbb{R}}:\quad \partial _{t}\mathrm{U}_{t,s}=-i\mathrm{H}%
_{t}\mathrm{U}_{t,s}\ ,\quad \mathrm{U}_{s,s}\doteq \mathbf{1}_{\mathcal{H}%
}\ .  \label{non auto explaination1}
\end{equation}%
This two--parameter\ family satisfies the cocycle (Chapman--Kolmogorov)
property%
\index{Cocycle property}
\begin{equation}
\forall s,r,t\in \mathbb{R}:\qquad \mathrm{U}_{t,s}=\mathrm{U}_{t,r}\mathrm{U%
}_{r,s}\ .  \label{Chapman--Kolmogorov}
\end{equation}

Equation (\ref{non auto explaination1}) is a \emph{non--autonomous}
evolution equation. The well--posedness of such non--auto%
\-%
nomous initial value problems requires some regularity properties of the
family $\{\mathrm{H}_{t}\}_{t\in \mathbb{R}}$ of self--adjoint operators.
For instance, if $\{\mathrm{H}_{t}\}_{t\in \mathbb{R}}\in C\left( \mathbb{R};%
\mathcal{B}(\mathcal{H})\right) $ is a continuous family of bounded
operators, the existence, uniqueness and even an explicit form of the
solution of (\ref{non auto explaination1}) on the space $\mathcal{B}(%
\mathcal{H})$ (that is, in the norm/uniform topology) is given by the
Dyson--Phillips series:%
\index{Dyson--Phillips series}
\begin{equation}
\mathrm{U}_{t,s}\doteq \mathbf{1}_{\mathcal{H}}+\sum\limits_{k\in {\mathbb{N}%
}}\left( -i\right) ^{k}\int_{s}^{t}\mathrm{d}s_{1}\cdots \int_{s}^{s_{k-1}}%
\mathrm{d}s_{k}\ \mathrm{H}_{s_{1}}\cdots \mathrm{H}_{s_{k}}\ ,\qquad s,t\in
\mathbb{R}\ .  \label{dyson series}
\end{equation}%
In this case, $\{\mathrm{U}_{t,s}\}_{s,t\in \mathbb{R}}$ is a
norm--continuous two--parameter group of \emph{unitary} operators. In
particular, the norm $\Vert \psi \left( t\right) \Vert _{\mathcal{H}}$\ of (%
\ref{wave functino non auto}) is constant for all times $t\in \mathbb{R}$
and the statistical interpretation of this wave function is still
meaningful. Moreover, since the map $B\mapsto B^{\ast }$ from $\mathcal{B}(%
\mathcal{H})$ to $\mathcal{B}(\mathcal{H})$ is continuous (in the
norm/uniform topology, see \cite[Theorem VI.3 (e)]{ReedSimon}),%
\index{Evolution equation!autonomous}
\begin{equation}
\forall s,t\in {\mathbb{R}}:\quad \partial _{t}\mathrm{U}_{t,s}^{\ast }=i%
\mathrm{U}_{t,s}^{\ast }\mathrm{H}_{t}\ ,\quad \mathrm{U}_{s,s}^{\ast
}\doteq \mathbf{1}_{\mathcal{H}}\ .  \label{non auto explaination2}
\end{equation}%
(This property is not that clear in the strong topology since the map $%
B\mapsto B^{\ast }$ is not continuous anymore, but it could still be proven.
See as an example \cite[Lemma 68]{Bru-Bach}.)

However, the well--posedness of non--autonomous evolution equations like (%
\ref{non auto explaination1}) is much more delicate for \emph{unbounded}
generators. It has been studied, after the first result of Kato in 1953 \cite%
{Kato1953}, for decades by many authors (Kato again \cite{Kato,Kato1973} but
also Yosida, Tanabe, Kisynski, Hackman, Kobayasi, Ishii, Goldstein,
Acquistapace, Terreni, Nickel, Schnaubelt, Caps, Tanaka, Zagrebnov,
Neidhardt, etc.), see, e.g., \cite%
{Bru-Bach,Katobis,Caps,Schnaubelt1,Pazy,Neidhardt-zagrebnov} and the
corresponding references cited therein. Yet, no unified theory of such
linear evolution equations that gives a complete characterization
analogously to the Hille--Yosida generation theorems \cite{EngelNagel} is
known.

Assuming the well--posedness of the non--autonomous evolution equation (\ref%
{non auto explaination1}), the expectation value of any observable $B$
(i.e., a self--adjoint operator acting on $\mathcal{H}$) is given, similarly
to (\ref{expectation value}), by the real number%
\index{Observable!S1-H2!expectation value}
\begin{equation}
\langle \psi \left( t\right) ,B\psi \left( t\right) \rangle _{\mathcal{H}%
}=\langle \psi \left( s\right) ,\mathrm{U}_{t,s}^{\ast }B\mathrm{U}%
_{t,s}\psi \left( s\right) \rangle _{\mathcal{H}}\ ,%
\text{\qquad }s,t\in \mathbb{R}\ .  \label{expectation value2}
\end{equation}%
By (\ref{non auto explaination1}) and (\ref{non auto explaination2}), in the
Heisenberg picture of Quantum Mechanics (H2), we observe for any family $\{%
\mathrm{H}_{t}\}_{t\in \mathbb{R}}\in C\left( \mathbb{R};\mathcal{B}(%
\mathcal{H})\right) $ that
\begin{equation*}
B\mapsto \tau _{t,s}\left( B\right) \doteq \mathrm{U}_{t,s}^{\ast }B\mathrm{U%
}_{t,s}\ ,\text{\qquad }s,t\in \mathbb{R}\ ,
\end{equation*}%
defines a two--parameter family $\{\tau _{t,s}\}_{s,t\in \mathbb{R}}$ of
automorphisms of $\mathcal{B}(\mathcal{H})$ satisfying the (reverse) cocycle
property%
\index{Cocycle property!reverse}
\begin{equation}
\forall s,r,t\in \mathbb{R}:\qquad \tau _{t,s}=\tau _{r,s}\tau _{t,r}\ ,
\label{reverse cocycle}
\end{equation}%
(cf. (\ref{Chapman--Kolmogorov})) as well as the evolution equation%
\index{Evolution equation!non--autonomous (H1)}%
\begin{equation}
\forall s,t\in {\mathbb{R}}:\quad \partial _{t}\tau _{t,s}=\tau _{t,s}\circ
\delta _{t}\ ,\qquad \tau _{s,s}=\mathbf{1}_{\mathcal{B}(\mathcal{H})}\ .
\label{non auto explaination3}
\end{equation}%
Here, $\delta _{t}$ is the time--dependent generator defined by
\begin{equation}
\delta _{t}\left( B\right) \doteq i\left[ \mathrm{H}_{t},B\right] \doteq
\mathrm{H}_{t}B-B\mathrm{H}_{t}\ ,%
\text{\qquad }B\in \mathcal{B}(\mathcal{H})\ .
\label{non auto explaination3bis}
\end{equation}%
Compare with Equations (\ref{automorphism2}) and (\ref{automorphism3}).
Equation (\ref{non auto explaination3}) is \emph{another type} of
non--autonomous evolution equation on the Banach space $\mathcal{B}(\mathcal{%
H})$, the well--posedness of which \emph{is much more easier} to prove than
the one of (\ref{non auto explaination1}) for unbounded generators.

Indeed, non--autonomous evolution equations in mathematics usually refer to
non--auto%
\-%
nomous initial value problems%
\index{Evolution equation!non--autonomous}%
\begin{equation}
\forall t\geq s:\qquad \partial _{t}\mathfrak{U}_{t,s}=-\mathfrak{G}_{t}%
\mathfrak{U}_{t,s}\ ,\quad \mathfrak{U}_{s,s}\doteq \mathbf{1}_{\mathcal{X}%
}\ ,  \label{non auto explaination4}
\end{equation}%
with generators $\mathfrak{G}_{t}$ acting on some Banach space $\mathcal{X}$
for times $t\geq s$. One important mathematical issue of (\ref{non auto
explaination1}) or (\ref{non auto explaination4}) for unbounded generators
is to find sufficient conditions to ensure that $\mathrm{H}_{t}\mathrm{U}%
_{t,s}$ or $\mathfrak{G}_{t}\mathfrak{U}_{t,s}$ are always well--defined on
some (possibly time--dependent) dense subset $\mathfrak{D}$\ of $\mathcal{H}$
or $\mathcal{X}$.

This problem does not appear in the non--auto%
\-%
nomous initial value problem (\ref{non auto explaination3}). In particular,
if the non--autonomous evolution equation%
\index{Evolution equation!non--autonomous}%
\begin{equation*}
\forall s,t\in {\mathbb{R}}:\qquad \partial _{s}\tau _{t,s}=-\delta
_{s}\circ \tau _{t,s}\ ,\qquad \tau _{t,t}=\mathbf{1}_{\mathcal{B}(\mathcal{H%
})}\ ,
\end{equation*}%
is well--posed for some possibly unbounded family $\{\delta _{t}\}_{t\in {%
\mathbb{R}}}$ of generators, then (\ref{non auto explaination3}) is also
well--posed, see for instance \cite[Lemma 93]{Bru-Bach}. The converse does
\emph{not} hold true, in general. Indeed, in contrast with (\ref{non auto
explaination1}) and (\ref{non auto explaination4}), there is no domain
conservation in (\ref{non auto explaination3}) to take care even if $%
\{\delta _{t}\}_{t\in {\mathbb{R}}}$ is a family of unbounded generators. An
example is given in Section \ref{section LR non-auto}, compare in particular
Corollary \ref{Theorem Lieb-Robinson copy(4)} (iii) with Theorem \ref{thm
non auto}.

As a consequence, for non--autonomous dynamics the Heisenberg picture of
Quantum Mechanics is mathematically more natural or technically advantageous
as compared to the Schr\"{o}dinger picture. This gives a first argument to
start the quantum formalism with the Heisenberg picture, instead of the Schr%
\"{o}dinger one as it is done in many elementary textbooks on quantum
physics. This approach refers to the so--called algebraic formulation of
Quantum Mechanics widely used in Quantum Statistical Mechanics and Quantum
Field Theory.

\subsection{Algebraic Formulation of Quantum Mechanics (H1--S2)\label{sect
Algebraic Formulation of Quantum Mechanics}}

Algebraic Quantum Mechanics is an approach, starting in the forties (cf. GNS
construction), which \emph{reverses} the view point presented in Sections %
\ref{Section Schro}--\ref{sect non auto} by taking the Heisenberg picture of
Quantum Mechanics (H1) as the more fundamental one. Therefore, instead of
starting with Hilbert spaces and the Schr\"{o}dinger equation, one uses $%
C^{\ast }$--dynamical systems, that is, a pair constituted of a $C^{\ast }$%
--algebra and a group of $\ast $--automorphisms. The first generalizes the
Banach space $\mathcal{B}(\mathcal{H})$ of all bounded linear operators
acting on some Hilbert space $\mathcal{H}$ and the second, the map (\ref%
{automorphism}). They are defined as follows:\medskip

\noindent
\underline{(i):} Let $\mathcal{X}\equiv (\mathcal{X},+,\cdot _{{\mathbb{C}}%
}) $ be a complex vector space with a product map defined on the Cartesian
product $\mathcal{X}\times \mathcal{X}$ by
\begin{equation*}
(B_{1},B_{2})\mapsto B_{1}B_{2}\ .
\end{equation*}%
$\mathcal{X}$ is an associative and distributive\emph{\ }algebra, when, for
any $B_{1},B_{2},B_{3}\in \mathcal{X}$ and all complex numbers $\alpha
_{1},\alpha _{2}\in {\mathbb{C}}$,%
\begin{eqnarray*}
(B_{1}+B_{2})B_{3} &=&B_{1}B_{3}+B_{2}B_{3}\ ,\quad
(B_{1}B_{2})B_{3}=B_{1}(B_{2}B_{3})\ , \\
B_{3}(B_{1}+B_{2}) &=&B_{3}B_{1}+B_{3}B_{2}\ ,\quad \alpha _{1}\alpha
_{2}(B_{1}B_{2})=(\alpha _{1}B_{1})(\alpha _{2}B_{2})\ .
\end{eqnarray*}%
In the sequel, an algebra%
\index{Algebra} carries, by definition, an associative and distributive
product. $\mathcal{X}$ is a commutative algebra if $B_{1}B_{2}=B_{2}B_{1}$
for any $B_{1},B_{2}\in \mathcal{X}$. $\mathbf{1}\in \mathcal{X}$ is the
unit (or identity) of $\mathcal{X}$ when $B\mathbf{1}=\mathbf{1}B=B$ for all
$B\in \mathcal{X}$. If $\mathbf{1}\in \mathcal{X}$ exists then it is unique
and $\mathcal{X}$ is named a unital algebra.%
\index{Unital algebra} \medskip

\noindent
\underline{(ii):} An involution%
\index{Involution} is a map $B\mapsto B^{\ast }$ from an algebra $\mathcal{X}
$ to $\mathcal{X}$ that, by definition, satisfies, for any $B_{1},B_{2}\in
\mathcal{X}$ and $\alpha _{1},\alpha _{2}\in {\mathbb{C}}$,
\begin{equation*}
(B_{1}^{\ast })^{\ast }=B_{1}\ ,\quad (B_{1}B_{2})^{\ast }=B_{2}^{\ast
}B_{1}^{\ast }\ ,\quad (\alpha _{1}B_{1}+\alpha _{2}B_{2})^{\ast }=%
\overline{\alpha _{1}}B_{1}^{\ast }+\overline{\alpha _{2}}B_{2}^{\ast }\ .
\end{equation*}%
An algebra $\mathcal{X}$ equipped with an involution is a $\ast $--algebra%
\index{*--algebra@$\ast $--algebra} and $B\in \mathcal{X}$ is self--adjoint
when $B=B^{\ast }$.\ In this case, by uniqueness of the unit, one checks
that a unit $\mathbf{1}$ has to be self--adjoint. \medskip

\noindent
\underline{(iii):} Let $\Vert \cdot \Vert _{\mathcal{X}}$ be a norm on a
vector space $\mathcal{X}$. Then, $\mathcal{X}\equiv (\mathcal{X},\Vert
\cdot \Vert _{\mathcal{X}})$ is a normed algebra whenever $\mathcal{X}$ is
an algebra and
\begin{equation*}
\left\Vert B_{1}B_{2}\right\Vert _{\mathcal{X}}\leq \left\Vert
B_{1}\right\Vert _{\mathcal{X}}\left\Vert B_{2}\right\Vert _{\mathcal{X}}\ ,%
\text{\qquad }B_{1},B_{2}\in \mathcal{X}\ .
\end{equation*}%
A normed algebra $\mathcal{X}$ is a Banach algebra%
\index{Banach algebra} if $\mathcal{X}$ is complete with respect to (w.r.t.)
the norm $\Vert \cdot \Vert _{\mathcal{X}}$. A Banach algebra $\mathcal{X}$
equipped with an involution such that
\begin{equation*}
\left\Vert B\right\Vert _{\mathcal{X}}=\left\Vert B^{\ast }\right\Vert _{%
\mathcal{X}}\ ,%
\text{\qquad }B\in \mathcal{X}\ ,
\end{equation*}%
is a Banach $\ast $--algebra. Then, a Banach $\ast $--algebra%
\index{C*--algebra@$C^{\ast }$--algebra} $\mathcal{X}$ is a $C^{\ast }$%
--algebra whenever
\begin{equation}
\left\Vert B^{\ast }B\right\Vert _{\mathcal{X}}=\left\Vert B\right\Vert _{%
\mathcal{X}}^{2}\ ,%
\text{\qquad }B\in \mathcal{X}\ .  \label{def c*}
\end{equation}%
If $\mathcal{X}$ is a Banach $\ast $--algebra, then there is a unique norm $%
\Vert \cdot \Vert _{\mathcal{X}}$ on $\mathcal{X}$ such that $(\mathcal{X}%
,\Vert \cdot \Vert _{\mathcal{X}})$ is a $C^{\ast }$--algebra. Note also
that in $C^{\ast }$--algebras there is a natural notion of spectrum, which
is a real subset for any self--adjoint element. \medskip

\noindent \underline{(iv):} Let $\mathcal{X}$ and $\mathcal{Y}$ be two $%
C^{\ast }$--algebras. A linear map $\mathbf{\pi }:\mathcal{X}\rightarrow
\mathcal{Y}$ is a $\ast $--homomorphism%
\index{Homomorphism} when it preserves the product and involution of the $%
C^{\ast }$--algebras, i.e., if, for all $B_{1},B_{2}\in \mathcal{X}$,%
\begin{equation*}
\mathbf{\pi }\left( B_{1}B_{2}\right) =\mathbf{\pi }\left( B_{1}\right)
\mathbf{\pi }\left( B_{2}\right) \qquad
\text{and}\qquad \mathbf{\pi }\left( B_{1}^{\ast }\right) =\mathbf{\pi }%
\left( B_{1}\right) ^{\ast }\ .
\end{equation*}%
Such maps $\mathbf{\pi }$ are automatically contractive \cite[Proposition
2.3.1]{BratteliRobinsonI} and even isometric when $\mathbf{\pi }$ is
injective \cite[Proposition 2.3.3]{BratteliRobinsonI}. Bijective $\ast $%
--homomorphisms are called $\ast $--isomorphisms.%
\index{Isomorphism} $C^{\ast }$--algebras $\mathcal{X}$ and $\mathcal{Y}$
are said to be $\ast $--isomorphic whenever there exists a $\ast $%
--isomorphism $\mathbf{\pi }:\mathcal{X}\rightarrow \mathcal{Y}$. $\ast $%
--isomorphisms from $\mathcal{X}$ to $\mathcal{X}$ are named $\ast $%
--automorphisms%
\index{Automorphism} of the $C^{\ast }$--algebra $\mathcal{X}$. \medskip

\noindent For more details on the theory of $C^{\ast }$--algebras, see,
e.g., \cite{BratteliRobinsonI,KR1,KR2}.

A well--known example of unital $C^{\ast }$--algebra is given by the
\index{Space of bounded operators}Banach space $\mathcal{B}(\mathcal{H})$ of
all bounded linear operators acting on some Hilbert space $\mathcal{H}$. The
norm on $\mathcal{B}(\mathcal{H})$ is of course the operator norm, as
before, and the involution is defined by taking the adjoint of operators.
The complex vector space of complex--valued, measurable, bounded functions
on some set equipped with the sup--norm and the point--wise product can also
be seen as a unital commutative $C^{\ast }$--algebra.

We are now in position to explain the algebraic approach of Quantum\
Mechanics, which starts as follows. \bigskip

\noindent \textbf{Heisenberg Picture of Quantum Mechanics (H1).}%
\index{Heisenberg picture!H1} A physical system is described by its physical
properties, i.e., by a non--empty set $\mathcal{O}\not=\emptyset $ of all
physical quantifies that can be measured in this system, as well as by the
relations between them. Elements $B\in \mathcal{O}$ are called
\index{Observable!H1-S2}\emph{observables} and are taken as self--adjoint
elements of a unital\footnote{%
The existence of a unit $\mathbf{1}\in \mathcal{X}$ is assumed to simplify
discussions.} $C^{\ast }$--algebra $\mathcal{X}$. Each self--adjoint element
$B$ represents some apparatus (or measuring device) and its spectrum
corresponds to all values that can come up by measuring the corresponding
physical quantity. The quantum dynamics is given by a $C_{0}$--group (that
is, a strongly continuous group) $\tau \doteq \{\tau _{t}\}_{t\in {\mathbb{R}%
}}$ of $\ast $--automorphisms generated \cite[1.2 Definition, 1.4 Theorem]%
{EngelNagel} by a symmetric derivation $\delta $ acting on the $C^{\ast }$%
--algebra $\mathcal{X}$. In particular, by \cite[1.3 Lemma (ii)]{EngelNagel}%
, it satisfies the (autonomous) evolution equation%
\index{Evolution equation!autonomous}
\begin{equation*}
\forall t\in \mathbb{R}:\quad \partial _{t}\tau _{t}=\tau _{t}\circ \delta
=\delta \circ \tau _{t}\ ,%
\text{\qquad }\tau _{0}=\mathbf{1}_{\mathcal{X}}\ ,
\end{equation*}%
with $\delta $ being a possibly unbounded operator acting on $\mathcal{X}$.
Compare with Equations (\ref{automorphism2})--(\ref{symmetric derivation}).
Recall also that symmetric derivations refer to (linear) operators
satisfying properties (\ref{symmetric derivation}) on $\mathcal{X}$. The
pair $(\mathcal{X},\tau )$ is known as a (autonomous)
\index{C*--dynamical system@$C^{\ast }$--dynamical system}$C^{\ast }$\emph{%
--dynamical system}. A similar automorphism family can be defined for
non--autonomous dynamics by using (\ref{reverse cocycle}) and (\ref{non auto
explaination3}) on the domain $\mathrm{Dom}(\delta _{t})\subseteq \mathcal{X}
$ of a time--dependent symmetric derivation $\delta _{t}$ for $t\in \mathbb{R%
}$. See for instance Corollary \ref{Theorem Lieb-Robinson copy(4)} (iii). In
this case, one speaks about non--autonomous $C^{\ast }$--dynamical systems.
\bigskip

\noindent \textbf{Schr\"{o}dinger Picture of Quantum Mechanics (S2).}%
\index{Schr\"{o}dinger picture!S2} States are not anymore defined from a
wave function within some Hilbert space, like in Section \ref{Section Schro}%
.
\index{States}States on the $C^{\ast }$--algebra $\mathcal{X}$ are, by
definition, continuous linear functionals $\rho \in \mathcal{X}^{\ast }$
which are normalized and positive, i.e., $\rho (\mathbf{1})=1$ and $\rho
(B^{\ast }B)\geq 0$ for all $B\in \mathcal{X}$. They represent the state of
the physical system. Observe for instance that Equation (\ref{expectation
value}), or (\ref{expectation value2}) in the non--autonomous situation,
defines a continuous linear functional on the $C^{\ast }$--algebra $\mathcal{%
B}(\mathcal{H})$ that is positive and normalized, provided $\Vert \psi
\left( 0\right) \Vert _{\mathcal{H}}=1$. Thus, a state $\rho $ represents
the statistical distribution of all measures of any observable%
\index{Observable!H1-S2!expectation value} $B\in \mathcal{X}$. For
commutative $C^{\ast }$--algebras, it corresponds to a probability
distribution. If $\{\tau _{t}\}_{t\in {\mathbb{R}}}$ is a $C_{0}$--group of $%
\ast $--automorphism of $\mathcal{X}$, then, for any time $t\in {\mathbb{R}}$
and state $\rho \in \mathcal{X}^{\ast }$,%
\begin{equation*}
\rho _{t}\doteq \rho \circ \tau _{t}\in \mathcal{X}^{\ast }
\end{equation*}%
is also a state. The same holds true if the dynamics would have been
non--autonomous. In the Schr\"{o}dinger picture, the dynamics is
consequently given by the family $\{\rho _{t}\}_{t\in {\mathbb{R}}}$ of
states.\bigskip

Therefore, in the algebraic formulation of Quantum Mechanics (H1--S2), there
is no a priori Hilbert space structure appearing in the mathematical
framework, in contrast with the approach S1--H2 presented in Sections \ref%
{Section Schro}--\ref{sect non auto}. In fact, by S1--H2 one fixes a unique
Hilbert space right from the beginning, whereas the use of H1--S2 can lead
to a (not necessarily unique) Hilbert space that depends on the choice of
the state.

By \cite[p. 274]{Haag-memoires}, I.E. Segal was the first who proposed to
leave the Hilbert space approach to consider quantum observables as elements
of certain involutive Banach algebras, now known as $C^{\ast }$--algebras.
The relation between the algebraic formulation and the usual Hilbert space
based formulation of Quantum Mechanics has been established via one
important result obtained in the forties: The celebrated \emph{GNS}
(Gel'fand--Naimark--Segal) representation of states%
\index{States!GNS}.

Indeed, by \cite[Lemma 2.3.10]{BratteliRobinsonI}, a positive linear
functional $\rho $ over a $\ast $--algebra $\mathcal{X}$ satisfies
\begin{equation*}
\rho (B_{1}^{\ast }B_{2})=%
\overline{\rho (B_{2}^{\ast }B_{1})}\ ,\text{\qquad }B_{1},B_{2}\in \mathcal{%
X}\ ,
\end{equation*}%
and the Cauchy--Schwarz inequality:
\begin{equation*}
|\rho (B_{1}^{\ast }B_{2})|^{2}\leq \rho (B_{1}^{\ast }B_{1})\rho
(B_{2}^{\ast }B_{2})\ ,\text{\qquad }B_{1},B_{2}\in \mathcal{X}\ .
\end{equation*}%
Therefore, if $\mathcal{X}$ is a unital $C^{\ast }$--algebra and $\rho \in
\mathcal{X}^{\ast }$ is a state then
\begin{equation}
\mathcal{L}_{\rho }\doteq \{B\in \mathcal{X}\;:\;\rho (B^{\ast }B)=0\}
\label{kernel}
\end{equation}%
is a closed left--ideal of $\mathcal{X}$, i.e., $\mathcal{L}_{\rho }$ is a
closed subspace such that $\mathcal{XL}_{\rho }\subset \mathcal{L}_{\rho }$,
and one can define a scalar product on the quotient $\mathcal{X}/\mathcal{L}%
_{\rho }$, which can be completed to get a Hilbert space $\mathcal{H}_{\rho
} $. For more details on the GNS construction, see \cite%
{BratteliRobinsonI,KR1}.

The GNS representation has led to very important applications of the
Tomita--Takesaki theory (see, e.g., \cite{BratteliRobinsonI,KR2}), developed
in seventies, to Quantum Field Theory and Statistical Mechanics. These
developments mark the beginning of the algebraic approach to Quantum
Mechanics and Quantum Field Theory. For more details, see, e.g., \cite{Emch}%
. In fact, the algebraic formulation turned out to be extremely important
and fruitful for the mathematical foundations of Quantum Statistical
Mechanics. See for instance discussions of Section \ref{section CAR}, in
particular Lemmata \ref{equivalence lemma1} and \ref{equivalence lemma2}. In
particular, it has been an important branch of research during decades with
lots of works on quantum spin and Fermi systems. See, e.g., \cite%
{BratteliRobinson,Israel} (spin) and \cite%
{Araki-Moriya,BruPedra2,BruPedra-homog} (Fermi).

\subsection{Representation Theory -- The importance of the Algebraic
Approach for Infinite Systems\label{Sec Representation theory}}

We discuss here how $C^{\ast }$--algebras can be represented by spaces of
bounded operators acting Hilbert spaces. A \emph{representation}%
\index{Representation} on the Hilbert space $\mathcal{H}$ of a $C^{\ast }$%
--algebra $\mathcal{X}$ is, by definition \cite[Definition 2.3.2]%
{BratteliRobinsonI}, a $\ast $--homomorphism $\mathbf{\pi }$ from $\mathcal{X%
}$ to the unital $C^{\ast }$--algebra $\mathcal{B}(\mathcal{H})$ of all
bounded linear operators acting on $\mathcal{H}$. In this case, $\mathcal{H}$
is named the representation (Hilbert) space and if it is finite (resp.
infinite) dimensional then we have a finite (resp. infinite) dimensional
representation of $\mathcal{X}$. Injective representations are called
faithful.

By the Gelfand--Naimark theorem \cite{Dix}, each $C^{\ast }$--algebra has,
at least, one faithful representation.%
\index{Gelfand--Naimark theorem} In particular, since faithful
representations are isometric \cite[Proposition 2.3.3]{BratteliRobinsonI},
any $C^{\ast }$--algebra can be identified with some $C^{\ast }$--subalgebra
of the $C^{\ast }$--algebra $\mathcal{B}(\mathcal{H})$\ of all bounded
linear operators acting on some Hilbert space $\mathcal{H}$. In fact, as
mentioned in Section \ref{sect Algebraic Formulation of Quantum Mechanics},
the algebraic formulation of Quantum Mechanics (H1--S2) leads to a Hilbert
space $\mathcal{H}_{\rho }$ for any state $\rho $, via its GNS
representation.%
\index{States!GNS} A faithful representation can be derived in this way if
there exists a state $\rho $ for which $\mathcal{L}_{\rho }=\{0\}$ (\ref%
{kernel}), i.e., if a faithful state exists for the algebra under
consideration.

Uniqueness of representations of $C^{\ast }$--algebras is clearly wrong.
Indeed, for any representation $\mathbf{\pi :}\mathcal{X}\rightarrow
\mathcal{B}(\mathcal{H})$, we can construct another one by doubling the
Hilbert space $\mathcal{H}$ and the map $\mathbf{\pi }$, via a direct sum $%
\mathcal{H}_{1}\oplus \mathcal{H}_{2}$ with $\mathcal{H}_{1},\mathcal{H}_{2}$
being two copies of $\mathcal{H}$. Therefore, one uses a notion of
\textquotedblleft minimal\textquotedblright\ representations of $C^{\ast }$%
--algebras: If $\mathbf{\pi }:\mathcal{X}\rightarrow \mathcal{B}(\mathcal{H}%
) $ is a representation of a $C^{\ast }$--algebra $\mathcal{X}$ on the
Hilbert space $\mathcal{H}$, we say that it is \emph{irreducible}, whenever $%
\{0\}$ and $\mathcal{H}$ are the only closed subspaces of $\mathcal{H}$
which are invariant w.r.t. to any operator of $\mathbf{\pi }(\mathcal{X}%
)\subset \mathcal{B}(\mathcal{H})$.

Now, it is well--known \cite{Na} that if a $C^{\ast }$--algebra $\mathcal{X}$
is isomorphic to the $C^{\ast }$--algebra $\mathcal{K}(\mathcal{H})\subset
\mathcal{B}(\mathcal{H})$ of all compact operators on some Hilbert space $%
\mathcal{H}$ , then, up to unitary equivalence, $\mathcal{X}$ has \emph{only
one} irreducible representation (the canonical one on $\mathcal{H}$). The
converse is true for \emph{separable} Hilbert spaces: If $\mathcal{X}$ is a $%
C^{\ast }$--algebra with a faithful representation on a separable Hilbert
space $\mathcal{H}$ and if all irreducible representations of $\mathcal{X}$
are unitarily equivalent, then $\mathcal{X}$ is isomorphic to the $C^{\ast }$%
--algebra $\mathcal{K}(\mathcal{H})$ of compact operators on some Hilbert
space $\mathcal{H}$. This result is known as the Rosenberg theorem \cite{Ros}%
. In other words, one gets the following theorem:

\begin{satz}[Uniqueness of irreducible representations -- I]
\label{Uniqueness of irreducible representations}\mbox{
}\newline
If a $C^{\ast }$--algebra $\mathcal{X}$ has a faithful representation on a
separable Hilbert space, then its irreducible representation is unique (up
to unitary equivalence) iff $\mathcal{X}$ is isomorphic to some $C^{\ast }$%
--algebra of compact operators on some Hilbert space.
\end{satz}

The question whether all $C^{\ast }$--algebras with a unique (up to unitary
equivalence) irreducible representation is isomorphic to an algebra of
compact operators on a non--separable Hilbert space is known as
\textquotedblleft Naimark's problem\textquotedblright . Indeed, this
question is highly non--trivial. It depends on the continuum hypothesis and
not only on the axioms of the Zermelo--Fraenkel set theory with the axiom of
choice (ZFC) \cite[Chapter 19]{Wea}.

In the finite dimensional situation, the $C^{\ast }$--algebra of compact
operators is of course equal to the whole $C^{\ast }$--algebra of bounded
operators. Therefore, Theorem \ref{Uniqueness of irreducible representations}
implies the following assertion:

\begin{koro}[Uniqueness of irreducible representations -- II]
\label{Uniqueness of irreducible representations --II}\mbox{
}\newline
If the $C^{\ast }$--algebra $\mathcal{X}$ is isomorphic to $\mathcal{B}(%
\mathcal{H})$ for some finite dimensional Hilbert space $\mathcal{H}$, then
its irreducible representation is unique, up to unitary equivalence. Any
isomorphism $\mathcal{X}\rightarrow \mathcal{B}(\mathcal{H})$ of $C^{\ast }$%
--algebras is such a irreducible representation.
\end{koro}

As a consequence, in the finite dimensional situation, the algebraic and
Hilbert space based approaches turns out to be equivalent to each other.
However, this is not anymore the case in the infinite dimensional situation
for unital $C^{\ast }$--algebras because the $C^{\ast }$--algebra of all
compact operators \emph{cannot} have a unit:

\begin{koro}[Non--uniqueness of irreducible representations]
\label{Uniqueness of irreducible representations --III}\mbox{
}\newline
Any unital $C^{\ast }$--algebra $\mathcal{X}$ with an infinite dimensional
faithful representation on a separable Hilbert space has more than one
unitarily non--equivalent irreducible representation.
\end{koro}

Because of Corollary \ref{Uniqueness of irreducible representations --III},
the algebraic approach is more general than the Hilbert space based
approach, in the case of infinite dimensional unital underlying $C^{\ast }$%
--algebras. In condensed matter physics the non--uniqueness of irreducible
representations is intimately related to the existence of various
thermodynamically stable phases of the same material. Because of this, no
reasonable microscopic theory of first order phase transitions is possible
within the Hilbert space based approach, and the use of the algebraic
setting is imperative.

This fact was first observed by Haag in 1962 \cite{Haag62}, who established
that the non--uniqueness of the ground state of the BCS model in infinite
volume is related to the existence of several inequivalent irreducible
representations \cite[Definition 2.3.2]{BratteliRobinsonI} of the
Hamiltonian, see also \cite{ThirWeh67,Emch}.

\section{Algebraic Setting for Interacting Fermions on the Lattice\label%
{sect 2.1 copy(1)}\label{Section main results}}

\subsection{{Single Quantum Particle on }Lattices\label{Section Qingle part}}

All quantum particles carry an intrinsic form of angular momentum, the
so--called \emph{spin}, first introduced by W. Pauli in the twenties. It is
reflected by a
\index{Spin quantum number}spin quantum number $\mathfrak{s}\in \mathbb{N}/2$
which gives rise to the finite spin set%
\index{Spin set}
\begin{equation}
\mathrm{S}\doteq \left\{ -\mathfrak{s},-\mathfrak{s}+1,\ldots \mathfrak{s}-1,%
\mathfrak{s}\right\} \subset \mathbb{N}\ .  \label{spin set}
\end{equation}%
In fact, $\mathrm{S}$ is the spectrum of the spin observable associated with
the quantum particle.

If $\mathfrak{s}\notin \mathbb{N}$ is half--integer then the corresponding
particles are named
\index{Fermion}\emph{fermions} while $\mathfrak{s}\in \mathbb{N}$ means by
definition that we have
\index{Boson}\emph{bosons}. For instance, among all elementary particles of
the standard model in Particle Physics, quarks and leptons (like
\index{Electron}electrons, $\mathfrak{s}=1/2$) are fermions while all the
other ones -- the gluon, photon, Z-- and W-- bosons as well as the Higgs
bosons -- are bosons. Many known composite particles like protons ($%
\mathfrak{s}=1/2$) are fermions. Others are bosons, like for instance Helium
4.

By the celebrated spin--statistics theorem, fermionic wave functions are
antisymmetric with respect to (w.r.t.) permutations of particles, whereas
the bosonic ones are symmetric. In the sequel, we consider the fermionic
case which is \emph{only} defined here via the antisymmetry of many--body
wave functions (Section \ref{sect Quantum Many--Body Systems}), or
equivalently by the Canonical Anti--commutation Relations (CAR) in the
algebraic formulation (Section \ref{section CAR}). Therefore, in order to
simplify notation, we omit the spin property of quantum particles because it
is \emph{completely irrelevant} in all our proofs, up to obvious
modifications. So, we consider the case $\mathfrak{s}=0$, i.e., $\mathrm{S}%
\doteq \{0\}$, even if this is not coherent with the definition explained
just above of fermions in Physics.

Additionally, the host material for the quantum particle is a cubic crystal,
i.e., a lattice%
\index{Lattice}
\begin{equation*}
\mathfrak{L}\doteq \mathbb{Z}^{d}\times \mathrm{S}\equiv \mathbb{Z}^{d}\
,\qquad d\in \mathbb{N}\ .
\end{equation*}%
This special choice is again not essential in our proofs. In fact, we could
take instead of $\mathbb{Z}^{d}$ any countable metric space $\mathfrak{L}$
as soon as it is regular, as defined in \cite[Section 3.1]{NS}. See also
Section \ref{Section Banach space interaction} for more details. (If $%
\mathfrak{s}\neq 0$ then it would suffice to equip the set $\mathfrak{L}%
\doteq \mathbb{Z}^{d}\times \mathrm{S}\neq \mathbb{Z}^{d}$ with the metric
of $\mathbb{Z}^{d}$ while omitting the spin variable.)

Therefore, the one--particle Hilbert space representing the set of all wave
functions of any quantum particle on the lattice is given by the space%
\index{One--particle Hilbert space}
\begin{equation*}
\ell ^{2}\left( \mathfrak{L}\right) \doteq \left\{ \psi :\mathfrak{L}%
\rightarrow \mathbb{C}%
\text{ such that }\sum\limits_{x\in \mathfrak{L}}\left\vert \psi \left(
x\right) \right\vert ^{2}<\infty \right\}
\end{equation*}%
of square--summable functions on the lattice $\mathfrak{L}$. Here, the
scalar product of $\ell ^{2}(\mathfrak{L})$ is defined by%
\begin{equation*}
\left\langle \psi ,\varphi \right\rangle _{\ell ^{2}\left( \mathfrak{L}%
\right) }\doteq \sum\limits_{x\in \mathfrak{L}}\overline{\psi \left(
x\right) }\varphi \left( x\right) \ ,\qquad \psi ,\varphi \in \ell
^{2}\left( \mathfrak{L}\right) \ .
\end{equation*}%
The canonical orthonormal basis of $\ell ^{2}(\mathfrak{L})$ is given by the
family $\left\{ \mathfrak{e}_{x}\right\} _{x\in \mathfrak{L}}$ defined by%
\index{One--particle Hilbert space!orthonormal basis}
\begin{equation}
\mathfrak{e}_{x}(y)\doteq \delta _{x,y}\ ,\qquad x,y\in \mathfrak{L}\ .
\label{ONB}
\end{equation}%
Here, $\delta _{k,l}$ is the Kronecker delta, that is, the function of two
variables defined by $\delta _{k,l}=1$ if $k=l$ and $\delta _{k,l}=0$
otherwise.

In real systems, the quantum particle is contained in an arbitrary large but
finite region. Therefore, we use the notation $\mathcal{P}_{f}(\mathfrak{L}%
)\subset 2^{\mathfrak{L}}$ for the set of all \emph{finite} subsets of $%
\mathfrak{L}$ and we meanwhile denote by%
\index{One--particle Hilbert space}
\begin{equation}
\ell ^{2}\left( \Lambda \right) \doteq \left\{ \psi \in \ell ^{2}\left(
\mathfrak{L}\right) :%
\text{ }\psi |_{\Lambda ^{c}}=0\right\} \subseteq \ell ^{2}\left( \mathfrak{L%
}\right)  \label{l2 space}
\end{equation}%
the Hilbert subspace of square--summable functions on any \emph{possibly
infinite} subset $\Lambda \subseteq \mathfrak{L}$ with complement $\Lambda
^{c}\doteq \mathfrak{L}\backslash \Lambda $. Clearly, the Hilbert subspace $%
\ell ^{2}\left( \Lambda \right) $ has $\left\{ \mathfrak{e}_{x}\right\}
_{x\in \Lambda }$ as canonical orthonormal basis and, for any $\Lambda \in
\mathcal{P}_{f}(\mathfrak{L})$, its dimension thus equals\ the volume $%
|\Lambda |$ of $\Lambda $.

Then, as explained in Section \ref{Section Schro}, the quantum dynamics is
defined by the Schr\"{o}dinger equation (\ref{Schrodinger equation}) for
some one--particle Hamiltonian $\mathrm{H}_{1}$ acting on $\mathcal{H}=\ell
^{2}\left( \Lambda \right) $ for any (possibly infinite) subset $\Lambda
\subseteq \mathfrak{L}$. A standard example of such self--adjoint operators
is given by%
\index{Kinetic term}%
\index{Hamiltonian!S1-H2!one--particle}
\begin{equation}
\lbrack \mathrm{H}_{1}(\psi )](x)=\sum\limits_{y\in \mathfrak{L}}h\left(
\left\vert x-y\right\vert \right) \psi (y)\ ,%
\text{\qquad }x\in \Lambda ,\ \psi \in \ell ^{2}\left( \Lambda \right) \ ,
\label{one particule Hamil}
\end{equation}%
for any function $h:[0,\infty )\rightarrow \mathbb{R}$, the absolute value
of which decreases sufficiently fast at infinity. This example includes $d$%
--dimensional discrete Laplacians, see Section \ref{Section Inter dis media}.

\subsection{Quantum Many--Body Systems on Lattices\label{sect Quantum
Many--Body Systems}}

Assume that quantum particles are within some (possibly infinite) subset $%
\Lambda \subseteq \mathfrak{L}$. A priori, the Hilbert space representing
the set of all wave functions of $n\in \mathbb{N}$ identical particles is
given by the\ Hilbert space%
\index{Many--body Hilbert space}
\begin{equation*}
\ell ^{2}(\Lambda )^{\otimes n}\doteq \ell ^{2}(\Lambda )\otimes \cdots
\otimes \ell ^{2}(\Lambda )\ ,
\end{equation*}%
the $n$--fold tensor product of $\ell ^{2}(\Lambda )$ with scalar product
defined by
\begin{equation*}
\left\langle \psi _{1}\otimes \cdots \otimes \psi _{n},\varphi _{1}\otimes
\cdots \otimes \varphi _{n}\right\rangle _{\ell ^{2}(\Lambda )^{\otimes
n}}\doteq \left\langle \psi _{1},\varphi _{1}\right\rangle _{\ell ^{2}\left(
\Lambda \right) }\cdots \left\langle \psi _{n},\varphi _{n}\right\rangle
_{\ell ^{2}\left( \Lambda \right) }\ ,
\end{equation*}%
for any $\psi _{1},\ldots ,\psi _{n},\varphi _{1},\ldots ,\varphi _{1}\in
\ell ^{2}\left( \Lambda \right) $. A canonical orthonormal basis of $\ell
^{2}(\Lambda )^{\otimes n}$ is given by the family%
\index{Many--body Hilbert space!orthonormal basis}
\begin{equation}
\left\{ \mathfrak{e}_{x_{1}}\otimes \cdots \otimes \mathfrak{e}%
_{x_{n}}\right\} _{x_{1},\ldots ,x_{n}\in \Lambda }\ ,  \label{ONBbis}
\end{equation}%
where we recall that $\left\{ \mathfrak{e}_{x}\right\} _{x\in \Lambda }$ is
the (canonical) orthonormal basis of $\ell ^{2}(\Lambda )$ defined by (\ref%
{ONB}). Because of (\ref{l2 space}), note that
\begin{equation*}
\ell ^{2}(\Lambda )^{\otimes n}\subseteq \ell ^{2}(\mathfrak{L})^{\otimes n}
\end{equation*}%
for any $\Lambda \subseteq \mathfrak{L}$ and $n\in \mathbb{N}$.

In Quantum Mechanics, however, quantum particles are \emph{indistinguishable}
(or indiscernible), i.e., we cannot distinguish them, even in principle.
Indistinguishability is a concept already used in Classical Mechanics, for
instance in Botlzmann's `Combinatorial Approach' to derive the so--called
Maxwell--%
\-%
Boltz%
\-%
mann statistics. Two classical objects are indeed indistinguishable when
they share the same properties, up to their spatio--temporal location. In
particular, by some form of impenetrability assumption, their
indistinguishability does not prevent them from being two \emph{different}
individuals and so, a spatio--temporal permutation of the two objects yields
another physical state.

This property is no longer true in Quantum Mechanics. Quoting E. Schr\"{o}%
dinger \cite{Shrodinger}: \textit{\textquotedblleft You cannot mark an
electron, you cannot paint it red. Indeed, you must not even think of it as
marked.\textquotedblright\ }This has an important mathematical consequence
on the modelling of composite quantum objects, the individuality of which
becomes philosophically questionable. This was implicitly used by M.K.E.L.
Planck in his famous study of thermal radiation law, but rather in \textit{%
ad hoc} way\footnote{%
He may have discovered it by working backwards from the thermal radiation
law, see \cite[p. 86]{FK}.}, without conceptual foundations. For more
details on that issue, including references, we strongly recommend \cite{FK}.

In fact, the expectation value (\ref{expectation value}) of any observable
must not depend on the arbitrary numbering of particles. In other words, the
wave function $\psi ^{(n)}\in \ell ^{2}(\Lambda )^{\otimes n}$ have to
satisfy the equality%
\index{Observable!S1-H2!expectation value}%
\begin{equation}
\langle \psi ^{(n)},B\psi ^{(n)}\rangle _{\ell ^{2}(\Lambda )^{\otimes
n}}=\langle S_{\pi }\psi ^{(n)},BS_{\pi }\psi ^{(n)}\rangle _{\ell
^{2}(\Lambda )^{\otimes n}}  \label{property}
\end{equation}%
for all $B=B^{\ast }\in \mathcal{B}(\ell ^{2}(\Lambda )^{\otimes n})$, where
$S_{\pi }\in \mathcal{B}(\ell ^{2}(\mathfrak{L})^{\otimes n})$ is the
unitary operator defined for any permutation $\pi $ of $n\in \mathbb{N}$
elements by the conditions%
\begin{equation}
S_{\pi }\left( \psi _{1}\otimes \cdots \otimes \psi _{n}\right) =\psi _{\pi
\left( 1\right) }\otimes \cdots \otimes \psi _{\pi \left( n\right) }\
,\qquad \psi _{1},\ldots ,\psi _{n}\in \ell ^{2}\left( \mathfrak{L}\right) \
.  \label{Spi}
\end{equation}%
This yields two drastically different situations:

\begin{itemize}
\item[(b)] For any permutation $\pi $ of $n\in \mathbb{N}$ elements, $S_{\pi
}\psi ^{(n)}=\psi ^{(n)}$, i.e., $\psi ^{(n)}$ is a \emph{symmetric} $n$%
--particle wave function. It corresponds to the
\index{Boson}\emph{boson} case.

\item[(f)] For any permutation $\pi $ of $n\in \mathbb{N}$ elements with
sign $(-1)^{\pi }$, $S_{\pi }\psi ^{(n)}=(-1)^{\pi }\psi ^{(n)}$, i.e., $%
\psi ^{(n)}$ is an \emph{antisymmetric} $n$--particle wave function. Quantum
particles are
\index{Fermion}\emph{fermions}.
\end{itemize}

\noindent Indeed, in contrast with particles with integer spins (boson
case), physical particles with half--integer spins (fermion case), obey the
\index{Pauli exclusion principle}Pauli exclusion principle, which says that
two identical fermions cannot occupy the same quantum state simultaneously.
The latter is reflected in the antisymmetry property of many--fermion wave
functions.

Therefore, we mathematically distinguish fermions and bosons only with
symmetry properties of wave functions w.r.t. to permutations. In fact, as
already mentioned in Section \ref{Section Qingle part}, the spin dependence
is, from the technical point of view of proofs, irrelevant here (up to
obvious modifications) and without loss of generality (w.l.o.g.) we consider
fermions without taking into account its spin in our notation.

\begin{bemerkung}[Anyons]
\mbox{
}\newline
\index{Anyon}By implementing the permutation symmetry property in the
configuration space before the \textquotedblleft
quantization\textquotedblright , in the two dimensional space $\mathbb{R}%
^{2} $, one has a continuum of (fractional) statistics ranging from the
fermionic to the bosonic cases. This refers to the existence of anyons \cite%
{G,LM,Wi}, which has been observed in the context of the fractional quantum
Hall effect. Anyons (like bosons as well) do not play any role in the sequel.
\end{bemerkung}

\begin{bemerkung}[Parastatistics]
\mbox{
}\newline
\index{Parastatistics}If $\left[ B,S_{\pi }\right] =0$ then Equation (\ref%
{property}) trivially holds true for all states $\psi ^{(n)}\in \ell
^{2}(\Lambda )^{\otimes n}$. One could thus assume that, for any permutation
$\pi $ of $n\in \mathbb{N}$ elements, the states $\psi ^{(n)}$ and $S_{\pi
}\psi ^{(n)}$ cannot be distinguished by any experiment. This view point
restricts the set of possible observables to those commuting with all
permutation operators $S_{\pi }$. Different statistics, again ranging from
the fermionic to the bosonic cases, can then be found from a mathematical
perspective. This refers to the so--called parastatistics (which is
invariant under the quantum dynamics). Philosophically, this view point has
the advantage to restore the individuality of quantum particles, in the
classical sense. A historical overview on this approach is given in \cite[%
Section 3.8]{FK}.
\end{bemerkung}

Therefore, for any fixed $n\in \mathbb{N}$, we define the orthogonal
projection $P_{n}\in \mathcal{B}(\ell ^{2}(\mathfrak{L})^{\otimes n})$ onto
the subspace of antisymmetric $n$--particle wave functions in $\ell ^{2}(%
\mathfrak{L})^{\otimes n}$ by%
\begin{equation}
P_{n}\doteq
\frac{1}{n!}\sum_{\pi \in \mathcal{P}_{n}}(-1)^{\pi }S_{\pi }
\label{definition projection operator}
\end{equation}%
with $S_{\pi }$ being the operator defined via (\ref{Spi}) and where%
\begin{equation}
\mathcal{P}_{n}\doteq \left\{ \pi :\left\{ 1,\ldots ,n\right\} \rightarrow
\left\{ 1,\ldots ,n\right\} \text{ bijective}\right\}  \label{permutations}
\end{equation}%
denotes the set of all permutations $\pi $ of $n\in \mathbb{N}$\ elements.
Then, for $n\in \mathbb{N}$, the Hilbert space representing the set of all $%
n $--fermion wave functions is given by the\ Hilbert subspace%
\index{Many--fermion Hilbert space}
\begin{equation*}
P_{n}\ell ^{2}(\Lambda )^{\otimes n}\subseteq \ell ^{2}(\Lambda )^{\otimes n}
\end{equation*}%
for any (possibly infinite) subset $\Lambda \subseteq \mathfrak{L}$.

As explained in Section \ref{Section Schro}, the quantum dynamics is defined
by the Schr\"{o}%
\-%
dinger equation (\ref{Schrodinger equation}) for some Hamiltonian $\mathrm{H}%
^{\otimes n}$ acting on $\mathcal{H}=P_{n}\ell ^{2}(\Lambda )^{\otimes n}$
at fixed $n\in \mathbb{N}$ and $\Lambda \subseteq \mathfrak{L}$. A standard
example of such self--adjoint operators is given by%
\index{Hamiltonian!S1-H2!many-fermion}
\begin{equation}
\lbrack \mathrm{H}^{\otimes n}(\psi )](x)=\mathrm{H}_{1}\otimes \mathbf{1}%
_{\ell ^{2}(\Lambda )}\cdots \otimes \mathbf{1}_{\ell ^{2}(\Lambda )}+\cdots
+\mathbf{1}_{\ell ^{2}(\Lambda )}\otimes \cdots \otimes \mathbf{1}_{\ell
^{2}(\Lambda )}\otimes \mathrm{H}_{1}+\mathrm{I}^{\otimes n}\ ,
\label{family Hn}
\end{equation}%
with the one--particle Hamiltonian $\mathrm{H}_{1}$ defined by (\ref{one
particule Hamil}) while $\mathrm{I}^{\otimes n}$ is defined by the conditions%
\index{Interparticle interactions}%
\index{Density--density interaction}
\begin{equation*}
\mathrm{I}^{\otimes n}P_{n}\left( \mathfrak{e}_{x_{1}}\otimes \cdots \otimes
\mathfrak{e}_{x_{n}}\right) =\sum\limits_{1\leq j<k\leq n}v\left( \left\vert
x_{j}-x_{k}\right\vert \right) P_{n}\left( \mathfrak{e}_{x_{1}}\otimes
\cdots \otimes \mathfrak{e}_{x_{n}}\right)
\end{equation*}%
for any $x_{1},\ldots ,x_{n}\in \Lambda \subseteq \mathfrak{L}$. $\mathrm{I}%
^{\otimes n}$ represents some interparticle forces which are characterized
by a function $v:[0,\infty )\rightarrow \mathbb{R}$, the absolute value of
which decreases sufficiently fast at infinity.

\subsection{Fermion Fock Spaces%
\index{Fermion Fock Space}}

In Quantum Statistical Mechanics we are interested in understanding the
physical behavior of macroscopic systems from the laws of Quantum Mechanics.
This means here that one studies physical properties in the limit $%
n\rightarrow \infty $ of infinite particles. However, the quantum dynamics
and even the mathematical framework, that is, the Hilbert space $P_{n}\ell
^{2}(\Lambda )^{\otimes n}$ of antisymmetric $n$--particle wave functions,
strongly depend on the particle number $n$, which may additionally be
unknown. Moreover, one is often interested in time--dependent particle
numbers, as in Quantum Field Theory.

To this end, in 1932 V.A. Fock introduced a space now known as the Fock
space defined for Fermi systems by (a priori infinite) direct sums:%
\begin{equation*}
\mathcal{F}_{\Lambda }\doteq \bigoplus_{n\in \mathbb{N}_{0}}P_{n}\ell
^{2}(\Lambda )^{\otimes n}\ ,%
\text{\qquad with\qquad }\ell ^{2}(\Lambda )^{\otimes 0}\doteq \mathbb{C}%
\text{\quad and\quad }P_{0}\doteq \mathbf{1}_{\mathbb{C}}\ ,
\end{equation*}%
for any (possibly infinite) subset $\Lambda $ of $\mathfrak{L}$. It is an
Hilbert space with scalar product defined on $\mathcal{F}_{\Lambda }\times
\mathcal{F}_{\Lambda }$ by
\begin{equation*}
\langle \psi ,\varphi \rangle _{\mathcal{F}_{\Lambda }}\doteq
\sum\limits_{n\in \mathbb{N}_{0}}\langle \psi ^{(n)},\varphi ^{(n)}\rangle
_{\ell ^{2}(\mathfrak{L})^{\otimes n}}\ ,\text{\qquad }\varphi =(\varphi
^{(n)})_{n\in \mathbb{N}_{0}},\psi =(\psi ^{(n)})_{n\in \mathbb{N}_{0}}\in
\mathcal{F}_{\Lambda }\ .
\end{equation*}%
The element $(1,0,0,\ldots )\in \mathcal{F}_{\Lambda }$ is the
zero--particle state, i.e., the so--called \emph{vacuum} vector of the Fock
space.

For any \emph{finite} subset $\Lambda \in \mathcal{P}_{f}(\mathfrak{L})$,
recall that $\ell ^{2}\left( \Lambda \right) $ (cf. (\ref{l2 space})) has
dimension equal to the volume $|\Lambda |$ of $\Lambda $. Therefore, because
of the antisymmetry of the $n$--particle wave function in $\mathcal{F}%
_{\Lambda }$,
\begin{equation}
\mathcal{F}_{\Lambda }=\bigoplus_{n=0}^{\left\vert \Lambda \right\vert
}P_{n}\ell ^{2}(\Lambda )^{\otimes n}\ ,\text{\qquad }\Lambda \in \mathcal{P}%
_{f}(\mathfrak{L})\ .  \label{loca fock space}
\end{equation}%
Using some elementary combinatorics, one checks in this case that the
fermion Fock space $\mathcal{F}_{\Lambda }$ is a finite dimensional Hilbert
space with dimension equal to $2^{\left\vert \Lambda \right\vert }$ for any $%
\Lambda \in \mathcal{P}_{f}(\mathfrak{L})$.

For any possibly infinite subset $\Lambda \subseteq \mathfrak{L}$, the
particle number becomes a self--adjoint (possibly unbounded)
\index{Particle number operator}operator $\mathbf{N}_{\Lambda }$ defined by%
\begin{equation}
\left( \mathbf{N}_{\Lambda }\psi \right) ^{(n)}\doteq n\psi ^{(n)}\ ,%
\text{\qquad }n\in \mathbb{N}_{0}\ ,  \label{paticle number operator1}
\end{equation}%
on the domain
\begin{equation}
\mathrm{Dom}(\mathbf{N}_{\Lambda })\doteq \left\{ \psi =(\psi ^{(n)})_{n\in
\mathbb{N}_{0}}\in \mathcal{F}_{\Lambda }:\sum\limits_{n\in \mathbb{N}%
_{0}}n^{2}\langle \psi ^{(n)},\psi ^{(n)}\rangle _{\ell ^{2}(\mathfrak{L}%
)^{\otimes n}}<\infty \right\} \ .  \label{paticle number operator2}
\end{equation}

Any family $\{\mathrm{H}^{\otimes n}\}_{n\in \mathbb{N}_{0}}$ of
Hamiltonians acting on $P_{n}\ell ^{2}(\Lambda )^{\otimes n}$, like those
defined by (\ref{family Hn}) for $n\in \mathbb{N}$, gives rise to an
operator $\mathbf{H}_{\Lambda }$ defined for any $n\in \mathbb{N}_{0}$ by%
\index{Hamiltonian!S1-H2!many-fermion}
\begin{equation}
\mathbf{H}_{\Lambda }\psi ^{(n)}\doteq \mathrm{H}^{\otimes n}\psi ^{(n)}\
,\qquad \psi ^{(n)}\in P_{n}\ell ^{2}(\Lambda )^{\otimes n}\subset \mathcal{F%
}_{\Lambda }\ .  \label{defininition H lambda}
\end{equation}%
It is clearly a symmetric operator on the subspace of $\mathcal{F}_{\Lambda
} $ constituted of sequences that eventually vanish. If $\Lambda \in
\mathcal{P}_{f}(\mathfrak{L})$, it means that $\mathbf{H}_{\Lambda }$ is
self--adjoint, by finite dimensionality of $\mathcal{F}_{\Lambda }$.

If $\Lambda \subseteq \mathfrak{L}$ is an infinite subset of $\mathfrak{L}$,
then $\mathbf{H}_{\Lambda }$ is closable because it is in any case
symmetric, see \cite[Theorem VIII.1, Section VIII.2]{ReedSimon}. By
self--adjointness of $\mathrm{H}^{\otimes n}$, there is additionally a dense
set of analytic vectors \cite[Section X.6]{ReedSimon2} on the subspace of $%
\mathcal{F}_{\Lambda }$ constituted of sequences that eventually vanish.
Therefore, by Nelson's analytic vector theorem \cite[Theorem X.39]%
{ReedSimon2}, $\mathbf{H}_{\Lambda }$ has a self--adjoint closure, again
denoted by $\mathbf{H}_{\Lambda }$.

Therefore, in any case, we obtain from (\ref{defininition H lambda}) a
Hamiltonian $\mathbf{H}_{\Lambda }$ that again defines a quantum dynamics on
the Hilbert space $\mathcal{H}=\mathcal{F}_{\Lambda }$ via the Schr\"{o}%
dinger equation (\ref{Schrodinger equation}). Typically, such kind of
Hamiltonian conserves the particle number in the sense that
\begin{equation*}
\mathrm{e}^{it\mathbf{H}_{\Lambda }}\mathbf{N}_{\Lambda }\mathrm{e}^{-it%
\mathbf{H}_{\Lambda }}=\mathbf{N}_{\Lambda }\ ,\qquad t\in \mathbb{R}\ .
\end{equation*}
In this case, the expectation value of the particle number observable $%
\mathbf{N}_{\Lambda }$ w.r.t. any solution of the Schr\"{o}dinger equation
equals%
\begin{equation*}
\langle \psi \left( 0\right) ,\mathrm{e}^{it\mathbf{H}_{\Lambda }}\mathbf{N}%
_{\Lambda }\mathrm{e}^{-it\mathbf{H}_{\Lambda }}\psi \left( 0\right) \rangle
_{\mathcal{F}_{\Lambda }}=\langle \psi \left( 0\right) ,\mathbf{N}_{\Lambda
}\psi \left( 0\right) \rangle _{\mathcal{F}_{\Lambda }}\ ,\qquad \psi \left(
0\right) \in \mathcal{F}_{\Lambda }\ .
\end{equation*}%
See Equations (\ref{Schrodinger equation}), (\ref{Schrodinger equation2})
and (\ref{expectation value}).

Compared to the first approach described in Section \ref{sect Quantum
Many--Body Systems}\ it is still not clear why the use of the Fock space can
be advantageous when the particle number is conserved. The utility of Fock
spaces comes from the use of so--called\emph{\ creation/annihilation
operators} explained below.

For more details on Fock spaces, see for instance \cite[Section 5.2.1]%
{BratteliRobinson}.

\subsection{Creation/Annihilation Operators\label{section creation ann}%
\index{Annihilation operator}%
\index{Creation operator}}

Apart from the fact that the Hilbert spaces of Section \ref{sect Quantum
Many--Body Systems} depend on the particle number $n$, one has always to
care about combinatorial issues because of the antisymmetry of wave
functions. This property of the wave function is encoded in the Fock space
in algebraic properties of the so--called creation and annihilation
operators $a_{x}^{\ast },a_{x}\in \mathcal{B}(\mathcal{F}_{\mathfrak{L}})$
of a fermion at lattice site $x\in \mathfrak{L}$: \medskip

\noindent
\underline{(i):} The
\index{Annihilation operator}\emph{annihilation} operator $a_{x}\in \mathcal{%
B}(\mathcal{F}_{\mathfrak{L}})$ of a fermion at lattice site $x\in \mathfrak{%
L}$ is the (linear) operator uniquely defined by the conditions $%
a_{x}(1,0,0,\ldots )=0$ and%
\begin{equation}
a_{x}\left( P_{n}\left( \psi _{1}\otimes \cdots \otimes \psi _{n}\right)
\right) \doteq
\frac{\sqrt{n}}{n!}\sum_{\pi \in \mathcal{P}_{n}}(-1)^{\pi }\left\langle
\mathfrak{e}_{x},\psi _{\pi \left( 1\right) }\right\rangle _{\ell ^{2}(%
\mathfrak{L})}\psi _{\pi \left( 2\right) }\otimes \cdots \otimes \psi _{\pi
\left( n\right) }  \label{annhilation}
\end{equation}%
for any $n\in \mathbb{N}$ and $\psi _{1},\ldots ,\psi _{n}\in \ell
^{2}\left( \mathfrak{L}\right) $, where we recall that $P_{n}$ is the
orthogonal projection (\ref{definition projection operator}) onto the
subspace of antisymmetric $n$--particle wave functions and $\mathcal{P}_{n}$
is the set of all permutations $\pi $ of $n$\ elements, see (\ref%
{permutations}). \medskip

\noindent \underline{(ii):} The
\index{Creation operator}\emph{creation} operator $a_{x}^{\ast }\in \mathcal{%
B}(\mathcal{F}_{\mathfrak{L}})$ of a fermion at lattice site $x\in \mathfrak{%
L}$, which turns out to be the adjoint of $a_{x}$, is defined by
\begin{equation}
\left( a_{x}^{\ast }\psi \right) ^{(0)}\doteq 0\qquad
\text{and}\qquad \left( a_{x}^{\ast }\psi \right) ^{(n)}\doteq \sqrt{n}%
P_{n}\left( \psi ^{(n-1)}\otimes \mathfrak{e}_{x}\right)  \label{creation}
\end{equation}%
for $n\in \mathbb{N}$, with $\psi =(\psi ^{(n)})_{n\in \mathbb{N}_{0}}\in
\mathcal{F}_{\mathfrak{L}}$. \medskip

Because of the antisymmetry property, note that, for any $x\in \mathfrak{L}$%
, $a_{x}^{\ast }a_{x}^{\ast }=0$, which reflects the
\index{Pauli exclusion principle}Pauli exclusion principle. In fact,
straightforward computations show that the family $\{a_{x},a_{x}^{\ast
}\}_{x\in \mathfrak{L}}\subset \mathcal{B}(\mathcal{F}_{\mathfrak{L}})$
satisfies the celebrated Canonical Anti--commutation Relations (CAR):%
\index{CAR} For any $x,y\in \mathfrak{L}$,%
\begin{equation}
a_{x}a_{y}+a_{y}a_{x}=0\ ,\qquad a_{x}a_{y}^{\ast }+a_{y}^{\ast
}a_{x}=\delta _{x,y}\mathbf{1}_{\mathcal{F}_{\mathfrak{L}}}\ .  \label{CAR}
\end{equation}%
They are pivotal relations coming from the antisymmetry property of wave
functions in the fermion Fock space. For instance, one deduces from (\ref%
{def c*}) and (\ref{CAR}) that in the $C^{\ast }$--algebra $\mathcal{B}(%
\mathcal{F}_{\mathfrak{L}})$, for any $x\in \mathfrak{L}$,
\begin{equation*}
\Vert a_{x}\Vert _{\mathcal{B}(\mathcal{F}_{\mathfrak{L}})}^{4}=\Vert
(a_{x}^{\ast }a_{x})^{2}\Vert _{\mathcal{B}(\mathcal{F}_{\mathfrak{L}%
})}=\Vert a_{x}^{\ast }a_{x}\Vert _{\mathcal{B}(\mathcal{F}_{\mathfrak{L}%
})}=\Vert a_{x}\Vert _{\mathcal{B}(\mathcal{F}_{\mathfrak{L}})}^{2}
\end{equation*}%
and since $a_{x}\neq 0$, we obtain that
\begin{equation}
\left\Vert a_{x}\right\Vert _{\mathcal{B}(\mathcal{F}_{\mathfrak{L}%
})}=\left\Vert a_{x}^{\ast }\right\Vert _{\mathcal{B}(\mathcal{F}_{\mathfrak{%
L}})}=1\ ,\qquad x\in \mathfrak{L}\ .  \label{equation stupide}
\end{equation}%
For more details on creation/annihilation operators, see \cite[Section 5.2.1]%
{BratteliRobinson}.

Meanwhile, the particle number operator (\ref{paticle number operator1})--(%
\ref{paticle number operator2}) in the (possibly infinite) subset $\Lambda
\subseteq \mathfrak{L}$ can be\ \emph{formally} written on the subspace of
antisymmetric $n$--particle wave functions ($n\in \mathbb{N}$) as%
\index{Particle number operator}
\begin{equation*}
\mathbf{N}_{\Lambda }\psi ^{(n)}=\left( \sum\limits_{x\in \Lambda }\mathbf{n}%
_{x}\right) \psi ^{(n)},\qquad \psi ^{(n)}\in P_{n}\ell ^{2}(\Lambda
)^{\otimes n}\ ,
\end{equation*}%
with
\begin{equation}
\mathbf{n}_{x}\doteq a_{x}^{\ast }a_{x}\in \mathcal{B}(\mathcal{F}_{%
\mathfrak{L}})  \label{particle number operators0}
\end{equation}%
being the so--called particle number operator at lattice site $x\in
\mathfrak{L}$. Because of the CAR (\ref{CAR}), note that $\mathbf{n}_{x}$ is
a projection operator.

For any finite subset $\Lambda \in \mathcal{P}_{f}(\mathfrak{L})$, $\mathbf{N%
}_{\Lambda }\in \mathcal{B}(\mathcal{F}_{\Lambda })$, by finite
dimensionality of the local fermion Fock space $\mathcal{F}_{\Lambda }$ (see
(\ref{loca fock space})) and we can see this particle number operator as
\begin{equation}
\mathbf{N}_{\Lambda }\equiv \sum\limits_{x\in \Lambda }\mathbf{n}_{x}\in
\mathcal{B}(\mathcal{F}_{\mathfrak{L}})\ .  \label{particle number operators}
\end{equation}

In the same way, the operator $\mathbf{H}_{\Lambda }$ (\ref{defininition H
lambda}) can be written in terms of creation and annihilation operators. For
instance, if $\mathbf{H}_{\Lambda }$ is defined from the Hamiltonian (\ref%
{family Hn}) for any $n\in \mathbb{N}$, then it can be \emph{formally}
written on the subspace of antisymmetric $n$--particle wave functions ($n\in
\mathbb{N}$) as%
\index{Hamiltonian!S1-H2!many-fermion}
\begin{equation*}
\mathbf{H}_{\Lambda }\psi ^{(n)}=\left( \sum\limits_{x,y\in \Lambda }h\left(
\left\vert x-y\right\vert \right) a_{x}^{\ast }a_{y}+\sum\limits_{x,y\in
\Lambda }v(|x-y|)\mathbf{n}_{x}\mathbf{n}_{y}\right) \psi ^{(n)}
\end{equation*}%
for all $\psi ^{(n)}\in P_{n}\ell ^{2}(\Lambda )^{\otimes n}$ and any
(possibly infinite) subset $\Lambda \subseteq \mathfrak{L}$.

If $\Lambda \in \mathcal{P}_{f}(\mathfrak{L})$ is a finite subset then $%
\mathbf{H}_{\Lambda }\in \mathcal{B}(\mathcal{F}_{\Lambda })\subset \mathcal{%
B}(\mathcal{F}_{\mathfrak{L}})$ can be seen as the operator
\begin{equation}
\mathbf{H}_{\Lambda }\equiv \sum\limits_{x,y\in \Lambda }h\left( \left\vert
x-y\right\vert \right) a_{x}^{\ast }a_{y}+\sum\limits_{x,y\in \Lambda
}v(|x-y|)\mathbf{n}_{x}\mathbf{n}_{y}\ ,  \label{H lambda}
\end{equation}%
again by finite dimensionality of $\mathcal{F}_{\Lambda }$. This formulation
of $\mathbf{H}_{\Lambda }$ can easily be interpreted: The term $a_{x}^{\ast
}a_{y}$ destroy a fermion at lattice site $y$ to create another one at
lattice site $x$. It thus gives rise to fermion transport properties in the
physical system and it is related to the (usual)
\index{Kinetic term}\emph{kinetic terms}. The second term depends on the
occupation number ($0$ or $1$) at lattice sites $x$ and $y$, and yields the
interaction energy. It is a so--called
\index{Interparticle interactions}%
\index{Density--density interaction}\emph{density--density interaction}.

\subsection{The Lattice CAR $C^{\ast }$--Algebra\label{section CAR}%
\index{CAR algebras}}

Sections \ref{Section Qingle part}--\ref{section creation ann} was
preliminary sections presenting many--fermion systems on lattices in the
usual context of Quantum Mechanics, as explained in Sections \ref{Section
Schro}--\ref{sect non auto}. In the sequel, however, we avoid to speak about
Hilbert space structures, Fock spaces, etc., by using the algebraic
formulation of Quantum Mechanics as explained in Section \ref{sect Algebraic
Formulation of Quantum Mechanics}. To this end, we have to define a $C^{\ast
}$--algebra, named the lattice CAR $C^{\ast }$--algebra, defined as follows:
\medskip

\noindent
\underline{(i):} Recall that $a_{x}^{\ast },a_{x}$ are the so--called
creation and annihilation operators of a fermion at lattice site $x\in
\mathfrak{L}$. In Section \ref{section creation ann} we explicitly define
them by using the fermion Fock space. In the algebraic approach, however, we
only assume the existence of a unit $\mathbf{1}$ and a
\index{Annihilation operator}%
\index{Creation operator}family $\{a_{x},a_{x}^{\ast }\}_{x\in \mathfrak{L}}$
satisfying the CAR:%
\index{CAR} For any $x,y\in \mathfrak{L}$,%
\begin{equation}
a_{x}a_{y}+a_{y}a_{x}=0\ ,\qquad a_{x}a_{y}^{\ast }+a_{y}^{\ast
}a_{x}=\delta _{x,y}\mathbf{1}\ .  \label{CARbis}
\end{equation}%
Compare with (\ref{CAR}), by observing that the unit $\mathbf{1}$ refers to
the identity map $\mathbf{1}_{\mathcal{F}_{\mathfrak{L}}}$ of $\mathcal{F}_{%
\mathfrak{L}}$. Such commutation relations are indeed sufficient to
characterize fermion systems via the Pauli exclusion principle. \medskip

\noindent
\underline{(ii):} \emph{Every} physical system of fermions on the lattice is
associated with some \emph{finite} region $\Lambda $ of lattice space. The
set of all finite subsets of the lattice $\mathfrak{L}$ is denoted by $%
\mathcal{P}_{f}(\mathfrak{L})\subset 2^{\mathfrak{L}}$.
\index{Observable!H1-S2}Observables one can then measure on many--fermion
systems within any $\Lambda \in \mathcal{P}_{f}(\mathfrak{L})$ are finite
sums of monomials of $\{a_{x}^{\ast },a_{x}\}_{x\in \Lambda }$ and $\mathbf{1%
}$. See (\ref{particle number operators0}), (\ref{particle number operators}%
) and (\ref{H lambda}) for explicit examples in the Fock space
representation.

For $\Lambda \in \mathcal{P}_{f}(\mathfrak{L})$, this yields to the \emph{%
local} CAR $C^{\ast }$--algebra $\mathcal{U}_{\Lambda }$
\index{CAR algebras!local}as the set of all finite sums of monomials
constructed from $\{a_{x}^{\ast },a_{x}\}_{x\in \Lambda }$ and the unit $%
\mathbf{1}$. See Section \ref{sect Algebraic Formulation of Quantum
Mechanics} for the definition of $C^{\ast }$--algebras. In particular, the
particle number operators $\mathbf{n}_{x}\doteq a_{x}^{\ast }a_{x}$, $x\in
\Lambda $, and the Hamiltonian of the fermion system are self--adjoint
elements of $\mathcal{U}_{\Lambda }$. See again (\ref{H lambda}) for an
example of Hamiltonians in the Fock space representation.

Note that one can define annihilation and creation operators of fermions
with wave functions $\psi \in \ell ^{2}(\Lambda )\subset \ell ^{2}(\mathfrak{%
L})$ for any $\Lambda \in \mathcal{P}_{f}(\mathfrak{L})$ by%
\index{Annihilation operator}%
\index{Creation operator}
\begin{equation}
a(\psi )\doteq \sum\limits_{x\in \Lambda }%
\overline{\psi (x)}a_{x}\in \mathcal{U}_{\Lambda }\ ,\quad a^{\ast }(\psi
)\doteq \sum\limits_{x\in \Lambda }\psi (x)a_{x}^{\ast }\in \mathcal{U}%
_{\Lambda }\ .  \label{definition}
\end{equation}%
Clearly, $a^{\ast }(\psi )=a(\psi )^{\ast }$ for all $\psi \in \ell
^{2}(\Lambda )$ and on the canonical orthonormal basis $\left\{ \mathfrak{e}%
_{x}\right\} _{x\in \mathfrak{L}}$ (\ref{ONB}), $a(\mathfrak{e}_{x})=a_{x}$
at all $x\in \Lambda $. The map $\psi \mapsto a(\psi )$ (resp. $\psi \mapsto
a^{\ast }(\psi )$) from $\ell ^{2}(\Lambda )$\ to $\mathcal{U}_{\Lambda }$
is anti--linear (resp. linear) and because of (\ref{CARbis}),%
\index{CAR}%
\begin{equation}
a(\psi )a(\varphi )+a(\varphi )a(\psi )=0\ ,\qquad a(\psi )a(\varphi )^{\ast
}+a(\varphi )^{\ast }a(\psi )=\langle \psi ,\varphi \rangle _{\ell ^{2}(%
\mathfrak{L})}\mathbf{1}  \label{CARbisbis}
\end{equation}%
for any $\varphi ,\psi \in \ell ^{2}(\Lambda )\subset \ell ^{2}(\mathfrak{L}%
) $. These CAR are a generalization of (\ref{CARbis}).

The relation with the local fermion Fock space $\mathcal{F}_{\Lambda }$ (\ref%
{loca fock space}) is now clear:

\begin{lemma}[CAR algebras and fermion Fock spaces for finite systems]
\label{equivalence lemma1}\mbox{ }\newline
For any $\Lambda \in \mathcal{P}_{f}(\mathfrak{L})$, the local CAR $C^{\ast
} $--algebra $\mathcal{U}_{\Lambda }$ is $\ast $--isomorphic to the $C^{\ast
}$--algebra $\mathcal{B}(\mathcal{F}_{\Lambda })$ of bounded operators
acting on $\mathcal{F}_{\Lambda }$. In particular, its dimension equals $%
2^{2\left\vert \Lambda \right\vert }$.%
\index{CAR algebras!local}
\end{lemma}

\begin{proof}
On the one hand, since $\ell ^{2}\left( \Lambda \right) $ has dimension $%
|\Lambda |$, for any $\Lambda \in \mathcal{P}_{f}(\mathfrak{L})$, we infer
from (\ref{definition})--(\ref{CARbisbis}) and \cite[Theorem 5.2.5]%
{BratteliRobinson} that $\mathcal{U}_{\Lambda }$\ is isomorphic to the $%
C^{\ast }$--algebra of $2^{\left\vert \Lambda \right\vert }\times
2^{\left\vert \Lambda \right\vert }$ complex matrices. $\mathcal{U}_{\Lambda
}$ has in particular dimension equal to $2^{2\left\vert \Lambda \right\vert
} $. On the other hand, as explained below Equation (\ref{loca fock space}),
for any $\Lambda \in \mathcal{P}_{f}(\mathfrak{L})$, the dimension of the
Fock space $\mathcal{F}_{\Lambda }$\ is equal to $2^{\left\vert \Lambda
\right\vert }$. Therefore, $\mathcal{U}_{\Lambda }$ is $\ast $--isomorphic
to $\mathcal{B}(\mathcal{F}_{\Lambda })$.
\end{proof}

\noindent Lemma \ref{equivalence lemma1} yields a \emph{faithful}
representation
\begin{equation}
\mathbf{\pi }_{\Lambda }:\mathcal{U}_{\Lambda }\rightarrow \mathcal{B}(%
\mathcal{F}_{\Lambda })  \label{map blabla}
\end{equation}%
of the local CAR $C^{\ast }$--algebra $\mathcal{U}_{\Lambda }$ on the
representation (Hilbert) space $\mathcal{F}_{\Lambda }$ for every $\Lambda
\in \mathcal{P}_{f}(\mathfrak{L})$. This representation is said to be \emph{%
canonical} when it maps $a_{x}^{\ast },a_{x}$ to creation and annihilation
operators defined by (\ref{annhilation}) and (\ref{creation}) on $\mathcal{F}%
_{\Lambda }$ for any $x\in \mathfrak{L}$. The (canonical) representation is
\emph{irreducible} and named the Fock representation. For more details on
the representation theory of $C^{\ast }$--algebras, see Section \ref{Sec
Representation theory}.

Because of Corollary \ref{Uniqueness of irreducible representations --II},
one can equivalently use the Fock space formulation within the usual context
of Quantum Mechanics or the algebraic approach, provided $\Lambda \in
\mathcal{P}_{f}(\mathfrak{L})$. Compare indeed Lemma \ref{equivalence lemma1}
with the Heisenberg picture of Quantum Mechanics described in\ Section \ref%
{Section Hei}. This fact is \emph{not anymore} true in the infinite--volume
situation, i.e., for \emph{infinite} subsets $\Lambda \subseteq \mathfrak{L}$%
, since the corresponding Fock space $\mathcal{F}_{\Lambda }$ has then
infinite dimension. See Corollary \ref{Uniqueness of irreducible
representations --III}. \medskip

\noindent
\underline{(iii):} Physical systems become macroscopic when they belong to
finite regions $\Lambda $ of lattice that become arbitrarily large.
Therefore, we consider a family of cubic boxes\footnote{%
It is a technically convenient choice to define the thermodynamic limit, but
one could also take other Van Hove nets. See for instance \cite[Remark 1.3]%
{BruPedra2}.} defined, for all $L\in \mathbb{R}_{0}^{+}$, by
\begin{equation}
\Lambda _{L}\doteq \{(x_{1},\ldots ,x_{d})\in \mathfrak{L}%
\,:\,|x_{1}|,\ldots ,|x_{d}|\leq L\}\in \mathcal{P}_{f}(\mathfrak{L})\ .
\label{eq:def lambda n}
\end{equation}%
Hence, $\{\mathcal{U}_{\Lambda _{L}}\}_{L\in \mathbb{R}_{0}^{+}}$ is an
increasing net of $C^{\ast }$--algebras and the set%
\begin{equation}
\mathcal{U}_{0}\doteq \underset{L\in \mathbb{R}_{0}^{+}}{\bigcup }\mathcal{U}%
_{\Lambda _{L}}  \label{simple}
\end{equation}%
of local elements is a normed $\ast $--algebra with $\left\Vert A\right\Vert
_{\mathcal{U}_{0}}=\left\Vert A\right\Vert _{\mathcal{U}_{\Lambda _{L}}}$for
all $A\in \mathcal{U}_{\Lambda _{L}}$ and $L\in \mathbb{R}_{0}^{+}$.\medskip

\noindent \underline{(iv):} For physical macroscopic systems one considers
the limit \textquotedblleft $\Lambda \rightarrow \mathfrak{L}$%
\textquotedblright . This is named the
\index{Thermodynamic limit}thermodynamic limit and one gets an infinite
fermion system. This approach yields the CAR $C^{\ast }$--algebra $\mathcal{U%
}$
\index{CAR algebras!non--local}of the infinite system, which is by
definition the completion of the normed $\ast $--algebra $\mathcal{U}_{0}$.
It is separable, by finite dimensionality of $\mathcal{U}_{\Lambda }$ for
any $\Lambda \in \mathcal{P}_{f}(\mathfrak{L})$. In other words, $\mathcal{U}
$ is the inductive limit of the finite dimensional $C^{\ast }$--algebras $\{%
\mathcal{U}_{\Lambda }\}_{\Lambda \in \mathcal{P}_{f}(\mathfrak{L})}$. In
this construction, $\mathcal{U}_{0}\subset \mathcal{U}$ can be seen as the
smallest normed $\ast $--algebra containing all generators $\{a_{x}\}_{x\in
\mathfrak{L}}$.

By replacing $\Lambda $ with $\mathfrak{L}$ in Equation (\ref{definition})
one can again define annihilation and creation operators $a(\psi ),a^{\ast
}(\psi )$ of fermions with wave functions $\psi \in \ell ^{2}(\mathfrak{L})$%
. These operators are still well--defined. Indeed, because of the CAR (\ref%
{CARbis}),
\begin{equation*}
\Vert a(\psi )\Vert _{\mathcal{U}}^{2}=\Vert a^{\ast }(\psi )\Vert _{%
\mathcal{U}}^{2}=\langle \psi ,\psi \rangle _{\ell ^{2}(\mathfrak{L})}\
,\qquad \psi \in \ell ^{2}(\mathfrak{L})\ .
\end{equation*}%
Compare with (\ref{equation stupide}). Hence, the anti--linear (resp.
linear) map $\psi \mapsto a(\psi )$ (resp. $\psi \mapsto a^{\ast }(\psi )$)
from $\ell ^{2}(\mathfrak{L})$\ to $\mathcal{U}$ is norm--continuous. Again,
$a^{\ast }(\psi )=a(\psi )^{\ast }$ for all $\psi \in \ell ^{2}(\mathfrak{L}%
) $ and the CAR (\ref{CARbisbis}) can be extended to all $\varphi ,\psi \in
\ell ^{2}(\mathfrak{L})$.%
\index{CAR}

The faithful and irreducible canonical representation $\mathbf{\pi }%
_{\Lambda }$ defined by (\ref{map blabla}) gives rise in the infinite volume
limit to a unique faithful and irreducible (canonical) representation $%
\mathbf{\pi }_{\mathfrak{L}}$ of the CAR $C^{\ast }$--algebra $\mathcal{U}$
such that
\begin{equation}
\mathbf{\pi }_{\mathfrak{L}}(\mathcal{U}_{\Lambda })=\mathbf{\pi }_{\Lambda
}(\mathcal{U}_{\Lambda })=\mathcal{B}(\mathcal{F}_{\Lambda })\ ,\qquad
\Lambda \in \mathcal{P}_{f}(\mathfrak{L})\ .  \label{simplesimple}
\end{equation}%
This representation is in fact defined via the unique $\ast $--isomorphisms $%
\mathbf{\pi }_{\{x\}}$, $x\in \mathfrak{L}$, mapping $a_{x}\in \mathcal{U}$
to the operator $\mathbf{\pi }_{\{x\}}(a_{x})\equiv a_{x}\in \mathcal{B}(%
\mathcal{F}_{\mathfrak{L}})$ defined by (\ref{annhilation}). It is again
named the Fock representation.

The Fock space $\mathcal{F}_{\mathfrak{L}}$ has infinite dimension and\ is
separable. Thus, by Corollary \ref{Uniqueness of irreducible representations
--III}, we emphasize again that the Fock space formulation within the usual
context of Quantum Mechanics and the algebraic approach are not equivalent
to each other. Again by Corollary \ref{Uniqueness of irreducible
representations --III}, the $C^{\ast }$--algebra $\mathcal{B}(\mathcal{F}_{%
\mathfrak{L}})$ has more than one unitarily non--equivalent irreducible
representation as well. Moreover, $\mathcal{U}$ is in some sense \emph{%
strictly smaller} than $\mathcal{B}(\mathcal{F}_{\mathfrak{L}})$:

\begin{lemma}[CAR algebras and fermion Fock spaces for infinite systems]
\label{equivalence lemma2}\mbox{ }\newline
\begin{equation*}
\mathbf{\pi }_{\mathfrak{L}}(\mathcal{U})=%
\overline{\underset{L\in \mathbb{R}_{0}^{+}}{\bigcup }\mathcal{B}(\mathcal{F}%
_{\Lambda _{L}})}\varsubsetneq \mathcal{B}(\mathcal{F}_{\mathfrak{L}})\ .
\end{equation*}
\end{lemma}

\begin{proof}
For any non--vanishing $\theta \in \mathbb{R}/(2\pi \mathbb{Z)}$, $\mathrm{e}%
^{i\theta \mathbf{N}_{\mathfrak{L}}}\in \mathcal{B}(\mathcal{F}_{\mathfrak{L}%
})$ but $\mathbf{\pi }_{\mathfrak{L}}^{-1}(\mathrm{e}^{i\theta \mathbf{N}_{%
\mathfrak{L}}})=\emptyset $, with the particle number operator $\mathbf{N}_{%
\mathfrak{L}}$ being the self--adjoint operator defined by (\ref{paticle
number operator1})--(\ref{paticle number operator2}) for $\Lambda =\mathfrak{%
L}$. Indeed, for all $\Lambda \in \mathcal{P}_{f}(\mathfrak{L})$ with
complement $\Lambda ^{c}\doteq \mathfrak{L}\backslash \Lambda $ and any
function $\psi \in \ell ^{2}(\mathfrak{L})\subset \mathcal{F}_{\mathfrak{L}}$%
, a direct computation shows that%
\begin{equation}
\left\Vert \left( \mathrm{e}^{i\theta \mathbf{N}_{\mathfrak{L}}}-\mathrm{e}%
^{i\theta \mathbf{N}_{\Lambda }}\right) \psi \right\Vert _{\mathcal{F}_{%
\mathfrak{L}}}^{2}=\left\vert \mathrm{e}^{i\theta }-1\right\vert
^{2}\sum\limits_{x\in \Lambda ^{c}}\left\vert \psi \left( x\right)
\right\vert ^{2}\ .  \label{eq conne}
\end{equation}%
Here, $\mathbf{N}_{\Lambda }\in \mathcal{B}(\mathcal{F}_{\mathfrak{L}})$ is
the particle number operator defined by (\ref{particle number operators})
for $\Lambda \in \mathcal{P}_{f}(\mathfrak{L})$. Therefore, $\mathrm{e}%
^{i\theta \mathbf{N}_{\Lambda }}$\ does not converge in $\mathcal{B}(%
\mathcal{F}_{\mathfrak{L}})$ (norm topology) to $\mathrm{e}^{i\theta \mathbf{%
N}_{\mathfrak{L}}}$ for $\theta \neq 0$.

Assume now that $\mathbf{\pi }_{\mathfrak{L}}^{-1}(\mathrm{e}^{i\theta
\mathbf{N}_{\mathfrak{L}}})\neq \emptyset $. Then, by Lemma \ref{equivalence
lemma1} and density of $\mathcal{U}_{0}$ in $\mathcal{U}$, there are two
families $\{\Lambda _{n}\}_{n\in \mathbb{N}}\subset \mathcal{P}_{f}(%
\mathfrak{L})$ and $\{U_{\Lambda _{n}}\}_{n\in \mathbb{N}}$ such that $%
U_{\Lambda _{n}}\in \mathcal{B}(\mathcal{F}_{\Lambda _{n}})\subset \mathcal{B%
}(\mathcal{F}_{\mathfrak{L}})$ converges in $\mathcal{B}(\mathcal{F}_{%
\mathfrak{L}})$ to $\mathrm{e}^{i\theta \mathbf{N}_{\mathfrak{L}}}$, as $%
n\rightarrow \infty $. From this and (\ref{eq conne}), one deduces that $%
\left( U_{\Lambda _{n}}-\mathrm{e}^{i\theta \mathbf{N}_{\Lambda
_{n}}}\right) $ must converge, as $n\rightarrow \infty $, to zero in $%
\mathcal{B}(\mathcal{F}_{\mathfrak{L}})$. The latter is not possible,
otherwise $\mathrm{e}^{i\theta \mathbf{N}_{\Lambda }}$\ would then converge
in $\mathcal{B}(\mathcal{F}_{\mathfrak{L}})$ to $\mathrm{e}^{i\theta \mathbf{%
N}_{\mathfrak{L}}}$.

Therefore, for any non--vanishing $\theta \in \mathbb{R}/(2\pi \mathbb{Z)}$,
$\mathbf{\pi }_{\mathfrak{L}}^{-1}(\mathrm{e}^{i\theta \mathbf{N}_{\mathfrak{%
L}}})=\emptyset $. The assertion then follows, by Equations (\ref{simple})
and (\ref{simplesimple}).
\end{proof}

\medskip

\noindent \underline{(v):} For any non--vanishing $\theta \in \mathbb{R}%
/(2\pi \mathbb{Z)}$, the unitary operator $\mathrm{e}^{i\theta \mathbf{N}_{%
\mathfrak{L}}}\notin \mathbf{\pi }_{\mathfrak{L}}(\mathcal{U})$ (see proof
of Lemma \ref{equivalence lemma2}) gives rise to a $\ast $--automorphism
\begin{equation*}
B\mapsto \mathrm{e}^{i\theta \mathbf{N}_{\mathfrak{L}}}\ B\ \mathrm{e}%
^{-i\theta \mathbf{N}_{\mathfrak{L}}}
\end{equation*}%
of $\mathcal{B}(\mathcal{F}_{\mathfrak{L}})$ defined via (\ref{paticle
number operator1})--(\ref{paticle number operator2}). One says that the
unitary operator $\mathrm{e}^{i\theta \mathbf{N}_{\mathfrak{L}}}\in \mathcal{%
B}(\mathcal{F}_{\mathfrak{L}})$ implements a \emph{global gauge
transformation}, see for instance \cite[Eq. (A.4)]{BruPedra1}.%
\index{Gauge transformation} A similar $\ast $--automorphism exists on the
CAR $C^{\ast }$--algebra $\mathcal{U}$: For any $\theta \in \mathbb{R}/(2\pi
\mathbb{Z)}$, it is the unique $\ast $--automorphism $\sigma _{\theta }$ of $%
\mathcal{U}$ satisfying the conditions%
\begin{equation}
\sigma _{\theta }(a_{x})=\mathrm{e}^{-i\theta }a_{x}\ ,%
\text{\qquad }x\in \mathfrak{L}\ .  \label{definition of gauge}
\end{equation}%
Note indeed that, using the Fock representation, one verifies that
\begin{equation*}
\mathbf{\pi }_{\mathfrak{L}}\left( \sigma _{\theta }(B)\right) =\mathrm{e}%
^{i\theta \mathbf{N}_{\mathfrak{L}}}\ \mathbf{\pi }_{\mathfrak{L}}\left(
B\right) \ \mathrm{e}^{-i\theta \mathbf{N}_{\mathfrak{L}}}\ ,\text{\qquad }%
B\in \mathcal{U}\ ,
\end{equation*}%
for any $\theta \in \mathbb{R}/(2\pi \mathbb{Z)}$.

A special role is played by $\sigma _{\pi }$: Elements $B_{1},B_{2}\in
\mathcal{U}$ satisfying $\sigma _{\pi }(B_{1})=B_{1}$ and $\sigma _{\pi
}(B_{2})=-B_{2}$ are respectively called \emph{even} and
\index{Algebra element!odd}\emph{odd}, while elements $B\in \mathcal{U}$
satisfying $\sigma _{\theta }(B)=B$ for any $\theta \in \mathbb{R}/(2\pi
\mathbb{Z)}$ are called
\index{Algebra element!gauge invariant}\emph{gauge invariant}. The set
\begin{equation}
\mathcal{U}^{+}\doteq \{B\in \mathcal{U}\;:\;B=\sigma _{\pi }(B)\}\subset
\mathcal{U}  \label{definition of even operators}
\end{equation}%
of all
\index{Algebra element!even}even elements and the set
\begin{equation}
\mathcal{U}^{\circ }\doteq \bigcap\limits_{\theta \in \mathbb{R}/(2\pi
\mathbb{Z)}}\{B\in \mathcal{U}\;:\;B=\sigma _{\theta }(B)\}\subset \mathcal{U%
}^{+}  \label{definition of gauge invariant operators}
\end{equation}%
of all gauge invariant elements are $\ast $--algebras. By continuity of $%
\sigma _{\theta }$, it follows that $\mathcal{U}^{+}$ and $\mathcal{U}%
^{\circ }$ are closed and hence $C^{\ast }$--algebras. $\mathcal{U}^{\circ }$
is known as the fermion
\index{Observable!H1-S2}\emph{observable} algebra because it equals the $%
C^{\ast }$--algebra of all self--adjoint elements of $\mathcal{U}$.

\subsection{Lattice Fermi versus Quantum Spin Systems\label{Quantum spin
systems}%
\index{Quantum spin systems}}

Quantum spin systems are models used to describe quantum phenomena appearing
at low temperatures in condensed matter physics. They are nowadays
particularly important in Quantum Information Theory. This subject appeared
right from the beginning, with the emergence of Quantum Mechanics in the
twenties. A concise introduction on its history is given in the paper \cite%
{N}, see also the corresponding references therein.

For completeness, we shortly recall that quantum spin systems are infinite
systems composed of elementary finite dimensional spaces, originally
referring to a spin variable (see (\ref{spin set})). Therefore, they are
constructed in a similar way as lattice Fermi systems. Mathematically
speaking, they are defined via the algebraic formulation of Quantum
Mechanics from the so--called \emph{spin} $C^{\ast }$--algebra $\mathcal{Q}$%
: \medskip

\noindent
\underline{(i):} With any lattice site $x\in \mathfrak{L}$ we associate a
finite dimensional Hilbert space $\mathcal{H}_{x}\equiv \mathbb{C}^{N}$ for
some $N\in \mathbb{N}$. Typically, the parameter $N$ is the cardinal $%
\left\vert \mathrm{S}\right\vert $ of the
\index{Spin set}spin set $\mathrm{S}$ (\ref{spin set}). Then, the algebra of
local observables over $\Lambda \in \mathcal{P}_{f}(\mathfrak{L})$ is the
subset of self--adjoint elements of the $C^{\ast }$--algebra%
\index{Spin algebra!local}%
\begin{equation*}
\mathcal{Q}_{\Lambda }\doteq \bigotimes_{x\in \Lambda }\mathcal{B}\left(
\mathcal{H}_{x}\right) \equiv \mathcal{B}\left( \bigotimes_{x\in \Lambda }%
\mathcal{H}_{x}\right) \ .
\end{equation*}%
Recall that $\mathcal{B}\left( \mathcal{H}_{x}\right) $ denotes the $C^{\ast
}$--algebra of bounded linear operators on $\mathcal{H}_{x}$ for $x\in
\mathfrak{L}$. The dimension of $\mathcal{Q}_{\Lambda }$ is equal to $%
N^{2\left\vert \Lambda \right\vert }$ for any $\Lambda \in \mathcal{P}_{f}(%
\mathfrak{L})$. Compare with the local CAR $C^{\ast }$--algebra $\mathcal{U}%
_{\Lambda }$, see in particular Lemma \ref{equivalence lemma1}. \medskip

\noindent
\underline{(ii):} For all $\Lambda ^{(1)},\Lambda ^{(2)}\in \mathcal{P}_{f}(%
\mathfrak{L})$ with $\Lambda ^{(1)}\subset \Lambda ^{(2)}$, there is a
(canonical) isometric inclusion $\mathcal{Q}_{\Lambda ^{(1)}}\hookrightarrow
\mathcal{Q}_{\Lambda ^{(2)}}$ defined by
\begin{equation*}
A\mapsto A\otimes \bigotimes_{x\in \Lambda ^{(2)}\backslash \Lambda ^{(1)}}%
\mathbf{1}_{\mathcal{H}_{x}}\;.
\end{equation*}%
In particular, using the sequence of cubic boxes defined by (\ref{eq:def
lambda n}) we observe that $\{\mathcal{Q}_{\Lambda _{L}}\}_{L\in \mathbb{R}%
_{0}^{+}}$ is also an increasing net of $C^{\ast }$--algebras. Compare with
the family $\{\mathcal{U}_{\Lambda _{L}}\}_{L\in \mathbb{R}_{0}^{+}}$.
\medskip

\noindent \underline{(iii):} Hence, the set%
\begin{equation*}
\mathcal{Q}_{0}\doteq \underset{L\in \mathbb{R}_{0}^{+}}{\bigcup }\mathcal{Q}%
_{\Lambda _{L}}
\end{equation*}%
of local elements is a normed $\ast $--algebra with $\left\Vert A\right\Vert
_{\mathcal{Q}_{0}}=\left\Vert A\right\Vert _{\mathcal{Q}_{\Lambda _{L}}}$for
all $A\in \mathcal{Q}_{\Lambda _{L}}$ and $L\in \mathbb{R}_{0}^{+}$. Compare
with the normed $\ast $--algebra $\mathcal{U}_{0}$ of local elements defined
by (\ref{simple}).

For any finite subsets $\Lambda ^{(1)},\Lambda ^{(2)}\in \mathcal{P}_{f}(%
\mathfrak{L})$ with $\Lambda ^{(1)}\cap \Lambda ^{(2)}=\emptyset $ we
observe that
\begin{equation*}
\left[ B_{1},B_{2}\right] \doteq B_{1}B_{2}-B_{2}B_{1}=0\ ,\qquad B_{1}\in
\mathcal{Q}_{\Lambda ^{(1)}},\ B_{2}\in \mathcal{Q}_{\Lambda ^{(2)}}\ .
\end{equation*}%
Because of the CAR (\ref{CARbis}), such a property is also satisfied for all
\emph{even} local elements $B_{1}\in \mathcal{U}_{\Lambda ^{(1)}}\cap
\mathcal{U}^{+}$ and $B_{2}\in \mathcal{U}_{\Lambda ^{(2)}}\cap \mathcal{U}%
^{+}$, see (\ref{definition of even operators}). However, it is wrong in
general for Fermi systems. For instance, the CAR (\ref{CARbis}) trivially
yield $[a_{x},a_{y}]=2a_{x}a_{y}$ for any $x,y\in \mathfrak{L}$.\medskip

\noindent \underline{(iv):} The \emph{spin} $C^{\ast }$--algebra $\mathcal{Q}
$
\index{Spin algebra!non--local}of the lattice $\mathfrak{L}$ is by
definition the completion of the normed $\ast $--algebra $\mathcal{Q}_{0}$.
It is separable, by finite dimensionality of $\mathcal{Q}_{\Lambda }$ for $%
\Lambda \in \mathcal{P}_{f}(\mathfrak{L})$. In other words, $\mathcal{Q}$ is
the inductive limit of the finite dimensional $C^{\ast }$--algebras $\{%
\mathcal{Q}_{\Lambda }\}_{\Lambda \in \mathcal{P}_{f}(\mathfrak{L})}$.
Compare with the CAR $C^{\ast }$--algebra $\mathcal{U}$.
\vspace{0.1cm}

Infinite--volume dynamics is then constructed via Lieb--Robinson bounds, as
done in Sections \ref{Generalized Lieb--Robinson Bounds}--\ref{section LR
non-auto} for CAR $C^{\ast }$--algebras. Here, we focus on lattice Fermi
systems which are, from a technical point of view, slightly more difficult
because of the non--commutativity of their elements on different lattice
sites, as explained above. However, all the results presented in Sections %
\ref{Generalized Lieb--Robinson Bounds}--\ref{section LR non-auto} hold true
for quantum spin systems, by restricting them on the $C^{\ast }$--algebra $%
\mathcal{U}^{+}\subset \mathcal{U}$ (\ref{definition of even operators}) of
all even elements and then by replacing $\mathcal{U}^{+}$ with the spin $%
C^{\ast }$--algebras $\mathcal{Q}$.

\section{Lieb--Robinson Bounds for Multi--Commutators\label{Generalized
Lieb--Robinson Bounds}}

Lieb--Robinson bounds for multi--commutators are studied here for fermion
systems, only. In the case of quantum spin systems, $\mathcal{U}$ has to be
replaced by the infinite tensor product $\mathcal{Q}$ of copies of some
finite dimensional $C^{\ast }$--algebra attached to each site $x\in
\mathfrak{L}$. See Section \ref{Quantum spin systems}. All results of this
section also hold in this situation. We concentrate our attention on fermion
algebras in view of applications to microscopic foundations of the theory of
electrical conduction \cite{brupedrahistoire, OhmV}. Moreover, as explained
in Section \ref{Quantum spin systems}, the fermionic case is, technically
speaking, more involved, because of the non--commutativity of elements of
the CAR algebra\ $\mathcal{U}$ sitting on different lattice sites.

\subsection{Interactions and Finite--Volume Dynamics\label{Finite Volume
Dynamics}}

Following the algebraic formulation of Quantum Mechanics (Section \ref{sect
Algebraic Formulation of Quantum Mechanics}), we have to define a $C_{0}$%
--group (that is, a strongly continuous group) $\{\tau _{t}\}_{t\in {\mathbb{%
R}}}$ of $\ast $--automorphisms of the CAR $C^{\ast }$--algebra $\mathcal{U}$%
. On the other hand, as explained in Section \ref{section CAR}, every
physical system of particles belongs to some finite region $\Lambda _{L}$ (%
\ref{eq:def lambda n}) of lattice space, and they become macroscopic when $%
L\rightarrow \infty $. Therefore, we define the $C_{0}$--group $\{\tau
_{t}\}_{t\in {\mathbb{R}}}$ as a limit $L\rightarrow \infty $ of
finite--volume dynamics.

We thus need to define a family of Hamiltonians $H_{L}\in \mathcal{U}%
_{\Lambda _{L}}$ for $L\in \mathbb{R}_{0}^{+}$. This is done by using the
notions of \emph{interactions} and \emph{potentials} defined as follows:

\begin{itemize}
\item Interactions
\index{Interaction}are by definition families $\Psi =\{\Psi _{\Lambda
}\}_{\Lambda \in \mathcal{P}_{f}(\mathfrak{L})}$ of even (cf. (\ref%
{definition of even operators})) and self--adjoint local elements $\Psi
_{\Lambda }=\Psi _{\Lambda }^{\ast }\in \mathcal{U}^{+}\cap \mathcal{U}%
_{\Lambda }$ with $\Psi _{\emptyset }=0$. Obviously, the set of all
interactions can be endowed with a real vector space structure:
\begin{equation*}
\left( \alpha _{1}\Phi +\alpha _{2}\Psi \right) _{\Lambda }\doteq \alpha
_{1}\Phi _{\Lambda }+\alpha _{2}\Psi _{\Lambda }
\end{equation*}%
for any interactions $\Phi $, $\Psi $, and any real numbers $\alpha
_{1},\alpha _{2}\in \mathbb{R}$.

\item By
\index{Potential}potential, we mean here a collection $\mathbf{V}\doteq \{%
\mathbf{V}_{\left\{ x\right\} }\}_{x\in \mathfrak{L}}$ of even (cf. (\ref%
{definition of even operators})) and self--adjoint elements such that $%
\mathbf{V}_{\left\{ x\right\} }=\mathbf{V}_{\left\{ x\right\} }^{\ast }\in
\mathcal{U}^{+}\cap \mathcal{U}_{\left\{ x\right\} }$ for all $x\in
\mathfrak{L}$. Such objects are sometimes called
\index{On--site interaction}\emph{on--site interactions}. Indeed, strictly
speaking, a potential is nothing but a special case of interaction. But, the
use of this special notion allows us to treat latter the cases for which (%
\ref{condition divergence}) holds true.
\end{itemize}

Take now any interaction $\Psi $ and\emph{\ }potential $\mathbf{V}$. With
such objects we associate the (internal) energy observable or Hamiltonian%
\index{Interaction!local Hamiltonian}%
\index{Potential!local Hamiltonian}%
\index{Hamiltonian!H1-S2!many-fermion}
\begin{equation}
H_{L}\doteq \sum\limits_{\Lambda \subseteq \Lambda _{L}}\Psi _{\Lambda
}+\sum\limits_{x\in \Lambda _{L}}\mathbf{V}_{\left\{ x\right\} }\ ,\qquad
L\in \mathbb{R}_{0}^{+}\ ,  \label{definition fininte vol dynam0}
\end{equation}%
of the cubic box $\Lambda _{L}$ defined by (\ref{eq:def lambda n}).

Then, similar to Equations (\ref{automorphism})--(\ref{automorphism3}), the
finite--volume dynamics corresponds to the continuous group $\{\tau
_{t}^{(L)}\}_{t\in {\mathbb{R}}}$ of $\ast $--auto%
\-%
morphisms of $\mathcal{U}$ defined by%
\index{Interaction!finite--volume dynamics}%
\index{Potential!finite--volume dynamics}
\begin{equation}
\tau _{t}^{(L)}(B)=\mathrm{e}^{itH_{L}}B\mathrm{e}^{-itH_{L}}\ ,\qquad B\in
\mathcal{U}\ ,  \label{definition fininte vol dynam}
\end{equation}%
for any $L\in \mathbb{R}_{0}^{+}$, interaction $\Psi $ and potential $%
\mathbf{V}$. Obviously, its generator is the bounded linear operator $\delta
^{(L)}$ defined on $\mathcal{U}$ by%
\index{Interaction!symmetric derivation}%
\index{Potential!symmetric derivation}
\begin{equation}
\delta ^{(L)}(B)\doteq i\sum\limits_{\Lambda \subseteq \Lambda _{L}}\left[
\Psi _{\Lambda },B\right] +i\sum\limits_{x\in \Lambda _{L}}\left[ \mathbf{V}%
_{\left\{ x\right\} },B\right] \ ,\qquad B\in \mathcal{U}\ .
\label{dynamic series}
\end{equation}%
It is a symmetric derivation on $\mathcal{U}$ because, for all $%
B_{1},B_{2}\in \mathcal{U}$,%
\index{Symmetric derivation}
\begin{equation*}
\delta ^{(L)}(B_{1}^{\ast })=\delta ^{(L)}(B_{1})^{\ast }\quad
\text{and}\quad \delta ^{(L)}(B_{1}B_{2})=\delta
^{(L)}(B_{1})B_{2}+B_{1}\delta ^{(L)}(B_{2})\ .
\end{equation*}%
Compare with Equation (\ref{symmetric derivation}).

Using two functions $h,v:[0,\infty )\rightarrow \mathbb{R}$, note that the
finite--volume Hamiltonian (\ref{definition fininte vol dynam0}) associated
with the interaction%
\index{Interaction!example} $\Psi ^{\left( h,v\right) }$ defined by
\begin{eqnarray}
\Psi _{\Lambda }^{\left( h,v\right) } &\doteq &h\left( \left\vert
x-y\right\vert \right) a_{x}^{\ast }a_{y}+\left( 1-\delta _{x,y}\right)
h\left( \left\vert x-y\right\vert \right) a_{y}^{\ast }a_{x}
\label{de interactino example} \\
&&+v\left( \left\vert x-y\right\vert \right) \left( a_{y}^{\ast
}a_{y}a_{x}^{\ast }a_{x}+\left( 1-\delta _{x,y}\right) a_{x}^{\ast
}a_{x}a_{y}^{\ast }a_{y}\right)  \notag
\end{eqnarray}%
whenever $\Lambda =\left\{ x,y\right\} $ for $x,y\in \mathfrak{L}$, and $%
\Psi _{\Lambda }^{\left( h,v\right) }\doteq 0$ otherwise, is equal in this
case to%
\index{Hamiltonian!H1-S2!many-fermion}
\begin{equation*}
H_{L}=\sum\limits_{x,y\in \Lambda _{L}}h\left( \left\vert x-y\right\vert
\right) a_{x}^{\ast }a_{y}+\sum\limits_{x,y\in \Lambda _{L}}v(|x-y|)\mathbf{n%
}_{x}\mathbf{n}_{y}\ ,\qquad L\in \mathbb{R}_{0}^{+}\ .
\end{equation*}%
Compare with (\ref{H lambda}) for $\Lambda =\Lambda _{L}$. This gives a very
important -- albeit very specific -- example of a Fermi model on the
lattice. For instance, it includes the celebrated
\index{Hubbard model}Hubbard model widely used in Physics. Other examples
are given in Section \ref{Section Inter dis media}.

\subsection{Banach Spaces of Short--Range Interactions\label{Section Banach
space interaction}%
\index{Banach space of interactions}%
\index{Interaction!short-range}}

The finite--volume dynamics we define in Section \ref{Finite Volume Dynamics}
should converge to an infinite--volume one to be able to understand
macroscopic systems. In other words, the limit $L\rightarrow \infty $ of the
continuous group $\{\tau _{t}^{(L)}\}_{t\in {\mathbb{R}}}$ of $\ast $--auto%
\-%
morphisms defined by (\ref{definition fininte vol dynam}) has to converge to
a $C_{0}$--group $\{\tau _{t}\}_{t\in {\mathbb{R}}}$ of $\ast $%
--automorphisms of the CAR $C^{\ast }$--algebra $\mathcal{U}$. In order to
ensure that property (cf. Section \ref{sect Lieb--Robinson}), we define
Banach spaces of short--range interactions by introducing specific norms for
interactions, taking into account space decay.

Following \cite[Eqs. (1.3)--(1.4)]{NOS}, we consider positive--valued and
non--increasing
\index{Decay function}decay functions $\mathbf{F}:\mathbb{R}%
_{0}^{+}\rightarrow \mathbb{R}^{+}$ satisfying the following properties:

\begin{itemize}
\item \emph{Summability on }$\mathfrak{L}$\emph{.}
\begin{equation}
\left\Vert \mathbf{F}\right\Vert _{1,\mathfrak{L}}\doteq \underset{y\in
\mathfrak{L}}{\sup }\sum_{x\in \mathfrak{L}}\mathbf{F}\left( \left\vert
x-y\right\vert \right) =\sum_{x\in \mathfrak{L}}\mathbf{F}\left( \left\vert
x\right\vert \right) <\infty \ .  \label{(3.1) NS}
\end{equation}

\item \emph{Bounded convolution constant.}
\index{Convolution constant}%
\begin{equation}
\mathbf{D}\doteq \underset{x,y\in \mathfrak{L}}{\sup }\sum_{z\in \mathfrak{L}%
}%
\frac{\mathbf{F}\left( \left\vert x-z\right\vert \right) \mathbf{F}\left(
\left\vert z-y\right\vert \right) }{\mathbf{F}\left( \left\vert
x-y\right\vert \right) }<\infty \ .  \label{(3.2) NS}
\end{equation}
\end{itemize}

\noindent Note that the idea of a bounded convolution constant is
also used in \cite[cf. Assumption 2.1.]{Hastings}.

In the case $\mathfrak{L}$ would be a general countable set with infinite
cardinality and some metric $\mathrm{d}$, the existence of such a function $%
\mathbf{F}$ satisfying (\ref{(3.1) NS})--(\ref{(3.2) NS}) with $\mathrm{d}%
(\cdot ,\cdot )$ instead of $\left\vert \cdot -\cdot \right\vert $ refers to
the so--called \emph{regular} property%
\index{Lattice!regular} of $\mathfrak{L}$. For any $d\in \mathbb{N}$, $%
\mathfrak{L}\doteq \mathbb{Z}^{d}$ is in this sense regular with the metric $%
\mathrm{d}(\cdot ,\cdot )$ $=$ $\left\vert \cdot -\cdot \right\vert $.
Indeed, a typical example of such a $\mathbf{F}$ for $\mathfrak{L}=\mathbb{Z}%
^{d}$, $d\in \mathbb{N}$, and the metric induced by $\left\vert \cdot
\right\vert $ is the function%
\index{Decay function!polynomial decay}
\begin{equation}
\mathbf{F}\left( r\right) \doteq \left( 1+r\right) ^{-(d+\epsilon )}\
,\qquad r\in \mathbb{R}_{0}^{+}\ ,  \label{example polynomial}
\end{equation}%
which has convolution constant $\mathbf{D}\leq 2^{d+1+\epsilon }\left\Vert
\mathbf{F}\right\Vert _{1,\mathfrak{L}}$ for $\epsilon \in \mathbb{R}^{+}$.
See \cite[Eq. (1.6)]{NOS} or \cite[Example 3.1]{S}. Note that the
exponential function $\mathbf{F}\left( r\right) =\mathrm{e}^{-\varsigma r}$,
$\varsigma \in \mathbb{R}^{+}$, satisfies (\ref{(3.1) NS}) but not (\ref%
{(3.2) NS}). Nevertheless, for every function $\mathbf{F}$ with bounded
convolution constant (\ref{(3.2) NS}) and any strictly positive parameter $%
\varsigma \in \mathbb{R}^{+}$, the function%
\index{Decay function!exponential decay}
\begin{equation*}
\mathbf{%
\tilde{F}}\left( r\right) =\mathrm{e}^{-\varsigma r}\mathbf{F}\left(
r\right) \ ,\qquad r\in \mathbb{R}_{0}^{+}\ ,
\end{equation*}%
clearly satisfies Assumption (\ref{(3.2) NS}) with a convolution constant
that is no bigger than the one of $\mathbf{F}$. In fact, as observed in \cite%
[Section 3.1]{S}, the multiplication of such a\ function $\mathbf{F}$ with a
non--increasing weight $f:\mathbb{R}_{0}^{+}\rightarrow \mathbb{R}^{+}$
satisfying $f\left( r+s\right) \geq f\left( r\right) f\left( s\right) $
\index{Logarithmically superadditive function}(logarithmically superadditive
function) does not increase the convolution constant $\mathbf{D}$.\ In the
sequel, (\ref{(3.1) NS})--(\ref{(3.2) NS}) are assumed to be satisfied.

The function $\mathbf{F}$ encodes the short--range property of interactions.
Indeed, an interaction $\Psi $ is said to be \emph{short--range} if%
\index{Interaction!norm}
\begin{equation}
\left\Vert \Psi \right\Vert _{\mathcal{W}}\doteq \underset{x,y\in \mathfrak{L%
}}{\sup }\sum\limits_{\Lambda \in \mathcal{P}_{f}(\mathfrak{L}),\;\Lambda
\supset \{x,y\}}%
\frac{\Vert \Psi _{\Lambda }\Vert _{\mathcal{U}}}{\mathbf{F}\left(
\left\vert x-y\right\vert \right) }<\infty \ .  \label{iteration0}
\end{equation}%
Since the map $\Psi \mapsto \Vert \Psi \Vert _{\mathcal{W}}$ defines a norm
on interactions, the space of short--range interactions w.r.t. to the decay
function $\mathbf{F}$ is the real separable Banach space $\mathcal{W}\equiv (%
\mathcal{W},\Vert \cdot \Vert _{\mathcal{W}}\mathcal{)}$ of all interactions
$\Psi $ with $\Vert \Psi \Vert _{\mathcal{W}}<\infty $. Note that a
short--range interaction $\Psi \in \mathcal{W}$ is not necessarily weak away
from the origin of $\mathfrak{L}$: Generally, the element $\Psi _{x+\Lambda
} $, $x\in \mathfrak{L}$, does not vanish when $|x|\rightarrow \infty $. It
turns out that all short--range interactions $\Psi \in \mathcal{W}$ define,
in a natural way, infinite--volume quantum dynamics, i.e., they define $%
C^{\ast }$--dynamical systems on $\mathcal{U}$. For more details, see
Section \ref{sect Lieb--Robinson}, in particular Theorem \ref{Theorem
Lieb-Robinson copy(3)}. (Recall that $C^{\ast }$--dynamical systems are
defined in Section \ref{sect Algebraic Formulation of Quantum Mechanics}.)

\begin{bemerkung}[Lattice Fermi models]
\mbox{
}\newline
The
\index{Interaction!example}interaction $\Psi ^{\left( h,v\right) }$ defined
in Section \ref{Finite Volume Dynamics} (see (\ref{de interactino example}))
belongs to $\mathcal{W}$ as soon as $h,v:[0,\infty )\rightarrow \mathbb{R}$
are real--valued and summable functions satisfying%
\begin{equation}
\underset{r\in \mathbb{R}_{0}^{+}}{\sup }\left\{
\frac{\left\vert h\left( r\right) \right\vert }{\mathbf{F}\left( r\right) }%
\right\} <\infty \text{\qquad and\qquad }\underset{r\in \mathbb{R}_{0}^{+}}{%
\sup }\left\{ \frac{\left\vert v\left( r\right) \right\vert }{\mathbf{F}%
\left( r\right) }\right\} <\infty \ .  \label{sdsdsdsds}
\end{equation}
\end{bemerkung}

\begin{bemerkung}[Quantum spin models]
\mbox{
}\newline
\index{Quantum spin systems!models}All important spin models with no mean
field term can be constructed from short--range interactions, as defined
above. As examples, we can mention the Ising model, the (quantum) Heisenberg
model, the XXZ model, the XY model, the XXZ model, the model \cite{AKLT},
etc. See for instance \cite{N} and references therein.
\end{bemerkung}

\subsection{Existence of Dynamics and Lieb--Robinson Bounds\label{sect
Lieb--Robinson}}

In Section \ref{Section Banach space interaction}, we define a Banach space $%
\mathcal{W}$ of short--range interactions by using a convenient norm $\Vert
\cdot \Vert _{\mathcal{W}}$ for interactions, see (\ref{iteration0}). $\Psi
\in \mathcal{W}$ ensures the existence of an infinite--volume derivation $%
\delta $ associated with $\Psi $ by taking the
\index{Thermodynamic limit}thermodynamic limit $L\rightarrow \infty $ of
commutators involving $\Psi _{\Lambda }$, $\Lambda \in \mathcal{P}_{f}(%
\mathfrak{L})$, see (\ref{dynamic series}). This also holds true for all
potentials $\mathbf{V}\doteq \{\mathbf{V}_{\left\{ x\right\} }\}_{x\in
\mathfrak{L}}$, as defined in\ Section \ref{Finite Volume Dynamics}. Every
generator of a $C^{\ast }$--dynamical system is a derivation, but the
converse does not generally hold. We show here that $\delta $ is the
generator of a $C^{\ast }$--dynamical system in $\mathcal{U}$ when $\Psi \in
\mathcal{W}$ and for all potentials $\mathbf{V}\doteq \{\mathbf{V}_{\left\{
x\right\} }\}_{x\in \mathfrak{L}}$. Note that the interaction representing $%
\mathbf{V}$ can possibly be outside $\mathcal{W}$ because we allow $\mathbf{V%
}$ to be unbounded, i.e., the case%
\begin{equation}
\sup_{x\in \mathfrak{L}}\left\Vert \mathbf{V}_{\left\{ x\right\}
}\right\Vert _{\mathcal{U}}=\infty  \label{condition divergence}
\end{equation}%
is included in the discussion below.

The key ingredient in this analysis are the so--called \emph{Lieb--Robinson
bounds}. Indeed, they lead, among other things, to the existence of the
infinite--volume dynamics for interacting particles. By using this, we
define a $C^{\ast }$--dynamical system in $\mathcal{U}$ for any short--range
interaction $\Psi \in \mathcal{W}$. These bounds are, moreover, a pivotal
ingredient to study transport properties of interacting fermion systems
later on. Thus, for the reader's convenience, below we review this topic in
detail.

It is convenient to introduce at this point the notation
\begin{equation}
\mathcal{S}_{\Lambda }(%
\tilde{\Lambda})\doteq \left\{ \mathcal{Z}\subset \Lambda :\mathcal{Z}\cap
\tilde{\Lambda}\neq 0\text{ and }\mathcal{Z}\cap \tilde{\Lambda}^{c}\neq
0\right\}  \label{notation1}
\end{equation}%
for any set $\tilde{\Lambda}\subset \Lambda \subset \mathfrak{L}$ with
complement $\tilde{\Lambda}^{c}\doteq \mathfrak{L}\backslash \tilde{\Lambda}$%
, as well as
\begin{equation*}
\partial _{\Psi }\Lambda \doteq \left\{ x\in \Lambda :\exists \mathcal{Z}\in
\mathcal{S}_{\mathfrak{L}}(\Lambda )\text{ with }x\in \mathcal{Z}\text{ and }%
\Psi _{\mathcal{Z}}\neq 0\right\}
\end{equation*}%
for any interaction $\Psi \doteq \{\Psi _{\mathcal{Z}}\}_{\mathcal{Z}\in
\mathcal{P}_{f}(\mathfrak{L})}$ and any finite subset $\Lambda \in \mathcal{P%
}_{f}(\mathfrak{L})$ of $\mathfrak{L}$. We are now in position to prove
Lieb--Robinson bounds for finite--volume fermion systems with short--range
interactions and in presence of potentials:%
\index{Lieb--Robinson bounds}

\begin{satz}[Lieb--Robinson bounds]
\label{Theorem Lieb-Robinson}\mbox{
}\newline
Let $\Psi \in \mathcal{W}$ and $\mathbf{V}$ be any potential. Then, for any $%
t\in \mathbb{R}$, $L\in \mathbb{R}_{0}^{+}$, and elements $B_{1}\in \mathcal{%
U}^{+}\cap \mathcal{U}_{\Lambda ^{(1)}}$, $B_{2}\in \mathcal{U}_{\Lambda
^{(2)}}$ with $\Lambda ^{(1)},\Lambda ^{(2)}\in \mathcal{P}_{f}(\mathfrak{L}%
) $ and $\Lambda ^{(1)}\cap \Lambda ^{(2)}=\emptyset $,%
\index{Lieb--Robinson bounds}
\begin{eqnarray}
\left\Vert \lbrack \tau _{t}^{(L)}\left( B_{1}\right) ,B_{2}]\right\Vert _{%
\mathcal{U}} &\leq &2\mathbf{D}^{-1}\left\Vert B_{1}\right\Vert _{\mathcal{U}%
}\left\Vert B_{2}\right\Vert _{\mathcal{U}}\left( \mathrm{e}^{2\mathbf{D}%
\left\vert t\right\vert \left\Vert \Psi \right\Vert _{\mathcal{W}}}-1\right)
\label{LR bounds} \\
&&\times \sum_{x\in \partial _{\Psi }\Lambda ^{(1)}}\sum_{y\in \Lambda
^{(2)}}\mathbf{F}\left( \left\vert x-y\right\vert \right) \ .  \notag
\end{eqnarray}%
The constant $\mathbf{D}\in \mathbb{R}^{+}$ is defined by (\ref{(3.2) NS}).
\end{satz}

\begin{proof}
The arguments are essentially the same as those proving \cite[Theorem 2.3.]%
{NS} for quantum spin systems. Here, we consider fermion systems and we give
the detailed proof for completeness and to prepare its extension to
time--dependent interactions and potentials, in Theorem \ref{Theorem
Lieb-Robinsonnew} (i). We fix $L\in \mathbb{R}_{0}^{+}$, $B_{1}\in \mathcal{U%
}^{+}\cap \mathcal{U}_{\Lambda ^{(1)}}$ and $B_{2}\in \mathcal{U}_{\Lambda
^{(2)}}$ with disjoint sets $\Lambda ^{(1)},\Lambda ^{(2)}\subsetneq \Lambda
_{L}$. [Note that $\Lambda ^{(1)}\cap \Lambda ^{(2)}=\emptyset $ yields $%
L\geq 1$.]

Let
\begin{equation*}
C_{B_{2}}\left( \Lambda ;t\right) \doteq \underset{B\in \mathcal{U}^{+}\cap
\mathcal{U}_{\Lambda },B\neq 0}{\sup }%
\frac{\left\Vert [\tau _{t}^{(L)}\left( B\right) ,B_{2}]\right\Vert _{%
\mathcal{U}}}{\left\Vert B\right\Vert _{\mathcal{U}}}\ ,\qquad t\in {\mathbb{%
R}}\ ,\ \Lambda \in \mathcal{P}_{f}(\mathfrak{L})\ .
\end{equation*}%
At time $t=0$, we observe that%
\begin{equation*}
\left\vert C_{B_{2}}\left( \Lambda ;0\right) \right\vert \leq 2\left\Vert
B_{2}\right\Vert _{\mathcal{U}}\mathbf{1}\left[ \Lambda \cap \Lambda
^{(2)}\neq \emptyset \right] \ ,
\end{equation*}%
while, for any $t\in {\mathbb{R}}$,
\begin{equation*}
C_{B_{2}}\left( \Lambda ;t\right) =\underset{B\in \mathcal{U}^{+}\cap
\mathcal{U}_{\Lambda },B\neq 0}{\sup }\frac{\left\Vert [\tau _{t}^{(L)}\circ
\tau _{-t}^{(\Lambda )}\left( B\right) ,B_{2}]\right\Vert _{\mathcal{U}}}{%
\left\Vert B\right\Vert _{\mathcal{U}}}\ .
\end{equation*}%
Here, $\{\tau _{t}^{(\Lambda )}\}_{t\in {\mathbb{R}}}$ is the continuous
group of $\ast $--automorphisms defined like $\{\tau _{t}^{(L)}\}_{t\in {%
\mathbb{R}}}$ by replacing the box $\Lambda _{L}$ with the (finite) set $%
\Lambda \in \mathcal{P}_{f}(\mathfrak{L})$.

Consider the function%
\begin{equation}
f\left( t\right) \doteq \left[ \tau _{t}^{(L)}\circ \tau _{-t}^{(\Lambda
^{(1)})}\left( B_{1}\right) ,B_{2}\right] \ ,\qquad t\in {\mathbb{R}}\ .
\label{iteration1new}
\end{equation}%
Then, using $B_{1}\in \mathcal{U}^{+}\cap \mathcal{U}_{\Lambda ^{(1)}}$ and $%
\Lambda ^{(1)}\subset \Lambda _{L}$, we deduce from (\ref{dynamic series})
and explicit computations that%
\begin{eqnarray}
\partial _{t}f\left( t\right) &=&i\sum\limits_{\mathcal{Z}\in \mathcal{S}%
_{\Lambda _{L}}(\Lambda ^{(1)})}\left[ \tau _{t}^{(L)}\left( \Psi _{\mathcal{%
Z}}\right) ,f\left( t\right) \right]  \label{derivativenew} \\
&&-i\sum\limits_{\mathcal{Z}\in \mathcal{S}_{\Lambda _{L}}(\Lambda ^{(1)})}%
\left[ \tau _{t}^{(L)}\circ \tau _{-t}^{(\Lambda ^{(1)})}\left( B_{1}\right)
,\left[ \tau _{t}^{(L)}\left( \Psi _{\mathcal{Z}}\right) ,B_{2}\right] %
\right] \ .  \notag
\end{eqnarray}%
Let $\mathfrak{g}_{t}\left( B\right) $ be the solution of
\begin{equation*}
\forall t\geq 0:\qquad \partial _{t}\mathfrak{g}_{t}\left( B\right)
=i\sum\limits_{\mathcal{Z}\in \mathcal{S}_{\Lambda _{L}}(\Lambda
^{(1)})}[\tau _{t}^{(L)}\left( \Psi _{\mathcal{Z}}\right) ,\mathfrak{g}%
_{t}\left( B\right) ]\ ,\qquad \mathfrak{g}_{0}\left( B\right) =B\in
\mathcal{U}\ .
\end{equation*}%
Since $\Vert \mathfrak{g}_{t}\left( B\right) \Vert _{\mathcal{U}}=\Vert
B\Vert _{\mathcal{U}}$ for any $B\in \mathcal{U}$, it follows from (\ref%
{derivativenew}), by variation of constants, that%
\begin{equation}
\left\Vert f\left( t\right) \right\Vert _{\mathcal{U}}\leq \left\Vert
f\left( 0\right) \right\Vert _{\mathcal{U}}+2\left\Vert B_{1}\right\Vert _{%
\mathcal{U}}\sum\limits_{\mathcal{Z}\in \mathcal{S}_{\Lambda _{L}}(\Lambda
^{(1)})}\int_{0}^{\left\vert t\right\vert }\left\Vert \left[ \tau _{\pm
s}^{(L)}\left( \Psi _{\mathcal{Z}}\right) ,B_{2}\right] \right\Vert _{%
\mathcal{U}}\mathrm{d}s\text{ }.  \label{iteration2new0}
\end{equation}%
[The sign of $s$ in $\pm s$ depends whether $t$ is positive or negative.]
Hence, as $\Lambda ^{(1)},\Lambda ^{(2)}$ are disjoint, for any $t\in {%
\mathbb{R}}$,
\begin{equation}
C_{B_{2}}\left( \Lambda ^{(1)};t\right) \leq 2\sum\limits_{\mathcal{Z}\in
\mathcal{S}_{\Lambda _{L}}(\Lambda ^{(1)})}\left\Vert \Psi _{\mathcal{Z}%
}\right\Vert _{\mathcal{U}}\int_{0}^{\left\vert t\right\vert
}C_{B_{2}}\left( \mathcal{Z};\pm s\right) \mathrm{d}s\ .
\label{iteration2new}
\end{equation}%
By estimating $C_{B_{2}}\left( \mathcal{Z};s\right) $ in a similar manner
and iterating this procedure, we show that, for every $L\in \mathbb{R}%
_{0}^{+}$, $t\in {\mathbb{R}}$ and all $B_{1}\in \mathcal{U}^{+}\cap
\mathcal{U}_{\Lambda ^{(1)}}$, $B_{2}\in \mathcal{U}_{\Lambda ^{(2)}}$ with
disjoint $\Lambda ^{(1)},\Lambda ^{(2)}\subset \Lambda _{L}$,
\begin{equation}
C_{B_{2}}\left( \Lambda ^{(1)};t\right) \leq 2\left\Vert B_{2}\right\Vert _{%
\mathcal{U}}\sum_{k\in \mathbb{N}}\frac{\left\vert 2t\right\vert ^{k}}{k!}%
u_{k}\ ,  \label{iteration3bis}
\end{equation}%
where, for any $k\in \mathbb{N}$,
\begin{equation*}
u_{k}\doteq \sum\limits_{\mathcal{Z}_{1}\in \mathcal{S}_{\Lambda
_{L}}(\Lambda ^{(1)})}\sum\limits_{\mathcal{Z}_{2}\in \mathcal{S}_{\Lambda
_{L}}(\mathcal{Z}_{1})}\cdots \sum\limits_{\mathcal{Z}_{k}\in \mathcal{S}%
_{\Lambda _{L}}(\mathcal{Z}_{k-1})}\mathbf{1}\left[ \mathcal{Z}_{k}\cap
\Lambda ^{(2)}\neq \emptyset \right] \underset{j=1}{\overset{k}{\prod }}%
\left\Vert \Psi _{\mathcal{Z}_{j}}\right\Vert _{\mathcal{U}}\ .
\end{equation*}%
The above series is absolutely and uniformly convergent for $L\in \mathbb{R}%
_{0}^{+}$ (with fixed $\Lambda ^{(1)},\Lambda ^{(2)}\subsetneq \Lambda _{L}$%
). Indeed, from straightforward estimates,%
\begin{equation}
u_{k}\leq \mathbf{D}^{k-1}\left\Vert \Psi \right\Vert _{\mathcal{W}%
}^{k}\sum_{x\in \partial _{\Psi }\Lambda ^{(1)}}\sum_{y\in \Lambda ^{(2)}}%
\mathbf{F}\left( \left\vert x-y\right\vert \right) \ ,  \label{final1}
\end{equation}%
by Equations (\ref{(3.2) NS}) and (\ref{iteration0}).

Note that (\ref{iteration3bis})--(\ref{final1}) yield (\ref{LR bounds}),
provided $\Lambda ^{(1)},\Lambda ^{(2)}\subsetneq \Lambda _{L}$. This last
condition can easily be removed by taking, at any fixed $L\in \mathbb{R}%
_{0}^{+}$, an interaction $\tilde{\Psi}^{(L)}\in \mathcal{W}$ defined by $%
\tilde{\Psi}_{\mathcal{Z}}^{(L)}\doteq \Psi _{\mathcal{Z}}$ for any $%
\mathcal{Z}\subseteq \Lambda _{L}$, while $\tilde{\Psi}_{\mathcal{Z}%
}^{(L)}\doteq 0$ when $\mathcal{Z}\nsubseteq \Lambda _{L}$. Indeed, for all $%
L\in \mathbb{R}_{0}^{+}$, we obviously have $\Vert \tilde{\Psi}^{(L)}\Vert _{%
\mathcal{W}}\leq \Vert \Psi \Vert _{\mathcal{W}}$. Furthermore, for all $L,%
\tilde{L}\in \mathbb{R}_{0}^{+}$ with $\tilde{L}>L$, $\tilde{\tau}_{t}^{(%
\tilde{L})}=\tau _{t}^{(L)}$, where $\{\tilde{\tau}_{t}^{(\tilde{L}%
)}\}_{t\in {\mathbb{R}}}$ is the (finite--volume) group of $\ast $--auto%
\-%
morphisms of $\mathcal{U}$ defined by (\ref{definition fininte vol dynam})
with $L=\tilde{L}$ and $\Psi =\tilde{\Psi}^{(L)}$. Therefore, it suffices to
apply (\ref{iteration3bis})--(\ref{final1}) to the interaction $\tilde{\Psi}%
^{(L)}$ for sufficiently large $\tilde{L}\in \mathbb{R}_{0}^{+}$ in order to
get the assertion without the condition $\Lambda ^{(1)},\Lambda
^{(2)}\subsetneq \Lambda _{L}$.
\end{proof}

As explained in \cite[Theorem 3.1]{NS} for quantum spin systems,
Lieb--Robinson bounds lead to the existence of the infinite--volume dynamics:

\begin{lemma}[Finite--volume dynamics as a Cauchy sequence]
\label{Theorem Lieb-Robinson copy(2)}\mbox{ }\newline
Let $\Psi \in \mathcal{W}$ and $\mathbf{V}$ be any potential. Then, for any $%
t\in \mathbb{R}$, $\Lambda \in \mathcal{P}_{f}(\mathfrak{L})$, $B\in
\mathcal{U}_{\Lambda }$ and $L_{1},L_{2}\in \mathbb{R}_{0}^{+}$ with $%
\Lambda \subset \Lambda _{L_{1}}\subsetneq \Lambda _{L_{2}}$,%
\index{Interaction!finite--volume dynamics}%
\index{Potential!finite--volume dynamics}
\begin{eqnarray*}
\left\Vert \tau _{t}^{(L_{2})}\left( B\right) -\tau _{t}^{(L_{1})}\left(
B\right) \right\Vert _{\mathcal{U}} &\leq &2\left\Vert B\right\Vert _{%
\mathcal{U}}\Vert \Psi \Vert _{\mathcal{W}}\left\vert t\right\vert \mathrm{e}%
^{4\mathbf{D}\left\vert t\right\vert \left\Vert \Psi \right\Vert _{\mathcal{W%
}}} \\
&&\times \sum\limits_{y\in \Lambda _{L_{2}}\backslash \Lambda
_{L_{1}}}\sum_{x\in \Lambda }\mathbf{F}\left( \left\vert x-y\right\vert
\right) \ .
\end{eqnarray*}
\end{lemma}

\begin{proof}
Again, the arguments are those proving \cite[Theorem 3.1.]{NS} for quantum
spin systems. We give them for completeness, having also in mind the
extension of the lemma to time--dependent interactions and potentials, in
Theorem \ref{Theorem Lieb-Robinsonnew} (ii). We fix in all the proof $%
\Lambda \in \mathcal{P}_{f}(\mathfrak{L})$ and $B\in \mathcal{U}_{\Lambda }$.

For any $L\in \mathbb{R}_{0}^{+}$ and $s,t\in \mathbb{R}$, define the
unitary element%
\begin{equation}
\mathbf{U}_{L}\left( t,s\right) \doteq \mathrm{e}^{it\mathbf{V}_{\Lambda
_{L}}}\mathrm{e}^{-i\left( t-s\right) H_{L}}\mathrm{e}^{-is\mathbf{V}%
_{\Lambda _{L}}}\in \mathcal{U}_{\Lambda _{L}}  \label{unitary propagator}
\end{equation}%
with
\begin{equation*}
\mathbf{V}_{\mathcal{Z}}\doteq \sum\limits_{x\in \mathcal{Z}}\mathbf{V}%
_{\left\{ x\right\} }\in \mathcal{U}^{+}\cap \mathcal{U}_{\mathcal{Z}}\
,\qquad \mathcal{Z}\in \mathcal{P}_{f}(\mathfrak{L})\ .
\end{equation*}%
Clearly, $\mathbf{U}_{L}\left( t,t\right) =\mathbf{1}_{\mathcal{U}}$ for all
$t\in \mathbb{R}$ while%
\begin{equation*}
\partial _{t}\mathbf{U}_{L}\left( t,s\right) =-iG_{L}\left( t\right) \mathbf{%
U}_{L}\left( t,s\right)
\text{\quad and\quad }\partial _{s}\mathbf{U}_{L}\left( t,s\right) =i\mathbf{%
U}_{L}\left( t,s\right) G_{L}\left( s\right)
\end{equation*}%
with%
\begin{equation*}
G_{L}\left( t\right) \doteq \sum\limits_{\mathcal{Z}\subseteq \Lambda _{L}}%
\mathrm{e}^{it\mathbf{V}_{\Lambda _{L}}}\Psi _{\mathcal{Z}}\mathrm{e}^{-it%
\mathbf{V}_{\Lambda _{L}}}\ .
\end{equation*}%
Let%
\begin{equation*}
\tilde{\tau}_{t}^{(L)}\left( B\right) \doteq \mathbf{U}_{L}\left( 0,t\right)
B\mathbf{U}_{L}\left( t,0\right) \ ,\text{\qquad }B\in \mathcal{U}_{\Lambda
}\ .
\end{equation*}%
For any $t\in \mathbb{R}$ and $L\in \mathbb{R}_{0}^{+}$ such that $\Lambda
\subset \Lambda _{L}$,
\begin{equation*}
\tau _{t}^{(L)}\left( B\right) =\tilde{\tau}_{t}^{(L)}\left( \mathrm{e}^{it%
\mathbf{V}_{\Lambda _{L}}}B\mathrm{e}^{-it\mathbf{V}_{\Lambda _{L}}}\right) =%
\tilde{\tau}_{t}^{(L)}\left( \mathrm{e}^{it\mathbf{V}_{\Lambda }}B\mathrm{e}%
^{-it\mathbf{V}_{\Lambda }}\right)
\end{equation*}%
and it suffices to study the net $\{\tilde{\tau}_{t}^{(L)}\left( B\right)
\}_{L\in \mathbb{R}_{0}^{+}}$ in $\mathcal{U}$. The equality above is
related to the so--called \textquotedblleft interaction
picture\textquotedblright\ (w.r.t. potentials)\ of the time--evolution
defined by the $\ast $--automorphism $\tau _{t}^{(L)}$.

Fix $L_{1},L_{2}\in \mathbb{R}_{0}^{+}$ with $\Lambda \subset \Lambda
_{L_{1}}\varsubsetneq \Lambda _{L_{2}}$. Note that, for any $t\in \mathbb{R}$%
,%
\begin{equation}
\tilde{\tau}_{t}^{(L_{2})}\left( B\right) -\tilde{\tau}_{t}^{(L_{1})}\left(
B\right) =\int_{0}^{t}\partial _{s}\left\{ \mathbf{U}_{L_{2}}\left(
0,s\right) \mathbf{U}_{L_{1}}\left( s,t\right) B\mathbf{U}_{L_{1}}\left(
t,s\right) \mathbf{U}_{L_{2}}\left( s,0\right) \right\} \mathrm{d}s\ .\
\label{eqaulity dyna1}
\end{equation}%
Straightforward computations yield%
\begin{eqnarray}
&&\partial _{s}\left\{ \mathbf{U}_{L_{2}}\left( 0,s\right) \mathbf{U}%
_{L_{1}}\left( s,t\right) B\mathbf{U}_{L_{1}}\left( t,s\right) \mathbf{U}%
_{L_{2}}\left( s,0\right) \right\}  \notag \\
&=&i\mathbf{U}_{L_{2}}\left( 0,s\right)
\Big[%
G_{L_{2}}\left( s\right) -G_{L_{1}}\left( s\right) ,\mathbf{U}_{L_{1}}\left(
s,t\right) B\mathbf{U}_{L_{1}}\left( t,s\right)
\Big]%
\mathbf{U}_{L_{2}}\left( s,0\right)  \notag \\
&=&i\mathbf{U}_{L_{2}}\left( 0,s\right) \mathrm{e}^{is\mathbf{V}_{\Lambda
_{L_{1}}}}%
\Big[%
B_{s},\tau _{t-s}^{(L_{1})}(\tilde{B}_{t})%
\Big]%
\mathrm{e}^{-is\mathbf{V}_{\Lambda _{L_{1}}}}\mathbf{U}_{L_{2}}\left(
s,0\right) \ ,  \label{eqaulity dyna2}
\end{eqnarray}%
where, for any $s,t\in \mathbb{R}$, we define%
\begin{equation}
B_{s}\doteq \mathrm{e}^{-is\mathbf{V}_{\Lambda _{L_{1}}}}\left(
G_{L_{2}}\left( s\right) -G_{L_{1}}\left( s\right) \right) \mathrm{e}^{is%
\mathbf{V}_{\Lambda _{L_{1}}}}\quad \text{and}\quad \tilde{B}_{t}\doteq
\mathrm{e}^{-it\mathbf{V}_{\Lambda }}B\mathrm{e}^{it\mathbf{V}_{\Lambda }}\ .
\label{eqaulity dyna3}
\end{equation}%
Thus, we infer from Equations (\ref{eqaulity dyna1})--(\ref{eqaulity dyna3})
that
\begin{equation}
\left\Vert \tilde{\tau}_{t}^{(L_{2})}\left( B\right) -\tilde{\tau}%
_{t}^{(L_{1})}\left( B\right) \right\Vert _{\mathcal{U}}\leq
\int_{0}^{\left\vert t\right\vert }\left\Vert \left[ \tau _{\pm
s-t}^{(L_{1})}\left( B_{\pm s}\right) ,\tilde{B}_{t}\right] \right\Vert _{%
\mathcal{U}}\mathrm{d}s\ .  \label{inequality dyn5}
\end{equation}%
[The sign of $s$ in $\pm s$ depends whether $t$ is positive or negative.]
Note that $\tilde{B}_{t}\in \mathcal{U}_{\Lambda }$ and%
\begin{equation*}
B_{s}=\sum\limits_{\mathcal{Z}\subseteq \Lambda _{L_{2}},\ \mathcal{Z}\cap
(\Lambda _{L_{2}}\backslash \Lambda _{L_{1}})\neq \emptyset }\mathrm{e}^{is%
\mathbf{V}_{\Lambda _{L_{2}}\backslash \Lambda _{L_{1}}}}\Psi _{\mathcal{Z}}%
\mathrm{e}^{-is\mathbf{V}_{\Lambda _{L_{2}}\backslash \Lambda _{L_{1}}}}\in
\mathcal{U}^{+}\cap \mathcal{U}_{\Lambda _{L_{2}}}
\end{equation*}%
where, for any $\mathcal{Z}\subseteq \Lambda _{L_{2}}$,
\begin{equation*}
\mathrm{e}^{is\mathbf{V}_{\Lambda _{L_{2}}\backslash \Lambda _{L_{1}}}}\Psi
_{\mathcal{Z}}\mathrm{e}^{-is\mathbf{V}_{\Lambda _{L_{2}}\backslash \Lambda
_{L_{1}}}}\in \mathcal{U}_{\mathcal{Z}}\ .
\end{equation*}%
Now, we apply the Lieb--Robinson bounds given by Theorem \ref{Theorem
Lieb-Robinson} to deduce that, for any $\Lambda \in \mathcal{P}_{f}(%
\mathfrak{L})$, $s,t\in \mathbb{R}$, $B\in \mathcal{U}_{\Lambda }$ and $%
L_{1},L_{2}\in \mathbb{R}_{0}^{+}$ with $\Lambda \subset \Lambda
_{L_{1}}\varsubsetneq \Lambda _{L_{2}}$,%
\begin{eqnarray}
\frac{\left\Vert \left[ \tau _{s-t}^{(L_{1})}\left( B_{s}\right) ,\tilde{B}_{t}\right] \right\Vert _{\mathcal{U}}}{2\left\Vert B\right\Vert _{\mathcal{U}}} &\leq &\mathbf{D}^{-1}\left( \mathrm{e}^{2\mathbf{D}\left\vert
s-t\right\vert \left\Vert \Psi \right\Vert _{\mathcal{W}}}-1\right)
\label{toto inter 1} \\
&&\times \sum\limits_{{\mathcal{Z}\subseteq  \Lambda _{L_{2}},} \atop {\mathcal{Z}\cap (\Lambda _{L_{2}}\backslash \Lambda _{L_{1}})\neq \emptyset ,\
\mathcal{Z}\cap \Lambda =\emptyset }}\left\Vert \Psi _{\mathcal{Z}}\right\Vert _{\mathcal{U}}\sum_{z\in \partial _{\Psi }\mathcal{Z}}\sum_{x\in \Lambda }\mathbf{F}\left( \left\vert x-z\right\vert \right)
\notag \\
&&+\sum\limits_{{\mathcal{Z}\subseteq  \Lambda _{L_{2}},} \atop {\mathcal{Z}\cap
(\Lambda _{L_{2}}\backslash \Lambda _{L_{1}})\neq \emptyset ,\ \mathcal{Z}\cap \Lambda \neq \emptyset }}\left\Vert \Psi _{\mathcal{Z}}\right\Vert _{\mathcal{U}}\ .  \notag
\end{eqnarray}%
Direct estimates using (\ref{(3.2) NS}) and (\ref{iteration0}) show that%
\begin{eqnarray}
&&\sum\limits_{\mathcal{Z}\subseteq \Lambda _{L_{2}},\ \mathcal{Z}\cap
(\Lambda _{L_{2}}\backslash \Lambda _{L_{1}})\neq \emptyset }\left\Vert \Psi
_{\mathcal{Z}}\right\Vert _{\mathcal{U}}\sum_{z\in \partial _{\Psi }\mathcal{%
Z}}\sum_{x\in \Lambda }\mathbf{F}\left( \left\vert x-z\right\vert \right)
\notag \\
&\leq &\sum\limits_{y\in \Lambda _{L_{2}}\backslash \Lambda
_{L_{1}}}\sum\limits_{\mathcal{Z}\subseteq \Lambda _{L_{2}},\ \mathcal{Z}%
\supset \{y\}}\left\Vert \Psi _{\mathcal{Z}}\right\Vert _{\mathcal{U}%
}\sum_{z\in \mathcal{Z}}\sum_{x\in \Lambda }\mathbf{F}\left( \left\vert
x-z\right\vert \right)  \notag \\
&\leq &\sum\limits_{y\in \Lambda _{L_{2}}\backslash \Lambda
_{L_{1}}}\sum\limits_{z\in \Lambda _{L_{2}}}\sum\limits_{\mathcal{Z}%
\subseteq \Lambda _{L_{2}},\ \mathcal{Z}\supset \{y,z\}}\left\Vert \Psi _{%
\mathcal{Z}}\right\Vert _{\mathcal{U}}\sum_{x\in \Lambda }\mathbf{F}\left(
\left\vert x-z\right\vert \right)  \notag \\
&\leq &\Vert \Psi \Vert _{\mathcal{W}}\sum\limits_{y\in \Lambda
_{L_{2}}\backslash \Lambda _{L_{1}}}\sum_{x\in \Lambda }\sum\limits_{z\in
\Lambda _{L_{2}}}\mathbf{F}\left( \left\vert y-z\right\vert \right) \mathbf{F%
}\left( \left\vert x-z\right\vert \right)  \notag \\
&\leq &\mathbf{D}\Vert \Psi \Vert _{\mathcal{W}}\sum\limits_{y\in \Lambda
_{L_{2}}\backslash \Lambda _{L_{1}}}\sum_{x\in \Lambda }\mathbf{F}\left(
\left\vert x-y\right\vert \right) \ ,  \label{estimate utile}
\end{eqnarray}%
while, by using (\ref{iteration0}) only,
\begin{eqnarray}
&&\sum\limits_{\mathcal{Z}\subseteq \Lambda _{L_{2}},\ \mathcal{Z}\cap
(\Lambda _{L_{2}}\backslash \Lambda _{L_{1}})\neq \emptyset ,\ \mathcal{Z}%
\cap \Lambda \neq \emptyset }\left\Vert \Psi _{\mathcal{Z}}\right\Vert _{%
\mathcal{U}}  \notag \\
&\leq &\sum\limits_{y\in \Lambda _{L_{2}}\backslash \Lambda
_{L_{1}}}\sum_{x\in \Lambda }\sum\limits_{\mathcal{Z}\subseteq \Lambda
_{L_{2}},\ \mathcal{Z}\supset \{x,y\}}\left\Vert \Psi _{\mathcal{Z}%
}\right\Vert _{\mathcal{U}}  \notag \\
&\leq &\left\Vert \Psi \right\Vert _{\mathcal{W}}\sum\limits_{y\in \Lambda
_{L_{2}}\backslash \Lambda _{L_{1}}}\sum_{x\in \Lambda }\mathbf{F}\left(
\left\vert x-y\right\vert \right) \ .  \label{estimate utile2}
\end{eqnarray}%
The lemma is then a direct consequence of (\ref{inequality dyn5})--(\ref%
{toto inter 1}) combined with the upper bounds (\ref{estimate utile})--(\ref%
{estimate utile2}).
\end{proof}

The infinite--volume dynamics is obtained from Lemma \ref{Theorem
Lieb-Robinson copy(2)} and the completeness of $\mathcal{U}$. Indeed, from
the above lemma, for all $t\in {\mathbb{R}}$, $\tau _{t}^{(L)}$ converges
strongly on $\mathcal{U}_{0}$ to $\tau _{t}$, as $L\rightarrow \infty $. By
density of $\mathcal{U}_{0}$ in the Banach space $\mathcal{U}$ and the fact
that $\tau _{t}^{(L)}$ are isometries for all $L\in \mathbb{R}_{0}^{+}$ and $%
t\in {\mathbb{R}}$, the limit $\tau _{t}$, $t\in {\mathbb{R}}$, uniquely
defines a $\ast $--automorphism, also denoted by $\tau _{t}$, of the $%
C^{\ast }$--algebra $\mathcal{U}$. $\{\tau _{t}\}_{t\in {\mathbb{R}}}$ is
clearly a group of $\ast $--auto%
\-%
morphisms on $\mathcal{U}$. Again by the above lemma, for any element $B$ in
the dense subset $\mathcal{U}_{0}\subset $ $\mathcal{U}$, the convergence of
$\tau _{t}^{(L)}(B)$, as $L\rightarrow \infty $, is uniform for $t$ on
compacta and $\{\tau _{t}\}_{t\in {\mathbb{R}}}$ thus defines a $C_{0}$%
--group on $\mathcal{U}$, that is, a strongly continuous group on $\mathcal{U%
}$.

We need in the sequel an explicit characterization of the infinitesimal
generator of this $C_{0}$--group. Since the generator equals (\ref{dynamic
series}) at finite--volume, one expects that the infinitesimal generator
equals on $\mathcal{U}_{0}$ the linear map $\delta $ from $\mathcal{U}_{0}$
to $\mathcal{U}$ defined by%
\index{Interaction!symmetric derivation}%
\index{Potential!symmetric derivation}%
\begin{equation}
\delta (B)\doteq i\sum\limits_{\Lambda \in \mathcal{P}_{f}(\mathfrak{L})}%
\left[ \Psi _{\Lambda },B\right] +i\sum\limits_{x\in \mathfrak{L}}\left[
\mathbf{V}_{\left\{ x\right\} },B\right] \ ,\qquad B\in \mathcal{U}_{0}\ ,
\label{dynamic seriesbis}
\end{equation}%
for any $\Psi \in \mathcal{W}$ and potential $\mathbf{V}$. Indeed, for any $%
\Lambda \in \mathcal{P}_{f}(\mathfrak{L})$ and local element $B\in \mathcal{U%
}_{\Lambda }$,%
\begin{eqnarray}
&&\sum\limits_{\mathcal{Z}\in \mathcal{P}_{f}(\mathfrak{L})}\left\Vert \left[
\Psi _{\mathcal{Z}},B\right] \right\Vert _{\mathcal{U}}+\sum\limits_{x\in
\mathfrak{L}}\left\Vert \left[ \mathbf{V}_{\left\{ x\right\} },B\right]
\right\Vert _{\mathcal{U}}  \label{inequlity utile} \\
&\leq &2\left\Vert B\right\Vert _{\mathcal{U}}\left( \left\vert \Lambda
\right\vert \mathbf{F}\left( 0\right) \left\Vert \Psi \right\Vert _{\mathcal{%
W}}+\sum\limits_{x\in \Lambda }\left\Vert \mathbf{V}_{\left\{ x\right\}
}\right\Vert _{\mathcal{U}}\right)  \notag
\end{eqnarray}%
and the series (\ref{dynamic seriesbis}) is absolutely convergent for all $%
B\in \mathcal{U}_{0}$. Moreover, by (\ref{dynamic series}), we obviously
have
\begin{equation}
\delta (B)=\underset{L\rightarrow \infty }{\lim }\delta ^{(L)}(B)\ ,\qquad
B\in \mathcal{U}_{0}\ .  \label{core limit}
\end{equation}%
To prove that the closure of the linear map $\delta :\mathcal{U}%
_{0}\rightarrow \mathcal{U}$ is the generator of the $C_{0}$--group $\{\tau
_{t}\}_{t\in {\mathbb{R}}}$ of $\ast $--auto%
\-%
morphisms we use the second Trotter--Kato approximation theorem \cite[Chap.
III, Sect. 4.9]{EngelNagel}.

To this end, we first show that the (generally unbounded) operator $\delta $
on $\mathcal{U}$ with dense domain $\mathrm{Dom}(\delta )=\mathcal{U}_{0}$
is closable. Observe that both $\pm \delta $ are symmetric derivations and $%
\delta $ is thus conservative \cite[Definition 3.1.13.]{BratteliRobinsonI},
by structure of the set $\mathcal{U}_{0}$ of local elements:

\begin{lemma}[Conservative infinite--volume derivation]
\label{lemma dissipative}\mbox{ }\newline
Let $\Psi \in \mathcal{W}$ and $\mathbf{V}$ be any potential. Then, the
derivation $\delta $ defined on $\mathcal{U}_{0}$ by (\ref{dynamic seriesbis}%
) is a conservative symmetric derivation.
\end{lemma}

\begin{proof}
Let $B\in \mathcal{U}_{0}$ satisfying $B\geq 0$. By definition of $\mathcal{U%
}_{0}$, $B\in \mathcal{U}_{\Lambda }$ for some $\Lambda \in \mathcal{P}_{f}(%
\mathfrak{L})$. Since $\mathcal{U}_{\Lambda }$ is a unital $C^{\ast }$%
--algebra, there is $B^{1/2}\in \mathcal{U}_{\Lambda }\subset \mathcal{U}%
_{0} $ such that $B^{1/2}\geq 0$ and $(B^{1/2})^{2}=B$. Therefore, the lemma
follows from \cite[Proposition 3.2.22]{BratteliRobinsonI}.
\end{proof}

\noindent It follows that the symmetric derivation $\delta $ is (norm--)
closable:

\begin{lemma}[Closure of the infinite--volume derivation]
\label{lemma norm closure}\mbox{ }\newline
Let $\Psi \in \mathcal{W}$ and $\mathbf{V}$ be any potential. Then, the
derivations $\pm \delta $ defined on $\mathcal{U}_{0}$ by (\ref{dynamic
seriesbis}) are closable and their closures, again denoted for simplicity by
$\pm \delta $, are conservative.
\end{lemma}

\begin{proof}
$\pm \delta $ are densely defined dissipative operators on the Banach space $%
\mathcal{U}$. Therefore, the lemma is an obvious application of \cite[%
Proposition 3.1.15.]{BratteliRobinsonI}.
\end{proof}

In order to apply the second Trotter--Kato approximation theorem \cite[Chap.
III, Sect. 4.9]{EngelNagel}, we also prove that the range $\mathrm{Ran}\{(x%
\mathbf{1}_{\mathcal{U}}\mp \delta )\}$ of the closed operators $x\mathbf{1}%
_{\mathcal{U}}\mp \delta $ are dense in the Banach space $\mathcal{U}$ for $%
x>0$. This is done in the following lemma:

\begin{lemma}[Range of the infinite--volume derivation]
\label{lemma range derivation}\mbox{ }\newline
Let $\Psi \in \mathcal{W}$ and $\mathbf{V}$ be any potential. Then, for any $%
x\in {\mathbb{R}}^{+}$,%
\begin{equation*}
\mathcal{U}_{0}\subseteq \mathrm{Ran}\{(x\mathbf{1}_{\mathcal{U}}\mp \delta
)\}\subseteq \mathcal{U}
\end{equation*}%
with $\mathbf{1}_{\mathcal{U}}$ being the identity on $\mathcal{U}$. In
particular, $\mathrm{Ran}\{(x\mathbf{1}_{\mathcal{U}}\mp \delta )\}$ is
dense in $\mathcal{U}$.
\end{lemma}

\begin{proof}
We only give the proof for the range of the operator $x\mathbf{1_{\mathcal{U}%
}-}\delta $, since the other case uses similar arguments.

Note that $\Vert \tau _{t}^{(L)}\Vert _{\mathcal{B}(\mathcal{U})}=1$ for any
$L\in \mathbb{R}_{0}^{+}$ and $t\in {\mathbb{R}}$. Here, $\mathcal{B}(%
\mathcal{U})$ is the Banach space of bounded linear operators acting on $%
\mathcal{U}$. Thus, for any $L\in \mathbb{R}_{0}^{+}$, $x\in \mathbb{R}^{+}$%
, and $B\in \mathcal{U}$, the improper Riemann integral
\begin{equation*}
\int_{0}^{\infty }\mathrm{e}^{-xs}\tau _{s}^{(L)}\left( B\right) \mathrm{d}%
s\doteq \underset{t\rightarrow \infty }{\lim }\int_{0}^{t}\mathrm{e}%
^{-xs}\tau _{s}^{(L)}\left( B\right) \mathrm{d}s
\end{equation*}%
exists. By \cite[Chap. II, Sect. 1.10]{EngelNagel}, it follows that, for any
$L\in \mathbb{R}_{0}^{+}$ and $x\in \mathbb{R}^{+}$, the resolvent $(x%
\mathbf{1}_{\mathcal{U}}-\delta ^{(L)})^{-1}$ of the generator $\delta
^{(L)} $ of the group $\{\tau _{t}^{(L)}\}_{t\in {\mathbb{R}}}$ also exists
and satisfies
\begin{equation}
(x\mathbf{1}_{\mathcal{U}}\mathbf{-}\delta ^{(L)})^{-1}(B)=\int_{0}^{\infty }%
\mathrm{e}^{-xs}\tau _{s}^{(L)}\left( B\right) \mathrm{d}s\
\label{laplace transform derivation finite--volume}
\end{equation}%
for all $B\in \mathcal{U}$. Now, take $B\in \mathcal{U}_{0}$, $x\in {\mathbb{%
R}}^{+}$, and consider the element
\begin{equation}
B_{L}\doteq (x\mathbf{1}_{\mathcal{U}}\mathbf{-}\delta ^{(L)})^{-1}(B)\in
\mathcal{U}  \label{B_L}
\end{equation}%
for some sufficiently large parameter $L\in \mathbb{R}_{0}^{+}$ such that $%
B\in \mathcal{U}_{\Lambda _{L}}$. Note that $\tau _{s}^{(L)}(\mathcal{U}%
_{\Lambda _{L}})\subset \mathcal{U}_{\Lambda _{L}}$ and $B_{L}\in \mathcal{U}%
_{\Lambda _{L}}\subset \mathcal{U}_{0}$ because of (\ref{laplace transform
derivation finite--volume}). Then, we observe that
\begin{equation*}
(x\mathbf{1}_{\mathcal{U}}\mathbf{-}\delta )(B_{L})=B+(\delta ^{(L)}\mathbf{-%
}\delta )(B_{L})\ ,
\end{equation*}%
where we recall that $L\in \mathbb{R}_{0}^{+}$, $x\in \mathbb{R}^{+}$, and $%
B\in \mathcal{U}_{0}$. Now, by the Lumer--Phillips theorem \cite[Theorem
3.1.16]{BratteliRobinsonI} (see also its proof), if there is $x\in {\mathbb{R%
}}^{+}$ such that%
\begin{equation}
\underset{L\rightarrow \infty }{\lim }\left\Vert (\delta \mathbf{-}\delta
^{(L)})(B_{L})\right\Vert _{\mathcal{U}}=0  \label{assertion bis}
\end{equation}%
for all $B\in \mathcal{U}_{0}$ then we obtain the assertion. Indeed, by
using Lemma \ref{Theorem Lieb-Robinson copy(2)} together with $\Vert \tau
_{t}^{(L)}\Vert _{\mathcal{B}(\mathcal{U})}=1$ and (\ref{laplace transform
derivation finite--volume}), one verifies that $\{B_{L}\}_{L\in \mathbb{R}%
_{0}^{+}}$ is a Cauchy net, thus a convergent one in $\mathcal{U}$, while $x%
\mathbf{1}_{\mathcal{U}}-\delta $ is a closed operator, by Lemma \ref{lemma
norm closure}.

To prove (\ref{assertion bis}) we use Lieb--Robinson bounds (Theorem \ref%
{Theorem Lieb-Robinson}) as follows: Since $B_{L}\in \mathcal{U}_{\Lambda
_{L}}$ for sufficiently large $L\in \mathbb{R}_{0}^{+}$, we can combine (\ref%
{dynamic series}) and (\ref{dynamic seriesbis}) with (\ref{laplace transform
derivation finite--volume})--(\ref{B_L}) to compute that%
\begin{equation}
(\delta \mathbf{-}\delta ^{(L)})(B_{L})=i\sum\limits_{\mathcal{Z}\in
\mathcal{P}_{f}(\mathfrak{L}),\ \mathcal{Z}\cap \Lambda _{L}^{c}\neq
\emptyset }\int_{0}^{\infty }\mathrm{e}^{-xs}\left[ \Psi _{\mathcal{Z}},\tau
_{s}^{(L)}\left( B\right) \right] \mathrm{d}s  \label{inequ toto inter1}
\end{equation}%
for any $x\in {\mathbb{R}}^{+}$, sufficiently large $L\in \mathbb{R}_{0}^{+}$%
, and $B\in \mathcal{U}_{0}$. Here, $\Lambda _{L}^{c}\doteq \mathfrak{L}%
\backslash \Lambda _{L}$. It suffices to consider the case $B\neq 0$. Using
now Theorem \ref{Theorem Lieb-Robinson}, similar to (\ref{toto inter 1}),
one gets that, for all $s\in {\mathbb{R}}^{+}$ and any sufficiently large $%
L\in \mathbb{R}_{0}^{+}$ such that $B\in \mathcal{U}_{\Lambda }\subset
\mathcal{U}_{\Lambda _{L}}$ with $\Lambda \in \mathcal{P}_{f}(\mathfrak{L})$%
,
\begin{eqnarray}
&&\sum\limits_{\mathcal{Z}\in \mathcal{P}_{f}(\mathfrak{L}),\ \mathcal{Z}%
\cap \Lambda _{L}^{c}\neq \emptyset }%
\frac{\left\Vert \left[ \Psi _{\mathcal{Z}},\tau _{s}^{(L)}\left( B\right) %
\right] \right\Vert _{\mathcal{U}}}{2\left\Vert B\right\Vert _{\mathcal{U}}}
\label{assertion +} \\
&\leq &\mathbf{D}^{-1}\left( \mathrm{e}^{2\mathbf{D}\left\vert s\right\vert
\left\Vert \Psi \right\Vert _{\mathcal{W}}}-1\right) \sum\limits_{\mathcal{Z}%
\in \mathcal{P}_{f}(\mathfrak{L}),\ \mathcal{Z}\cap \Lambda _{L}^{c}\neq
\emptyset ,\ \mathcal{Z}\cap \Lambda =\emptyset }\left\Vert \Psi _{\mathcal{Z%
}}\right\Vert _{\mathcal{U}}\sum_{x\in \partial _{\Psi }\mathcal{Z}%
}\sum_{y\in \Lambda }\mathbf{F}\left( \left\vert x-y\right\vert \right)
\notag \\
&&+\sum\limits_{\mathcal{Z}\in \mathcal{P}_{f}(\mathfrak{L}),\ \mathcal{Z}%
\cap \Lambda _{L}^{c}\neq \emptyset ,\ \mathcal{Z}\cap \Lambda \neq
\emptyset }\left\Vert \Psi _{\mathcal{Z}}\right\Vert _{\mathcal{U}}\ .
\notag
\end{eqnarray}%
Similar to Inequalities (\ref{estimate utile})--(\ref{estimate utile2}), we
thus infer from (\ref{(3.2) NS}) and (\ref{iteration0}) that%
\begin{equation}
\sum\limits_{\mathcal{Z}\in \mathcal{P}_{f}(\mathfrak{L}),\ \mathcal{Z}\cap
\Lambda _{L}^{c}\neq \emptyset }\frac{\left\Vert \left[ \Psi _{\mathcal{Z}%
},\tau _{s}^{(L)}\left( B\right) \right] \right\Vert _{\mathcal{U}}}{%
2\left\Vert B\right\Vert _{\mathcal{U}}}\leq \Vert \Psi \Vert _{\mathcal{W}%
}\ \mathrm{e}^{2\mathbf{D}\left\vert s\right\vert \left\Vert \Psi
\right\Vert _{\mathcal{W}}}\sum\limits_{y\in \Lambda _{L}^{c}}\sum_{x\in
\Lambda }\mathbf{F}\left( \left\vert x-y\right\vert \right) \ ,
\label{assertion bisbisbis}
\end{equation}%
while
\begin{equation}
\underset{L\rightarrow \infty }{\lim }\sum\limits_{y\in \Lambda
_{L}^{c}}\sum_{x\in \Lambda }\mathbf{F}\left( \left\vert x-y\right\vert
\right) =0\ ,  \label{assertion bisbisbisbis}
\end{equation}%
because of (\ref{(3.1) NS}). Therefore, by (\ref{inequ toto inter1})--(\ref%
{assertion bisbisbisbis}), we deduce (\ref{assertion bis}) for all $x>2%
\mathbf{D}\Vert \Psi \Vert _{\mathcal{W}}$ and $B\in \mathcal{U}_{0}$.
\end{proof}

We now apply the second Trotter--Kato approximation theorem \cite[Chap. III,
Sect. 4.9]{EngelNagel} to deduce that $\delta $ is the generator of the
group $\{\tau _{t}\}_{t\in {\mathbb{R}}}$ of $\ast $--auto%
\-%
morphisms and resume all the main results, so far, in the following theorem:

\begin{satz}[Infinite--volume dynamics and its generator]
\label{Theorem Lieb-Robinson copy(3)}\mbox{
}\newline
Let $\Psi \in \mathcal{W}$, $\mathbf{V}$ be any potential, and $\mathbf{D}%
\in \mathbb{R}^{+}$ be defined by (\ref{(3.2) NS}).\newline
\emph{(i)}
\index{Interaction!infinite--volume dynamics}%
\index{Potential!infinite--volume dynamics}Infinite--volume dynamics. The
continuous groups $\{\tau _{t}^{(L)}\}_{t\in {\mathbb{R}}}$, $L\in \mathbb{R}%
_{0}^{+}$, defined by (\ref{definition fininte vol dynam}) converge strongly
to a $C_{0}$--group $\{\tau _{t}\}_{t\in {\mathbb{R}}}$ of $\ast $--auto%
\-%
morphisms with generator $\delta $.\newline
\emph{(ii)}
\index{Interaction!symmetric derivation}%
\index{Potential!symmetric derivation}Infinitesimal generator. $\delta $ is
a conservative closed symmetric derivation which is equal on its core $%
\mathcal{U}_{0}$ to
\begin{equation*}
\delta (B)=i\sum\limits_{\Lambda \in \mathcal{P}_{f}(\mathfrak{L})}\left[
\Psi _{\Lambda },B\right] +i\sum\limits_{x\in \mathfrak{L}}\left[ \mathbf{V}%
_{\left\{ x\right\} },B\right] \ ,\qquad B\in \mathcal{U}_{0}\ .
\end{equation*}%
\emph{(iii)} Rate of convergence. For any $\Lambda \in \mathcal{P}_{f}(%
\mathfrak{L})$, $B\in \mathcal{U}_{\Lambda }$ and $L\in \mathbb{R}_{0}^{+}$
such that $\Lambda \subset \Lambda _{L}$,%
\begin{equation*}
\left\Vert \tau _{t}\left( B\right) -\tau _{t}^{(L)}\left( B\right)
\right\Vert _{\mathcal{U}}\leq 2\left\Vert B\right\Vert _{\mathcal{U}}\Vert
\Psi \Vert _{\mathcal{W}}\left\vert t\right\vert \mathrm{e}^{4\mathbf{D}%
\left\vert t\right\vert \left\Vert \Psi \right\Vert _{\mathcal{W}%
}}\sum\limits_{y\in \mathfrak{L}\backslash \Lambda _{L}}\sum_{x\in \Lambda }%
\mathbf{F}\left( \left\vert x-y\right\vert \right) \ .
\end{equation*}%
\emph{(iv)}
\index{Lieb--Robinson bounds}Lieb--Robinson bounds. For any $t\in \mathbb{R}$
and $B_{1}\in \mathcal{U}^{+}\cap \mathcal{U}_{\Lambda ^{(1)}}$, $B_{2}\in
\mathcal{U}_{\Lambda ^{(2)}}$ with disjoint sets $\Lambda ^{(1)},\Lambda
^{(2)}\in \mathcal{P}_{f}(\mathfrak{L})$,
\begin{multline*}
\left\Vert \left[ \tau _{t}\left( B_{1}\right) ,B_{2}\right] \right\Vert _{%
\mathcal{U}}\leq 2\mathbf{D}^{-1}\left\Vert B_{1}\right\Vert _{\mathcal{U}%
}\left\Vert B_{2}\right\Vert _{\mathcal{U}}\left( \mathrm{e}^{2\mathbf{D}%
\left\vert t\right\vert \left\Vert \Psi \right\Vert _{\mathcal{W}}}-1\right)
\\
\times \sum_{x\in \partial _{\Psi }\Lambda ^{(1)}}\sum_{y\in \Lambda ^{(2)}}%
\mathbf{F}\left( \left\vert x-y\right\vert \right) \ .
\end{multline*}
\end{satz}

\begin{proof}
By Lemma \ref{lemma norm closure}, the set $\mathcal{U}_{0}$ of local
elements is a core of the dissipative derivation $\delta $ and one obtains
(ii), see (\ref{dynamic seriesbis}). Moreover, $\delta ^{(L)}\left( B\right)
\rightarrow \delta \left( B\right) $ for all $B\in \mathcal{U}_{0}$, see (%
\ref{core limit}). Recall that $\delta ^{(L)}$ is the generator of the group
$\{\tau _{t}^{(L)}\}_{t\in {\mathbb{R}}}$ for any $L\in \mathbb{R}_{0}^{+}$.
Therefore, since one also has Lemma \ref{lemma range derivation}, (i) is a
direct consequence of \cite[Chap. III, Sect. 4.9]{EngelNagel}. The third
statement (iii) thus follows from Lemma \ref{Theorem Lieb-Robinson copy(2)}.
(iv) is an obvious consequence of Theorem \ref{Theorem Lieb-Robinson} and
the first assertion (i).
\end{proof}

\subsection{Lieb--Robinson Bounds for Multi--Commutators\label{section LR
multi}}

Recall that multi--commutators are defined by induction as follows:%
\index{Multi--commutators}%
\begin{equation}
{[}B_{1},B_{0}{]}^{(2)}\doteq \lbrack B_{1},B_{0}]\doteq
B_{1}B_{0}-B_{0}B_{1}\ ,\qquad B_{0},B_{1}\in \mathcal{U}\ ,
\label{multi1-0}
\end{equation}%
and, for all integers $k\geq 2$,
\begin{equation}
{[}B_{k},B_{k-1},\ldots ,B_{0}{]}^{(k+1)}\doteq {[}B_{k},{[}B_{k-1},\ldots
,B_{0}{]}^{(k)}{]}\ ,\quad B_{0},\ldots ,B_{k}\in \mathcal{U}\ .
\label{multi2-0}
\end{equation}%
The aim of this subsection is to extend Theorem \ref{Theorem Lieb-Robinson
copy(3)} (iv) to multi--com%
\-%
mutators. The arguments we use below to prove Lieb--Robinson bounds for
multi--com%
\-%
mutators are not a generalization of the proof of Theorem \ref{Theorem
Lieb-Robinson} or Theorem \ref{Theorem Lieb-Robinson copy(3)} (iv). Instead,
we use a pivotal lemma deduced from Theorem \ref{Theorem Lieb-Robinson
copy(3)} (iii), which in turn results from finite--volume Lieb--Robinson
bounds of Theorem \ref{Theorem Lieb-Robinson}. This lemma expresses the $%
C_{0}$--group $\{\tau _{t}\}_{t\in {\mathbb{R}}}$ of Theorem \ref{Theorem
Lieb-Robinson copy(3)} (i) as \emph{telescoping} series.

To this end, it is convenient to introduce the family $\{\chi _{x}\}_{x\in
\mathfrak{L}}$ of $\ast $--automor%
\-%
phisms of $\mathcal{U}$, which implements the action of the
\index{CAR algebras!group of lattice translations}group of lattice
translations on the CAR $C^{\ast }$--algebra $\mathcal{U}$. This family is
uniquely defined by the conditions
\begin{equation}
\chi _{x}(a_{y})=a_{y+x}\ ,\qquad x,y\in \mathfrak{L}\ .  \label{translation}
\end{equation}%
We also define, for any $n\in \mathbb{N}_{0}$, $x\in \mathfrak{L}$, $\Psi
\in \mathcal{W}$ and potential $\mathbf{V}$, a \emph{space translated}
finite--volume dynamics which is the continuous group $\{\tau
_{t}^{(n,x)}\}_{t\in {\mathbb{R}}}$ of $\ast $--auto%
\-%
morphisms of $\mathcal{U}$ generated by the symmetric and bounded derivation
\begin{equation*}
\delta ^{(n,x)}(B)\doteq i\sum\limits_{\Lambda \subseteq x+\Lambda _{n}}
\left[ \Psi _{\Lambda },B\right] +i\sum\limits_{y\in x+\Lambda _{n}}\left[
\mathbf{V}_{\left\{ y\right\} },B\right] \ ,\qquad B\in \mathcal{U}\ .
\end{equation*}%
Note that the fermion system is generally \emph{not} translation invariant
and, in general,%
\begin{equation*}
\tau _{t}^{(n,x)}\circ \chi _{x}\neq \chi _{x}\circ \tau _{t}^{(n)}\ ,\qquad
x\in \mathfrak{L},\ n\in \mathbb{N}_{0}\ ,\ t\in {\mathbb{R}}\ .
\end{equation*}%
For $m\in \mathbb{N}_{0}$, $x\in \mathfrak{L}$, $B\in \mathcal{U}_{\Lambda
_{m}}$ and $t\in \mathbb{R}$, we finally introduce the local elements
\begin{equation}
\mathfrak{B}_{B,t,x}\left( m\right) \equiv \mathfrak{B}_{B,t,x}^{\left(
m\right) }\left( m\right) \doteq \tau _{t}^{(m,x)}\circ \chi _{x}\left(
B\right) \in \mathcal{U}_{\Lambda _{m}+x}  \label{def B frac1}
\end{equation}%
and
\begin{equation}
\mathfrak{B}_{B,t,x}\left( n\right) \equiv \mathfrak{B}_{B,t,x}^{\left(
m\right) }\left( n\right) \doteq (\tau _{t}^{(n,x)}-\tau
_{t}^{(n-1,x)})\circ \chi _{x}(B)\in \mathcal{U}_{\Lambda _{n}+x}\ ,\qquad
n\geq m+1\ .  \label{def B frac2}
\end{equation}%
The family $\{\mathfrak{B}_{B,t,x}\left( n\right) \}_{n\geq m}\subset
\mathcal{U}_{0}$ is used to define telescoping series:%
\index{Telescoping series}

\begin{lemma}[Infinite--volume dynamics as telescoping series]
\label{Lemma Series representation of the dynamics}\mbox{ }\newline
Let $\Psi \in \mathcal{W}$ and $\mathbf{V}$ be any potential. Then, for any $%
m\in \mathbb{N}_{0}$, $x\in \mathfrak{L}$, $B\in \mathcal{U}_{\Lambda _{m}}$
and $t\in \mathbb{R}$:
\begin{equation}
\sum_{n=m}^{\infty }\mathfrak{B}_{B,t,x}\left( n\right) =\tau _{t}\circ \chi
_{x}\left( B\right) \ .  \label{bound B fract0}
\end{equation}%
The above telescoping series is absolutely convergent in $\mathcal{U}$ with
\begin{equation}
\left\Vert \mathfrak{B}_{B,t,x}\left( n\right) \right\Vert _{\mathcal{U}%
}\leq 2\left\Vert B\right\Vert _{\mathcal{U}}\left\Vert \Psi \right\Vert _{%
\mathcal{W}}\left\vert t\right\vert \mathrm{e}^{4\mathbf{D}\left\vert
t\right\vert \left\Vert \Psi \right\Vert _{\mathcal{W}}}\sum\limits_{y\in
\Lambda _{n}\backslash \Lambda _{n-1}}\sum_{z\in \Lambda _{m}}\mathbf{F}%
\left( \left\vert z-y\right\vert \right)  \label{bound B fract}
\end{equation}%
for any $n\geq m+1$, while $\left\Vert \mathfrak{B}_{B,t,x}\left( m\right)
\right\Vert _{\mathcal{U}}=\left\Vert B\right\Vert _{\mathcal{U}}$.
\end{lemma}

\begin{proof}
Since, for any $N\in \mathbb{N}_{0}$ so that $N\geq m$,
\begin{equation}
\sum_{n=m}^{N}\mathfrak{B}_{B,t,x}\left( n\right) =\tau _{t}^{(N,x)}\circ
\chi _{x}\left( B\right) \ ,  \label{telescoping series}
\end{equation}%
it suffices to study the limit $N\rightarrow \infty $ of the group $\{\tau
_{t}^{(N,x)}\}_{t\in {\mathbb{R}}}$ at any fixed $x\in \mathfrak{L}$.
Similar to the proof of Theorem \ref{Theorem Lieb-Robinson copy(3)} (i), $%
\delta ^{(N,x)}\left( B\right) \rightarrow \delta \left( B\right) $ for all $%
B\in \mathcal{U}_{0}$, as $N\rightarrow \infty $. By Lemma \ref{lemma range
derivation} and \cite[Chap. III, Sect. 4.9]{EngelNagel}, the translated
groups $\{\tau _{t}^{(N,x)}\}_{t\in {\mathbb{R}}}$, $N\in \mathbb{N}_{0}$,
converge strongly to the $C_{0}$--group $\{\tau _{t}\}_{t\in {\mathbb{R}}}$
for any $x\in \mathfrak{L}$. In other words, we deduce Equation (\ref{bound
B fract0}) from (\ref{telescoping series}) in the limit $N\rightarrow \infty
$. Moreover, one easily checks that Theorem \ref{Theorem Lieb-Robinson} and
thus Lemma \ref{Theorem Lieb-Robinson copy(2)} also hold for the (space
translated) groups $\{\tau _{t}^{(n,x)}\}_{t\in {\mathbb{R}}}$, $n\in
\mathbb{N}_{0}$, at any fixed $x\in \mathfrak{L}$. This yields Inequality (%
\ref{bound B fract}) for $n>m$, while $\left\Vert \mathfrak{B}_{B,t,x}\left(
m\right) \right\Vert _{\mathcal{U}}=\left\Vert B\right\Vert _{\mathcal{U}}$,
because $\tau _{t}^{(m,x)}$ is a $\ast $--auto%
\-%
morphism on $\mathcal{U}_{\Lambda _{m}}$. It follows that%
\begin{equation*}
\sum_{n=m+1}^{\infty }\left\Vert \mathfrak{B}_{B,t,x}\left( n\right)
\right\Vert _{\mathcal{U}}\leq 2\left\Vert B\right\Vert _{\mathcal{U}%
}\left\Vert \Psi \right\Vert _{\mathcal{W}}\left\vert t\right\vert \mathrm{e}%
^{4\mathbf{D}\left\vert t\right\vert \left\Vert \Psi \right\Vert _{\mathcal{W%
}}}\sum_{z\in \Lambda _{m}}\sum_{n\in \mathbb{N}}\sum\limits_{y\in \Lambda
_{n}\backslash \Lambda _{n-1}}\mathbf{F}\left( \left\vert z-y\right\vert
\right) \ .
\end{equation*}%
Finally, by Assumption (\ref{(3.1) NS}),
\begin{equation*}
\sum_{z\in \Lambda _{m}}\sum_{n\in \mathbb{N}}\sum\limits_{y\in \Lambda
_{n}\backslash \Lambda _{n-1}}\mathbf{F}\left( \left\vert z-y\right\vert
\right) \leq \sum_{z\in \Lambda _{m}}\sum_{y\in \mathfrak{L}}\mathbf{F}%
\left( \left\vert z-y\right\vert \right) =\left\vert \Lambda _{m}\right\vert
\left\Vert \mathbf{F}\right\Vert _{1,\mathfrak{L}}<\infty \ .
\end{equation*}
\end{proof}

To extend Lieb--Robinson bounds to multi--commutators we combine Lemma \ref%
{Lemma Series representation of the dynamics} with tree decompositions of
sequences of clustering subsets of $\mathfrak{L}$ (cf. (\ref{cluster-tree}%
)): Let $\mathcal{T}_{2}$ be the set of all (non--oriented) trees with
exactly two vertices. This set contains a unique tree $T=\{\{0,1\}\}$ which,
in turn, contains the unique bond $\{0,1\}$, i.e., $\mathcal{T}_{2}\doteq
\left\{ \left\{ \{0,1\}\right\} \right\} $.
\index{Trees}Then, for each integer $k\geq 2$, we recursively define a set $%
\mathcal{T}_{k+1}$ of trees with $k+1$ vertices by
\begin{equation}
\mathcal{T}_{k+1}\doteq \Big\{\{\{j,k\}\}\cup T%
\text{ }:\text{ }j=0,\ldots ,k-1,\quad T\in \mathcal{T}_{k}\Big\}\ .
\label{def.tree}
\end{equation}%
Therefore, for $k\in \mathbb{N}$ and any tree $T\in \mathcal{T}_{k+1}$,
there is a map%
\begin{equation}
\mathrm{P}_{T}:\{1,\ldots ,k\}\rightarrow \{0,\ldots ,k-1\}
\label{def tree2}
\end{equation}%
such that $\mathrm{P}_{T}(j)<j$, $\mathrm{P}_{T}(1)=0$, and
\begin{equation}
T=\bigcup_{j=1}^{k}\{\{\mathrm{P}_{T}(j),j\}\}\ .  \label{def tree3}
\end{equation}%
For any $k\in \mathbb{N}$, $T\in \mathcal{T}_{k+1}$, and every sequence $%
\left\{ (n_{j},x_{j})\right\} _{j=0}^{k}$ in $\mathbb{N}_{0}\times \mathfrak{%
L}$ with length $k+1$, we define
\begin{equation}
\varkappa _{T}\left( \left\{ (n_{j},x_{j})\right\} _{j=0}^{k}\right) \doteq
\overset{k}{\prod_{j=1}}\mathbf{1}\left[ (\Lambda _{n_{j}}+x_{j})\cap
(\Lambda _{n_{\mathrm{P}_{T}(j)}}+x_{\mathrm{P}_{T}(j)})\neq \emptyset %
\right] \in \{0,1\}\ ,  \label{definition chqrqcteristic}
\end{equation}%
while, for all $\ell \in \left\{ 1,\ldots ,k\right\} $,%
\begin{equation}
\mathcal{S}_{\ell ,k}\doteq \left\{ \pi \text{ }|\text{ }\pi :\left\{ \ell
,\ldots ,k\right\} \rightarrow \left\{ 1,\ldots ,k\right\} \text{ such that }%
\pi \left( i\right) <\pi \left( j\right) \text{ when }i<j\right\} \ .
\label{SLK}
\end{equation}%
Then, one gets the following bound on multi--commutators:

\begin{satz}[Lieb--Robinson bounds for multi--commutators -- Part I]
\label{Theorem Lieb-Robinson copy(1)}\mbox{
}\newline
Let $\Psi \in \mathcal{W}$ and $\mathbf{V}$ be any potential. Then, for any
integer $k\in \mathbb{N}$, $\{m_{j}\}_{j=0}^{k}\subset \mathbb{N}_{0}$,
times $\{s_{j}\}_{j=1}^{k}\subset \mathbb{R}$, lattice sites $%
\{x_{j}\}_{j=0}^{k}\subset \mathfrak{L}$, and local elements $B_{0}\in
\mathcal{U}_{0}$, $\{B_{j}\}_{j=1}^{k}\subset \mathcal{U}_{0}\cap \mathcal{U}%
^{+}$ such that $B_{j}\in \mathcal{U}_{\Lambda _{m_{j}}}$ for $j\in
\{0,\ldots ,k\}$,%
\index{Lieb--Robinson bounds!multi--commutators}%
\begin{eqnarray*}
&&\left\Vert \left[ \tau _{s_{k}}\circ \chi _{x_{k}}(B_{k}),\ldots ,\tau
_{s_{1}}\circ \chi _{x_{1}}(B_{1}),\chi _{x_{0}}(B_{0})\right]
^{(k+1)}\right\Vert _{\mathcal{U}} \\
&\leq &2^{k}\prod\limits_{j=0}^{k}\left\Vert B_{j}\right\Vert _{\mathcal{U}%
}\sum_{T\in \mathcal{T}_{k+1}}\left( \varkappa _{T}\left( \left\{
(m_{j},x_{j})\right\} _{j=0}^{k}\right) +\Re _{T,\Vert \Psi \Vert _{\mathcal{%
W}}}\right)
\end{eqnarray*}%
with, for any $\alpha \in \mathbb{R}_{0}^{+}$,
\begin{eqnarray}
\Re _{T,\alpha } &\doteq &\sum_{\ell =1}^{k}\left( 2\alpha \right) ^{k-\ell
+1}\sum_{\pi \in \mathcal{S}_{\ell ,k}}\left( \prod\limits_{j\in \{\pi (\ell
),\ldots ,\pi (k)\}}\left\vert s_{j}\right\vert \mathrm{e}^{4\mathbf{D}%
\alpha \left\vert s_{j}\right\vert }\right)  \label{def AT} \\
&&\sum_{n_{\pi \left( \ell \right) }=m_{\pi \left( \ell \right) }+1}^{\infty
}\sum_{z_{\pi \left( \ell \right) }\in \Lambda _{m_{\pi \left( \ell \right)
}}}\sum\limits_{y_{\pi \left( \ell \right) }\in \Lambda _{n_{\pi \left( \ell
\right) }}\backslash \Lambda _{n_{\pi \left( \ell \right) }-1}}\cdots  \notag
\\
&&\cdots \sum_{n_{\pi \left( k\right) }=m_{\pi \left( k\right) }+1}^{\infty
}\sum_{z_{\pi \left( k\right) }\in \Lambda _{m_{\pi \left( k\right)
}}}\sum\limits_{y_{\pi \left( k\right) }\in \Lambda _{n_{\pi \left( k\right)
}}\backslash \Lambda _{n_{\pi \left( k\right) }-1}}  \notag \\
&&\quad \quad \quad \quad \varkappa _{T}\left( \left\{ (n_{j},x_{j})\right\}
_{j=0}^{k}\right) \prod\limits_{j\in \{\pi (\ell ),\ldots ,\pi (k)\}}\mathbf{%
F}\left( \left\vert z_{j}-y_{j}\right\vert \right) \ .  \notag
\end{eqnarray}%
In the right--hand side (r.h.s.) of (\ref{def AT}), we set $n_{j}\doteq
m_{j} $ if
\begin{equation*}
j\in \left\{ 0,\ldots ,k\right\} \backslash \left\{ \pi \left( \ell \right)
,\ldots ,\pi \left( k\right) \right\} .
\end{equation*}%
The constant $\mathbf{D}\in \mathbb{R}^{+}$ is defined by (\ref{(3.2) NS}).
\end{satz}

\begin{proof}
Fix $k\in \mathbb{N}$, $\{m_{j}\}_{j=0}^{k}\subset \mathbb{N}_{0}$, $%
\{s_{j}\}_{j=1}^{k}\subset \mathbb{R}$, $\{x_{j}\}_{j=0}^{k}\subset
\mathfrak{L}$ and elements $\{B_{j}\}_{j=0}^{k}\subset \mathcal{U}_{0}$ such
that the conditions of the theorem are satisfied. From Lemma \ref{Lemma
Series representation of the dynamics},
\begin{eqnarray}
&&\left[ \tau _{s_{k}}\circ \chi _{x_{k}}(B_{k}),\ldots ,\tau _{s_{1}}\circ
\chi _{x_{1}}(B_{1}),\chi _{x_{0}}(B_{0})\right] ^{(k+1)}
\label{ineq tot sup} \\
&=&\sum_{n_{1}=m_{1}}^{\infty }\cdots \sum_{n_{k}=m_{k}}^{\infty }\left[
\mathfrak{B}_{B_{k},s_{k},x_{k}}\left( n_{k}\right) ,\ldots ,\mathfrak{B}%
_{B_{1},s_{1},x_{1}}\left( n_{1}\right) ,\chi _{x_{0}}(B_{0})\right]
^{(k+1)}\ .  \notag
\end{eqnarray}%
Since $B_{j}\in \mathcal{U}_{\Lambda _{m_{j}}}\cap \mathcal{U}^{+}$ for $%
j\in \{1,\ldots ,k\}$, we infer from (\ref{def B frac1})--(\ref{def B frac2}%
) that
\begin{eqnarray}
&&\left[ \mathfrak{B}_{B_{k},s_{k},x_{k}}\left( n_{k}\right) ,\ldots ,%
\mathfrak{B}_{B_{1},s_{1},x_{1}}\left( n_{1}\right) ,\chi _{x_{0}}(B_{0})%
\right] ^{(k+1)}  \notag \\
&=&\prod_{j=1}^{k}\mathbf{1}\left[ \bigcup_{i=0}^{j-1}\left( \Lambda
_{n_{j}}+x_{j}\right) \cap \left( \Lambda _{n_{i}}+x_{i}\right) \neq
\emptyset \right]  \label{ineq tot sup1} \\
&&\left[ \mathfrak{B}_{B_{k},s_{k},x_{k}}\left( n_{k}\right) ,\ldots ,%
\mathfrak{B}_{B_{1},s_{1},x_{1}}\left( n_{1}\right) ,\chi _{x_{0}}(B_{0})%
\right] ^{(k+1)}  \notag
\end{eqnarray}%
for all integers $\{n_{j}\}_{j=0}^{k}\subset \mathbb{N}_{0}$ with $%
n_{0}\doteq m_{0}$ and $n_{j}\geq m_{j}$ when $j\in \{1,\ldots ,k\}$. The
conditions inside characteristic functions in (\ref{ineq tot sup1}) refer to
the fact that the sequence of sets $\{\Lambda _{n_{j}}\}_{j=0}^{k}$ has to
be a cluster to have a non--zero multi--commutator. Note further that%
\begin{equation}
\prod_{j=1}^{k}\mathbf{1}\left[ \bigcup_{i=0}^{j-1}\left( \Lambda
_{n_{j}}+x_{j}\right) \cap \left( \Lambda _{n_{i}}+x_{i}\right) \neq
\emptyset \right] \leq \sum_{T\in \mathcal{T}_{k+1}}\varkappa _{T}\left(
\left\{ (n_{j},x_{j})\right\} _{j=0}^{k}\right)
\text{ }.  \label{cluster-tree}
\end{equation}%
Using (\ref{ineq tot sup})--(\ref{cluster-tree}) one then shows that%
\begin{eqnarray}
&&\left\Vert \left[ \tau _{s_{k}}\circ \chi _{x_{k}}(B_{k}),\ldots ,\tau
_{s_{1}}\circ \chi _{x_{1}}(B_{1}),\chi _{x_{0}}(B_{0})\right]
^{(k+1)}\right\Vert _{\mathcal{U}}  \notag \\
&\leq &2^{k}\left\Vert B_{0}\right\Vert _{\mathcal{U}}\sum_{T\in \mathcal{T}%
_{k+1}}\sum_{n_{1}=m_{1}}^{\infty }\cdots \sum_{n_{k}=m_{k}}^{\infty
}\varkappa _{T}\left( \left\{ (n_{j},x_{j})\right\} _{j=0}^{k}\right)  \notag
\\
&&\qquad \qquad \qquad \qquad \qquad \qquad \times
\prod\limits_{j=1}^{k}\left\Vert \mathfrak{B}_{B_{j},s_{j},x_{j}}\left(
n_{j}\right) \right\Vert _{\mathcal{U}}\ .  \label{inequality sup LR multi}
\end{eqnarray}%
This inequality combined with (\ref{bound B fract}) yields the assertion.
\end{proof}

\noindent The above theorem extends Lieb--Robinson bounds to
multi--commutators. Indeed, if $\mathbf{F}(r)$ decays fast enough as $%
r\rightarrow \infty $, then Theorem \ref{Theorem Lieb-Robinson copy(1)} and
Lebesgue's dominated convergence theorem imply that, for any $j\in
\{0,\ldots ,k\}$,
\begin{equation}
\underset{\left\vert x_{j}\right\vert \rightarrow \infty }{\lim }\left\Vert %
\left[ \tau _{s_{k}}\circ \chi _{x_{k}}(B_{k}),\ldots ,\tau _{s_{1}}\circ
\chi _{x_{1}}(B_{1}),\chi _{x_{0}}(B_{0})\right] ^{(k+1)}\right\Vert _{%
\mathcal{U}}=0\ .  \label{mult comm a borner}
\end{equation}

The rate of convergence if this multi--commutator towards zero is, however,
a priori unclear. Hence, to obtain bounds on the space decay of the above
multi--commutator, more in the spirit of the original Lieb--Robinson bounds
for commutators, we consider two situations w.r.t. the behavior of the
function $\mathbf{F}:\mathbb{R}_{0}^{+}\rightarrow \mathbb{R}^{+}$ at large
arguments:

\begin{itemize}
\item \emph{Polynomial decay.}
\index{Decay function!polynomial decay}There is a constant $\varsigma \in
\mathbb{R}^{+}$ and, for all $m\in \mathbb{N}_{0}$, an absolutely summable
sequence $\{\mathbf{u}_{n,m}\}_{n\in \mathbb{N}}\in \ell ^{1}(\mathbb{N})$
such that, for all $n\in \mathbb{N}$ with $n>m$,%
\begin{equation}
|\Lambda _{n}\backslash \Lambda _{n-1}|\sum_{z\in \Lambda _{m}}\max_{y\in
\Lambda _{n}\backslash \Lambda _{n-1}}\mathbf{F}\left( \left\vert
z-y\right\vert \right) \leq
\frac{\mathbf{u}_{n,m}}{\left( 1+n\right) ^{\varsigma }}\text{ }.
\label{(3.3) NS generalized0}
\end{equation}

\item \emph{Exponential decay.}
\index{Decay function!exponential decay}There is $\varsigma \in \mathbb{R}%
^{+}$ and, for $m\in \mathbb{N}_{0}$, a constant $\mathbf{C}_{m}\in \mathbb{R%
}^{+}$ such that, for all $n\in \mathbb{N}$ with $n>m$,%
\begin{equation}
|\Lambda _{n}\backslash \Lambda _{n-1}|\sum_{z\in \Lambda _{m}}\max_{y\in
\Lambda _{n}\backslash \Lambda _{n-1}}\mathbf{F}\left( \left\vert
z-y\right\vert \right) \leq \mathbf{C}_{m}\mathrm{e}^{-2\varsigma n}%
\text{ }.  \label{(3.3) NS generalized}
\end{equation}
\end{itemize}

\noindent For sufficiently large $\epsilon \in \mathbb{R}^{+}$, the function
(\ref{example polynomial}) clearly satisfies Condition (\ref{(3.3) NS
generalized0}), while (\ref{(3.3) NS generalized0})--(\ref{(3.3) NS
generalized}) hold for the choice
\begin{equation}
\mathbf{F}\left( r\right) =\mathrm{e}^{-2\varsigma r}(1+r)^{-(d+\epsilon )}\
,\qquad r\in \mathbb{R}_{0}^{+}\ ,  \label{example}
\end{equation}%
with arbitrary $\varsigma ,\epsilon \in \mathbb{R}^{+}$. Under one of these
both very general assumptions, one can put the upper bound of Theorem \ref%
{Theorem Lieb-Robinson copy(1)} in a much more convenient form. In fact, one
obtains an estimate on the norm of the multi--commutator (\ref{mult comm a
borner}) as a function of the distances between the points $\{x_{0},\ldots
,x_{k}\}$, like in the usual Lieb--Robinson bounds (i.e., the special case $%
k=2$). To formulate such bounds, we need some preliminary definitions
related to properties of trees.

For any $k\in \mathbb{N}$ and $T\in \mathcal{T}_{k+1}$, we define the
sequence $\mathfrak{d}_{T}\equiv \{\mathfrak{d}_{T}(j)\}_{j=0}^{k}$ in $%
\{1,\ldots ,k\}$ by%
\begin{equation*}
\mathfrak{d}_{T}(j)\doteq |\{b\in T\text{ }:\text{ }j\in b\}|\text{ },\qquad
j\in \{0,\ldots ,k\}\text{ },
\end{equation*}%
i.e., $\mathfrak{d}_{T}(j)$ is the
\index{Trees!degree of a vertex}\emph{degree} of the $j$--th vertex of the
tree $T$. For $k\in \mathbb{N}$ and $T\in \mathcal{T}_{k+1}$, observe that%
\begin{equation}
\mathfrak{d}_{T}(0)+\cdots +\mathfrak{d}_{T}(k)=2k\ .
\label{inequality easy}
\end{equation}%
We also introduce the following notation:%
\begin{equation*}
\mathfrak{d}_{T}!\doteq \mathfrak{d}_{T}(0)!\cdots \mathfrak{d}_{T}(k)!
\end{equation*}%
for any tree $T\in \mathcal{T}_{k+1}$, $k\in \mathbb{N}$. The degree of any
vertex of a tree is at least $1$, by connectedness of such a graph, and (\ref%
{inequality easy}) yields
\begin{equation}
\mathfrak{d}_{T}!\leq k!%
\text{ },\qquad k\in \mathbb{N}\ ,\ T\in \mathcal{T}_{k+1}\text{ }.
\label{inequality easy2}
\end{equation}%
For any $k\in \mathbb{N}$, $T\in \mathcal{T}_{k+1}$, and any sequence $f:%
\mathbb{N}_{0}\rightarrow \mathbb{R}^{+}$, note that%
\begin{equation}
\prod\limits_{j=0}^{k}\left\{ f\left( j\right) \right\} ^{\mathfrak{d}%
_{T}(j)}=\prod\limits_{j=1}^{k}f\left( j\right) f\left( \mathrm{P}_{T}\left(
j\right) \right) \ .  \label{equality pivotal}
\end{equation}%
This property is elementary but pivotal to estimate the remainder $\Re
_{T,\alpha }$, defined by (\ref{def AT}), of Theorem \ref{Theorem
Lieb-Robinson copy(1)}.

\begin{satz}[Lieb--Robinson bounds for multi--commutators -- Part II]
\label{theorem exp tree decay copy(1)}\mbox{
}\newline
Let $\alpha \in \mathbb{R}_{0}^{+}$, $k\in \mathbb{N}$, $\{m_{j}\}_{j=0}^{k}%
\subset \mathbb{N}_{0}$, $\{s_{j}\}_{j=1}^{k}\subset \mathbb{R}$, $%
\{x_{j}\}_{j=0}^{k}\subset \mathfrak{L}$, and $T\in \mathcal{T}_{k+1}$.
Depending on decay properties of the function $\mathbf{F}:\mathbb{R}%
_{0}^{+}\rightarrow \mathbb{R}^{+}$, the coefficient $\Re _{T,\alpha }\in
\mathbb{R}_{0}^{+}$ defined by (\ref{def AT}) satisfies the following bounds:%
\index{Lieb--Robinson bounds!multi--commutators}\newline
\emph{(i)} Polynomial decay: Assume (\ref{(3.3) NS generalized0}). Then,%
\begin{eqnarray*}
\Re _{T,\alpha } &\leq &d^{%
\frac{\varsigma k}{2}}\sum_{\ell =1}^{k}\left( 2\alpha \right) ^{k-\ell
+1}\sum_{\pi \in \mathcal{S}_{\ell ,k}}\left( \prod\limits_{j\in \{\pi (\ell
),\ldots ,\pi (k)\}}\left\Vert \mathbf{u}_{\cdot ,m_{j}}\right\Vert _{\ell
^{1}(\mathbb{N})}\left\vert s_{j}\right\vert \mathrm{e}^{4\mathbf{D}%
\left\vert s_{j}\right\vert \alpha }\right) \\
&&\left( \prod\limits_{j\in \left\{ 0,\ldots ,k\right\} \backslash \left\{
\pi \left( \ell \right) ,\ldots ,\pi \left( k\right) \right\}
}(1+m_{j})^{\varsigma }\right) \prod\limits_{\{j,l\}\in T}\frac{1}{%
(1+\left\vert x_{j}-x_{l}\right\vert )^{\varsigma \left( \max \{\mathfrak{d}%
_{T}(j),\mathfrak{d}_{T}(l)\}\right) ^{-1}}}\ .
\end{eqnarray*}%
\emph{(ii)} Exponential decay: Assume (\ref{(3.3) NS generalized}). Then,%
\begin{eqnarray*}
\Re _{T,\alpha } &\leq &\sum_{\ell =1}^{k}\left( \frac{2\alpha }{\mathrm{e}%
^{\varsigma }-1}\right) ^{k-\ell +1}\sum_{\pi \in \mathcal{S}_{\ell
,k}}\left( \prod\limits_{j\in \{\pi (\ell ),\ldots ,\pi (k)\}}\mathbf{C}%
_{m_{j}}\left\vert s_{j}\right\vert \mathrm{e}^{4\mathbf{D}\left\vert
s_{j}\right\vert \alpha -\varsigma m_{j}}\right) \\
&&\left( \prod\limits_{j\in \left\{ 0,\ldots ,k\right\} \backslash \left\{
\pi \left( \ell \right) ,\ldots ,\pi \left( k\right) \right\} }\mathrm{e}%
^{\varsigma m_{j}}\right) \prod\limits_{\{j,l\}\in T}\exp \left( -\frac{%
\varsigma \left\vert x_{j}-x_{l}\right\vert }{\sqrt{d}\max \{\mathfrak{d}%
_{T}(j),\mathfrak{d}_{T}(l)\}}\right) \ .
\end{eqnarray*}
\end{satz}

\begin{proof}
(i) Fix all parameters of the theorem. We infer from (\ref{def AT}) and (\ref%
{(3.3) NS generalized0}) that%
\begin{eqnarray*}
\Re _{T,\alpha } &\leq &\sum_{\ell =1}^{k}\left( 2\alpha \right) ^{k-\ell
+1}\sum_{\pi \in \mathcal{S}_{\ell ,k}}\left( \prod\limits_{j\in \{\pi (\ell
),\ldots ,\pi (k)\}}\left\vert s_{j}\right\vert \mathrm{e}^{4\mathbf{D}%
\left\vert s_{j}\right\vert \alpha }\right) \sum_{n_{\pi \left( \ell \right)
}=m_{\pi \left( \ell \right) }+1}^{\infty } \\
&&\cdots \sum_{n_{\pi \left( k\right) }=m_{\pi \left( k\right) }+1}^{\infty
}\varkappa _{T}\left( \left\{ (n_{j},x_{j})\right\} _{j=0}^{k}\right)
\prod\limits_{j\in \{\pi (\ell ),\ldots ,\pi (k)\}}\frac{\mathbf{u}%
_{n_{j},m_{j}}}{\left( 1+n_{j}\right) ^{\varsigma }}\ .
\end{eqnarray*}%
Recall that $n_{j}\doteq m_{j}$ when $j\in \left\{ 0,\ldots ,k\right\}
\backslash \left\{ \pi \left( \ell \right) ,\ldots ,\pi \left( k\right)
\right\} $. By H\"{o}lder's inequality, it follows that%
\begin{eqnarray}
\Re _{T,\alpha } &\leq &\sum_{\ell =1}^{k}\left( 2\alpha \right) ^{k-\ell
+1}\sum_{\pi \in \mathcal{S}_{\ell ,k}}\left( \prod\limits_{j\in \{\pi (\ell
),\ldots ,\pi (k)\}}\left\Vert \mathbf{u}_{\cdot ,m_{j}}\right\Vert _{\ell
^{1}(\mathbb{N})}\left\vert s_{j}\right\vert \mathrm{e}^{4\mathbf{D}%
\left\vert s_{j}\right\vert \alpha }\right)  \label{ohmV0} \\
&&\times \max_{n_{\pi (\ell )},\ldots ,n_{\pi (k)}\in \mathbb{N}}\left\{
\varkappa _{T}\left( \left\{ (n_{j},x_{j})\right\} _{j=0}^{k}\right)
\prod\limits_{j\in \left\{ \pi \left( \ell \right) ,\ldots ,\pi \left(
k\right) \right\} }\frac{1}{\left( 1+n_{j}\right) ^{\varsigma }}\right\} \ .
\notag
\end{eqnarray}%
Therefore, it suffices to bound the above maximum in an appropriate way.
Using (\ref{equality pivotal}), note that%
\begin{eqnarray}
\prod\limits_{j=0}^{k}\frac{1}{\left( 1+n_{j}\right) ^{\varsigma }}
&=&\prod\limits_{j=0}^{k}\left( \frac{1}{\left( 1+n_{j}\right) ^{\frac{%
\varsigma }{\mathfrak{d}_{T}(j)}}}\right) ^{\mathfrak{d}_{T}(j)}  \notag \\
&=&\prod\limits_{j=1}^{k}\frac{1}{\left( 1+n_{j}\right) ^{\frac{\varsigma }{%
\mathfrak{d}_{T}(j)}}\left( 1+n_{\mathrm{P}_{T}(j)}\right) ^{\frac{\varsigma
}{\mathfrak{d}_{T}(\mathrm{P}_{T}(j))}}}  \notag \\
&\leq &\prod\limits_{j=1}^{k}\frac{1}{\left( 1+n_{j}+n_{\mathrm{P}%
_{T}(j)}\right) ^{\frac{\varsigma }{\mathfrak{m}_{T}(j)}}}\text{ },
\label{ine cool ohmV1bis}
\end{eqnarray}%
where, for $k\in \mathbb{N}$, any tree $T\in \mathcal{T}_{k+1}$, and $j\in
\left\{ 1,\ldots ,k\right\} $,%
\begin{equation*}
\mathfrak{m}_{T}(j)\doteq \max \{\mathfrak{d}_{T}(j),\mathfrak{d}_{T}(%
\mathrm{P}_{T}(j))\}\text{ }.
\end{equation*}%
Meanwhile, the condition
\begin{equation*}
(\Lambda _{n_{j}}+x_{j})\cap (\Lambda _{n_{\mathrm{P}_{T}(j)}}+x_{\mathrm{P}%
_{T}(j)})\neq \emptyset
\end{equation*}%
implies
\begin{equation}
\sqrt{d}(n_{j}+n_{\mathrm{P}_{T}(j)})\geq |x_{j}-x_{\mathrm{P}_{T}(j)}|\ .
\label{ine cool ohmV2}
\end{equation}%
Therefore, we infer from (\ref{ine cool ohmV1bis})--(\ref{ine cool ohmV2})
that%
\begin{eqnarray*}
&&\max_{n_{\pi (\ell )},\ldots ,n_{\pi (k)}\in \mathbb{N}}\left\{ \varkappa
_{T}\left( \left\{ (n_{j},x_{j})\right\} _{j=0}^{k}\right)
\prod\limits_{j\in \left\{ \pi \left( \ell \right) ,\ldots ,\pi \left(
k\right) \right\} }\frac{1}{\left( 1+n_{j}\right) ^{\varsigma }}\right\} \\
&\leq &\left( \prod\limits_{j\in \left\{ 0,\ldots ,k\right\} \backslash
\left\{ \pi \left( \ell \right) ,\ldots ,\pi \left( k\right) \right\}
}(1+n_{j})^{\varsigma }\right) \prod\limits_{j=1}^{k}\frac{d^{\frac{%
\varsigma }{2}}}{(1+|x_{j}-x_{\mathrm{P}_{T}(j)}|)^{\frac{\varsigma }{%
\mathfrak{m}_{T}(j)}}}\ .
\end{eqnarray*}%
Combined with (\ref{ohmV0}), this last inequality yields Assertion (i).

(ii) The second assertion is proven exactly in the same way. We omit the
details.
\end{proof}

We defined in \cite[Section 4]{OhmI} the concept of \emph{tree--decay bounds}
for pairs $(\rho ,\tau )$, where $\rho \in \mathcal{U}^{\ast }$\ and $\tau
\equiv \{\tau _{t}\}_{t\in {\mathbb{R}}}$ are respectively any state and any
one--parameter group of $\ast $--automor%
\-%
phisms on the $C^{\ast }$--algebra $\mathcal{U}$. They are a useful tool to
control multi--commu%
\-%
tators of products of annihilation and creation operators. Such bounds are
related to cluster or graph expansions in statistical physics. For more
details see the preliminary discussions of \cite[Section 4]{OhmI}. As a
straightforward corollary of Theorems \ref{Theorem Lieb-Robinson copy(1)}--%
\ref{theorem exp tree decay copy(1)} we give below an extension of the
tree--decay bounds\emph{\ }\cite[Section 4]{OhmI} to the case of interacting
fermions on lattices:

\begin{koro}[Tree--decay bounds]
\label{theorem exp tree decay}\mbox{
}\newline
Let $\Psi \in \mathcal{W}$, $\mathbf{V}$ be any potential, $k\in \mathbb{N}$%
, $m_{0}\in \mathbb{N}_{0}$, $t\in \mathbb{R}_{0}^{+}$, $\{s_{j}\}_{j=1}^{k}%
\subset \lbrack -t,t]$, $B_{0}\subset \mathcal{U}_{\Lambda _{m_{0}}}$,\ and $%
\{x_{j}\}_{j=0}^{k},\{z_{j}\}_{j=1}^{k}\subset \mathfrak{L}$ such that $%
|z_{j}|=1$ for $j\in \{1,\ldots ,k\}$.%
\index{Tree--decay bounds} \newline
\emph{(i)} Polynomial decay: Assume (\ref{(3.3) NS generalized0}) for $m=1$.
Then,%
\begin{eqnarray*}
&&\left\Vert \left[ \tau _{s_{k}}(a_{x_{k}}^{\ast }a_{x_{k}+z_{k}}),\ldots
,\tau _{s_{1}}(a_{x_{1}}^{\ast }a_{x_{1}+z_{1}}),\chi _{x_{0}}(B_{0})\right]
^{(k+1)}\right\Vert _{\mathcal{U}} \\
&\leq &\left\Vert B_{0}\right\Vert _{\mathcal{U}}(1+m_{0})^{\varsigma }%
\mathbf{K}_{0}^{k}\sum_{T\in \mathcal{T}_{k+1}}\prod\limits_{\{j,l\}\in T}%
\frac{1}{(1+\left\vert x_{j}-x_{l}\right\vert )^{\varsigma \left( \max \{%
\mathfrak{d}_{T}(j),\mathfrak{d}_{T}(l)\}\right) ^{-1}}}
\end{eqnarray*}%
with
\begin{equation*}
\mathbf{K}_{0}\doteq 2d^{\frac{\varsigma }{2}}\left( 2^{\varsigma
}+2\left\Vert \mathbf{u}_{\cdot ,1}\right\Vert _{\ell ^{1}(\mathbb{N})}\Vert
\Psi \Vert _{\mathcal{W}}\left\vert t\right\vert \mathrm{e}^{4\mathbf{D}%
\left\vert t\right\vert \left\Vert \Psi \right\Vert _{\mathcal{W}}}\right) \
.
\end{equation*}%
\emph{(ii)} Exponential decay: Assume (\ref{(3.3) NS generalized}) for $m=1$%
. Then,%
\begin{eqnarray*}
&&\left\Vert \left[ \tau _{s_{k}}(a_{x_{k}}^{\ast }a_{x_{k}+z_{k}}),\ldots
,\tau _{s_{1}}(a_{x_{1}}^{\ast }a_{x_{1}+z_{1}}),\chi _{x_{0}}(B_{0})\right]
^{(k+1)}\right\Vert _{\mathcal{U}} \\
&\leq &\left\Vert B_{0}\right\Vert _{\mathcal{U}}\mathrm{e}^{m_{0}\varsigma }%
\mathbf{K}_{1}^{k}\sum_{T\in \mathcal{T}_{k+1}}\prod\limits_{\{j,l\}\in
T}\exp \left( -\frac{\varsigma \left\vert x_{j}-x_{l}\right\vert }{\sqrt{d}%
\max \{\mathfrak{d}_{T}(j),\mathfrak{d}_{T}(l)\}}\right)
\end{eqnarray*}%
with%
\begin{equation*}
\mathbf{K}_{1}\doteq 2\left( \mathrm{e}^{\varsigma }+\frac{2\mathbf{C}%
_{1}\Vert \Psi \Vert _{\mathcal{W}}\left\vert t\right\vert \mathrm{e}^{4%
\mathbf{D}\left\vert t\right\vert \left\Vert \Psi \right\Vert _{\mathcal{W}}}%
}{\mathrm{e}^{2\varsigma }-\mathrm{e}^{\varsigma }}\right) \ .
\end{equation*}
\end{koro}

\begin{proof}
For all $k\in \mathbb{N}$, $T\in \mathcal{T}_{k+1}$, and any sequence $%
\{\left( m_{j},x_{j}\right) \}_{j=0}^{k}$ in $\mathbb{N}_{0}\times \mathfrak{%
L}$ of length $k+1$, the following upper bounds hold for $\varkappa _{T}$
(see (\ref{definition chqrqcteristic})):
\begin{equation}
\varkappa _{T}\left( \left\{ \left( m_{j},x_{j}\right) \right\}
_{j=0}^{k}\right) \leq d^{\frac{k\varsigma }{2}}\prod%
\limits_{j=0}^{k}(1+m_{j})^{\varsigma }\prod\limits_{\{j,l\}\in T}\frac{1}{%
(1+|x_{j}-x_{l}|)^{\frac{\varsigma }{\max \{\mathfrak{d}_{T}(j),\mathfrak{d}%
_{T}(l)\}}}}  \label{ki trivial estimate1}
\end{equation}%
while%
\begin{equation}
\varkappa _{T}\left( \left\{ \left( m_{j},x_{j}\right) \right\}
_{j=0}^{k}\right) \leq \mathrm{e}^{(m_{0}+\cdots +m_{k})\varsigma
}\prod\limits_{\{j,l\}\in T}\exp \left( -\frac{\varsigma \left\vert
x_{j}-x_{l}\right\vert }{\sqrt{d}\max \{\mathfrak{d}_{T}(j),\mathfrak{d}%
_{T}(l)\}}\right) \text{ }.  \label{ki trivial estimate2}
\end{equation}%
Cf. proof of Theorem \ref{theorem exp tree decay copy(1)}. Therefore, the
corollary is a direct consequence of Theorems \ref{Theorem Lieb-Robinson
copy(1)} and \ref{theorem exp tree decay copy(1)} together with the two
previous inequalities.
\end{proof}

\noindent Up to the powers $1/\max \{\mathfrak{d}_{T}(j),\mathfrak{d}%
_{T}(l)\}$, Corollary \ref{theorem exp tree decay} gives for interacting
systems upper bounds for multi--commutators like \cite[Eq. (4.14)]{OhmI} for
the free case. We show in the next subsection how to use these bounds to
obtain results similar to\ \cite[Theorem 3.4]{OhmI} on the dynamics
perturbed by the presence of external electromagnetic fields.

\begin{bemerkung}
\label{rermar trivial1}\mbox{
}\newline
All results of this subsection depend on Theorem \ref{Theorem Lieb-Robinson
copy(3)} (iii), i.e., the rate of convergence, as $n\rightarrow \infty $, of
the family $\{\tau ^{(n,x)}\}_{n\in \mathbb{N}_{0}}$ of finite--volume
groups introduced in the preliminary discussions before Lemma \ref{Lemma
Series representation of the dynamics}. It is the only information on the
Fermi system we needed here.
\end{bemerkung}

\begin{bemerkung}
\mbox{
}\newline
The
\index{Lieb--Robinson bounds}Lieb--Robinson bound for multi--commutators
given by Theorems \ref{Theorem Lieb-Robinson copy(1)}--\ref{theorem exp tree
decay copy(1)} at $k=1$ is not as good as the previous Lieb--Robinson bound
of Theorem \ref{Theorem Lieb-Robinson copy(3)} (iv). Nevertheless, they are
qualitatively equivalent in the following sense: For interactions with
polynomial decay, the first bound also has polynomial decay, even if with
lower degree than the second one. For interactions with exponential decay,
both bounds are exponentially decaying, even if the first one has a worse
prefactor and exponential rate than the second one.
\end{bemerkung}

\subsection{Application to Perturbed Autonomous Dynamics\label{section
Energy Increments as Power Series}}

Let $\Psi \in \mathcal{W}$ and $\mathbf{V}$ be a potential. For any $l\in
\mathbb{R}_{0}^{+}$, we consider a map $\eta \mapsto \mathbf{W}^{(l,\eta )}$
from $\mathbb{R}$ to the subspace of self--adjoint elements of $\mathcal{U}%
_{\Lambda _{l}}$. In the case that interests us, the following property
holds:
\begin{equation}
\left\Vert \mathbf{W}^{(l,\eta )}\right\Vert _{\mathcal{U}}=\mathcal{O}(\eta
\left\vert \Lambda _{l}\right\vert )%
\text{ }.  \label{bound assumption}
\end{equation}%
More precisely, we consider elements $\mathbf{W}^{(l,\eta )}$ of the form%
\begin{equation}
\mathbf{W}^{(l,\eta )}\doteq \sum\limits_{x\in \Lambda
_{l}}\sum\limits_{z\in \mathfrak{L},|z|\leq 1}\mathbf{w}_{x,x+z}(\eta
)a_{x}^{\ast }a_{x+z}\ ,\qquad l\in \mathbb{R}_{0}^{+}\ ,
\label{bound assumption2}
\end{equation}%
where $\{\mathbf{w}_{x,y}\}_{x,y\in \mathfrak{L}}$ are complex--valued
functions of $\eta \in \mathbb{R}$ with
\begin{equation}
\overline{\mathbf{w}_{x,y}}=\mathbf{w}_{y,x}\qquad \text{and}\qquad \mathbf{w%
}_{x,y}(0)=0  \label{assumption boundedness00}
\end{equation}%
for all $x,y\in \mathfrak{L}$.

Equation (\ref{bound assumption2}) has the form
\begin{equation}
\mathbf{W}^{(l,\eta )}=\sum\limits_{x\in \Lambda _{l}}W_{x}(\eta )
\label{w more general}
\end{equation}%
where, for some fixed radius $R\in \mathbb{R}^{+}$ and any $x\in \mathfrak{L}
$, $W_{x}(\eta )$ is a self--adjoint even element of $\mathcal{U}_{x+\Lambda
_{R}}$ that depends on the real parameter $\eta $. All results below in this
subsection hold for the more general case (\ref{w more general}) as well,
with obvious modifications. Indeed, we could even consider more general
perturbations with $R=\infty $, see proofs of Inequality (\ref{estimate poly}%
) and Theorem \ref{thm non auto copy(1)}.

We refrain from treating cases more general than (\ref{bound assumption2})
to keep technical aspects as simple as possible. Observe that perturbations
due to the presence of external electromagnetic fields are included in the
class of perturbations defined by (\ref{bound assumption2}). In fact, as
discussed in the introduction, our final aim is the microscopic quantum
theory of electrical conduction \cite{brupedrahistoire, OhmV, OhmVI}.
Indeed, at fixed $l\in \mathbb{R}_{0}^{+}$, $\mathbf{W}^{(l,\eta )}$ defined
by (\ref{bound assumption2}) is related to perturbations of dynamics caused
by constant external electromagnetic fields that vanish outside the box $%
\Lambda _{l}$.

We assume that $\{\mathbf{w}_{x,y}\}_{x,y\in \mathfrak{L}}$ are uniformly
bounded and Lipschitz continuous: There is a constant $K_{1}\in \mathbb{R}%
^{+}$ such that, for all $\eta ,\eta _{0}\in \mathbb{R}$,%
\begin{equation}
\sup_{x,y\in \mathfrak{L}}\left\vert \mathbf{w}_{x,y}(\eta )-\mathbf{w}%
_{x,y}(\eta _{0})\right\vert \leq K_{1}\left\vert \eta -\eta _{0}\right\vert
\text{ \ \ and \ \ }\sup_{x,y\in \mathfrak{L}}\sup_{\eta \in \mathbb{R}%
}\left\vert \mathbf{w}_{x,y}(\eta )\right\vert \leq K_{1}\ .
\label{assumption boundedness0}
\end{equation}%
These two uniformity conditions could hold for parameters $\eta ,\eta _{0}$
on compact sets only, but we refrain again from considering this more
general case, for simplicity.

The perturbed dynamics is defined via the symmetric derivation
\begin{equation}
\delta ^{(l,\eta )}\doteq \delta +i\left[ \mathbf{W}^{(l,\eta )},\ \cdot \ %
\right] \ ,\qquad l\in \mathbb{R}_{0}^{+},\ \eta \in \mathbb{R}\ .
\label{bounded pertubation}
\end{equation}%
Recall that $\delta $ is the symmetric derivation of Theorem \ref{Theorem
Lieb-Robinson copy(3)} which generates the $C_{0}$--group $\{\tau
_{t}\}_{t\in {\mathbb{R}}}$ on $\mathcal{U}$. The second term in the r.h.s.
of (\ref{bounded pertubation}) is a bounded perturbation of $\delta $.
Hence, $\delta ^{(l,\eta )}$ generates a $C_{0}$--group $\{\tilde{\tau}%
_{t}^{(l,\eta )}\}_{t\in {\mathbb{R}}}$ on $\mathcal{U}$, see \cite[Chap.
III, Sect. 1.3]{EngelNagel}. By Lemma \ref{lemma norm closure}, the
(generally unbounded) closed operator $\delta ^{(l,\eta )}$ is a
conservative symmetric derivation and $\tilde{\tau}_{t}^{(l,\eta )}$ is a $%
\ast $--auto%
\-%
morphism of $\mathcal{U}$ for all $t\in \mathbb{R}$.

Let $\Phi $ be any interaction with energy observables
\begin{equation}
U_{\Lambda _{L}}^{\Phi }\doteq \sum\limits_{\Lambda \subseteq \Lambda
_{L}}\Phi _{\Lambda }\ ,\qquad L\in \mathbb{R}_{0}^{+}\ .
\label{energy observable}
\end{equation}%
The main aim of this subsection is to study the energy increment%
\index{Increment}
\begin{equation}
\mathbf{T}_{t,s}^{(l,\eta ,L)}\doteq
\tilde{\tau}_{t-s}^{(l,\eta )}(U_{\Lambda _{L}}^{\Phi })-\tau
_{t-s}(U_{\Lambda _{L}}^{\Phi })\ ,\qquad l,L\in \mathbb{R}_{0}^{+},\
s,t,\eta \in \mathbb{R}\ ,  \label{generic Dyson--Phillips series2}
\end{equation}%
in the limit $L\rightarrow \infty $ to obtain similar results as \cite[%
Theorem 3.4]{OhmI}. This can be done by using the (partial)
\index{Dyson--Phillips series}Dyson--Phillips series:
\begin{eqnarray}
&&\mathbf{T}_{t,s}^{(l,\eta ,L)}-\mathbf{T}_{t,s}^{(l,\eta _{0},L)}
\label{finite dyson} \\
&=&\sum\limits_{k=1}^{m}i^{k}\int_{s}^{t}\mathrm{d}s_{1}\cdots
\int_{s}^{s_{k-1}}\mathrm{d}s_{k}\left[ \mathbf{X}_{s_{k},s}^{(l,\eta
_{0},\eta )},\ldots ,\mathbf{X}_{s_{1},s}^{(l,\eta _{0},\eta )},%
\tilde{\tau}_{t-s}^{(l,\eta _{0})}(U_{\Lambda _{L}}^{\Phi })\right] ^{(k+1)}
\notag \\
&&+i^{m+1}\int_{s}^{t}\mathrm{d}s_{1}\cdots \int_{s}^{s_{m}}\mathrm{d}s_{m+1}
\notag \\
&&\qquad \tilde{\tau}_{s_{m+1}-s}^{(l,\eta )}\left( \left[ \mathbf{W}%
^{(l,\eta )}-\mathbf{W}^{(l,\eta _{0})},\mathbf{X}_{s_{m},s_{m+1}}^{(l,\eta
_{0},\eta )},\ldots ,\mathbf{X}_{s_{1},s_{m+1}}^{(l,\eta _{0},\eta )},\tilde{%
\tau}_{t-s_{m+1}}^{(l,\eta _{0})}(U_{\Lambda _{L}}^{\Phi })\right]
^{(m+2)}\right)  \notag
\end{eqnarray}%
for any $m\in \mathbb{N}$, where
\begin{equation}
\mathbf{X}_{t,s}^{(l,\eta _{0},\eta )}\doteq \tilde{\tau}_{t-s}^{(l,\eta
_{0})}(\mathbf{W}^{(l,\eta )}-\mathbf{W}^{(l,\eta _{0})})\ ,\qquad l\in
\mathbb{R}_{0}^{+},\ s,t,\eta _{0},\eta \in \mathbb{R}\ .
\label{finite dysonbis}
\end{equation}%
By (\ref{assumption boundedness00}), note that $\mathbf{T}_{t,s}^{(l,0,L)}=0$%
.

By (\ref{bound assumption}), naive bounds on the r.h.s. of (\ref{finite
dyson}) predict that
\begin{equation*}
\left[ \mathbf{X}_{s_{k},s}^{(l,\eta _{0},\eta )},\ldots ,\mathbf{X}%
_{s_{1},s}^{(l,\eta _{0},\eta )},\tilde{\tau}_{t-s}^{(l,\eta
_{0})}(U_{\Lambda _{L}}^{\Phi })\right] ^{(k+1)}=\mathcal{O}(\left\vert
\Lambda _{l}\right\vert ^{k}\left\vert \Lambda _{L}\right\vert )\ .
\end{equation*}%
To obtain more accurate estimates, we use the tree--decay bounds on
multi--commutators of Corollary \ref{theorem exp tree decay}.

To this end, for any $x\in \mathfrak{L}$ and $m\in \mathbb{N}$, we define%
\begin{equation}
\mathcal{D}\left( x,m\right) \doteq \left\{ \Lambda \in \mathcal{P}_{f}(%
\mathfrak{L}):x\in \Lambda ,\text{ }\Lambda \subseteq \Lambda _{m}+x,\
\Lambda \nsubseteq \Lambda _{m-1}+x\right\} \subset 2^{\mathfrak{L}}\ .
\label{definition D1}
\end{equation}%
All elements of $\mathcal{D}(x,m)$ are finite subsets of the lattice $%
\mathfrak{L}$ that contain at least two sites which are separated by a
distance greater or equal than $m$. Using, for any $x\in \mathfrak{L}$ and $%
m=0$, the convention%
\begin{equation}
\mathcal{D}\left( x,0\right) \doteq \left\{ \left\{ x\right\} \right\} \ ,
\end{equation}%
we obviously have that
\begin{equation}
\mathcal{P}_{f}\left( \mathfrak{L}\right) =\underset{x\in \mathfrak{L},\
m\in \mathbb{N}_{0}}{\bigcup }\mathcal{D}\left( x,m\right) \ .
\label{set eq}
\end{equation}%
We now consider the following assumption on interactions $\Phi $:%
\index{Decay function!polynomial decay}%
\index{Decay function!exponential decay}%
\begin{equation}
\underset{x\in \mathfrak{L}}{\sup }\sum\limits_{m\in \mathbb{N}_{0}}\mathbf{v%
}_{m}\sum\limits_{\Lambda \in \mathcal{D}\left( x,m\right) }\left\Vert \Phi
_{\Lambda }\right\Vert _{\mathcal{U}}<\infty  \label{assumption boundedness2}
\end{equation}%
for some (generally diverging) sequence $\{\mathbf{v}_{m}\}_{m\in \mathbb{N}%
_{0}}\subset \mathbb{R}_{0}^{+}$. For instance, if $\Phi \in \mathcal{W}$
and Condition (\ref{(3.3) NS generalized0}) holds true, then one easily
verifies (\ref{assumption boundedness2}) with $\mathbf{v}_{m}=\left(
1+m\right) ^{\varsigma }$. In the case (\ref{(3.3) NS generalized}) holds
and $\Phi \in \mathcal{W}$, then (\ref{assumption boundedness2}) is also
satisfied even with $\mathbf{v}_{m}=\mathrm{e}^{m\varsigma }$.

We are now in position to state the first main result of this section, which
is an extension of \cite[Theorem 3.4 (i)]{OhmI} to interacting fermions:

\begin{satz}[Taylor's theorem for increments]
\label{Thm Heat production as power series copy(2)}\mbox{
}\newline
Let $l,\mathrm{T}\in \mathbb{R}_{0}^{+}$, $s,t\in \left[ -\mathrm{T},\mathrm{%
T}\right] $, $\eta ,\eta _{0}\in \mathbb{R}$, $\Psi \in \mathcal{W}$, and $%
\mathbf{V}$ be any potential.%
\index{Taylor's theorem}%
\index{Increment} Assume (\ref{(3.3) NS generalized0}) with $\varsigma >d$, (%
\ref{assumption boundedness00}) and (\ref{assumption boundedness0}). Take an
interaction $\Phi $ satisfying (\ref{assumption boundedness2}) with $\mathbf{%
v}_{m}=\left( 1+m\right) ^{\varsigma }$. Then:\newline
\emph{(i)} The map $\eta \mapsto \mathbf{T}_{t,s}^{(l,\eta ,L)}$ converges
uniformly on $\mathbb{R}$, as $L\rightarrow \infty $, to a continuous
function $\mathbf{T}_{t,s}^{(l,\eta )}$ of $\eta $ and
\begin{equation*}
\mathbf{T}_{t,s}^{(l,\eta )}-\mathbf{T}_{t,s}^{(l,\eta
_{0})}=\sum\limits_{\Lambda \in \mathcal{P}_{f}(\mathfrak{L})}i\int_{s}^{t}%
\mathrm{d}s_{1}%
\tilde{\tau}_{s_{1}-s}^{(l,\eta )}\left( \left[ \mathbf{W}^{(l,\eta )}-%
\mathbf{W}^{(l,\eta _{0})},\tilde{\tau}_{t-s_{1}}^{(l,\eta _{0})}(\Phi
_{\Lambda })\right] \right) \ .
\end{equation*}%
\emph{(ii)} For any $m\in \mathbb{N}$ satisfying $d(m+1)<\varsigma $,
\begin{eqnarray}
&&\mathbf{T}_{t,s}^{(l,\eta )}-\mathbf{T}_{t,s}^{(l,\eta _{0})}=
\label{diff T-T 1} \\
&&\sum\limits_{k=1}^{m}\sum\limits_{\Lambda \in \mathcal{P}_{f}(\mathfrak{L}%
)}i^{k}\int_{s}^{t}\mathrm{d}s_{1}\cdots \int_{s}^{s_{k-1}}\mathrm{d}s_{k}%
\left[ \mathbf{X}_{s_{k},s}^{(l,\eta _{0},\eta )},\ldots ,\mathbf{X}%
_{s_{1},s}^{(l,\eta _{0},\eta )},\tilde{\tau}_{t-s}^{(l,\eta _{0})}(\Phi
_{\Lambda })\right] ^{(k+1)}  \notag \\
&&+\sum\limits_{\Lambda \in \mathcal{P}_{f}(\mathfrak{L})}i^{m+1}\int_{s}^{t}%
\mathrm{d}s_{1}\cdots \int_{s}^{s_{m}}\mathrm{d}s_{m+1}  \notag \\
&&\qquad \tilde{\tau}_{s_{m+1}-s}^{(l,\eta )}\left( \left[ \mathbf{W}%
^{(l,\eta )}-\mathbf{W}^{(l,\eta _{0})},\mathbf{X}_{s_{m},s_{m+1}}^{(l,\eta
_{0},\eta )},\ldots ,\mathbf{X}_{s_{1},s_{m+1}}^{(l,\eta _{0},\eta )},\tilde{%
\tau}_{t-s_{m+1}}^{(l,\eta _{0})}(\Phi _{\Lambda })\right] ^{(m+2)}\right) .
\notag
\end{eqnarray}%
\emph{(iii)} All the above series in $\Lambda $ absolutely converge: For any
$m\in \mathbb{N}$ satisfying $d(m+1)<\varsigma $, $k\in \{1,\ldots ,m\}$,
and $\{s_{j}\}_{j=1}^{m+1}\subset \lbrack -\mathrm{T},\mathrm{T}]$,
\begin{equation*}
\sum\limits_{\Lambda \in \mathcal{P}_{f}(\mathfrak{L})}\left\Vert \left[
\mathbf{X}_{s_{k},s}^{(l,\eta _{0},\eta )},\ldots ,\mathbf{X}%
_{s_{1},s}^{(l,\eta _{0},\eta )},\tilde{\tau}_{t-s}^{(l,\eta _{0})}(\Phi
_{\Lambda })\right] ^{(k+1)}\right\Vert _{\mathcal{U}}\leq D\left\vert
\Lambda _{l}\right\vert \left\vert \eta -\eta _{0}\right\vert ^{k}
\end{equation*}%
and%
\begin{multline*}
\sum\limits_{\Lambda \in \mathcal{P}_{f}(\mathfrak{L})}\left\Vert \tilde{\tau%
}_{s_{m+1}-s}^{(l,\eta )}\left( \left[ \mathbf{W}^{(l,\eta )}-\mathbf{W}%
^{(l,\eta _{0})},\mathbf{X}_{s_{m},s_{m+1}}^{(l,\eta _{0},\eta )},\ldots ,%
\mathbf{X}_{s_{1},s_{m+1}}^{(l,\eta _{0},\eta )},\tilde{\tau}%
_{t-s_{m+1}}^{(l,\eta _{0})}(\Phi _{\Lambda })\right] ^{(m+2)}\right)
\right\Vert _{\mathcal{U}} \\
\leq D\left\vert \Lambda _{l}\right\vert \left\vert \eta -\eta
_{0}\right\vert ^{m+1}\text{ },
\end{multline*}%
for some constant $D\in \mathbb{R}^{+}$ depending only on $m,d,\mathrm{T}%
,\Psi ,K_{1},\Phi ,\mathbf{F}$. The last assertion also holds for $m=0$.
\end{satz}

\begin{proof}
We only prove (ii)--(iii), Assertion (i) being easier to prove by very
similar arguments. For simplicity, we assume w.l.o.g. $\eta _{0}=s=0$ and $%
m\in \mathbb{N}$. Because of Equations (\ref{bound assumption2}), (\ref%
{finite dyson}), (\ref{finite dysonbis}) and (\ref{set eq}), we first
control the multi--commu%
\-%
tator sum
\begin{eqnarray*}
\digamma _{k,L} &\doteq &\sum\limits_{x_{0}\in \mathfrak{L}\backslash
\Lambda _{L}}\ \sum\limits_{m_{0}\in \mathbb{N}_{0}}\ \sum\limits_{\Lambda
\in \mathcal{D}\left( x_{0},m_{0}\right) }\ \sum\limits_{x_{1}\in \Lambda
_{l}}\ \sum\limits_{z_{1}\in \mathfrak{L},|z_{1}|\leq 1}\ \cdots \
\sum\limits_{x_{k}\in \Lambda _{l}}\ \sum\limits_{z_{k}\in \mathfrak{L}%
,|z_{k}|\leq 1} \\
&&\quad \left\Vert \xi _{x_{1},z_{1},\ldots ,x_{k},z_{k}}\left[ \tau
_{s_{k}}(a_{x_{k}}^{\ast }a_{x_{k}+z_{k}}),\ldots ,\tau
_{s_{1}}(a_{x_{1}}^{\ast }a_{x_{1}+z_{1}}),\tau _{t}(\Phi _{\Lambda })\right]
^{(k+1)}\right\Vert _{\mathcal{U}}
\end{eqnarray*}%
for any fixed $k\in \{1,\ldots ,m\}$, $\mathrm{T}\in \mathbb{R}_{0}^{+}$, $%
\{s_{j}\}_{j=1}^{k}\subset \lbrack -\mathrm{T},\mathrm{T}]$ and $L\in
\mathbb{R}_{0}^{+}\cup \{-1\}$, where we use the convention $\Lambda
_{-1}\doteq \emptyset $ and%
\begin{equation}
\xi _{x_{1},z_{1},\ldots ,x_{k},z_{k}}\doteq \prod\limits_{j=1}^{k}\mathbf{w}%
_{x_{j},x_{j}+z_{j}}(\eta )\ .  \label{assumption boundedness30}
\end{equation}%
By (\ref{assumption boundedness00})--(\ref{assumption boundedness0}), there
is a constant $D\in \mathbb{R}^{+}$ (depending on $K_{1}$) such that%
\begin{equation}
\sup_{x_{1},z_{1},\ldots ,x_{k},z_{k}\in \mathfrak{L}}\sup_{\eta \in \mathbb{%
R}}\left\vert \xi _{x_{1},z_{1},\ldots ,x_{k},z_{k}}\right\vert \leq D\ \
\text{and}\ \ \sup_{x_{1},z_{1},\ldots ,x_{k},z_{k}\in \mathfrak{L}%
}\left\vert \xi _{x_{1},z_{1},\ldots ,x_{k},z_{k}}\right\vert \leq D|\eta
|^{k}\ .  \label{assumption boundedness3}
\end{equation}%
At fixed $k\in \{1,\ldots ,m\}$ observe further that the condition $%
\varsigma >dk$ yields%
\begin{equation}
\max_{x\in \mathfrak{L}}\sum_{y\in \mathfrak{L}}\frac{1}{\left( 1+\left\vert
y-x\right\vert \right) ^{\varsigma \left( \max \{\mathfrak{d}_{T}(j),%
\mathfrak{d}_{T}(l)\}\right) ^{-1}}}\leq \sum_{y\in \mathfrak{L}}\frac{1}{%
\left( 1+\left\vert y\right\vert \right) ^{\frac{\varsigma }{k}}}<\infty
\label{assumption boundedness3bis}
\end{equation}%
for any tree $T\in \mathcal{T}_{k+1}$ and all $j,l\in \{0,\ldots ,k\}$.
Using (\ref{assumption boundedness2}) with $\mathbf{v}_{m}=\left( 1+m\right)
^{\varsigma }$, (\ref{assumption boundedness3})--(\ref{assumption
boundedness3bis}) and the equality%
\begin{eqnarray}
&&\left\Vert \left[ \tau _{s_{k}}(a_{x_{k}}^{\ast }a_{x_{k}+z_{k}}),\ldots
,\tau _{s_{1}}(a_{x_{1}}^{\ast }a_{x_{1}+z_{1}}),\tau _{t}(\Phi _{\Lambda })%
\right] ^{(k+1)}\right\Vert _{\mathcal{U}}  \notag \\
&=&\left\Vert \left[ \tau _{s_{k}-t}(a_{x_{k}}^{\ast
}a_{x_{k}+z_{k}}),\ldots ,\tau _{s_{1}-t}(a_{x_{1}}^{\ast
}a_{x_{1}+z_{1}}),\Phi _{\Lambda }\right] ^{(k+1)}\right\Vert _{\mathcal{U}%
}\ ,  \label{ineq conne}
\end{eqnarray}%
we obtain from Corollary \ref{theorem exp tree decay} that, for any $m\in
\mathbb{N}$ and $k\in \{1,\ldots ,m\}$ with $\varsigma >dk$, $\digamma
_{k,-1}\leq D|\Lambda _{l}||\eta |^{k}$ for some constant $D\in \mathbb{R}%
^{+}$ depending only on $m,d,\mathrm{T},\Psi ,K_{1},\Phi ,\mathbf{F}$.

Hence, by Lebesgue's dominated convergence theorem, for any $k\in \mathbb{N}$
satisfying $\varsigma >dk$, there is $R\in \mathbb{R}^{+}$ such that $%
\digamma _{k,L}<\varepsilon $ for any $L\geq R$. This ensures the
convergence of the first $k$ multi--commutators of (\ref{finite dyson}) to
the first $k$ multi--commutators of (\ref{diff T-T 1}) as well as the
corresponding absolute summability. Cf. Assertions (ii)--(iii). The
convergence is even uniform for $\eta \in \mathbb{R}$ because of the first
assertion of (\ref{assumption boundedness3}).

Because $\tilde{\tau}_{t}^{(l,\eta )}$ is an isometry for any time $t\in
\mathbb{R}$, the same arguments are used to control the multi--commutator%
\begin{equation}
\tilde{\tau}_{s_{m+1}-s}^{(l,\eta )}\left( \left[ \mathbf{W}^{(l,\eta )}-%
\mathbf{W}^{(l,\eta _{0})},\mathbf{X}_{s_{m},s_{m+1}}^{(l,\eta _{0},\eta
)},\ldots ,\mathbf{X}_{s_{1},s_{m+1}}^{(l,\eta _{0},\eta )},\tilde{\tau}%
_{t-s_{m+1}}^{(l,\eta _{0})}(\Phi _{\Lambda })\right] ^{(m+2)}\right)
\label{milti esti mate plus}
\end{equation}%
in (\ref{finite dyson}). By (\ref{assumption boundedness0}), notice
additionally that there is a constant $D\in \mathbb{R}^{+}$ and a family $%
\{\Psi ^{(l,\eta )}\}_{l\in \mathbb{R}_{0}^{+},\eta \in \mathbb{R}}\subset
\mathcal{W}$ such that
\begin{equation*}
\sup_{\eta \in \mathbb{R}}\sup_{l\in \mathbb{R}_{0}^{+}}\left\Vert \Psi
^{(l,\eta )}\right\Vert _{\mathcal{W}}\leq D<\infty \text{ }
\end{equation*}%
and, for all $l\in \mathbb{R}_{0}^{+}$ and $\eta \in \mathbb{R}$, $\{\tilde{%
\tau}_{t}^{(l,\eta )}\}_{t\in {\mathbb{R}}}$ is the $C_{0}$--group of $\ast $%
--automorphisms on $\mathcal{U}$ associated with the interaction $\Psi
^{(l,\eta )}$ and the potential $\mathbf{V}$. The norm $\left\Vert \cdot
\right\Vert _{\mathcal{W}}$ in the last inequality, which defines the space $%
\mathcal{W}$ of interactions, is of course defined w.r.t.\ the same function
$\mathbf{F}$ to which the conditions of the theorem are imposed. This
property justifies the simplifying assumption $\eta _{0}=0$ at the beginning
of the proof. This concludes the proof of Assertions (ii)--(iii).

Assertion (i) is proven in the same way and we omit the details. Note only
that the convergence of $\digamma _{1,L}$ as $L\rightarrow \infty $ is
uniform for $\eta \in \mathbb{R}$ because of the first assertion of (\ref%
{assumption boundedness3}). The latter implies the continuity of the map $%
\eta \mapsto \mathbf{T}_{t,s}^{(l,\eta )}$ for $\eta \in \mathbb{R}$.
\end{proof}

\noindent A direct consequence of Theorem \ref{Thm Heat production as power
series copy(2)} is that $\mathbf{T}_{t,s}^{(l,\eta )}=\mathcal{O}(|\Lambda
_{l}|)$. Note furthermore that Theorem \ref{Thm Heat production as power
series copy(2)} also holds when the cubic box $\Lambda _{l}$ is replaced by
\emph{any} finite subset $\Lambda \in \mathcal{P}_{f}(\mathfrak{L})$. The
assumptions of this theorem are fulfilled for any interactions $\Psi ,\Phi
\in \mathcal{W}$ with the decay function (\ref{example polynomial}),
provided the parameter $\epsilon \in \mathbb{R}^{+}$ is sufficiently large.
Theorem \ref{Thm Heat production as power series copy(2)} is thus a \emph{%
significant} extension of \cite[Theorem 3.4 (i)]{OhmI} in the sense that
very general inter--particle interactions and the full range of parameters $%
\eta \in \mathbb{R}$\ are now allowed.

In the case of exponentially decaying interactions we can bound the
derivatives $|\Lambda _{l}|^{-1}\partial _{\eta }^{m}\mathbf{T}%
_{t,s}^{(l,\eta )}$\ for all $m\in \mathbb{N}$, uniformly w.r.t. $l\in
\mathbb{R}_{0}^{+}$. We thus extend \cite[Theorem 3.4 (ii)]{OhmI} for
interactions $\Phi $ satisfying\ (\ref{assumption boundedness2}).

Under these conditions, we show below that the map $\eta \mapsto |\Lambda
_{l}|^{-1}\mathbf{T}_{t,s}^{(l,\eta )}$ from $\mathbb{R}$ to $\mathcal{U}$
is bounded in the sense of Gevrey norms, uniformly w.r.t. $l\in \mathbb{R}%
_{0}^{+}$. Note that real analytic functions (cf. \cite[Theorem 3.4 (ii)]%
{OhmI}) are a special case of Gevrey functions.

\begin{satz}[Increments as Gevrey maps]
\label{Thm Heat production as power series copy(1)}\mbox{
}\newline
Let $l,\mathrm{T}\in \mathbb{R}_{0}^{+}$, $s,t\in \left[ -\mathrm{T},\mathrm{%
T}\right] $, $\Psi \in \mathcal{W}$, and $\mathbf{V}$ be any potential.%
\index{Gevrey map}%
\index{Increment} Assume (\ref{(3.3) NS generalized}) and take an
interaction $\Phi $ satisfying (\ref{assumption boundedness2}) with $\mathbf{%
v}_{m}=\mathrm{e}^{m\varsigma }$. Assume further the real analyticity of the
maps $\eta \mapsto \mathbf{w}_{x,y}(\eta )$, $x,y\in \mathfrak{L}$, from $%
\mathbb{R}$ to $\mathbb{C}$ as well as the existence of $r\in \mathbb{R}^{+}$
such that%
\begin{equation}
K_{2}\doteq \sup_{x,y\in \mathfrak{L}}\ \sup_{m\in \mathbb{N}}\ \sup_{\eta
\in \mathbb{R}}%
\frac{r^{m}\partial _{\eta }^{m}\mathbf{w}_{x,y}(\eta )}{m!}<\infty \ .
\label{assumption boundedness}
\end{equation}%
\emph{(i)} Smoothness. As a function of $\eta \in \mathbb{R}$, $\mathbf{T}%
_{t,s}^{(l,\eta )}\in C^{\infty }(\mathbb{R};\mathcal{U})$ and for any $m\in
\mathbb{N}$,
\begin{eqnarray*}
\partial _{\eta }^{m}\mathbf{T}_{t,s}^{(l,\eta )}
&=&\sum\limits_{k=1}^{m}\sum\limits_{\Lambda \in \mathcal{P}_{f}(\mathfrak{L}%
)}i^{k}\int_{s}^{t}\mathrm{d}s_{1}\cdots \int_{s}^{s_{k-1}}\mathrm{d}s_{k} \\
&&\qquad \left. \partial _{\varepsilon }^{m}\left[ \mathbf{X}%
_{s_{k},s}^{(l,\eta ,\eta +\varepsilon )},\ldots ,\mathbf{X}%
_{s_{1},s}^{(l,\eta ,\eta +\varepsilon )},\tilde{\tau}_{t,s}^{(l,\eta
)}(\Phi _{\Lambda })\right] ^{(k+1)}\right\vert _{\varepsilon =0}\ .
\end{eqnarray*}%
The above series in $\Lambda $ are absolutely convergent.\newline
\emph{(ii)} Uniform boundedness of the Gevrey norm of density of increments.
There exist $\tilde{r}\equiv \tilde{r}_{d,\mathrm{T},\Psi ,K_{2},\mathbf{F}%
}\in \mathbb{R}^{+}$ and $D\equiv D_{\mathrm{T},\Psi ,K_{2},\Phi }\in
\mathbb{R}^{+}$ such that, for all $l\in \mathbb{R}_{0}^{+}$, $\eta \in
\mathbb{R}$ and $s,t\in \left[ -\mathrm{T},\mathrm{T}\right] $,
\begin{equation*}
\sum_{m\in \mathbb{N}}\frac{\tilde{r}^{m}}{\left( m!\right) ^{d}}\sup_{l\in
\mathbb{R}_{0}^{+}}\left\Vert \left\vert \Lambda _{l}\right\vert
^{-1}\partial _{\eta }^{m}\mathbf{T}_{t,s}^{(l,\eta )}\right\Vert _{\mathcal{%
U}}\leq D\ .
\end{equation*}
\end{satz}

Before giving the proof, note first that the assumptions of Theorem \ref{Thm
Heat production as power series copy(1)} are satisfied for any interactions $%
\Psi ,\Phi \in \mathcal{W}$ with the decay function (\ref{example}).
Moreover, under conditions of Theorem \ref{Thm Heat production as power
series copy(1)}, the family $\{|\Lambda _{l}|^{-1}\mathbf{T}_{t,s}^{(l,\eta
)}\}_{l\in \mathbb{R}_{0}^{+}}$ of functions of the variable $\eta $ at
dimension $d=1$ is uniformly bounded w.r.t. analytic norms. In particular,
for $d=1$ and any state $\varrho \in \mathcal{U}^{\ast }$, the limit of the
increment density $|\Lambda _{l}|^{-1}\varrho (\mathbf{T}_{t,s}^{(l,\eta )})$%
, as $l\rightarrow \infty $ (possibly along subsequences), is either
identically vanishing for all $\eta \in \mathbb{R}$, or is different from
zero for $\eta $ outside a discrete subset of $\mathbb{R}$. Note that, by
contrast, general non--vanishing Gevrey functions can have arbitrarily small
support. We discuss this with more details at the end of Section \ref%
{Section existence dynamics}.

We now conclude this subsection by proving Theorem \ref{Thm Heat production
as power series copy(1)}. To this end, we need the following estimate:

\begin{proposition}
\label{Coro estimate super importante}\mbox{ }\newline
There is a constant $D\in \mathbb{R}^{+}$ such that, for all $k\in {\mathbb{N%
}}$,%
\begin{equation*}
\sum_{T\in \mathcal{T}_{k+1}}\max_{j\in \{0,\ldots ,k\}}\max_{x_{j}\in
\mathfrak{L}}\sum_{x_{0},\ldots ,\leavevmode\hbox{\rlap{\thinspace/}{$x$}}%
_{j},\ldots ,x_{k}\in \mathfrak{L}}\ \prod\limits_{\{p,l\}\in T}\mathrm{e}^{-%
\frac{\varsigma \left\vert x_{p}-x_{l}\right\vert }{\sqrt{d}\max \{\mathfrak{%
d}_{T}(p),\mathfrak{d}_{T}(l)\}}}\leq D^{k}(k!)^{d}\text{ }.
\end{equation*}
\end{proposition}

\noindent The proof of this upper bound uses the fact that trees with
vertices of large degree are \textquotedblleft \emph{rare}%
\textquotedblright\ in a way that summing up the numbers $(\mathfrak{d}%
_{T}!)^{\alpha }$ for $T\in \mathcal{T}_{k+1}$ and any $\alpha \in \mathbb{R}%
^{+}$ gives factors behaving, at worse, like $D^{k}(k!)^{\alpha }$. The
arguments are standard results of finite mathematics. We prove them below
for completeness, in two simple lemmata.

Let $k\in {\mathbb{N}}$. For any fixed sequence $\mathfrak{d}=(\mathfrak{d}%
(0),\ldots ,\mathfrak{d}(k))\in {\mathbb{N}}^{k+1}$ define the set $\mathcal{%
T}_{k+1}(\mathfrak{d})\subset \mathcal{T}_{k+1}$ by%
\begin{equation*}
\mathcal{T}_{k+1}(\mathfrak{d})\doteq \{T\in \mathcal{T}_{k+1}:\mathfrak{d}%
_{T}\equiv (\mathfrak{d}_{T}(0),\ldots ,\mathfrak{d}_{T}(k))=\mathfrak{d}\}%
\text{ }.
\end{equation*}%
In other words, $\mathcal{T}_{k+1}(\mathfrak{d})$ is the set of all trees of
$\mathcal{T}_{k+1}$ with vertices having their degree fixed by the sequence $%
\mathfrak{d}$. The cardinality of this set is bounded as follows:

\begin{lemma}[Number of trees with vertices of fixed degrees]
\label{Lemma fix degrees}\mbox{ }\newline
For all $k\in {\mathbb{N}}$ and $\mathfrak{d}\in {\mathbb{N}}^{k+1}$,%
\index{Trees}%
\begin{equation*}
|\mathcal{T}_{k+1}(\mathfrak{d})|\leq
\frac{(k-1)!}{(\mathfrak{d}(0)-1)!\cdots (\mathfrak{d}(k)-1)!}\text{ }.
\end{equation*}
\end{lemma}

\begin{proof}
The bound can be proven, for instance, by using so--called \textquotedblleft
Pr\"{u}fer codes\textquotedblright . We give here a proof based on a
simplified version of such codes, well adapted to the particular sets of
trees $\mathcal{T}_{k+1}$. At fixed $k\in {\mathbb{N}}$, define the map $%
\mathfrak{C}:\mathcal{T}_{k+1}\rightarrow \{0,\ldots ,k-1\}^{k-1}$ by%
\begin{equation*}
\mathfrak{C}(T)\doteq (\mathrm{P}_{T}(2),\ldots ,\mathrm{P}_{T}(k))\text{ }.
\end{equation*}%
See (\ref{def.tree})--(\ref{def tree3}). This map is clearly injective and
if $j\in \{0,\ldots ,k\}$ is a vertex of degree $\mathfrak{d}_{T}(j)$, then
it appears exactly $(\mathfrak{d}_{T}(j)-1)$ times in the sequence $%
\mathfrak{C}(T)$. Note that $\mathfrak{d}_{T}(k)=1$ for all $T\in \mathcal{T}%
_{k+1}$. To finish the proof, fix $\mathfrak{d}=(\mathfrak{d}(0),\ldots ,%
\mathfrak{d}(k))\in {\mathbb{N}}^{k+1}$ and observe that if $\mathfrak{d}%
(0)+\cdots +\mathfrak{d}(k)=2k$ then there are exactly%
\begin{equation*}
\frac{(k-1)!}{(\mathfrak{d}(0)-1)!\cdots (\mathfrak{d}(k)-1)!}
\end{equation*}%
sequences in $\{0,\ldots ,k-1\}^{k-1}$ with $j\in \{0,\ldots ,k\}$ appearing
exactly $(\mathfrak{d}(j)-1)$ times in such sequences. If $\mathfrak{d}%
(0)+\cdots +\mathfrak{d}(k)\neq 2k$ then such a sequence does not exist.
\end{proof}

\begin{lemma}
\label{Lemma estimate sums}\mbox{ }\newline
For all $k\in {\mathbb{N}}$,%
\begin{equation*}
\sum\limits_{\mathfrak{d}(0),\ldots ,\mathfrak{d}(k)\in {\mathbb{N}}}\mathbf{%
1}[\mathfrak{d}(0)+\cdots +\mathfrak{d}(k)=2k]\leq 4^{k}\ .
\end{equation*}
\end{lemma}

\begin{proof}
For $k\in {\mathbb{N}}$, the coefficient $c_{2k}$ of the analytic function%
\begin{equation*}
z\mapsto \frac{z^{k+1}}{(1-z)^{k+1}}=\sum_{m=1}^{\infty }c_{m}z^{m}
\end{equation*}%
on the complex disc $\{z\in \mathbb{C}$ $:$ $|z|<1\}$ is exactly the finite
sum%
\begin{equation*}
\sum\limits_{\mathfrak{d}(0),\ldots ,\mathfrak{d}(k)\in {\mathbb{N}}}\mathbf{%
1}[\mathfrak{d}(0)+\cdots +\mathfrak{d}(k)=2k]\text{ }.
\end{equation*}%
In particular,%
\begin{equation*}
\sum\limits_{\mathfrak{d}(0),\ldots ,\mathfrak{d}(k)\in {\mathbb{N}}}\mathbf{%
1}[\mathfrak{d}(0)+\cdots +\mathfrak{d}(k)=2k]=\frac{1}{2\pi i}%
\oint\limits_{|z|=1/2}\frac{1}{z^{k}(1-z)^{k+1}}\mathrm{d}z\text{ },
\end{equation*}%
which combined with the inequality%
\begin{equation*}
\left\vert \frac{1}{2\pi i}\oint\limits_{|z|=1/2}\frac{1}{z^{k}(1-z)^{k+1}}%
\mathrm{d}z\right\vert \leq 4^{k}
\end{equation*}%
yields the assertion.
\end{proof}

\noindent By using the two above lemmata, we now prove Proposition \ref{Coro
estimate super importante}:\medskip

\begin{proof}
Fix $\alpha \in \mathbb{R}^{+}$ and note first that, for all $d\in {\mathbb{N%
}}$,
\begin{equation*}
\lim_{g\rightarrow \infty }\frac{1}{g^{d}}\sum_{x\in \mathfrak{L}}\mathrm{e}%
^{-\frac{\alpha |x|}{g\sqrt{d}}}=\int_{\mathbb{R}^{d}}\mathrm{e}^{-\frac{%
\alpha |x|}{\sqrt{d}}}\mathrm{d}^{d}x<\infty \text{ }.
\end{equation*}%
Hence, for $d\in {\mathbb{N}}$, there is a constant $S_{d}\in \mathbb{R}^{+}$
such that%
\begin{equation*}
\sum_{x\in \mathfrak{L}}\mathrm{e}^{-\frac{\alpha |x|}{g\sqrt{d}}}\leq
S_{d}g^{d}\ ,\qquad g\in \mathbb{N}\ .
\end{equation*}%
From this estimate and by using the Stirling--type bounds \cite{Robb1}%
\index{Stirling--type bounds}
\begin{equation}
g^{g}\mathrm{e}^{-g}\mathrm{e}^{%
\frac{1}{12g+1}}\sqrt{2\pi g}\leq g!\leq g^{g}\mathrm{e}^{-g}\mathrm{e}^{%
\frac{1}{12g}}\sqrt{2\pi g}\ ,\qquad g\in \mathbb{N}\ ,
\label{stirling lowerbound}
\end{equation}%
we obtain%
\begin{eqnarray}
&&\max_{j\in \{0,\ldots ,k\}}\max_{x_{j}\in \mathfrak{L}}\sum_{x_{0},\ldots ,%
\leavevmode\hbox{\rlap{\thinspace/}{$x$}}_{j},\ldots ,x_{k}\in \mathfrak{L}%
}\ \prod\limits_{\{p,l\}\in T}\exp \left( -\frac{\varsigma \left\vert
x_{p}-x_{l}\right\vert }{\sqrt{d}\max \{\mathfrak{d}_{T}(p),\mathfrak{d}%
_{T}(l)\}}\right)  \notag \\
&\leq &S_{d}^{k}\prod\limits_{j=0}^{k}\mathfrak{d}_{T}(j)^{\mathfrak{d}%
_{T}(j)d}\leq S_{d}^{k}\mathrm{e}^{\mathfrak{d}_{T}(j)d}(\mathfrak{d}%
_{T}!)^{d}  \label{ine pas si easy}
\end{eqnarray}%
for all $d,k\in {\mathbb{N}}$ and $T\in \mathcal{T}_{k+1}$. We infer from (%
\ref{inequality easy2}) that%
\begin{equation}
\sum_{T\in \mathcal{T}_{k+1}}\left( \mathfrak{d}_{T}!\right) ^{d}\leq
(k!)^{d-1}\sum_{T\in \mathcal{T}_{k+1}}\left( \mathfrak{d}_{T}!\right) \ .
\label{ine easy2}
\end{equation}%
We use now Lemma \ref{Lemma fix degrees} to get
\begin{eqnarray*}
\sum_{T\in \mathcal{T}_{k+1}}\left( \mathfrak{d}_{T}!\right) &=&\sum\limits_{%
\mathfrak{d}(0),\ldots ,\mathfrak{d}(k)\in {\mathbb{N}}}\mathbf{1}[\mathfrak{%
d}(0)+\cdots +\mathfrak{d}(k)=2k]\sum_{T\in \mathcal{T}_{k+1}((\mathfrak{d}%
(0),\ldots ,\mathfrak{d}(k)))}(\mathfrak{d}_{T}!) \\
&\leq &k!\sum\limits_{\mathfrak{d}(0),\ldots ,\mathfrak{d}(k)\in {\mathbb{N}}%
}\mathbf{1}[\mathfrak{d}(0)+\cdots +\mathfrak{d}(k)=2k]\ \mathfrak{d}%
(0)\cdots \mathfrak{d}(k) \\
&\leq &k!\sum\limits_{\mathfrak{d}(0),\ldots ,\mathfrak{d}(k)\in {\mathbb{N}}%
}\mathbf{1}[\mathfrak{d}(0)+\cdots +\mathfrak{d}(k)=2k]\ \mathrm{e}^{%
\mathfrak{d}(0)}\cdots \mathrm{e}^{\mathfrak{d}(k)}\ .
\end{eqnarray*}%
We invoke (\ref{inequality easy}) and Lemma \ref{Lemma estimate sums} to
arrive at%
\begin{equation}
\sum_{T\in \mathcal{T}_{k+1}}\left( \mathfrak{d}_{T}!\right) \leq (k!)%
\mathrm{e}^{2k}\sum\limits_{\mathfrak{d}(0),\ldots ,\mathfrak{d}(k)\in {%
\mathbb{N}}}\mathbf{1}[\mathfrak{d}(0)+\cdots +\mathfrak{d}(k)=2k]\leq (k!)(4%
\mathrm{e}^{2})^{k}\text{ }.  \label{ine easy3}
\end{equation}%
Proposition \ref{Coro estimate super importante} is then a consequence of (%
\ref{ine pas si easy}), (\ref{ine easy2}) and (\ref{ine easy3}).
\end{proof}

\noindent We are now in position to prove Theorem \ref{Thm Heat production
as power series copy(1)}:\smallskip

\begin{proof}
(i) Observe that
\begin{equation}
\partial _{\eta }^{m}\mathbf{T}_{t,s}^{(l,\eta ,L)}=\partial _{\varepsilon
}^{m}\left. (\mathbf{T}_{t,s}^{(l,\eta +\varepsilon ,L)}-\mathbf{T}%
_{t,s}^{(l,\eta ,L)})\right\vert _{\varepsilon =0}\ .  \label{ge}
\end{equation}%
The difference $\mathbf{T}_{t,s}^{(l,\eta +\varepsilon ,L)}-\mathbf{T}%
_{t,s}^{(l,\eta ,L)}$ is explicitly given by a
\index{Dyson--Phillips series}Dyson--Phillips series involving
multi--commutators (\ref{multi1-0})--(\ref{multi2-0}): Use (\ref{finite
dyson}) to produce an infinite series. As the function $\eta \mapsto \mathbf{%
W}^{(l,\eta )}$ is, by assumption, real analytic, it follows that
\begin{eqnarray}
&&\partial _{\varepsilon }^{m}\left. (\mathbf{T}_{t,s}^{(l,\eta +\varepsilon
,L)}-\mathbf{T}_{t,s}^{(l,\eta ,L)})\right\vert _{\varepsilon =0}=
\label{generic Dyson--Phillips series} \\
&&\sum\limits_{k=1}^{m}i^{k}\int_{s}^{t}\mathrm{d}s_{1}\cdots
\int_{s}^{s_{k-1}}\mathrm{d}s_{k}\partial _{\varepsilon }^{m}\left. \left[
\mathbf{X}_{s_{k},s}^{(l,\eta ,\eta +\varepsilon )},\ldots ,\mathbf{X}%
_{s_{1},s}^{(l,\eta ,\eta +\varepsilon )},%
\tilde{\tau}_{t,s}^{(l,\eta )}(U_{\Lambda _{L}}^{\Phi })\right]
^{(k+1)}\right\vert _{\varepsilon =0}  \notag
\end{eqnarray}%
for any $m\in \mathbb{N}$, $l\in \mathbb{R}_{0}^{+}$, and $s,t,\eta \in
\mathbb{R}$. Set%
\begin{equation*}
\xi _{x_{1},z_{1},\ldots ,x_{k},z_{k}}\doteq \left. \partial _{\varepsilon
}^{m}\left\{ \prod\limits_{j=1}^{k}\left( \mathbf{w}_{x_{j},x_{j}+z_{j}}(%
\eta +\varepsilon )-\mathbf{w}_{x_{j},x_{j}+z_{j}}(\eta )\right) \right\}
\right\vert _{\varepsilon =0}\ .
\end{equation*}%
By (\ref{assumption boundedness}), these coefficients are uniformly bounded
w.r.t. $x_{1},z_{1},\ldots ,x_{k},z_{k}$ and $\eta $:
\begin{equation}
\sup_{x_{1},z_{1},\ldots ,x_{k},z_{k}\in \mathfrak{L}}\sup_{\eta \in \mathbb{%
R}}|\xi _{x_{1},z_{1},\ldots ,x_{k},z_{k}}|\leq D^{m}m!  \label{totoes}
\end{equation}%
for some constant $D\in \mathbb{R}^{+}$ depending on $K_{2}$ but not on $%
m\geq k$. Bounding the above multi--commutators exactly as done for the
proof of Theorem \ref{Thm Heat production as power series copy(2)} and by
taking the limit $L\rightarrow \infty $, we deduce from (\ref{ge})--(\ref%
{generic Dyson--Phillips series}) that, for any $m\in \mathbb{N}$ and $%
s,t,\eta \in \mathbb{R}$,
\begin{eqnarray}
\underset{L\rightarrow \infty }{\lim }\partial _{\eta }^{m}\mathbf{T}%
_{t,s}^{(l,\eta ,L)} &=&\sum\limits_{k=1}^{m}\sum\limits_{\Lambda \in
\mathcal{P}_{f}(\mathfrak{L})}i^{k}\int_{s}^{t}\mathrm{d}s_{1}\cdots
\int_{s}^{s_{k-1}}\mathrm{d}s_{k}  \label{limit} \\
&&\qquad \left. \partial _{\varepsilon }^{m}\left[ \mathbf{X}%
_{s_{k},s}^{(l,\eta ,\eta +\varepsilon )},\ldots ,\mathbf{X}%
_{s_{1},s}^{(l,\eta ,\eta +\varepsilon )},\tilde{\tau}_{t,s}^{(l,\eta
)}(\Phi _{\Lambda })\right] ^{(k+1)}\right\vert _{\varepsilon =0}\ .  \notag
\end{eqnarray}%
This limit is uniform for $\eta \in \mathbb{R}$ because of (\ref{totoes}).
As in Theorem \ref{Thm Heat production as power series copy(2)} (ii), the
above series in $\Lambda $ are absolutely convergent. Moreover, the uniform
convergence of $\partial _{\eta }^{m}\mathbf{T}_{t,s}^{(l,\eta ,L)}$, $m\in
\mathbb{N}$, together with Theorem \ref{Thm Heat production as power series
copy(2)} (i) implies that the energy increment limit $\mathbf{T}%
_{t,s}^{(l,\eta )}$ is a smooth function of $\eta $ with $m$--derivatives%
\begin{equation*}
\partial _{\eta }^{m}\mathbf{T}_{t,s}^{(l,\eta )}=\underset{L\rightarrow
\infty }{\lim }\partial _{\eta }^{m}\mathbf{T}_{t,s}^{(l,\eta ,L)}
\end{equation*}%
for all $m\in \mathbb{N}$ and $s,t,\eta \in \mathbb{R}$. Because of (\ref%
{limit}), Assertion (i) thus follows.

\noindent (ii) is a direct consequence of (i), Corollary \ref{theorem exp
tree decay}, and Proposition \ref{Coro estimate super importante} together
with (\ref{totoes}) and
\begin{equation*}
\int_{s}^{t}\mathrm{d}s_{1}\cdots \int_{s}^{s_{k-1}}\mathrm{d}s_{k}\leq
\frac{(2\mathrm{T})^{k}}{k!}\ .
\end{equation*}
\end{proof}

\section{Lieb--Robinson Bounds for Non--Autonomous Dynamics \label{section
LR non-auto}}

Like in Section \ref{Generalized Lieb--Robinson Bounds}, we only consider
fermion systems, but all results can easily be extended to quantum spin
systems (Section \ref{Quantum spin systems}). For quantum spin systems, note
that Lieb--Robinson bounds for non--autonomous dynamics have already been
considered in \cite{BMNS}. However, \cite{BMNS} only proves Lieb--Robinson
bounds for commutators, while the multi--commutator case was not considered,
in contrast with results of this section. Observe also that some aspects of
the non--autonomous case can be treated in a similar way to the autonomous
case. However, several important arguments cannot be directly extended to
the non--autonomous situation. Here, we only address in detail the technical
issues which are specific to the non--autonomous problem. See for instance
Corollary \ref{Theorem Lieb-Robinson copy(4)} (iii), Lemma \ref{Lemma Series
representation of the dynamics copy(1)}, Theorem \ref{thm non auto}, and
Theorem \ref{pertubed dynam thm}.

\subsection{Existence of Non--Autonomous Dynamics}

We now consider time--dependent models. So, let $\Psi \doteq \{\Psi
^{(t)}\}_{t\in \mathbb{R}}$ be a map from $\mathbb{R}$ to $\mathcal{W}$ such
that%
\index{Interaction!time--dependent}
\begin{equation*}
\left\Vert \Psi \right\Vert _{\infty }\doteq \sup_{t\in \mathbb{R}%
}\left\Vert \Psi ^{(t)}\right\Vert _{\mathcal{W}}<\infty \ .
\end{equation*}%
I.e., $\{\Psi ^{(t)}\}_{t\in \mathbb{R}}$ is a \emph{bounded} family in $%
\mathcal{W}$. We could easily extend the study of this section to families $%
\{\Psi ^{(t)}\}_{t\in \mathbb{R}}$ which are only bounded for $t$ on
compacta. We refrain from considering this more general case, for
simplicity. Take, furthermore, any collection $\{\mathbf{V}^{(t)}\}_{t\in
\mathbb{R}}$ of potentials. Note that (\ref{condition divergence}) is
allowed for any $t\in \mathbb{R}$.

For all $x\in \mathfrak{L}$ and $\Lambda \in \mathcal{P}_{f}(\mathfrak{L})$,
assume the continuity of the two maps $t\mapsto \Psi _{\Lambda }^{(t)}$, $%
t\mapsto \mathbf{V}_{\left\{ x\right\} }^{(t)}$ from $\mathbb{R}$ to $%
\mathcal{U}$, i.e., $\Psi _{\Lambda },\mathbf{V}_{\left\{ x\right\} }\in
C\left( \mathbb{R};\mathcal{U}\right) $. For any $L\in \mathbb{R}_{0}^{+}$,
this yields the existence, uniqueness and an explicit expression, as a
Dyson--Phillips series (cf. (\ref{dyson series})), of the solution $\{\tau
_{t,s}^{(L)}\}_{_{s,t\in \mathbb{R}}}$ of the (finite--volume) non--auto%
\-%
nomous evolutions equations%
\index{Evolution equation!non--autonomous}%
\index{Interaction!finite--volume dynamics}%
\index{Potential!finite--volume dynamics}%
\begin{equation}
\forall s,t\in {\mathbb{R}}:\qquad \partial _{s}\tau _{t,s}^{(L)}=-\delta
_{s}^{(L)}\circ \tau _{t,s}^{(L)}\ ,\qquad \tau _{t,t}^{(L)}=\mathbf{1}_{%
\mathcal{U}}\ ,  \label{cauchy1}
\end{equation}%
and%
\index{Evolution equation!non--autonomous (H1)}%
\begin{equation}
\forall s,t\in {\mathbb{R}}:\qquad \partial _{t}\tau _{t,s}^{(L)}=\tau
_{t,s}^{(L)}\circ \delta _{t}^{(L)}\ ,\qquad \tau _{s,s}^{(L)}=\mathbf{1}_{%
\mathcal{U}}\ .  \label{cauchy2}
\end{equation}%
Here, for any $t\in {\mathbb{R}}$ and $L\in \mathbb{R}_{0}^{+}$, the bounded
linear operator $\delta _{t}^{(L)}$ is defined on $\mathcal{U}$ by
\begin{equation*}
\delta _{t}^{(L)}(B)\doteq i\sum\limits_{\Lambda \subseteq \Lambda _{L}}
\left[ \Psi _{\Lambda }^{(t)},B\right] +i\sum\limits_{x\in \Lambda _{L}}%
\left[ \mathbf{V}_{\left\{ x\right\} }^{(t)},B\right] \ ,\qquad B\in
\mathcal{U}\ .
\end{equation*}%
Compare this definition with (\ref{dynamic series}). As explained in\
Section \ref{sect non auto} (see in particular Equations (\ref{non auto
explaination3})--(\ref{non auto explaination3bis})), recall that the natural
non--auto%
\-%
nomous evolution equation in Quantum Mechanics is (\ref{cauchy2}), but, by
boundedness of $\delta _{t}^{(L)}$ for all times, (\ref{cauchy1}) and (\ref%
{cauchy2}) are both satisfied.

Similar to the autonomous case, for any $L\in \mathbb{R}_{0}^{+}$, $\{\tau
_{t,s}^{(L)}\}_{_{s,t\in \mathbb{R}}}$ is a continuous two--para%
\-%
meter family of bounded operators that satisfies the (reverse) cocycle
property%
\index{Cocycle property!reverse}
\begin{equation}
\forall s,r,t\in \mathbb{R}:\qquad \tau _{t,s}^{(L)}=\tau _{r,s}^{(L)}\tau
_{t,r}^{(L)}\ .  \label{cocycle}
\end{equation}%
Its time--dependent generator $\delta _{t}^{(L)}$ is clearly a symmetric
derivation and $\tau _{t,s}^{(L)}$ is thus a $\ast $--auto%
\-%
morphism on $\mathcal{U}$ for all $L\in \mathbb{R}_{0}^{+}$ and $s,t\in
\mathbb{R}$. Moreover, similar to the autonomous case (cf. Theorem \ref%
{Theorem Lieb-Robinson} and Lemma \ref{Theorem Lieb-Robinson copy(2)}), for
all $L\in \mathbb{R}_{0}^{+}$ and $s,t\in \mathbb{R}$, $\tau _{t,s}^{(L)}$
satisfies Lieb--Robinson bounds and thus converges in the strong sense on $%
\mathcal{U}_{0}$, as $L\rightarrow \infty $:

\begin{satz}[Properties of non--autonomous finite--volume dynamics]
\label{Theorem Lieb-Robinsonnew}\mbox{
}\newline
Let $\Psi \doteq \{\Psi ^{(t)}\}_{t\in \mathbb{R}}$ be a bounded family on $%
\mathcal{W}$ (i.e., $\left\Vert \Psi \right\Vert _{\infty }<\infty $ ) and $%
\{\mathbf{V}^{(t)}\}_{t\in \mathbb{R}}$ a collection of potentials. For any $%
x\in \mathfrak{L}$ and $\Lambda \in \mathcal{P}_{f}(\mathfrak{L})$, assume $%
\Psi _{\Lambda },\mathbf{V}_{\left\{ x\right\} }\in C\left( \mathbb{R};%
\mathcal{U}\right) $. Fix $s,t\in \mathbb{R}$. \newline
\emph{(i)}
\index{Lieb--Robinson bounds}Lieb--Robinson bounds. For any $L\in \mathbb{R}%
_{0}^{+}$, $B_{1}\in \mathcal{U}^{+}\cap \mathcal{U}_{\Lambda ^{(1)}}$, and $%
B_{2}\in \mathcal{U}_{\Lambda ^{(2)}}$ with $\Lambda ^{(1)},\Lambda
^{(2)}\subsetneq \Lambda _{L}$ and $\Lambda ^{(1)}\cap \Lambda
^{(2)}=\emptyset $,
\begin{eqnarray*}
&&\left\Vert [\tau _{t,s}^{(L)}\left( B_{1}\right) ,B_{2}]\right\Vert _{%
\mathcal{U}} \\
&\leq &2\mathbf{D}^{-1}\left\Vert B_{1}\right\Vert _{\mathcal{U}}\left\Vert
B_{2}\right\Vert _{\mathcal{U}}\left( \mathrm{e}^{2\mathbf{D}\left\vert
t-s\right\vert \left\Vert \Psi \right\Vert _{\infty }}-1\right) \sum_{x\in
\partial _{\Psi }\Lambda ^{(1)}}\sum_{y\in \Lambda ^{(2)}}\mathbf{F}\left(
\left\vert x-y\right\vert \right) \ .
\end{eqnarray*}%
\emph{(ii)} Convergence of the finite--volume dynamics. For any $\Lambda \in
\mathcal{P}_{f}(\mathfrak{L})$, $B\in \mathcal{U}_{\Lambda }$, and $%
L_{1},L_{2}\in \mathbb{R}_{0}^{+}$ with $\Lambda \subset \Lambda
_{L_{1}}\varsubsetneq \Lambda _{L_{2}}$,
\begin{eqnarray*}
&&\left\Vert \tau _{t,s}^{(L_{2})}\left( B\right) -\tau
_{t,s}^{(L_{1})}\left( B\right) \right\Vert _{\mathcal{U}} \\
&\leq &2\left\Vert B\right\Vert _{\mathcal{U}}\left\Vert \Psi \right\Vert
_{\infty }\left\vert t-s\right\vert \mathrm{e}^{4\mathbf{D}\left\vert
t-s\right\vert \left\Vert \Psi \right\Vert _{\infty }}\sum\limits_{y\in
\Lambda _{L_{2}}\backslash \Lambda _{L_{1}}}\sum_{x\in \Lambda }\mathbf{F}%
\left( \left\vert x-y\right\vert \right) \ .
\end{eqnarray*}
\end{satz}

\begin{proof}
(i) The arguments are a straightforward extension of those proving Theorem %
\ref{Theorem Lieb-Robinson} to non--auto%
\-%
nomous dynamics: Fix $L\in \mathbb{R}_{0}^{+}$, $B_{1}\in \mathcal{U}%
^{+}\cap \mathcal{U}_{\Lambda ^{(1)}}$ and $B_{2}\in \mathcal{U}_{\Lambda
^{(2)}}$ with disjoint sets $\Lambda ^{(1)},\Lambda ^{(2)}\subsetneq \Lambda
_{L}$. Similar to (\ref{iteration1new})--(\ref{derivativenew}), we infer
from (\ref{cauchy1})--(\ref{cauchy2}) that the derivative w.r.t. to $t$ of
the function
\begin{equation*}
f\left( s,t\right) \doteq \left[ \tau _{t,s}^{(L)}\circ \tau
_{s,t}^{(\Lambda ^{(1)})}\left( B_{1}\right) ,B_{2}\right] \ ,\qquad s,t\in {%
\mathbb{R}}\ ,
\end{equation*}%
equals
\begin{eqnarray}
\partial _{t}f\left( s,t\right) &=&i\sum\limits_{\mathcal{Z}\in \mathcal{S}%
_{\Lambda _{L}}(\Lambda ^{(1)})}\left[ \tau _{t,s}^{(L)}(\Psi _{\mathcal{Z}%
}^{(t)}),f\left( s,t\right) \right]  \label{derivativenewderivativenew} \\
&&-i\sum\limits_{\mathcal{Z}\in \mathcal{S}_{\Lambda _{L}}(\Lambda ^{(1)})}%
\left[ \tau _{t,s}^{(L)}\circ \tau _{s,t}^{(\Lambda ^{(1)})}\left(
B_{1}\right) ,\left[ \tau _{t,s}^{(L)}(\Psi _{\mathcal{Z}}^{(t)}),B_{2}%
\right] \right] \ .  \notag
\end{eqnarray}%
Exactly like (\ref{iteration2new0}), it follows that%
\begin{equation*}
\left\Vert f\left( s,t\right) \right\Vert _{\mathcal{U}}\leq \left\Vert
f\left( s,s\right) \right\Vert _{\mathcal{U}}+2\left\Vert B_{1}\right\Vert _{%
\mathcal{U}}\sum\limits_{\mathcal{Z}\in \mathcal{S}_{\Lambda _{L}}(\Lambda
^{(1)})}\int_{\min \{s,t\}}^{\max \{s,t\}}\left\Vert \left[ \tau _{\alpha
,s}^{(L)}(\Psi _{\mathcal{Z}}^{(\alpha )}),B_{2}\right] \right\Vert _{%
\mathcal{U}}\mathrm{d}\alpha
\end{equation*}%
for any $s,t\in \mathbb{R}$. Therefore, by using estimates that are similar
to (\ref{iteration2new})--(\ref{final1}), we deduce Assertion (i).

\noindent (ii) The arguments are extensions to the non--auto%
\-%
nomous case of those proving Lemma \ref{Theorem Lieb-Robinson copy(2)}:
Since $\Psi _{\Lambda },\mathbf{V}_{\left\{ x\right\} }\in C\left( \mathbb{R}%
;\mathcal{U}\right) $ for any $x\in \mathfrak{L}$ and $\Lambda \in \mathcal{P%
}_{f}(\mathfrak{L})$, the time--dependent energy observables%
\begin{equation*}
H_{L}^{(t)}\doteq \sum\limits_{\Lambda \subseteq \Lambda _{L}}\Psi _{\Lambda
}^{(t)}+\sum\limits_{x\in \Lambda _{L}}\mathbf{V}_{\left\{ x\right\}
}^{(t)}\ ,\qquad L\in \mathbb{R}_{0}^{+}\ ,\ \ t\in \mathbb{R}\ ,
\end{equation*}%
and potentials
\begin{equation*}
\mathbf{V}_{\mathcal{Z}}^{\left( t\right) }\doteq \sum\limits_{x\in \mathcal{%
Z}}\mathbf{V}_{\left\{ x\right\} }^{\left( t\right) }\in \mathcal{U}^{+}\cap
\mathcal{U}_{\mathcal{Z}}\ ,\qquad \mathcal{Z}\in \mathcal{P}_{f}(\mathfrak{L%
})\ ,\ \ t\in \mathbb{R}\ ,
\end{equation*}%
generate two solutions $\{\mathcal{V}_{s,t}(H_{L})\}_{s,t\in \mathbb{R}}$
and $\{\mathcal{V}_{s,t}(\mathbf{V}_{\mathcal{Z}})\}_{s,t\in \mathbb{R}}$,
respectively, of the non--autonomous evolution equations%
\begin{equation}
\partial _{t}\left( \mathcal{V}_{s,t}(X)\right) =i\mathcal{V}%
_{s,t}(X)X^{(t)}\qquad
\text{and}\qquad \partial _{s}\left( \mathcal{V}_{s,t}(X)\right)
=-iX^{\left( s\right) }\mathcal{V}_{s,t}(X)  \label{non-auto evol eq}
\end{equation}%
with $X^{(t)}=H_{L}^{(t)}$ or $\mathbf{V}_{\mathcal{Z}}^{\left( t\right) }$.
These evolution families satisfy $\mathcal{V}_{t,t}(X)=\mathbf{1}_{\mathcal{U%
}}$ for $t\in \mathbb{R}$ as well as the (usual) cocycle
(Chapman--Kolmogorov) property%
\index{Cocycle property}
\begin{equation}
\forall t,r,s\in \mathbb{R}:\qquad \mathcal{V}_{s,t}(X)=\mathcal{V}_{s,r}(X)%
\mathcal{V}_{r,t}(X)\ .  \label{cocylce}
\end{equation}%
For any $L\in \mathbb{R}_{0}^{+}$ and $s,t,\alpha \in \mathbb{R}$, we then
replace (\ref{unitary propagator}) in the proof of Lemma \ref{Theorem
Lieb-Robinson copy(2)} with
\begin{equation}
\mathbf{U}_{L}\left( t,\alpha \right) \doteq \mathcal{V}_{s,t}(\mathbf{V}%
_{\Lambda _{L}})\mathcal{V}_{t,\alpha }(H_{L})\mathcal{V}_{\alpha ,s}(%
\mathbf{V}_{\Lambda _{L}})%
\text{ }.  \label{interactino picture0}
\end{equation}%
By (\ref{cocylce}), $\mathbf{U}_{L}\left( t,t\right) =\mathbf{1}_{\mathcal{U}%
}$ for all $t\in \mathbb{R}$ while%
\begin{equation}
\partial _{t}\mathbf{U}_{L}\left( t,\alpha \right) =-iG_{L}\left( t\right)
\mathbf{U}_{L}\left( t,\alpha \right) \text{\quad and\quad }\partial
_{\alpha }\mathbf{U}_{L}\left( t,\alpha \right) =i\mathbf{U}_{L}\left(
t,\alpha \right) G_{L}\left( \alpha \right)  \label{interactino picture1}
\end{equation}%
with%
\begin{equation}
G_{L}\left( t\right) \doteq \sum\limits_{\mathcal{Z}\subseteq \Lambda _{L}}%
\mathcal{V}_{s,t}(\mathbf{V}_{\Lambda _{L}})\ \Psi _{\mathcal{Z}}\ \mathcal{V%
}_{t,s}(\mathbf{V}_{\Lambda _{L}})\ .  \label{interactino picture1bis}
\end{equation}%
Using the notation%
\begin{equation}
\tilde{\tau}_{t,s}^{(L)}\left( B\right) \doteq \mathbf{U}_{L}\left(
s,t\right) B\mathbf{U}_{L}\left( t,s\right) \ ,\text{\qquad }B\in \mathcal{U}%
_{\Lambda }\ ,  \label{interactino picture}
\end{equation}%
for any $s,t\in \mathbb{R}$ and $L\in \mathbb{R}_{0}^{+}$ such that $\Lambda
\subset \Lambda _{L}$, observe that%
\begin{equation}
\tau _{t,s}^{(L)}\left( B\right) =\mathcal{V}_{s,t}(H_{L})B\mathcal{V}%
_{t,s}(H_{L})=\tilde{\tau}_{t,s}^{(L)}\left( \mathcal{V}_{s,t}(\mathbf{V}%
_{\Lambda })B\mathcal{V}_{t,s}(\mathbf{V}_{\Lambda })\right) \text{ }.
\label{interactino picture2}
\end{equation}%
Note that, for any $s,t\in \mathbb{R}$, $\Lambda ,\mathcal{Z}\in \mathcal{P}%
_{f}(\mathfrak{L})$ and $B\in \mathcal{U}_{\Lambda }$,
\begin{equation}
\mathcal{V}_{s,t}(\mathbf{V}_{\mathcal{Z}})B\mathcal{V}_{t,s}(\mathbf{V}_{%
\mathcal{Z}})\in \mathcal{U}_{\Lambda }\quad \text{and}\quad \left\Vert
\mathcal{V}_{s,t}(\mathbf{V}_{\mathcal{Z}})B\mathcal{V}_{t,s}(\mathbf{V}_{%
\mathcal{Z}})\right\Vert _{\mathcal{U}}=\left\Vert B\right\Vert _{\mathcal{U}%
}\ .  \label{inequality idiote utile}
\end{equation}%
Hence, it suffices to study the net $\{\tilde{\tau}_{t,s}^{(L)}\left(
B\right) \}_{L\in \mathbb{R}_{0}^{+}}$ with $B\in \mathcal{U}_{\Lambda }$.
Up to straightforward modifications taking into account the initial time $%
s\in \mathbb{R}$, the remaining part of the proof is now identical to the
arguments starting from Equation (\ref{eqaulity dyna1}) in the proof of
Lemma \ref{Theorem Lieb-Robinson copy(2)}.
\end{proof}

\begin{koro}[Infinite--volume dynamics]
\label{Theorem Lieb-Robinson copy(4)}\mbox{
}\newline
Under the conditions of Theorem \ref{Theorem Lieb-Robinsonnew},
finite--volume families $\{\tau _{t,s}^{(L)}\}_{s,t\in {\mathbb{R}}}$, $L\in
\mathbb{R}_{0}^{+}$, converge strongly and uniformly for $s,t$ on compact
sets to a strongly continuous two--para%
\-%
meter family%
\index{Interaction!infinite--volume dynamics}%
\index{Potential!infinite--volume dynamics} $\{\tau _{t,s}\}_{s,t\in {%
\mathbb{R}}}$ of $\ast $--auto%
\-%
morphisms on $\mathcal{U}$ satisfying the following properties:\newline
\emph{(i)}
\index{Cocycle property!reverse}Reverse cocycle property.
\begin{equation*}
\forall s,r,t\in \mathbb{R}:\qquad \tau _{t,s}=\tau _{r,s}\tau _{t,r}\ .
\end{equation*}%
\emph{(ii)}
\index{Lieb--Robinson bounds}Lieb--Robinson bounds. For any $s,t\in \mathbb{R%
}$, $B_{1}\in \mathcal{U}^{+}\cap \mathcal{U}_{\Lambda ^{(1)}}$, and $%
B_{2}\in \mathcal{U}_{\Lambda ^{(2)}}$ with disjoint sets $\Lambda
^{(1)},\Lambda ^{(2)}\in \mathcal{P}_{f}(\mathfrak{L})$,
\begin{eqnarray*}
&&\left\Vert [\tau _{t,s}\left( B_{1}\right) ,B_{2}]\right\Vert _{\mathcal{U}%
} \\
&\leq &2\mathbf{D}^{-1}\left\Vert B_{1}\right\Vert _{\mathcal{U}}\left\Vert
B_{2}\right\Vert _{\mathcal{U}}\left( \mathrm{e}^{2\mathbf{D}\left\vert
t-s\right\vert \left\Vert \Psi \right\Vert _{\infty }}-1\right) \sum_{x\in
\partial _{\Psi }\Lambda ^{(1)}}\sum_{y\in \Lambda ^{(2)}}\mathbf{F}\left(
\left\vert x-y\right\vert \right) \ .
\end{eqnarray*}%
\emph{(iii)}
\index{Evolution equation!non--autonomous (H1)}Non--autonomous evolution
equation. If $\Psi \in C(\mathbb{R};\mathcal{W})$ then $\{\tau
_{t,s}\}_{s,t\in {\mathbb{R}}}$ is the unique family of bounded operators on
$\mathcal{U}$ satisfying, in the strong sense on the dense domain $\mathcal{U%
}_{0}\subset \mathcal{U}$,%
\begin{equation}
\forall s,t\in {\mathbb{R}}:\qquad \partial _{t}\tau _{t,s}=\tau _{t,s}\circ
\delta _{t}\ ,\qquad \tau _{s,s}=\mathbf{1}_{\mathcal{U}}\ .
\label{cauchy trivial1}
\end{equation}%
Here, $\delta _{t}$, $t\in {\mathbb{R}}$, are the conservative closed
symmetric derivations, with common core $\mathcal{U}_{0}$, associated with
the interactions $\Psi ^{(t)}\in \mathcal{W}$ and the potentials $\mathbf{V}%
^{(t)}$. See Theorem \ref{Theorem Lieb-Robinson copy(3)}.
\end{koro}

\begin{proof}
The existence of a strongly continuous two--parameter family $\{\tau
_{t,s}\}_{s,t\in {\mathbb{R}}}$ of $\ast $--automorphisms satisfying
Lieb--Robinson bounds (ii) is a direct consequence of Theorem \ref{Theorem
Lieb-Robinsonnew} together with the density of $\mathcal{U}_{0}\subset
\mathcal{U}$ and completeness of $\mathcal{U}$. This limiting family also
satisfies the reverse cocycle property (i) because of (\ref{cocycle}).

\noindent (iii) For any $B\in \mathcal{U}_{0}\subset \mathrm{Dom}(\delta
_{t})$, the map $t\mapsto \tau _{t,s}\circ \delta _{t}(B)$ from $\mathbb{R}$
to $\mathcal{U}$ is continuous. Indeed, for any $B\in \mathcal{U}_{0}$ and $%
\alpha ,t\in {\mathbb{R}}$,
\begin{equation*}
\left\Vert \tau _{\alpha ,s}\circ \delta _{\alpha }\left( B\right) -\tau
_{t,s}\circ \delta _{t}\left( B\right) \right\Vert _{\mathcal{U}}\leq
\left\Vert \left( \tau _{\alpha ,s}-\tau _{t,s}\right) \circ \delta
_{t}\left( B\right) \right\Vert _{\mathcal{U}}+\left\Vert \delta _{\alpha
}\left( B\right) -\delta _{t}\left( B\right) \right\Vert _{\mathcal{U}}.
\end{equation*}%
By applying (\ref{inequlity utile}) to the interaction $\Psi ^{(t)}-\Psi
^{(\alpha )}$ and the potential $\mathbf{V}^{(t)}-\mathbf{V}^{(\alpha )}$
together with the strong continuity of $\{\tau _{t,s}\}_{s,t\in {\mathbb{R}}%
} $, one sees that, in the limit $\alpha \rightarrow t$, the r.h.s of the
above inequality vanishes when $B\in \mathcal{U}_{0}$ and $\Psi \in C(%
\mathbb{R};\mathcal{W})$. Now, because of (\ref{cauchy2}), for any $L\in
\mathbb{R}_{0}^{+}$, $B\in \mathcal{U}_{0}$, and $s,t\in \mathbb{R}$,
\begin{eqnarray}
\left\Vert \tau _{t,s}\left( B\right) -B-\int_{s}^{t}\tau _{\alpha ,s}\circ
\delta _{\alpha }\left( B\right) \mathrm{d}\alpha \right\Vert _{\mathcal{U}}
&\leq &\left\Vert \tau _{t,s}\left( B\right) -\tau _{t,s}^{(L)}\left(
B\right) \right\Vert _{\mathcal{U}}  \label{eq stupide1} \\
&&+\int_{s}^{t}\left\Vert \left( \tau _{\alpha ,s}^{(L)}-\tau _{\alpha
,s}\right) \circ \delta _{\alpha }\left( B\right) \right\Vert _{\mathcal{U}}%
\mathrm{d}\alpha  \notag \\
&&+\int_{s}^{t}\left\Vert \delta _{\alpha }^{(L)}\left( B\right) -\delta
_{\alpha }\left( B\right) \right\Vert _{\mathcal{U}}\mathrm{d}\alpha \ .
\notag
\end{eqnarray}%
By using the strong convergence of $\tau _{t,s}^{(L)}$ towards $\tau _{t,s}$
as well as (\ref{inequlity utile}) and (\ref{core limit}) together with
Lebesgue's dominated convergence theorem, one checks that the r.h.s. of (\ref%
{eq stupide1}) vanishes when $B\in \mathcal{U}_{0}$ and $L\rightarrow \infty
$. Because of the continuity of the map $t\mapsto \tau _{t,s}\circ \delta
_{t}(B)$, (\ref{cauchy trivial1}) is verified on the dense set $\mathcal{U}%
_{0}\subset \mathrm{Dom}(\delta _{t})$.

To prove uniqueness, assume that $\{%
\hat{\tau}_{t,s}\}_{s,t\in {\mathbb{R}}}$ is any family of bounded operators
on $\mathcal{U}$ satisfying (\ref{cauchy trivial1}) on $\mathcal{U}_{0}$. By
(\ref{cauchy1}) and because $\tau _{t,s}^{(L)}(B)\in \mathcal{U}_{0}$ for
any $B\in \mathcal{U}_{0}$,%
\begin{equation}
\hat{\tau}_{t,s}\left( B\right) -\tau _{t,s}^{(L)}\left( B\right)
=\int_{s}^{t}\hat{\tau}_{\alpha ,s}\circ \left( \delta _{\alpha }-\delta
_{\alpha }^{(L)}\right) \circ \tau _{t,\alpha }^{(L)}\left( B\right) \mathrm{%
d}\alpha  \label{eq similar}
\end{equation}%
for any $B\in \mathcal{U}_{0}$, $L\in \mathbb{R}_{0}^{+}$ and $s,t\in
\mathbb{R}$. Similar to (\ref{inequ toto inter1})--(\ref{assertion bisbisbis}%
), we infer from Theorem \ref{Theorem Lieb-Robinsonnew} (i) that, for any $%
\Lambda \in \mathcal{P}_{f}(\mathfrak{L})$, $B\in \mathcal{U}_{\Lambda }$, $%
\alpha ,t\in \mathbb{R}$ and sufficiently large $L\in \mathbb{R}_{0}^{+}$,%
\begin{equation*}
\left\Vert \left( \delta _{\alpha }-\delta _{\alpha }^{(L)}\right) \circ
\tau _{t,\alpha }^{(L)}\left( B\right) \right\Vert _{\mathcal{U}}\leq
\left\Vert \Psi \right\Vert _{\infty }\mathrm{e}^{2\mathbf{D}\left\vert
t-\alpha \right\vert \left\Vert \Psi \right\Vert _{\infty
}}\sum\limits_{y\in \Lambda _{L}^{c}}\sum_{x\in \Lambda }\mathbf{F}\left(
\left\vert x-y\right\vert \right) \ .
\end{equation*}%
In particular, by (\ref{assertion bisbisbisbis}), for any $B\in \mathcal{U}%
_{0}$ and $\alpha ,t\in \mathbb{R}$,%
\begin{equation}
\lim_{L\rightarrow \infty }\left\Vert \left( \delta _{\alpha }-\delta
_{\alpha }^{(L)}\right) \circ \tau _{t,\alpha }^{(L)}\left( B\right)
\right\Vert _{\mathcal{U}}=0  \label{limit unique}
\end{equation}%
uniformly for $\alpha $ on compacta. Because of (\ref{eq similar}) and $\{%
\hat{\tau}_{t,s}\}_{s,t\in {\mathbb{R}}}\subset \mathcal{B}(\mathcal{U})$,
we then conclude from (\ref{limit unique}) that, for every $s,t\in {\mathbb{R%
}}$, $\hat{\tau}_{t,s}$ coincides on the dense set $\mathcal{U}_{0}$ with
the limit $\tau _{t,s}$ of $\tau _{t,s}^{(L)}$, as $L\rightarrow \infty $.
By continuity, $\tau _{t,s}=\hat{\tau}_{t,s}$ on $\mathcal{U}$ for any $%
s,t\in {\mathbb{R}}$.
\end{proof}

\noindent The solution of (\ref{cauchy trivial1}) exists under very weak
conditions on interactions and potentials, i.e., their continuity, like in
the finite--volume case. It yields a fundamental solution for the states of
the interacting lattice fermions driven by the time--dependent interaction $%
\{\Psi ^{(t)}\}_{t\in \mathbb{R}}$. More precisely, for any fixed $\rho
_{s}\in \mathcal{U}^{\ast }$ at time $s\in {\mathbb{R}}$, the family $\{\rho
_{s}\circ \tau _{t,s}\}_{t\in {\mathbb{R}}}$ solves the following ordinary
differential equations, for each $B\in \mathcal{U}_{0}$:%
\index{Evolution equation!non--autonomous (H1)}%
\begin{equation}
\forall t\in {\mathbb{R}}:\qquad \partial _{t}\rho _{t}(B)=\rho _{t}\circ
\delta _{t}(B)\ .  \label{state non auto}
\end{equation}%
By Corollary \ref{Theorem Lieb-Robinson copy(4)}, the initial value problem
on $\mathcal{U}^{\ast }$ associated with the above infinite system of
ordinary differential equations is \emph{well--posed}. Indeed, the solution
of (\ref{state non auto}) is unique: Take any solution $\{\rho _{t}\}_{t\in {%
\mathbb{R}}}$ of (\ref{state non auto}) and, similar to (\ref{eq similar}),
use the equality
\begin{equation*}
\rho _{t}\left( B\right) -\rho _{s}\circ \tau _{t,s}^{(L)}\left( B\right)
=\int_{s}^{t}\rho _{\alpha }\left( \left( \delta _{\alpha }-\delta _{\alpha
}^{(L)}\right) \circ \tau _{t,\alpha }^{(L)}\left( B\right) \right) \mathrm{d%
}\alpha
\end{equation*}%
for any $\rho _{s}\in \mathcal{U}^{\ast }$, $B\in \mathcal{U}_{0}$, $L\in
\mathbb{R}_{0}^{+}$ and $s,t\in \mathbb{R}$ together with (\ref{limit unique}%
) and the weak$^{\ast }$--convergence of $\rho _{s}\circ \tau _{t,s}^{(L)}$
to $\rho _{s}\circ \tau _{t,s}$, as $L\rightarrow \infty $, by Corollary \ref%
{Theorem Lieb-Robinson copy(4)}.

Note again that (\ref{cauchy trivial1}) is the non--autonomous evolution
equation one formally obtains from the Schr\"{o}dinger equation for
automorphisms of the algebra of observables. See Section \ref{sect non auto}%
, in particular Equations (\ref{non auto explaination3})--(\ref{non auto
explaination3bis}). A similar remark can be done for the infinite system (%
\ref{state non auto}) of ordinary differential equations.

It is a priori unclear whether $\{\tau _{t,s}\}_{s,t\in {\mathbb{R}}}$
solves the non--auto%
\-%
nomous Cauchy initial value problem%
\index{Evolution equation!non--autonomous}
\begin{equation}
\forall s,t\in {\mathbb{R}}:\qquad \partial _{s}\tau _{t,s}=-\delta
_{s}\circ \tau _{t,s}\ ,\qquad \tau _{t,t}=\mathbf{1}_{\mathcal{U}}\ ,
\label{cauchy trivial2}
\end{equation}%
on some dense domain. The generators $\{\delta _{t}\}_{t\in \mathbb{R}}$ are
generally unbounded operators acting on $\mathcal{U}$ and their domains can
additionally depend on time. As explained in Section \ref{sect non auto}, no
unified theory of such linear evolution equations, similar to the
Hille--Yosida generation theorems in the autonomous case, is available. See,
e.g., \cite{Katobis,Caps,Schnaubelt1,Pazy,Bru-Bach} and the corresponding
references therein.

By using Lieb--Robinson bounds for multi--commutators, we show below in
Theorem \ref{thm non auto} that the evolution equation (\ref{cauchy trivial2}%
) \emph{also holds} on the dense set $\mathcal{U}_{0}$, under conditions
like polynomial decays of interactions and boundedness of the external
potential. Another example -- more restrictive in which concerns the
time--dependency of the generator of dynamics, but less restrictive w.r.t.
the behavior at large distances of the potential ${\mathbf{V}}$ -- for which
(\ref{cauchy trivial2}) holds is given by Theorem \ref{pertubed dynam thm}
(i) in Section \ref{Section existence dynamics}.

\subsection{Lieb--Robinson Bounds for Multi--Commutators}

As explained in Remark \ref{rermar trivial1}, all results of Section \ref%
{section LR multi} depend on Theorem \ref{Theorem Lieb-Robinson copy(3)}
(iii). It is the crucial ingredient we need in order to prove Lemma \ref%
{Lemma Series representation of the dynamics}, from which we derive
Lieb--Robinson bounds for multi--commutators. Theorem \ref{Theorem
Lieb-Robinsonnew} (ii) together with Corollary \ref{Theorem Lieb-Robinson
copy(4)} extend Theorem \ref{Theorem Lieb-Robinson copy(3)} (iii) to
time--dependent interactions and potentials. This allows us to prove Lemma %
\ref{Lemma Series representation of the dynamics} in the non--auto%
\-%
nomous case as well. It is then straightforward to extend Lieb--Robinson
bounds for multi--commutators to time--dependent interactions and potentials.

Recall that the proof of Lemma \ref{Lemma Series representation of the
dynamics} uses that the space translated finite--volume groups $\{\tau
_{t}^{(n,x)}\}_{t\in {\mathbb{R}}}$, $x\in \mathfrak{L}$, have all the same
limit $\{\tau _{t}\}_{t\in {\mathbb{R}}}$, as $n\rightarrow \infty $. This
also holds in the non--autonomous case. Indeed, for any $n\in \mathbb{N}_{0}$%
, $x\in \mathfrak{L}$, every bounded family $\Psi \doteq \{\Psi
^{(t)}\}_{t\in \mathbb{R}}$ on $\mathcal{W}$ (i.e., $\left\Vert \Psi
\right\Vert _{\infty }<\infty $ ), and each collection $\{\mathbf{V}%
^{(t)}\}_{t\in \mathbb{R}}$ of potentials, consider the (space) translated
family $\{\tau _{t,s}^{(n,x)}\}_{s,t\in {\mathbb{R}}}$ of finite--volume $%
\ast $--auto%
\-%
morphisms generated (cf. (\ref{cauchy1}) and (\ref{cauchy2})) by the
symmetric bounded derivation
\begin{equation*}
\delta _{t}^{(n,x)}(B)\doteq i\sum\limits_{\Lambda \subseteq \Lambda _{n}+x}
\left[ \Psi _{\Lambda }^{(t)},B\right] +i\sum\limits_{y\in \Lambda _{n}+x}%
\left[ \mathbf{V}_{\left\{ y\right\} }^{(t)},B\right] \ ,\qquad B\in
\mathcal{U}\ .
\end{equation*}%
In the autonomous case the strong convergence of these evolution families
towards $\{\tau _{t,s}\}_{s,t\in {\mathbb{R}}}$ easily follows from the
second Trotter--Kato approximation theorem \cite[Chap. III, Sect. 4.9]%
{EngelNagel}. We use the Lieb--Robinson bound of Theorem \ref{Theorem
Lieb-Robinsonnew} (i) to prove it in the non--autonomous case:

\begin{lemma}[Limit of translated dynamics]
\label{Lemma Series representation of the dynamics copy(1)}\mbox{ }\newline
Let $\Psi \doteq \{\Psi ^{(t)}\}_{t\in \mathbb{R}}$ be a bounded family on $%
\mathcal{W}$ (i.e., $\left\Vert \Psi \right\Vert _{\infty }<\infty $ ) and $%
\{\mathbf{V}^{(t)}\}_{t\in \mathbb{R}}$ a collection of potentials. For any $%
y\in \mathfrak{L}$ and $\Lambda \in \mathcal{P}_{f}(\mathfrak{L})$, assume $%
\Psi _{\Lambda },\mathbf{V}_{\left\{ y\right\} }\in C\left( \mathbb{R};%
\mathcal{U}\right) $. Then%
\begin{equation*}
\underset{n\rightarrow \infty }{\lim }\tau _{t,s}^{(n,x)}\left( B\right)
=\tau _{t,s}\left( B\right) \ ,\qquad B\in \mathcal{U}\ ,\ \ x\in \mathfrak{L%
},\ s,t\in \mathbb{R}\ .
\end{equation*}
\end{lemma}

\begin{proof}
For any $n\in \mathbb{N}_{0}$ and $x\in \mathfrak{L}$, the translated
finite--volume family $\{\tau _{s,t}^{(n,x)}\}_{s,t\in {\mathbb{R}}}$ solves
non--auto%
\-%
nomous evolution equations like (\ref{cauchy1})--(\ref{cauchy2}). Therefore,
similar to (\ref{eq similar}), for any $n\in \mathbb{N}_{0}$, $x\in
\mathfrak{L}$, $\Lambda \in \mathcal{P}_{f}(\mathfrak{L})$, $B\in \mathcal{U}%
_{\Lambda }$ and $s,t\in \mathbb{R}$,
\begin{equation}
\tau _{t,s}^{(n,x)}\left( B\right) -\tau _{t,s}^{(n,0)}\left( B\right)
=\int_{s}^{t}\tau _{\alpha ,s}^{(n,x)}\circ \left( \delta _{\alpha
}^{(n,x)}-\delta _{\alpha }^{(n,0)}\right) \circ \tau _{t,\alpha
}^{(n,0)}\left( B\right) \mathrm{d}\alpha \ .  \label{eq cool 0}
\end{equation}%
For sufficiently large $n\in \mathbb{N}_{0}$ such that $\Lambda \subset
\left( \Lambda _{n}+x\right) \cap \Lambda _{n}$, note that%
\begin{equation*}
\left\Vert \left( \delta _{\alpha }^{(n,x)}-\delta _{\alpha }^{(n,0)}\right)
\circ \tau _{t,\alpha }^{(n,0)}\left( B\right) \right\Vert _{\mathcal{U}%
}\leq \sum\limits_{\mathcal{Z}\in \mathcal{P}_{f}(\mathfrak{L}),\ \mathcal{Z}%
\cap \left( \left( \Lambda _{n}+x\right) ^{c}\cup \Lambda _{n}^{c}\right)
\neq \emptyset }\left\Vert \left[ \Psi _{\Lambda }^{(t)},\tau _{t,\alpha
}^{(n,0)}\left( B\right) \right] \right\Vert _{\mathcal{U}}
\end{equation*}%
with $\mathcal{Z}^{c}\doteq \mathfrak{L}\backslash \mathcal{Z}$ being the
complement of any set $\mathcal{Z}\in \mathcal{P}_{f}(\mathfrak{L})$. Then,
similar to Inequality (\ref{assertion bisbisbis}), by using Theorem \ref%
{Theorem Lieb-Robinsonnew} (i), one verifies that, for any $x\in \mathfrak{L}
$, $\Lambda \in \mathcal{P}_{f}(\mathfrak{L})$, $B\in \mathcal{U}_{\Lambda }$%
, $\alpha ,t\in \mathbb{R}$, and sufficiently large $n\in \mathbb{N}_{0}$,%
\begin{eqnarray}
&&\left\Vert \left( \delta _{\alpha }^{(n,x)}-\delta _{\alpha
}^{(n,0)}\right) \circ \tau _{t,\alpha }^{(n,0)}\left( B\right) \right\Vert
_{\mathcal{U}}  \label{eq cool1} \\
&\leq &2\left\Vert B\right\Vert _{\mathcal{U}}\left\Vert \Psi \right\Vert
_{\infty }\ \mathrm{e}^{2\mathbf{D}\left\vert t-\alpha \right\vert
\left\Vert \Psi \right\Vert _{\infty }}\sum\limits_{y\in \left( \Lambda
_{n}+x\right) ^{c}\cup \Lambda _{n}^{c}}\sum_{z\in \Lambda }\mathbf{F}\left(
\left\vert z-y\right\vert \right) \ ,  \notag
\end{eqnarray}%
while
\begin{equation}
\underset{n\rightarrow \infty }{\lim }\sum\limits_{y\in \left( \Lambda
_{n}+x\right) ^{c}\cup \Lambda _{n}^{c}}\sum_{z\in \Lambda }\mathbf{F}\left(
\left\vert z-y\right\vert \right) =0\ ,  \label{eq cool2}
\end{equation}%
because of (\ref{(3.1) NS}). We thus deduce from (\ref{eq cool1})--(\ref{eq
cool2}) that%
\begin{equation*}
\underset{n\rightarrow \infty }{\lim }\left\Vert \left( \delta _{\alpha
}^{(n,x)}-\delta _{\alpha }^{(n,0)}\right) \circ \tau _{t,\alpha
}^{(n,0)}\left( B\right) \right\Vert _{\mathcal{U}}=0
\end{equation*}%
uniformly for $\alpha $ on compacta. Combined with (\ref{eq cool 0}) and
Corollary \ref{Theorem Lieb-Robinson copy(4)}, this uniform limit implies
the assertion.
\end{proof}

With the above result and the introducing remarks of this subsection, it is
now straightforward to extend Theorem \ref{Theorem Lieb-Robinson copy(1)} to
the non--auto%
\-%
nomous case:

\begin{satz}[Lieb--Robinson bounds for multi--commutators -- Part I]
\label{Theorem Lieb-Robinson copy(1)-bis}\mbox{
}\newline
Let $\Psi \doteq \{\Psi ^{(t)}\}_{t\in \mathbb{R}}$ be a bounded family on $%
\mathcal{W}$ (i.e., $\left\Vert \Psi \right\Vert _{\infty }<\infty $ ), $\{%
\mathbf{V}^{(t)}\}_{t\in \mathbb{R}}$ a collection of potentials, and $s\in
\mathbb{R}$. For any $y\in \mathfrak{L}$ and $\Lambda \in \mathcal{P}_{f}(%
\mathfrak{L})$, assume $\Psi _{\Lambda },\mathbf{V}_{\left\{ y\right\} }\in
C\left( \mathbb{R};\mathcal{U}\right) $. Then, for any integer $k\in \mathbb{%
N}$, $\{m_{j}\}_{j=0}^{k}\subset \mathbb{N}_{0}$, times $\{s_{j}\}_{j=1}^{k}%
\subset \mathbb{R}$, lattice sites $\{x_{j}\}_{j=0}^{k}\subset \mathfrak{L}$%
, and elements $B_{0}\in \mathcal{U}_{0}$, $\{B_{j}\}_{j=1}^{k}\subset
\mathcal{U}_{0}\cap \mathcal{U}^{+}$ such that $B_{j}\in \mathcal{U}%
_{\Lambda _{m_{j}}}$ for $j\in \{0,\ldots ,k\}$,%
\index{Lieb--Robinson bounds!multi--commutators}%
\begin{eqnarray*}
&&\left\Vert \left[ \tau _{s_{k},s}\circ \chi _{x_{k}}(B_{k}),\ldots ,\tau
_{s_{1},s}\circ \chi _{x_{1}}(B_{1}),\chi _{x_{0}}(B_{0})\right]
^{(k+1)}\right\Vert _{\mathcal{U}} \\
&\leq &2^{k}\prod\limits_{j=0}^{k}\left\Vert B_{j}\right\Vert _{\mathcal{U}%
}\sum_{T\in \mathcal{T}_{k+1}}\left( \varkappa _{T}\left( \left\{
(m_{j},x_{j})\right\} _{j=0}^{k}\right) +\Re _{T,\left\Vert \Psi \right\Vert
_{\infty }}\right) \ ,
\end{eqnarray*}%
where $\varkappa _{T}$ and $\Re _{T,\alpha }$ are respectively defined by (%
\ref{definition chqrqcteristic}) and (\ref{def AT}) for $T\in \mathcal{T}%
_{k+1}$ and $\alpha \in \mathbb{R}_{0}^{+}$, the times $\{s_{j}\}_{j=1}^{k}$
in (\ref{def AT}) being replaced with $\{(s_{j}-s)\}_{j=1}^{k}$.
\end{satz}

\begin{proof}
One easily checks that Theorem \ref{Theorem Lieb-Robinsonnew} (ii) holds for
$\{\tau _{t,s}^{(n,x)}\}_{s,t\in {\mathbb{R}}}$ at any fixed $x\in \mathfrak{%
L}$ and $n\in \mathbb{N}_{0}$. By Lemma \ref{Lemma Series representation of
the dynamics copy(1)}, Lemma \ref{Lemma Series representation of the
dynamics} also holds in the non--auto%
\-%
nomous case and the assertion follows from (\ref{inequality sup LR multi})
with the $\ast $--auto%
\-%
morphism $\tau _{s_{j}}$ being replaced by $\tau _{s_{j},s}$ for every $j\in
\left\{ 1,\ldots ,k\right\} $.
\end{proof}

By Theorems \ref{theorem exp tree decay copy(1)} and \ref{Theorem
Lieb-Robinson copy(1)-bis}, we obtain
\index{Lieb--Robinson bounds!multi--commutators}Lieb--Robinson bounds for
multi--com%
\-%
mutators as well as a version of Corollary \ref{theorem exp tree decay} in
the \emph{non--auto%
\-%
nomous} case. I.e., interacting and non--auto%
\-%
nomous systems also satisfy the so--called
\index{Tree--decay bounds}tree--decay bounds.

Another application of Theorems \ref{theorem exp tree decay copy(1)} and \ref%
{Theorem Lieb-Robinson copy(1)-bis} is a proof of existence of a fundamental
solution for the non--auto%
\-%
nomous abstract Cauchy initial value problem for observables%
\index{Evolution equation!non--autonomous}
\begin{equation}
\forall s\in {\mathbb{R}}:\qquad \partial _{s}B_{s}=-\delta _{s}(B_{s})\
,\qquad B_{t}=B\in \mathcal{U}_{0}\ ,  \label{non auto hyperbolic}
\end{equation}%
in the Banach space $\mathcal{U}$, i.e., a proof of existence of a solution
of the evolution equation (\ref{cauchy trivial2}). The latter is a
non--trivial statement, as previously discussed, among other things because
the domain of $\delta _{s}$ depends, in general, on the time $s\in \mathbb{R}
$. [Here, $t\in {\mathbb{R}}$ is the \textquotedblleft
initial\textquotedblright\ time.]

To this end, like in (\ref{definition D1})--(\ref{assumption boundedness2}),
we add the following condition on interactions $\Phi $:%
\index{Decay function!polynomial decay}

\begin{itemize}
\item \emph{Polynomial decay.} Assume (\ref{(3.3) NS generalized0}) and the
existence of constants $\upsilon ,D\in \mathbb{R}^{+}$ such that%
\begin{equation}
\underset{x\in \mathfrak{L}}{\sup }\sum\limits_{\Lambda \in \mathcal{D}%
\left( x,m\right) }\left\Vert \Phi _{\Lambda }\right\Vert _{\mathcal{U}}\leq
D\left( m+1\right) ^{-\upsilon }\ ,\qquad m\in \mathbb{N}_{0}\ ,
\label{assumption boundedness2bis}
\end{equation}%
while the sequence $\{\mathbf{u}_{n,m}\}_{n\in \mathbb{N}}\in \ell ^{1}(%
\mathbb{N})$ of (\ref{(3.3) NS generalized0}) satisfies
\begin{equation}
\sum_{m,n\in \mathbb{N}}m^{-\upsilon }\left\vert \mathbf{u}_{n,m}\right\vert
<\infty \ .  \label{assumption boundedness4}
\end{equation}
\end{itemize}

\noindent As $\mathbf{F}(|x|)>0$ for all $x\in \mathfrak{L}$, note that (\ref%
{(3.3) NS generalized0}) implies
\begin{equation*}
\sum_{n\in \mathbb{N}}\left\vert \mathbf{u}_{n,m}\right\vert \geq
Dm^{\varsigma }
\end{equation*}%
for some $D\in $ $\mathbb{R}^{+}$ and all $m\in \mathbb{N}_{0}$. Hence, the
inequality (\ref{assumption boundedness4}) imposes $\upsilon >\varsigma +1$.

Then, one gets the following assertion:

\begin{satz}[Dynamics and non--autonomous evolution equations]
\label{thm non auto}\mbox{
}\newline
Let $\Psi \doteq \{\Psi ^{(t)}\}_{t\in \mathbb{R}}\in C(\mathbb{R};\mathcal{W%
})$ be a bounded family on $\mathcal{W}$ (i.e., $\left\Vert \Psi \right\Vert
_{\infty }<\infty $) and $\{\mathbf{V}_{\left\{ x\right\} }^{(t)}\}_{x\in
\mathfrak{L},t\in \mathbb{R}}$ a bounded family on $\mathcal{U}$ of
potentials with $\mathbf{V}_{\left\{ x\right\} }\in C\left( \mathbb{R};%
\mathcal{U}\right) $ for any $x\in \mathfrak{L}$. Assume$\ $(\ref{(3.3) NS
generalized0}) with $\varsigma >2d$ and that (\ref{assumption
boundedness2bis})--(\ref{assumption boundedness4}) with $\Phi =\Psi ^{(t)}$
and $\nu >\varsigma +1$ hold uniformly for $t\in \mathbb{R}$. Then, for any $%
s,t\in \mathbb{R}$, $\tau _{t,s}\left( \mathcal{U}_{0}\right) \subset
\mathrm{Dom}(\delta _{s})$ and $\{\tau _{t,s}\}_{s,t\in {\mathbb{R}}}$
solves the non--autonomous evolution equation%
\index{Evolution equation!non--autonomous}%
\index{Interaction!infinite--volume dynamics}%
\index{Potential!infinite--volume dynamics}%
\begin{equation}
\forall s,t\in {\mathbb{R}}:\qquad \partial _{s}\tau _{t,s}=-\delta
_{s}\circ \tau _{t,s}\ ,\qquad \tau _{t,t}=\mathbf{1}_{\mathcal{U}}\ ,
\label{cauchy trivial1bis}
\end{equation}%
in the strong sense on the dense set $\mathcal{U}_{0}$.
\end{satz}

\begin{proof}
\noindent \textbf{1.} Let $s,t\in \mathbb{R}$, $\Lambda \in \mathcal{P}_{f}(%
\mathfrak{L})$ and take any element $B\in \mathcal{U}_{\Lambda }$. As a
preliminary step, we prove that $\{\delta _{s}\circ \tau _{t,s}^{(L)}\left(
B\right) \}_{L\in \mathbb{R}_{0}^{+}}$ converges to $\delta _{s}\circ \tau
_{t,s}\left( B\right) $, as $L\rightarrow \infty $. In particular, $\tau
_{t,s}\left( \mathcal{U}_{0}\right) \subset \mathrm{Dom}(\delta _{s})$. By
using similar arguments as in the proof of Theorem \ref{Theorem
Lieb-Robinsonnew} (ii), it suffices to study the limit of $\{\delta
_{s}\circ
\tilde{\tau}_{t,s}^{(L)}\left( B\right) \}_{L\in \mathbb{R}_{0}^{+}}$, see (%
\ref{interactino picture}).

Similar to (\ref{inequality dyn5}), from (\ref{cocylce})--(\ref{interactino
picture2}) and straightforward computations, for any $L_{1},L_{2}\in \mathbb{%
R}_{0}^{+}$ with $\Lambda \subset \Lambda _{L_{1}}\varsubsetneq \Lambda
_{L_{2}}$,
\begin{eqnarray}
&&\left\Vert \delta _{s}\circ \left( \tilde{\tau}_{t,s}^{(L_{2})}\left(
B\right) -\tilde{\tau}_{t,s}^{(L_{1})}\left( B\right) \right) \right\Vert _{%
\mathcal{U}}  \label{uper bound3} \\
&\leq &\int_{\min \{s,t\}}^{\max \{s,t\}}\sum\limits_{\mathcal{Z}\in
\mathcal{P}_{f}(\mathfrak{L})}\left\Vert \left[ \hat{\tau}_{s,\alpha
}^{(L_{1},L_{2})}(\Psi _{\mathcal{Z}}^{(s)}),B_{\alpha
}^{(L_{1},L_{2})},\tau _{t,\alpha }^{(L_{1})}(\tilde{B}_{t})\right]
^{(3)}\right\Vert _{\mathcal{U}}\mathrm{d}\alpha \ ,  \notag
\end{eqnarray}%
where $\tilde{B}_{t}\doteq \mathcal{V}_{t,s}(\mathbf{V}_{\Lambda })B\mathcal{%
V}_{s,t}(\mathbf{V}_{\Lambda })$,
\begin{equation}
\hat{\tau}_{s,\alpha }^{(L_{1},L_{2})}\left( B\right) \doteq \mathcal{V}%
_{s,\alpha }(\mathbf{V}_{\Lambda _{L_{2}}\backslash \Lambda _{L_{1}}})\tau
_{s,\alpha }^{(L_{2})}(B)\mathcal{V}_{\alpha ,s}(\mathbf{V}_{\Lambda
_{L_{2}}\backslash \Lambda _{L_{1}}})\ ,\quad B\in \mathcal{U},\ s,\alpha
\in \mathbb{R}\ ,  \label{def new}
\end{equation}%
and%
\begin{equation*}
B_{\alpha }^{(L_{1},L_{2})}\doteq \sum\limits_{\mathcal{Z}\subseteq \Lambda
_{L_{2}},\ \mathcal{Z}\cap (\Lambda _{L_{2}}\backslash \Lambda _{L_{1}})\neq
\emptyset }\mathcal{V}_{\alpha ,s}(\mathbf{V}_{\Lambda _{L_{2}}\backslash
\Lambda _{L_{1}}})\Psi _{\mathcal{Z}}\mathcal{V}_{s,\alpha }(\mathbf{V}%
_{\Lambda _{L_{2}}\backslash \Lambda _{L_{1}}})\in \mathcal{U}^{+}\cap
\mathcal{U}_{\Lambda _{L_{2}}}.
\end{equation*}%
Using (\ref{inequality idiote utile}), observe that, for all $\mathcal{Z}%
\subseteq \Lambda _{L_{2}}$ and $\alpha ,s\in \mathbb{R}$,
\begin{equation}
\mathcal{V}_{\alpha ,s}(\mathbf{V}_{\Lambda _{L_{2}}\backslash \Lambda
_{L_{1}}})\Psi _{\mathcal{Z}}\mathcal{V}_{s,\alpha }(\mathbf{V}_{\Lambda
_{L_{2}}\backslash \Lambda _{L_{1}}})\in \mathcal{U}^{+}\cap \mathcal{U}_{%
\mathcal{Z}}  \label{obs B 1}
\end{equation}%
with
\begin{equation}
\left\Vert \mathcal{V}_{\alpha ,s}(\mathbf{V}_{\Lambda _{L_{2}}\backslash
\Lambda _{L_{1}}})\Psi _{\mathcal{Z}}\mathcal{V}_{s,\alpha }(\mathbf{V}%
_{\Lambda _{L_{2}}\backslash \Lambda _{L_{1}}})\right\Vert _{\mathcal{U}%
}=\left\Vert \Psi _{\mathcal{Z}}\right\Vert _{\mathcal{U}}\ .
\end{equation}%
Similarly, for all $t\in \mathbb{R}$,
\begin{equation}
\tilde{B}_{t}\in \mathcal{U}_{\Lambda }\text{\qquad and\qquad }\Vert \tilde{B%
}_{t}\Vert _{\mathcal{U}}=\Vert B\Vert _{\mathcal{U}}\ .  \label{obs B 2}
\end{equation}

In order to bound the sum%
\begin{equation}
\sum\limits_{\mathcal{Z}\in \mathcal{P}_{f}(\mathfrak{L})}\left[ \hat{\tau}%
_{s,\alpha }^{(L_{1},L_{2})}(\Psi _{\mathcal{Z}}^{(s)}),B_{\alpha
}^{(L_{1},L_{2})},\tau _{t,\alpha }^{(L_{1})}(\tilde{B}_{t})\right] ^{(3)}
\label{sum bis}
\end{equation}%
of multi--commutators of order three we represent it as a convenient series,
whose summability is uniform w.r.t. $L_{1},L_{2}\in \mathbb{R}_{0}^{+}$ ($%
\Lambda \subset \Lambda _{L_{1}}\varsubsetneq \Lambda _{L_{2}}$). To this
end, first develop $\tau _{t,\alpha }^{(L_{1})}(\tilde{B}_{t})$ as a
telescoping series: Let $m_{0}\in \mathbb{N}_{0}$ be the smallest integer
such that $\Lambda \subset \Lambda _{m_{0}}$. Then, similar to Lemma \ref%
{Lemma Series representation of the dynamics} \ (autonomous case) and as
explained in the proof of Theorem \ref{Theorem Lieb-Robinson copy(1)-bis},
for any $\alpha ,t\in {\mathbb{R}}$ and $L_{1}\in \mathbb{R}_{0}^{+}$,%
\begin{equation*}
\tau _{t,\alpha }^{(L_{1})}(\tilde{B}_{t})=\sum_{n=m_{0}}^{\infty }\mathfrak{%
\tilde{B}}_{t,\alpha }(n)\text{ }.
\end{equation*}%
Here, for all integers $n\geq m_{0}$, $\mathfrak{\tilde{B}}_{t,\alpha
}(n)\in \mathcal{U}_{\Lambda _{n}}$ where $\Vert \mathfrak{\tilde{B}}%
_{t,\alpha }(m_{0})\Vert _{\mathcal{U}}=\Vert B\Vert _{\mathcal{U}}$ (see (%
\ref{obs B 2})) and, for all $n\in \mathbb{N}$ with $n>m_{0}$,%
\begin{equation}
\Vert \mathfrak{\tilde{B}}_{t,\alpha }(n)\Vert _{\mathcal{U}}\leq 2\Vert
B\Vert _{\mathcal{U}}\left\Vert \Psi \right\Vert _{\infty }\left\vert
t-\alpha \right\vert \mathrm{e}^{4\mathbf{D}\left\vert t-\alpha \right\vert
\left\Vert \Psi \right\Vert _{\infty }}\frac{\mathbf{u}_{n,m_{0}}}{\left(
1+n\right) ^{\varsigma }}\ ,  \label{born telescoping}
\end{equation}%
by Theorem \ref{Theorem Lieb-Robinsonnew} (ii) and Assumption (\ref{(3.3) NS
generalized0}). Of course, $\mathfrak{\tilde{B}}_{t,\alpha }(n)=0$ for any
integer $n>L_{1}$ and $\alpha ,t\in {\mathbb{R}}$ because $\{\tau
_{t,s}^{(L_{1})}\}_{s,t\in {\mathbb{R}}}$ is a finite--volume dynamics.
Meanwhile, because of (\ref{inequality idiote utile}), Theorem \ref{Theorem
Lieb-Robinson copy(1)-bis} holds by replacing $\{\tau _{t,s}\}_{s,t\in
\mathbb{R}}$ with $\{\hat{\tau}_{t,s}^{(L_{1},L_{2})}\}_{s,t\in \mathbb{R}}$
at sufficiently large $L_{1},L_{2}\in \mathbb{R}_{0}^{+}$ ($\Lambda
_{L_{1}}\subsetneq \Lambda _{L_{2}}$). Using this together with (\ref%
{assumption boundedness2bis})--(\ref{assumption boundedness4}) for $\Phi
=\Psi ^{(t)}$, Equations (\ref{obs B 1})--(\ref{obs B 2}), Theorem \ref%
{theorem exp tree decay copy(1)}, as well as the assumptions $\nu >\varsigma
+1$ and $\varsigma >2d$,%
\begin{eqnarray}
&&\sum_{n_{0}=m_{0}}^{\infty }\ \ \sum_{x_{2}\in \mathfrak{L}}\sum_{m_{2}\in
\mathbb{N}_{0}}\sum\limits_{\mathcal{Z}_{2}\in \mathcal{D}(x_{2},m)}\ \
\sum_{x_{1}\in \mathfrak{L}}\sum_{m_{1}\in \mathbb{N}_{0}}\sum\limits_{%
\mathcal{Z}_{1}\in \mathcal{D}(x_{1},m)}  \label{estimate poly} \\
&&\left\Vert \left[ \hat{\tau}_{s,\alpha }^{(L_{1},L_{2})}(\Psi _{\mathcal{Z}%
_{2}}^{(s)}),\mathcal{V}_{\alpha ,s}(\mathbf{V}_{\Lambda _{L_{2}}\backslash
\Lambda _{L_{1}}})\Psi _{\mathcal{Z}_{1}}\mathcal{V}_{s,\alpha }(\mathbf{V}%
_{\Lambda _{L_{2}}\backslash \Lambda _{L_{1}}}),\mathfrak{\tilde{B}}%
_{t,\alpha }(n)\right] ^{(3)}\right\Vert _{\mathcal{U}}  \notag \\
&\leq &D\left\Vert B\right\Vert _{\mathcal{U}}\left\Vert \mathbf{u}_{\cdot
,m_{0}}\right\Vert _{\ell ^{1}(\mathbb{N})}\left( \sum_{m_{1}\in \mathbb{N}%
_{0}}\left( m_{1}+1\right) ^{\varsigma -\upsilon }\right)  \notag \\
&&\times \sum_{m_{2}\in \mathbb{N}_{0}}\left( m_{2}+1\right) ^{-\upsilon
}\left( \sum_{n_{2}\in \mathbb{N}}\mathbf{u}_{n_{2},m_{2}}+\left(
m_{2}+1\right) ^{\varsigma }\right) <\infty \ .  \notag
\end{eqnarray}%
Similar to (\ref{born telescoping}) and because (\ref{assumption
boundedness2bis})--(\ref{assumption boundedness4}) with $\Phi =\Psi ^{(t)}$
hold uniformly for $t\in \mathbb{R}$, the strictly positive constant $D\in
\mathbb{R}^{+}$ is uniformly bounded for $s,t,\alpha $ on compacta and $%
L_{1},L_{2}\in \mathbb{R}_{0}^{+}$ ($\Lambda \subset \Lambda
_{L_{1}}\varsubsetneq \Lambda _{L_{2}}$). The last sum is an \emph{upper
bound} of the integrand of the r.h.s. of (\ref{uper bound3}). Indeed, we
deduce from (\ref{set eq}) that%
\begin{eqnarray*}
B_{\alpha }^{(L_{1},L_{2})} &=&\sum_{x\in \Lambda _{L_{2}}\backslash \Lambda
_{L_{1}}}\sum_{m\in \mathbb{N}_{0}}\sum\limits_{\mathcal{Z}\subseteq \Lambda
_{L_{2}},\ \mathcal{Z}\in \mathcal{D}(x,m)} \\
&&\frac{1}{|\mathcal{Z}\cap \Lambda _{L_{2}}\backslash \Lambda _{L_{1}}|}%
\mathcal{V}_{\alpha ,s}(\mathbf{V}_{\Lambda _{L_{2}}\backslash \Lambda
_{L_{1}}})\Psi _{\mathcal{Z}}\mathcal{V}_{s,\alpha }(\mathbf{V}_{\Lambda
_{L_{2}}\backslash \Lambda _{L_{1}}})\text{ }
\end{eqnarray*}%
and%
\begin{equation*}
\sum\limits_{\mathcal{Z}\in \mathcal{P}_{f}(\mathfrak{L}),\ \mathcal{Z}\cap
\Lambda _{L_{2}}\neq \emptyset }\hat{\tau}_{s,\alpha }^{(L_{1},L_{2})}(\Psi
_{\mathcal{Z}}^{(s)})=\sum_{x\in \Lambda _{L_{2}}}\sum_{m\in \mathbb{N}%
_{0}}\sum\limits_{\mathcal{Z}\in \mathcal{D}(x,m)}\frac{1}{|\mathcal{Z}\cap
\Lambda _{L_{2}}|}\hat{\tau}_{s,\alpha }^{(L_{1},L_{2})}(\Psi _{\mathcal{Z}%
}^{(s)})\ .
\end{equation*}%
[Compare this last sum with (\ref{sum bis}) by using (\ref{obs B 1}) and (%
\ref{obs B 2}) to restrict the whole sum over $\mathcal{Z}\in \mathcal{P}%
_{f}(\mathfrak{L})$ to finite sets $\mathcal{Z}$ so that $\mathcal{Z}\cap
\Lambda _{L_{2}}\neq \emptyset $.]

As a consequence, for any $s,t\in \mathbb{R}$ and $B\in \mathcal{U}_{0}$, we
infer from (\ref{uper bound3}), (\ref{estimate poly}), and Lebesgue's
dominated convergence theorem that $\{\delta _{s}\circ \tilde{\tau}%
_{t,s}^{(L)}\left( B\right) \}_{L\in \mathbb{R}_{0}^{+}}$, and hence $%
\{\delta _{s}\circ \tau _{t,s}^{(L)}\left( B\right) \}_{L\in \mathbb{R}%
_{0}^{+}}$, are Cauchy nets within the complete space $\mathcal{U}$. By
Corollary \ref{Theorem Lieb-Robinson copy(4)}, $\{\tau _{t,s}^{(L)}\}_{L\in
\mathbb{R}_{0}^{+}}$ converges strongly to $\tau _{t,s}$ for every $s,t\in
\mathbb{R}$. Recall meanwhile that the operator $\delta _{s}$ is the \emph{%
closed} operator described in Theorem \ref{Theorem Lieb-Robinson copy(3)}
for the interaction $\Psi ^{(s)}\in \mathcal{W}$ and the potential $\mathbf{V%
}^{(s)}$ at fixed $s\in {\mathbb{R}}$. Therefore, $\tau _{t,s}\left(
B\right) \in \mathrm{Dom}(\delta _{s})$ and the family $\{\delta _{s}\circ
\tau _{t,s}^{(L)}\left( B\right) \}_{L\in \mathbb{R}_{0}^{+}}$ converges to $%
\delta _{s}\circ \tau _{t,s}\left( B\right) $, i.e.,%
\begin{equation}
\underset{L\rightarrow \infty }{\lim }\left\Vert \delta _{s}\circ \left(
\tau _{t,s}\left( B\right) -\tau _{t,s}^{(L)}\left( B\right) \right)
\right\Vert _{\mathcal{U}}=0\ .  \label{inportant upper boundbis+3}
\end{equation}%
In particular, $\tau _{t,s}\left( \mathcal{U}_{0}\right) \subset \mathrm{Dom}%
(\delta _{s})$.

Now, by using (\ref{cauchy1}) one gets that, for $L\in \mathbb{R}_{0}^{+}$, $%
s,t,h\in \mathbb{R}$, $h\neq 0$, and $B\in \mathcal{U}_{0}$,
\begin{eqnarray}
&&\left\Vert \left\vert h\right\vert ^{-1}\left( \tau _{t,s+h}\left(
B\right) -\tau _{t,s}\left( B\right) \right) +\delta _{s}\circ \tau
_{t,s}\left( B\right) \right\Vert _{\mathcal{U}}  \notag \\
&\leq &\left\Vert \delta _{s}\circ \left( \tau _{t,s}\left( B\right) -\tau
_{t,s}^{(L)}\left( B\right) \right) \right\Vert _{\mathcal{U}}
\label{inportant upper boundbis} \\
&&+\underset{\alpha \in \left[ s-\left\vert h\right\vert ,s+\left\vert
h\right\vert \right] }{\sup }\left\Vert \delta _{s}^{(L)}\circ \tau
_{t,s}^{(L)}\left( B\right) -\delta _{\alpha }^{(L)}\circ \tau _{t,\alpha
}^{(L)}\left( B\right) \right\Vert _{\mathcal{U}}  \notag \\
&&+\left\Vert \left( \delta _{s}^{(L)}-\delta _{s}\right) \circ \tau
_{t,s}^{(L)}\left( B\right) \right\Vert _{\mathcal{U}}  \notag \\
&&+2\left\vert h\right\vert ^{-1}\underset{\alpha \in \left[ s-\left\vert
h\right\vert ,s+\left\vert h\right\vert \right] }{\sup }\left\Vert \tau
_{t,\alpha }\left( B\right) -\tau _{t,\alpha }^{(L)}\left( B\right)
\right\Vert _{\mathcal{U}}\ .  \notag
\end{eqnarray}%
We proceed by estimating the four terms in the upper bound of (\ref%
{inportant upper boundbis}). The first one is already analyzed, see (\ref%
{inportant upper boundbis+3}). So, we start with the second. If nothing is
explicitly mentioned, the parameters $L\in \mathbb{R}_{0}^{+}$, $s,t,h\in
\mathbb{R}$, $\Lambda \in \mathcal{P}_{f}(\mathfrak{L})$ and $B\in \mathcal{U%
}_{\Lambda }$ are fixed.\smallskip

\noindent \textbf{2.} For any $\alpha \in \mathbb{R}$, observe that
\begin{eqnarray}
\left\Vert \delta _{s}^{(L)}\circ \tau _{t,s}^{(L)}\left( B\right) -\delta
_{\alpha }^{(L)}\circ \tau _{t,\alpha }^{(L)}\left( B\right) \right\Vert _{%
\mathcal{U}} &\leq &\left\Vert \left( \delta _{s}^{(L)}-\delta _{\alpha
}^{(L)}\right) \circ \tau _{t,\alpha }^{(L)}\left( B\right) \right\Vert _{%
\mathcal{U}}  \label{eq upper bound1-1} \\
&&+\left\Vert \delta _{s}^{(L)}\circ \left( \tau _{t,s}^{(L)}-\tau
_{t,\alpha }^{(L)}\right) \left( B\right) \right\Vert _{\mathcal{U}}\ .
\notag
\end{eqnarray}%
By using first (\ref{dynamic seriesbis}) for the interaction $\Psi ^{(s)}$
and potential $\mathbf{V}^{(s)}$ and then Lieb--Robinson bounds (Theorem \ref%
{Theorem Lieb-Robinsonnew} (i)) in the same way as (\ref{assertion bisbisbis}%
), one verifies that, for any $\alpha \in \mathbb{R}$ and $B\neq 0$,%
\begin{eqnarray}
&&\frac{\left\Vert \left( \delta _{s}^{(L)}-\delta _{\alpha }^{(L)}\right)
\circ \tau _{t,\alpha }^{(L)}\left( B\right) \right\Vert _{\mathcal{U}}}{%
2\left\Vert B\right\Vert _{\mathcal{U}}}  \label{ineq autre} \\
&\leq &\left\Vert \Psi ^{(s)}-\Psi ^{(\alpha )}\right\Vert _{\mathcal{W}}\
\mathrm{e}^{2\mathbf{D}\left\vert t-\alpha \right\vert \left\Vert \Psi
\right\Vert _{\infty }}|\Lambda |\left\Vert \mathbf{F}\right\Vert _{1,%
\mathfrak{L}}+\sum\limits_{x\in \Lambda }\Vert \mathbf{V}_{\left\{ x\right\}
}^{(\alpha )}-\mathbf{V}_{\left\{ x\right\} }^{(s)}\Vert _{\mathcal{U}}
\notag \\
&&+\mathbf{D}^{-1}\left( \mathrm{e}^{2\mathbf{D}\left\vert t-\alpha
\right\vert \left\Vert \Psi \right\Vert _{\infty }}-1\right)
\sum\limits_{x\in \mathfrak{L}\backslash \Lambda }\Vert \mathbf{V}_{\left\{
x\right\} }^{(\alpha )}-\mathbf{V}_{\left\{ x\right\} }^{(s)}\Vert _{%
\mathcal{U}}\sum_{y\in \Lambda }\mathbf{F}\left( \left\vert x-y\right\vert
\right) \ .  \notag
\end{eqnarray}%
By assumption, $\Psi \in C(\mathbb{R};\mathcal{W})$, $\{\mathbf{V}_{\left\{
x\right\} }^{(t)}\}_{x\in \mathfrak{L},t\in \mathbb{R}}$ is a bounded family
in $\mathcal{U}$, and $\mathbf{V}_{\left\{ x\right\} }\in C\left( \mathbb{R};%
\mathcal{U}\right) $ for any $x\in \mathfrak{L}$. So, by Lebesgue's
dominated convergence theorem, it follows from (\ref{ineq autre}) that
\begin{equation}
\lim_{h\rightarrow 0}\underset{\alpha \in \left[ s-\left\vert h\right\vert
,s+\left\vert h\right\vert \right] }{\sup }\left\Vert \left( \delta
_{s}^{(L)}-\delta _{\alpha }^{(L)}\right) \circ \tau _{t,\alpha
}^{(L)}\left( B\right) \right\Vert _{\mathcal{U}}=0\ .
\label{eq upper bound1-2}
\end{equation}%
On the other hand, by (\ref{cauchy1}),
\begin{equation}
\underset{\alpha \in \left[ s-\left\vert h\right\vert ,s+\left\vert
h\right\vert \right] }{\sup }\left\Vert \delta _{s}^{(L)}\circ \left( \tau
_{t,s}^{(L)}-\tau _{t,\alpha }^{(L)}\right) \left( B\right) \right\Vert _{%
\mathcal{U}}\leq \int_{s-\left\vert h\right\vert }^{s+\left\vert
h\right\vert }\left\Vert \delta _{s}^{(L)}\circ \delta _{\alpha }^{(L)}\circ
\tau _{t,\alpha }^{(L)}\left( B\right) \right\Vert _{\mathcal{U}}\mathrm{d}%
\alpha \ ,  \label{upper bound0}
\end{equation}%
where
\begin{eqnarray}
\left\Vert \delta _{s}^{(L)}\circ \delta _{\alpha }^{(L)}\circ \tau
_{t,\alpha }^{(L)}\left( B\right) \right\Vert _{\mathcal{U}} &\leq
&\sum\limits_{\mathcal{Z}_{1},\mathcal{Z}_{2}\in \mathcal{P}_{f}(\mathfrak{L}%
)}\left\Vert \left[ \Psi _{\mathcal{Z}_{1}}^{(s)},\Psi _{\mathcal{Z}%
_{2}}^{(\alpha )},\tau _{t,\alpha }^{(L)}\left( B\right) \right]
^{(3)}\right\Vert _{\mathcal{U}}  \label{upper bound} \\
&&+\sum\limits_{\mathcal{Z}\in \mathcal{P}_{f}(\mathfrak{L}%
)}\sum\limits_{x\in \mathfrak{L}}\left\Vert \left[ \Psi _{\mathcal{Z}}^{(s)},%
\mathbf{V}_{\left\{ x\right\} }^{(\alpha )},\tau _{t,\alpha }^{(L)}\left(
B\right) \right] ^{(3)}\right\Vert _{\mathcal{U}}  \notag \\
&&+\sum\limits_{\mathcal{Z}\in \mathcal{P}_{f}(\mathfrak{L}%
)}\sum\limits_{x\in \mathfrak{L}}\left\Vert \left[ \mathbf{V}_{\left\{
x\right\} }^{(s)},\Psi _{\mathcal{Z}}^{(\alpha )},\tau _{t,\alpha
}^{(L)}\left( B\right) \right] ^{(3)}\right\Vert _{\mathcal{U}}  \notag \\
&&+\sum\limits_{x,y\in \mathfrak{L}}\left\Vert \left[ \mathbf{V}_{\left\{
x\right\} }^{(s)},\mathbf{V}_{\left\{ y\right\} }^{(\alpha )},\tau
_{t,\alpha }^{(L)}\left( B\right) \right] ^{(3)}\right\Vert _{\mathcal{U}}\ .
\notag
\end{eqnarray}%
Similar to (\ref{estimate poly}), we use Theorems \ref{theorem exp tree
decay copy(1)}\ (i)\ and \ref{Theorem Lieb-Robinson copy(1)-bis} for $k=2$
to derive an upper bound for the r.h.s. of (\ref{upper bound}), uniformly
w.r.t. large $L\in \mathbb{R}_{0}^{+}$ and $\alpha \in \left[ s-1,s+1\right]
$. By (\ref{upper bound0}), it follows that
\begin{equation*}
\lim_{h\rightarrow 0}\underset{\alpha \in \left[ s-\left\vert h\right\vert
,s+\left\vert h\right\vert \right] }{\sup }\left\Vert \delta _{s}^{(L)}\circ
\left( \tau _{t,s}^{(L)}-\tau _{t,\alpha }^{(L)}\right) \left( B\right)
\right\Vert _{\mathcal{U}}=0\ .
\end{equation*}%
Combined with (\ref{eq upper bound1-1}) and (\ref{eq upper bound1-2}) this
yields
\begin{equation}
\lim_{h\rightarrow 0}\underset{\alpha \in \left[ s-\left\vert h\right\vert
,s+\left\vert h\right\vert \right] }{\sup }\left\Vert \delta _{s}^{(L)}\circ
\tau _{t,s}^{(L)}\left( B\right) -\delta _{\alpha }^{(L)}\circ \tau
_{t,\alpha }^{(L)}\left( B\right) \right\Vert _{\mathcal{U}}=0\ .
\label{inportant upper boundbis+1}
\end{equation}%
\smallskip

\noindent \textbf{3.} Similar to (\ref{assertion bisbisbis}), one gets from
Lieb--Robinson bounds (Theorem \ref{Theorem Lieb-Robinsonnew} (i)) that%
\begin{equation*}
\left\Vert \left( \delta _{s}^{(L)}-\delta _{s}\right) \circ \tau
_{t,s}^{(L)}\left( B\right) \right\Vert _{\mathcal{U}}\leq \left\Vert \Psi
\right\Vert _{\infty }\ \mathrm{e}^{2\mathbf{D}\left\vert t-s\right\vert
\left\Vert \Psi \right\Vert _{\infty }}\sum\limits_{y\in \Lambda
_{L}^{c}}\sum_{x\in \Lambda }\mathbf{F}\left( \left\vert x-y\right\vert
\right) \ ,
\end{equation*}%
which combined with (\ref{assertion bisbisbisbis}) gives%
\begin{equation}
\lim_{L\rightarrow \infty }\left\Vert \left( \delta _{s}^{(L)}-\delta
_{s}\right) \circ \tau _{t,s}^{(L)}\left( B\right) \right\Vert _{\mathcal{U}%
}=0\ .  \label{inportant upper boundbis+2}
\end{equation}%
\smallskip \noindent \textbf{4.} In the limit $h\rightarrow 0$, we take $%
L_{h}\rightarrow \infty $ such that
\begin{equation}
\underset{h\rightarrow 0}{\lim }\left\vert h\right\vert ^{-1}\underset{%
\alpha \in \left[ s-\left\vert h\right\vert ,s+\left\vert h\right\vert %
\right] }{\sup }\left\Vert \tau _{t,\alpha }\left( B\right) -\tau _{t,\alpha
}^{(L_{h})}\left( B\right) \right\Vert _{\mathcal{U}}=0\ .
\label{inportant upper boundbis+4}
\end{equation}%
This is possible because $\tau _{t,s}^{(L)}\left( B\right) $ converges to $%
\tau _{t,s}\left( B\right) $, uniformly for $t,s$ on compacta, by Corollary %
\ref{Theorem Lieb-Robinson copy(4)}. We eventually combine (\ref{inportant
upper boundbis+3}), (\ref{inportant upper boundbis+1}), (\ref{inportant
upper boundbis+2}), and (\ref{inportant upper boundbis+4}) with Inequality (%
\ref{inportant upper boundbis}) to arrive at the assertion.
\end{proof}

\noindent Note that uniqueness of the solution of the non--auto%
\-%
nomous evolution equation (\ref{cauchy trivial1bis}) cannot be proven as
done for the proof of uniqueness in Corollary \ref{Theorem Lieb-Robinson
copy(4)} (iii). Indeed, take any family $\{\hat{\tau}_{t,s}\}_{s,t\in {%
\mathbb{R}}}$ of bounded operators on $\mathcal{U}$ satisfying (\ref{cauchy
trivial1bis}) on $\mathcal{U}_{0}$. Then, as before in the proof of
Corollary \ref{Theorem Lieb-Robinson copy(4)} (iii), for any $B\in \mathcal{U%
}_{0}$, $L\in \mathbb{R}_{0}^{+}$ and $s,t\in \mathbb{R}$,%
\begin{equation}
\tau _{t,s}^{(L)}\left( B\right) -\hat{\tau}_{t,s}\left( B\right)
=\int_{s}^{t}\tau _{\alpha ,s}^{(L)}\circ \left( \delta _{\alpha
}^{(L)}-\delta _{\alpha }\right) \circ \hat{\tau}_{t,\alpha }\left( B\right)
\mathrm{d}\alpha \ ,  \label{inequality marche pas}
\end{equation}%
by using (\ref{cauchy2}). However, it is not clear this time whether the norm%
\begin{equation*}
\left\Vert \tau _{\alpha ,s}^{(L)}\circ \left( \delta _{\alpha
}^{(L)}-\delta _{\alpha }\right) \circ \hat{\tau}_{t,\alpha }\left( B\right)
\right\Vert _{\mathcal{U}}=\left\Vert \left( \delta _{\alpha }-\delta
_{\alpha }^{(L)}\right) \circ \hat{\tau}_{t,\alpha }\left( B\right)
\right\Vert _{\mathcal{U}}
\end{equation*}%
vanishes, as $L\rightarrow \infty $, even if (\ref{core limit}) for $\delta
_{\alpha }$ and $\delta _{\alpha }^{(L)}$ holds true, because $\hat{\tau}%
_{t,\alpha }\left( B\right) \in \mathrm{Dom}(\delta _{\alpha })$ can be
\emph{outside} $\mathcal{U}_{0}$. The strong convergence of $\delta _{\alpha
}^{(L)}$ to $\delta _{\alpha }$ on some core of $\delta _{\alpha }$ does not
imply, in general, the strong convergence on any core of $\delta _{\alpha }$%
. The equality (\ref{inequality marche pas}) with $\tau _{t,s},\delta _{s}$
replacing $\tau _{t,s}^{(L)},\delta _{s}^{(L)}$ is also not clear because (%
\ref{cauchy trivial1}) is only known to hold true on $\mathcal{U}_{0}$ and
\emph{a priori not} on the whole domain $\mathrm{Dom}(\delta _{\alpha })$ of
$\delta _{\alpha }$.

The non--autonomous evolution equation (\ref{non auto hyperbolic}) of
Theorem \ref{thm non auto} is not parabolic because the symmetric derivation
$\delta _{t}$, $t\in \mathbb{R}$, is generally \emph{not} the generator of
an analytic semigroup. Note also that no H\"{o}lder continuity condition is
imposed on $\{\delta _{t}\}_{t\in \mathbb{R}}$, like in the class of
parabolic evolution equations introduced in \cite[Hypotheses I--II]%
{parabolic def}. See also \cite{Schnaubelt1} or \cite[Sect. 5.6.]{Pazy} for
more simplified studies.

In fact, (\ref{non auto hyperbolic}) is rather related to Kato's \emph{%
hyperbolic} evolution equations \cite{Kato,Kato1973,Katobis}. The so--called
\emph{Kato quasi--stability}%
\index{Kato quasi--stability} is satisfied by the family of generators $%
\{\delta _{t}\}_{t\in \mathbb{R}}$ because they are always dissipative
operators, by Lemma \ref{lemma dissipative}. $\{\delta _{t}\}_{t\in \mathbb{R%
}}$ is also strongly continuous on the dense set $\mathcal{U}_{0}$, which is
a common core of all $\delta _{t}$, $t\in \mathbb{R}$. However, in general,
even for \emph{finite} range interactions $\Psi \in \mathcal{W}$, the
strongly continuous two--para%
\-%
meter family $\{\tau _{t,s}\}_{s,t\in {\mathbb{R}}}$ does \emph{not}
conserve the dense set $\mathcal{U}_{0}$, i.e., $\tau _{t,s}\left( \mathcal{U%
}_{0}\right) \nsubseteq \mathcal{U}_{0}$\ for any $s\neq t$. In some
specific situations one can directly show that the completion of the core $%
\mathcal{U}_{0}$ w.r.t. a conveniently chosen norm defines a so--called
admissible Banach space $\mathcal{Y}\supset \mathcal{U}_{0}$ of the
generator at any time, which satisfies further technical conditions leading
to Kato's hyperbolic conditions \cite{Kato,Kato1973,Katobis}. See also \cite[%
Sect. 5.3.]{Pazy} and \cite[Sect. VII.1]{Bru-Bach}, which is used in the
proof of Theorem \ref{pertubed dynam thm} (i). Nevertheless, the existence
of such a Banach space $\mathcal{Y}$ is a priori unclear in the general case
treated in Theorem \ref{thm non auto}. See for instance the uniqueness
problem explained just above.

Note that we only assume in Theorem \ref{thm non auto} some polynomial decay
for the interaction with (\ref{(3.3) NS generalized0}) and (\ref{assumption
boundedness2bis})--(\ref{assumption boundedness4}) (uniformly in time).
Recall that these assumptions are fulfilled for any interaction $\Psi \in
\mathcal{W}$ with (\ref{example polynomial}), provided the parameter $%
\epsilon \in \mathbb{R}^{+}$ is sufficiently large. In the case of
exponential decays, stronger results can be deduced from Lieb--Robinson
bounds for multi--commutators. For the interested reader, we give below one
example, which is based on interactions $\Phi $ satisfying the following
condition:%
\index{Decay function!exponential decay}

\begin{itemize}
\item \emph{Exponential decay.} Assume (\ref{(3.3) NS generalized}) and the
existence of constants $\upsilon >\varsigma $ and $D\in \mathbb{R}^{+}$ such
that%
\begin{equation}
\underset{x\in \mathfrak{L}}{\sup }\sum\limits_{\Lambda \in \mathcal{D}%
\left( x,m\right) }\left\Vert \Phi _{\Lambda }\right\Vert _{\mathcal{U}}\leq
D\mathrm{e}^{-\upsilon m}\ ,\qquad m\in \mathbb{N}_{0}\ ,
\label{assumption boundedness 5}
\end{equation}%
while%
\begin{equation}
\sum_{m\in \mathbb{N}}\mathbf{C}_{m}\mathrm{e}^{-\left( \varsigma +\upsilon
\right) m}<\infty \ .  \label{assumption boundedness 5bis}
\end{equation}
\end{itemize}

\begin{satz}[Graph norm convergence and Gevrey vectors]
\label{thm non auto copy(1)}\mbox{
}\newline
Let $\Psi \doteq \{\Psi ^{(t)}\}_{t\in \mathbb{R}}$ be a bounded family on $%
\mathcal{W}$ (i.e., $\left\Vert \Psi \right\Vert _{\infty }<\infty $ ), $\{%
\mathbf{V}^{(t)}\}_{t\in \mathbb{R}}$ a collection of potentials, and $B\in
\mathcal{U}_{0}$.
\index{Gevrey vector}For any $x\in \mathfrak{L}$ and $\Lambda \in \mathcal{P}%
_{f}(\mathfrak{L})$, $\Psi _{\Lambda },\mathbf{V}_{\left\{ x\right\} }\in
C\left( \mathbb{R};\mathcal{U}\right) $. Assume$\ $that (\ref{(3.3) NS
generalized}) and (\ref{assumption boundedness 5})--(\ref{assumption
boundedness 5bis}) hold for $\Phi =\Psi ^{(t)}$, uniformly in time.\newline
\emph{(i)} Graph norm convergence. As $L\rightarrow \infty $, $\tau
_{t,s}^{(L)}(B)$ converges, uniformly for $s,t$ on compacta, to $\tau
_{t,s}(B)$ within the normed space $(\mathrm{Dom}(\delta
_{s}^{m}),\left\Vert \cdot \right\Vert _{\delta _{s}^{m}})$, where, for all $%
m\in \mathbb{N}_{0}$, $\left\Vert \cdot \right\Vert _{\delta _{s}^{m}}$
stands for the graph norm of the densely defined operator $\delta _{s}^{m}$%
.\smallskip \newline
\emph{(ii)} Gevrey vectors. If $\{\mathbf{V}_{\left\{ x\right\}
}^{(t)}\}_{x\in \mathfrak{L},t\in \mathbb{R}}$ is a bounded family on $%
\mathcal{U}$ then, for any $\mathrm{T}\in {\mathbb{R}}_{0}^{+}$, there exist
$r\equiv r_{d,\mathrm{T},\Psi ,\mathbf{V},\mathbf{F}}\in \mathbb{R}^{+}$ and
$D\equiv D_{\mathrm{T},\Psi ,\mathbf{V}}\in \mathbb{R}^{+}$ such that, for
all $s,t\in \left[ -\mathrm{T},\mathrm{T}\right] $, $m_{0}\in \mathbb{N}_{0}$
and $B\in \mathcal{U}_{\Lambda _{m_{0}}}$,%
\begin{equation*}
\sum_{m\in \mathbb{N}}%
\frac{r^{m}}{\left( m!\right) ^{d}}\left\Vert \delta _{s}^{m}\circ \tau
_{t,s}(B)\right\Vert _{\mathcal{U}}\leq D\mathrm{e}^{m_{0}\varsigma
}\left\Vert B\right\Vert _{\mathcal{U}}\ .
\end{equation*}
\end{satz}

\begin{proof}
(i) The case $m=0$ follows from Corollary \ref{Theorem Lieb-Robinson copy(4)}%
. Let $m\in \mathbb{N}$ and $B\in \mathcal{U}_{0}$. Similar to (\ref{uper
bound3}), for any sufficiently large $L_{1},L_{2}\in \mathbb{R}_{0}^{+}$, $%
\Lambda _{L_{1}}\subsetneq \Lambda _{L_{2}}$,
\begin{align}
& \left\Vert \delta _{s}^{m}\circ \left( \tilde{\tau}_{t,s}^{(L_{2})}\left(
B\right) -\tilde{\tau}_{t,s}^{(L_{1})}\left( B\right) \right) \right\Vert _{%
\mathcal{U}}  \notag \\
& \leq \int_{\min \{s,t\}}^{\max \{s,t\}}\sum\limits_{\mathcal{Z}_{1},\ldots
,\mathcal{Z}_{m}\in \mathcal{P}_{f}(\mathfrak{L})}\left\Vert \left[ \hat{\tau%
}_{s,\alpha }^{(L_{1},L_{2})}(\Psi _{\mathcal{Z}_{m}}^{(s)}),\ldots ,\hat{%
\tau}_{s,\alpha }^{(L_{1},L_{2})}(\Psi _{\mathcal{Z}_{1}}^{(s)}),\right.
\right.  \notag \\
& \qquad \qquad \qquad \qquad \qquad \qquad \left. \left. ,B_{\alpha
}^{(L_{1},L_{2})},\tau _{t,\alpha }^{(L_{1})}(\tilde{B}_{t})\right]
^{(m+2)}\right\Vert _{\mathcal{U}}\mathrm{d}\alpha \ ,
\label{estimate polybis}
\end{align}%
see (\ref{def new}). From a straightforward generalization of (\ref{estimate
poly}) for multi--commutators of degree $m+2$ and the same kind of arguments
used in point \textbf{1.} of the proof of Theorem \ref{thm non auto}, the
r.h.s. of the above inequality tends to zero in the limit of large $%
L_{1},L_{2}\in \mathbb{R}_{0}^{+}$ ($\Lambda _{L_{1}}\subsetneq \Lambda
_{L_{2}}$). This holds for every $m\in \mathbb{N}$ because the interaction
has, by assumption, exponential decay, see (\ref{(3.3) NS generalized}) and (%
\ref{assumption boundedness 5})--(\ref{assumption boundedness 5bis}).

Consequently, $\{\delta _{s}^{m}\circ \tilde{\tau}_{t,s}^{(L)}\left(
B\right) \}_{L\in \mathbb{R}_{0}^{+}}$, and hence $\{\delta _{s}^{m}\circ
\tau _{t,s}^{(L)}\left( B\right) \}_{L\in \mathbb{R}_{0}^{+}}$, are Cauchy
nets in $\mathcal{U}$ for any fixed $s,t\in \mathbb{R}$ and $m\in \mathbb{N}$%
. At $m=0$, the limit is $\tau _{t,s}(B)$. As the operator $\delta _{s}$ is
closed, by induction, for any $m\in \mathbb{N}$ and $s,t\in \mathbb{R}$, $%
\tau _{t,s}(B)\in \mathrm{Dom}(\delta _{s}^{m})$ and $\delta _{s}^{m}\circ
\tau _{t,s}^{(L)}\left( B\right) $ converges to $\delta _{s}^{m}\circ \tau
_{t,s}\left( B\right) $, as $L\rightarrow \infty $.

\noindent (ii) For any $m\in \mathbb{N}$, $B\in \mathcal{U}_{0}$, and
sufficiently large $L\in \mathbb{R}_{0}^{+}$,%
\begin{eqnarray*}
&&\left\Vert \delta _{s}^{m}\circ \tau _{t,s}^{(L)}(B)\right\Vert _{\mathcal{%
U}} \\
&\leq &\sum_{\ell =0}^{m}\sum_{\pi \in \mathcal{S}_{\ell ,m}}\sum_{x_{\pi
\left( \ell \right) }\in \mathfrak{L}}\cdots \sum_{x_{\pi \left( m\right)
}\in \mathfrak{L}}\sum_{\mathcal{Z}_{1}\in \mathcal{P}_{f}(\mathfrak{L}%
)}\cdots \sum_{\mathcal{Z}_{\pi \left( \ell \right) -1}\in \mathcal{P}_{f}(%
\mathfrak{L})}\sum_{\mathcal{Z}_{\pi \left( \ell \right) +1}\in \mathcal{P}%
_{f}(\mathfrak{L})}\cdots \\
&&\cdots \sum_{\mathcal{Z}_{\pi \left( m\right) -1}\in \mathcal{P}_{f}(%
\mathfrak{L})}\sum_{\mathcal{Z}_{\pi \left( m\right) +1}\in \mathcal{P}_{f}(%
\mathfrak{L})}\cdots \sum_{\mathcal{Z}_{m}\in \mathcal{P}_{f}(\mathfrak{L})}
\\
&&\qquad \left\Vert \left[ \Psi _{\mathcal{Z}_{1}}^{(s)},\ldots ,\Psi _{%
\mathcal{Z}_{\pi \left( \ell \right) -1}}^{(s)},\mathbf{V}_{\{x_{\pi \left(
\ell \right) }\}}^{(s)},\Psi _{\mathcal{Z}_{\pi \left( \ell \right)
+1}}^{(s)},\right. \right. \\
&&\qquad \qquad \left. \left. \ldots ,\Psi _{\mathcal{Z}_{\pi \left(
m\right) -1}}^{(s)},\mathbf{V}_{\{x_{\pi \left( m\right) }\}}^{(s)},\Psi _{%
\mathcal{Z}_{\pi \left( m\right) +1}}^{(s)},\ldots ,\Psi _{\mathcal{Z}%
_{m}}^{(s)},\tau _{t,s}^{(L)}(B)\right] ^{(m+1)}\right\Vert _{\mathcal{U}}\ ,
\end{eqnarray*}%
with $\mathcal{S}_{\ell ,m}$ being defined by (\ref{SLK}) for $\ell \in
\{1,\ldots ,m\}$. For $\ell =0$, we use here the convention $\mathcal{S}%
_{0,m}\doteq \emptyset $ and all sums involving the maps $\pi $ in the
r.h.s. of the above inequality disappear in this case. Similar to (\ref%
{estimate polybis}), Lieb--Robinson bounds for multi--commutators imply
that, if $B\in \mathcal{U}_{\Lambda _{m_{0}}}$, $m_{0}\in \mathbb{N}_{0}$,
then the r.h.s. of the above inequality is bounded by $D(m!)^{d}r^{m}\mathrm{%
e}^{m_{0}\varsigma }\Vert B\Vert _{\mathcal{U}}$, uniformly for $s,t$ on
compacta, where $r\equiv r_{d,\mathrm{T},\Psi ,\mathbf{V},\mathbf{F}}\in
\mathbb{R}^{+}$ and $D\equiv D_{\mathrm{T},\Psi ,\mathbf{V}}\in \mathbb{R}%
^{+}$. We omit the details. By Assertion (i), the same bound thus holds for
the norm $\Vert \delta _{s}^{m}\circ \tau _{t,s}(B)\Vert _{\mathcal{U}}$ of
the limiting vector.
\end{proof}

The assumptions of Theorem \ref{thm non auto copy(1)} are satisfied for
interactions $\Psi ^{(t)}\in \mathcal{W}$ with (\ref{example}). Note
additionally that Theorem \ref{thm non auto copy(1)} for $s=t$ shows that
\begin{equation*}
\mathcal{U}_{0}\subseteq \underset{s\in \mathbb{R},m\in \mathbb{N}}{\bigcap }%
\mathrm{Dom}\left( \delta _{s}^{m}\right) \subset \mathcal{U}\ .
\end{equation*}%
In fact, $\mathcal{U}_{0}$ is a common core for $\{\delta _{s}\}_{s\in
\mathbb{R}}$ and thus the intersection of domains
\begin{equation*}
\underset{s\in \mathbb{R},m\in \mathbb{N}}{\bigcap }\mathrm{Dom}\left(
\delta _{s}^{m}\right) \subset \mathcal{U}
\end{equation*}%
is also a common core of $\{\delta _{s}\}_{s\in \mathbb{R}}$. Observe that,
at fixed $s\in \mathbb{R}$, the dense space
\begin{equation*}
\mathrm{Dom}\left( \delta _{s}^{\infty }\right) \doteq \underset{m\in
\mathbb{N}}{\bigcap }\mathrm{Dom}\left( \delta _{s}^{m}\right) \subset
\mathcal{U}
\end{equation*}%
is always a core of $\delta _{s}$. See, e.g., \cite[Chap. II, 1.8 Proposition%
]{EngelNagel}.

\subsection{Application to Response Theory\label{Section existence dynamics}}

In the present subsection we extend to the time--dependent case the
assertions of Section \ref{section Energy Increments as Power Series}. As
previously discussed, these results can be proven, also in the
non--autonomous case, for more general (time--dependent) perturbations of
the form (\ref{w more general}). See also proofs of Inequality (\ref%
{estimate poly}) and Theorem \ref{thm non auto copy(1)}. Similar to Section %
\ref{section Energy Increments as Power Series}, the case of perturbations
considered below is the relevant one to study linear and non--linear
responses of interacting fermions to time--dependent external
electromagnetic fields.

Let $\Psi \in \mathcal{W}$ and $\mathbf{V}$ be a potential. [So, these
objects do\emph{\ not} depend on time.] For any $l\in \mathbb{R}_{0}^{+}$,
we consider a map $(\eta ,t)\mapsto \mathbf{W}_{t}^{(l,\eta )}$ from $%
\mathbb{R}^{2}$ to the subspace of self--adjoint elements of $\mathcal{U}%
_{\Lambda _{l}}$. Like (\ref{bound assumption2}), we consider elements of
the form
\begin{equation}
\mathbf{W}_{t}^{(l,\eta )}\doteq \sum\limits_{x\in \Lambda
_{l}}\sum\limits_{z\in \mathfrak{L},|z|\leq 1}\mathbf{w}_{x,x+z}(\eta
,t)a_{x}^{\ast }a_{x+z}\ ,\qquad (\eta ,t)\in \mathbb{R}^{2},\ l\in \mathbb{R%
}_{0}^{+}\ ,  \label{bounded pertubationbis0}
\end{equation}%
where $\{\mathbf{w}_{x,y}\}_{x,y\in \mathfrak{L}}$ are complex--valued
functions of $(\eta ,t)\in \mathbb{R}^{2}$ with
\begin{equation}
\overline{\mathbf{w}_{x,y}\left( \eta ,t\right) }=\mathbf{w}_{y,x}\left(
\eta ,t\right) \qquad \text{and}\qquad \mathbf{w}_{x,y}(0,t)=0
\label{condition1}
\end{equation}%
for all $x,y\in \mathfrak{L}$ and $(\eta ,t)\in \mathbb{R}^{2}$. We assume
that $\{\mathbf{w}_{x,y}(\eta ,\cdot )\}_{x,y\in \mathfrak{L,}\eta \in
\mathbb{R}}$ is a family of continuous and uniformly bounded functions (of
time): There is $K_{1}\in \mathbb{R}^{+}$ such that%
\begin{equation}
\sup_{x,y\in \mathfrak{L}}\sup_{\eta ,t\in \mathbb{R}}\left\vert \mathbf{w}%
_{x,y}(\eta ,t)\right\vert \leq K_{1}\ .  \label{condition2}
\end{equation}%
The self--adjoint elements $\mathbf{W}_{t}^{(l,\eta )}$ of $\mathcal{U}$ are
related to perturbations of dynamics caused by time--dependent external
electromagnetic fields that vanish outside the box $\Lambda _{l}$. By the
above conditions on $\mathbf{w}_{x,y}$, for all $l,\eta \in \mathbb{R}$, $%
t\mapsto \mathbf{W}_{t}^{(l,\eta )}$ is a continuous map from $\mathbb{R}$
to $\mathcal{B}(\mathcal{U})$.

We now denote the perturbed dynamics by the family $\{\tilde{\tau}%
_{t,s}^{(l,\eta )}\}_{s,t\in {\mathbb{R}}}$ of $\ast $--automor%
\-%
phisms generated by the symmetric derivation
\begin{equation}
\delta _{t}^{(l,\eta )}\doteq \delta +i\left[ \mathbf{W}_{t}^{(l,\eta )},\
\cdot \ \right] \ ,\qquad l\in \mathbb{R}_{0}^{+},\ \eta \in \mathbb{R}\ ,
\label{bounded pertubationbis}
\end{equation}%
in the sense of Corollary \ref{Theorem Lieb-Robinson copy(4)}. [This family
of $\ast $--automor%
\-%
phisms has \emph{nothing} to do with (\ref{interactino picture}).] Recall
that $\delta $ is the symmetric derivation of Theorem \ref{Theorem
Lieb-Robinson copy(3)}. The last term in the r.h.s. of (\ref{bounded
pertubationbis}) is clearly a perturbation of $\delta $ which depends
continuously on time, in the sense of the operator norm on $\mathcal{B}(%
\mathcal{U})$. It is easy to prove in this case that $\{\tilde{\tau}%
_{t,s}^{(l,\eta )}\}_{s,t\in {\mathbb{R}}}$ is the unique
\index{Evolution equation!fundamental solution}\emph{fundamental solution}
of (\ref{cauchy trivial2}). It means that $\{%
\tilde{\tau}_{t,s}^{(l,\eta )}\}_{s,t\in \mathbb{R}}$ is strongly
continuous, conserves the domain
\begin{equation*}
\mathrm{Dom}(\delta _{t}^{(l,\eta )})=\mathrm{Dom}(\delta )\ ,
\end{equation*}%
satisfies%
\begin{equation*}
\tilde{\tau}_{t,\cdot }^{(l,\eta )}(B)\in C^{1}(\mathbb{R};(\mathrm{Dom}%
(\delta ),\left\Vert \cdot \right\Vert _{\mathcal{U}}))\ ,\quad \tilde{\tau}%
_{\cdot ,s}^{(l,\eta )}(B)\in C^{1}(\mathbb{R};(\mathrm{Dom}(\delta
),\left\Vert \cdot \right\Vert _{\mathcal{U}}))
\end{equation*}%
for all $B\in \mathrm{Dom}(\delta )$, and solves the abstract Cauchy initial
value\ problem (\ref{cauchy trivial2}) on $\mathrm{Dom}(\delta )$.

To explicitly verify this, define the family $\{\mathfrak{V}_{t,s}\}_{s,t\in
\mathbb{R}}\subset \mathcal{U}$ of unitary elements by the absolutely
summable series
\begin{equation}
\mathfrak{V}_{t,s}\doteq \mathbf{1}_{\mathcal{U}}\mathbf{+}\sum\limits_{k\in
{\mathbb{N}}}i^{k}\int_{s}^{t}\mathrm{d}s_{1}\cdots \int_{s}^{s_{k-1}}%
\mathrm{d}s_{k}\mathbf{W}_{s_{k},s_{k}}^{(l,\eta )}\cdots \mathbf{W}%
_{s_{1},s_{1}}^{(l,\eta )}\ ,  \label{Dyson interact}
\end{equation}%
where
\begin{equation*}
\mathbf{W}_{t,s}^{(l,\eta )}\doteq \tau _{t}(\mathbf{W}_{s}^{(l,\eta )})\in
\mathrm{Dom}(\delta )\ ,\quad \ l\in \mathbb{R}_{0}^{+},\ \eta ,s,t\in
\mathbb{R}\ .
\end{equation*}%
By using this unitary family, we obtain the following additional properties
of the perturbed dynamics:

\begin{satz}[Properties of the perturbed dynamics]
\label{pertubed dynam thm}\mbox{
}\newline
Let $\Psi \in \mathcal{W}$, $l\in \mathbb{R}_{0}^{+}$, $\eta ,\eta _{0}\in
\mathbb{R}$, and $\mathbf{V}$ be a potential. Assume Conditions (\ref%
{condition1})--(\ref{condition2}) with $\{\mathbf{w}_{x,y}(\eta ,\cdot
)\}_{x,y\in \mathfrak{L,}\eta \in \mathbb{R}}$ being a family of continuous
functions (of time). Then, the family $\{\tilde{\tau}_{t,s}^{(l,\eta
)}\}_{s,t\in {\mathbb{R}}}$ of $\ast $--automor%
\-%
phisms has the following properties: \newline
\emph{(i)}
\index{Evolution equation!non--autonomous}Non--autonomous evolution
equation. It is the unique fundamental solution of%
\begin{equation*}
\forall s,t\in {\mathbb{R}}:\qquad \partial _{s}%
\tilde{\tau}_{t,s}^{(l,\eta )}=-\delta _{s}^{(l,\eta )}\circ \tilde{\tau}%
_{t,s}^{(l,\eta )}\ ,\qquad \tilde{\tau}_{t,t}^{(l,\eta )}=\mathbf{1}_{%
\mathcal{U}}\ .
\end{equation*}%
\emph{(ii)}
\index{Interaction picture}Interaction picture. For any $s,t\in {\mathbb{R}}$%
,%
\begin{equation*}
\tilde{\tau}_{t,s}^{(l,\eta )}(B)=\tau _{-s}\left( \mathfrak{V}_{t,s}\tau
_{t}(B)\mathfrak{V}_{t,s}^{\ast }\right) \ ,\qquad B\in \mathcal{U}\ .
\end{equation*}%
\emph{(iii)}
\index{Dyson--Phillips series}Dyson--Phillips series. For any $s,t\in
\mathbb{R}$ and $B\in \mathcal{U}$,
\begin{eqnarray}
\tilde{\tau}_{t,s}^{(l,\eta )}\left( B\right) &=&\tilde{\tau}_{t,s}^{(l,\eta
_{0})}\left( B\right) +\sum\limits_{k\in {\mathbb{N}}}i^{k}\int_{s}^{t}%
\mathrm{d}s_{1}\cdots \int_{s}^{s_{k-1}}\mathrm{d}s_{k}  \label{Dyson tau 1}
\\
&&\qquad \qquad \left[ \mathbf{X}_{s_{k},s,s_{k}}^{(l,\eta _{0},\eta
)},\ldots ,\mathbf{X}_{s_{1},s,s_{1}}^{(l,\eta _{0},\eta )},\tilde{\tau}%
_{t,s}^{(l,\eta _{0})}(B)\right] ^{(k+1)}\ .  \notag
\end{eqnarray}%
Here, the series absolutely converges and%
\begin{equation}
\mathbf{X}_{t,s,\alpha }^{(l,\eta _{0},\eta )}\doteq \tilde{\tau}%
_{t,s}^{(l,\eta _{0})}\left( \mathbf{W}_{\alpha }^{(l,\eta )}-\mathbf{W}%
_{\alpha }^{(l,\eta _{0})}\right) \ ,\qquad l\in \mathbb{R}_{0}^{+},\ \alpha
,s,t,\eta _{0},\eta \in \mathbb{R}\ .
\end{equation}
\end{satz}

\begin{proof}
Before starting, note that Assertion (i) cannot be deduced from Theorem \ref%
{thm non auto} because the cases for which (\ref{condition divergence})
holds for some time $t\in \mathbb{R}$ is excluded by assumptions of that
theorem.\smallskip

\noindent \textbf{1.} Assertion (i) could be deduced from \cite[Theorem 6.1]%
{Kato}. Here, we use \cite[Theorem 88]{Bru-Bach} because it is proven from
three conditions (B1--B3) that are elementary to verify:

\begin{itemize}
\item[B1] \textit{(Kato quasi--stability).}%
\index{Kato quasi--stability} For any $t\in {\mathbb{R}}$, the generator $%
\delta _{t}^{(l,\eta )}$ is conservative, by Lemma \ref{lemma dissipative},
and Condition B1 of \cite[Section VII.1]{Bru-Bach} is clearly satisfied for $%
\lambda _{1},\ldots ,\lambda _{n}\in {\mathbb{R}}^{+}$, even with \emph{%
non--ordered} and all real times $t_{1},\ldots ,t_{n}\in {\mathbb{R}}$.
Indeed, $\{\delta _{t}^{(l,\eta )}\}_{t\in {\mathbb{R}}}$, $l\in {\mathbb{R}}%
_{0}^{+}$, generate strongly continuous groups, and not only $C_{0}$%
--semigroups.

\item[B2] \textit{(Domains and continuity).} $\{\mathbf{w}_{x,y}(\eta ,\cdot
)\}_{x,y\in \mathfrak{L,}\eta \in \mathbb{R}}$ is by assumption a family of
continuous functions (of time) and thus, the map $t\mapsto \lbrack \mathbf{W}%
_{t}^{(l,\eta )},\cdot ]$ from $\mathbb{R}$ to $\mathcal{B}(\mathcal{U})$ is
continuous in operator norm. It follows that Condition B2 of \cite[Section
VII.1]{Bru-Bach} holds with the Banach space
\begin{equation}
\mathcal{Y}\doteq (\mathrm{Dom}(\delta ),\Vert \cdot \Vert _{\delta })\ ,
\label{space Y}
\end{equation}%
$\Vert \cdot \Vert _{\delta }$ being the graph norm of the closed operator $%
\delta $.

\item[B3] \textit{(Intertwining condition).} Since $\delta $ is a symmetric
derivation with core $\mathcal{U}_{0}$ (Theorem \ref{Theorem Lieb-Robinson
copy(3)} (ii)) and $\mathbf{W}_{t}^{(l,\eta )}\in \mathcal{U}_{\Lambda _{l}}$%
, for any $l\in \mathbb{R}_{0}^{+}$, $\eta \in \mathbb{R}$, $t\in {\mathbb{R}%
}$ and $B\in \mathrm{Dom}(\delta )$,
\begin{equation*}
\delta \left( \left[ \mathbf{W}_{t}^{(l,\eta )},B\right] \right) -\left[
\mathbf{W}_{t}^{(l,\eta )},\delta \left( B\right) \right] =\left[ \delta
\left( \mathbf{W}_{t}^{(l,\eta )}\right) ,B\right] \in \mathcal{U}
\end{equation*}%
while, by using (\ref{inequlity utile}), one verifies that
\begin{eqnarray*}
\left\Vert \left[ \delta \left( \mathbf{W}_{t}^{(l,\eta )}\right) ,B\right]
\right\Vert _{\mathcal{U}} &\leq &4\Vert B\Vert _{\mathcal{U}}\Vert \mathbf{W%
}_{t}^{(l,\eta )}\Vert _{\mathcal{U}} \\
&&\times \left( \left\vert \Lambda _{l}\right\vert \mathbf{F}\left( 0\right)
\left\Vert \Psi \right\Vert _{\mathcal{W}}+\sum\limits_{x\in \Lambda
_{l}}\left\Vert \mathbf{V}_{\left\{ x\right\} }\right\Vert _{\mathcal{U}%
}\right) \ .
\end{eqnarray*}%
In particular, Condition B3 of \cite[Section VII.1]{Bru-Bach} holds true
with $\Theta =\delta $.
\end{itemize}

\noindent Therefore, similar to \cite[Theorem 70 (v)]{Bru-Bach}, we infer
from an extension of \cite[Theorem 88]{Bru-Bach}, which takes into account
the fact that B1 holds with non--ordered real times (see, e.g., the proof of
\cite[Lemma 89]{Bru-Bach}), the existence of a unique solution $\{\mathfrak{W%
}_{s,t}\}_{_{s,t\in \mathbb{R}}}$ of the non--auto%
\-%
nomous evolution equation%
\index{Evolution equation!non--autonomous}
\begin{equation}
\forall s,t\in {\mathbb{R}}:\qquad \partial _{s}\mathfrak{W}_{s,t}=-\delta
_{s}^{(l,\eta )}\circ \mathfrak{W}_{s,t}\ ,\qquad \mathfrak{W}_{t,t}=\mathbf{%
1}_{\mathcal{U}}\ ,  \label{non-autonomous evolution equation}
\end{equation}%
in the strong sense on $\mathrm{Dom}(\delta )\subset \mathcal{U}$. Here, $\{%
\mathfrak{W}_{s,t}\}_{_{s,t\in \mathbb{R}}}$ is an evolution family of $%
\mathcal{B}\left( \mathcal{U}\right) $, that is, a strongly continuous
two--parameter family of bounded operators acting on $\mathcal{U}$ that
satisfies the cocycle (Chapman--Kolmogorov) property%
\index{Cocycle property}
\begin{equation*}
\forall t,r,s\in \mathbb{R}:\qquad \mathfrak{W}_{s,t}=\mathfrak{W}_{s,r}%
\mathfrak{W}_{r,t}\ .
\end{equation*}%
\smallskip

\noindent \textbf{2.} Note now that the family $\{\mathfrak{V}%
_{t,s}\}_{s,t\in \mathbb{R}}$ was already studied in the proof of \cite[%
Theorem 5.3]{OhmI} for general closed symmetric derivations $\delta $ on $%
\mathcal{U}$: The series (\ref{Dyson interact}) absolutely converges in the
Banach space $\mathcal{Y}$ (\ref{space Y}). Additionally, for any $s,t\in
\mathbb{R}$,%
\begin{equation*}
\partial _{t}\mathfrak{V}_{t,s}=i\mathfrak{V}_{t,s}\mathbf{W}_{t,t}^{(l,\eta
)}\qquad
\text{and}\qquad \partial _{s}\mathfrak{V}_{t,s}=-i\mathbf{W}_{s,s}^{(l,\eta
)}\mathfrak{V}_{t,s}
\end{equation*}%
hold in the sense of the Banach space $\mathcal{Y}$, and thus also in the
sense of $\mathcal{U}$. Therefore, for any $s,t\in \mathbb{R}$,
\begin{equation}
\mathfrak{W}_{s,t}\left( B\right) =\tau _{-s}\left( \mathfrak{V}_{t,s}\tau
_{t}(B)\mathfrak{V}_{t,s}^{\ast }\right) \ ,\qquad B\in \mathcal{U}\ .
\label{definition tho chap inv}
\end{equation}%
To show this equality, use the fact that the r.h.s. of this equation defines
an evolution family that is a fundamental solution of (\ref{non-autonomous
evolution equation}), see \cite[Eqs. (5.24)--(5.26)]{OhmI}. \smallskip

\noindent \textbf{3.} Since $\{\tau _{t}\}_{t\in {\mathbb{R}}}$ is a group
of $\ast $--auto%
\-%
morphisms and $\{\mathfrak{V}_{t,s}\}_{s,t\in \mathbb{R}}$ is a family of
unitary elements of $\mathcal{U}$, we deduce from (\ref{definition tho chap
inv}) that $\{\mathfrak{W}_{s,t}\}_{s,t\in \mathbb{R}}$ is a collection of $%
\ast $--auto%
\-%
morphisms of the $C^{\ast }$--algebra $\mathcal{U}$. We also infer from (\ref%
{definition tho chap inv}) that the two--parameter evolution family $\{%
\mathfrak{W}_{s,t}\}_{s,t\in \mathbb{R}}$ solves on $\mathrm{Dom}(\delta )$
the abstract Cauchy initial value\ problem%
\index{Evolution equation!non--autonomous}
\begin{equation}
\forall s,t\in {\mathbb{R}}:\qquad \partial _{t}\mathfrak{W}_{s,t}=\mathfrak{%
W}_{s,t}\circ \delta _{t}^{(l,\eta )}\ ,\qquad \mathfrak{W}_{s,s}=\mathbf{1}%
_{\mathcal{U}}\ .  \label{non-autonomous evolution equation2}
\end{equation}%
The solution of (\ref{non-autonomous evolution equation2}) is unique in $%
\mathcal{B}(\mathcal{U})$, by Corollary \ref{Theorem Lieb-Robinson copy(4)}
(iii). We thus arrive at Assertions (i)--(ii) with the equality%
\begin{equation}
\tilde{\tau}_{t,s}^{(l,\eta )}=\mathfrak{W}_{s,t}\ ,\qquad l\in \mathbb{R}%
_{0}^{+},\ \eta ,s,t\in \mathbb{R}\ .  \label{definition dynamics}
\end{equation}%
\smallskip

\noindent \textbf{4.} For any $l\in \mathbb{R}_{0}^{+}$, $s,t\in \mathbb{R}$%
, $\eta ,\eta _{0}\in \mathbb{R}$, and $B\in \mathcal{U}$, define
\begin{align}
\hat{\tau}_{t,s}^{(l,\eta ,\eta _{0})}\left( B\right) & \doteq \tilde{\tau}%
_{t,s}^{(l,\eta _{0})}\left( B\right) +\sum\limits_{k\in {\mathbb{N}}%
}i^{k}\int_{s}^{t}\mathrm{d}s_{1}\cdots \int_{s}^{s_{k-1}}\mathrm{d}s_{k}
\label{dyson1} \\
& \qquad \qquad \left[ \mathbf{X}_{s_{k},s,s_{k}}^{(l,\eta _{0},\eta
)},\ldots ,\mathbf{X}_{s_{1},s,s_{1}}^{(l,\eta _{0},\eta )},\tilde{\tau}%
_{t,s}^{(l,\eta _{0})}(B)\right] ^{(k+1)}\ .  \notag
\end{align}%
This series is well--defined and absolutely convergent. Indeed, because of (%
\ref{condition2}), there is a constant $D\in \mathbb{R}^{+}$ such that, for
all $l\in \mathbb{R}_{0}^{+}$ and $\eta ,\eta _{0}\in \mathbb{R}$,%
\begin{equation*}
\underset{t\in \mathbb{R}}{\sup }\
\big \|%
\delta _{t}^{(l,\eta )}-\delta _{t}^{(l,\eta _{0})}%
\big \|%
_{\mathcal{B}\left( \mathcal{U}\right) }<D\ .
\end{equation*}%
It follows that
\begin{equation}
\big \|%
\hat{\tau}_{t,s}^{(l,\eta ,\eta _{0})}%
\big \|%
_{\mathcal{B}\left( \mathcal{U}\right) }\leq \mathrm{e}^{D\left( t-s\right)
}\ ,\qquad l\in \mathbb{R}_{0}^{+},\ s,t\in \mathbb{R},\ \eta ,\eta _{0}\in
\mathbb{R}\ .  \label{borne utile}
\end{equation}%
See, e.g., \cite[Chap. 5, Theorems 2.3 and 3.1]{Pazy}. Now, for any $l\in
\mathbb{R}_{0}^{+}$, $s,t\in \mathbb{R}$, $\eta ,\eta _{0}\in \mathbb{R}$,
and $B\in \mathcal{U}$, note that (\ref{dyson1}) yields
\begin{equation*}
\hat{\tau}_{t,s}^{(l,\eta ,\eta _{0})}\left( B\right) =\tilde{\tau}%
_{t,s}^{(l,\eta _{0})}\left( B\right) +i\int_{s}^{t}\mathrm{d}s_{1}\hat{\tau}%
_{s_{1},s}^{(l,\eta ,\eta _{0})}\left( \left[ \mathbf{W}_{s_{1}}^{(l,\eta )}-%
\mathbf{W}_{s_{1}}^{(l,\eta _{0})},\tilde{\tau}_{t,s_{1}}^{(l,\eta _{0})}(B)%
\right] \right)
\end{equation*}%
from which we deduce that $\{\hat{\tau}_{t,s}^{(l,\eta )}\}_{s,t\in {\mathbb{%
R}}}$ solves (\ref{cauchy trivial1}), by (\ref{non-autonomous evolution
equation2})--(\ref{definition dynamics}), (\ref{borne utile}) and continuity
of the maps $t\mapsto \mathbf{W}_{t}^{(l,\eta )}$ and $t\mapsto \tilde{\tau}%
_{t,s}^{(l,\eta _{0})}(B)$ from $\mathbb{R}$ to $\mathcal{U}$. Hence, by
Corollary \ref{Theorem Lieb-Robinson copy(4)} (iii), $\hat{\tau}%
_{t,s}^{(l,\eta ,\eta _{0})}=\tilde{\tau}_{t,s}^{(l,\eta )}$ for any $l\in
\mathbb{R}_{0}^{+}$, $s,t\in \mathbb{R}$ and $\eta ,\eta _{0}\in \mathbb{R}$.
\end{proof}

Now, by assuming the uniform Lipschitz continuity of the family
\begin{equation*}
\{\mathbf{w}_{x,y}(\cdot ,t)\}_{x,y\in \mathfrak{L,}t\in \mathbb{R}}
\end{equation*}
of functions (of $\eta $), i.e., for all parameters $\eta ,\eta _{0}\in
\mathbb{R}$,
\begin{equation}
\sup_{x,y\in \mathfrak{L}}\sup_{t\in \mathbb{R}}\left\vert \mathbf{w}%
_{x,y}(\eta ,t)-\mathbf{w}_{x,y}(\eta _{0},t)\right\vert \leq
K_{1}\left\vert \eta -\eta _{0}\right\vert \ ,  \label{condition3}
\end{equation}%
we can extend Theorem \ref{Thm Heat production as power series copy(2)} to
the non--autonomous case.

To this end, for some interaction $\Phi $ with energy observables $%
U_{\Lambda _{L}}^{\Phi }$ defined by (\ref{energy observable}) we study the
increment (\ref{generic Dyson--Phillips series2}), which now equals%
\index{Increment}
\begin{equation}
\mathbf{T}_{t,s}^{(l,\eta ,L)}\doteq
\tilde{\tau}_{t,s}^{(l,\eta )}(U_{\Lambda _{L}}^{\Phi })-\tau
_{t,s}(U_{\Lambda _{L}}^{\Phi })\ ,\qquad l,L\in \mathbb{R}_{0}^{+},\
s,t,\eta \in \mathbb{R}\ .  \label{increment}
\end{equation}%
By (\ref{condition1}), note\ again that $\mathbf{T}_{t,s}^{(l,0,L)}=0$.
Exactly like in the proof of Theorem \ref{Thm Heat production as power
series copy(2)}, we prove a version of Taylor's theorem for increments in
the non--auto%
\-%
nomous case:

\begin{satz}[Taylor's theorem for increments]
\label{Thm Heat production as power series copy(3)}\mbox{
}\newline
Let $l,\mathrm{T}\in \mathbb{R}_{0}^{+}$, $s,t\in \left[ -\mathrm{T},\mathrm{%
T}\right] $, $\eta ,\eta _{0}\in \mathbb{R}$, $\Psi \in \mathcal{W}$, and $%
\mathbf{V}$ be any potential.%
\index{Taylor's theorem}%
\index{Increment} Assume (\ref{(3.3) NS generalized0}) with $\varsigma >d$, (%
\ref{condition1})--(\ref{condition2}) and (\ref{condition3}), with $\{%
\mathbf{w}_{x,y}(\eta ,\cdot )\}_{x,y\in \mathfrak{L,}\eta \in \mathbb{R}}$
being a family of continuous functions (of time). Take an interaction $\Phi $
satisfying (\ref{assumption boundedness2}) with $\mathbf{v}_{m}=\left(
1+m\right) ^{\varsigma }$. Then:\newline
\emph{(i)} The map $\eta \mapsto \mathbf{T}_{t,s}^{(l,\eta ,L)}$ converges
uniformly on $\mathbb{R}$, as $L\rightarrow \infty $, to a continuous
function $\mathbf{T}_{t,s}^{(l,\eta )}$ of $\eta $ and
\begin{equation*}
\mathbf{T}_{t,s}^{(l,\eta )}-\mathbf{T}_{t,s}^{(l,\eta
_{0})}=\sum\limits_{\Lambda \in \mathcal{P}_{f}(\mathfrak{L})}i\int_{s}^{t}%
\mathrm{d}s_{1}%
\tilde{\tau}_{s_{1},s}^{(l,\eta )}\left( \left[ \mathbf{W}_{s_{1}}^{(l,\eta
)}-\mathbf{W}_{s_{1}}^{(l,\eta _{0})},\tilde{\tau}_{t,s_{1}}^{(l,\eta
_{0})}(\Phi _{\Lambda })\right] \right) \ .
\end{equation*}%
\emph{(ii)} For any $m\in \mathbb{N}$ satisfying $d(m+1)<\varsigma $,%
\begin{eqnarray*}
&&\mathbf{T}_{t,s}^{(l,\eta )}-\mathbf{T}_{t,s}^{(l,\eta _{0})} \\
&=&\sum\limits_{k=1}^{m}\sum\limits_{\Lambda \in \mathcal{P}_{f}(\mathfrak{L}%
)}i^{k}\int_{s}^{t}\mathrm{d}s_{1}\cdots \int_{s}^{s_{k-1}}\mathrm{d}s_{k}%
\left[ \mathbf{X}_{s_{k},s,s_{k}}^{(l,\eta _{0},\eta )},\ldots ,\mathbf{X}%
_{s_{1},s,s_{1}}^{(l,\eta _{0},\eta )},\tilde{\tau}_{t,s}^{(l,\eta
_{0})}(\Phi _{\Lambda })\right] ^{(k+1)} \\
&&+\sum\limits_{\Lambda \in \mathcal{P}_{f}(\mathfrak{L})}i^{m+1}\int_{s}^{t}%
\mathrm{d}s_{1}\cdots \int_{s}^{s_{m}}\mathrm{d}s_{m+1} \\
&&\tilde{\tau}_{s_{m+1},s}^{(l,\eta )}\left( \left[ \mathbf{W}%
_{s_{m+1}}^{(l,\eta )}-\mathbf{W}_{s_{m+1}}^{(l,\eta _{0})},\mathbf{X}%
_{s_{m},s_{m+1},s_{m}}^{(l,\eta _{0},\eta )},\ldots ,\mathbf{X}%
_{s_{1},s_{m+1},s_{1}}^{(l,\eta _{0},\eta )},\tilde{\tau}_{t,s_{m+1}}^{(l,%
\eta _{0})}(\Phi _{\Lambda })\right] ^{(m+2)}\right) \ .
\end{eqnarray*}%
\emph{(iii)} All the above series in $\Lambda $ absolutely converge: For any
$m\in \mathbb{N}$ satisfying $d(m+1)<\varsigma $, $k\in \{1,\ldots ,m\}$,
and $\{s_{j}\}_{j=1}^{m+1}\subset \lbrack -\mathrm{T},\mathrm{T}]$,
\begin{equation*}
\sum\limits_{\Lambda \in \mathcal{P}_{f}(\mathfrak{L})}\left\Vert \left[
\mathbf{X}_{s_{k},s,s_{k}}^{(l,\eta _{0},\eta )},\ldots ,\mathbf{X}%
_{s_{1},s,s_{1}}^{(l,\eta _{0},\eta )},\tilde{\tau}_{t,s}^{(l,\eta
_{0})}(\Phi _{\Lambda })\right] ^{(k+1)}\right\Vert _{\mathcal{U}}\leq
D\left\vert \Lambda _{l}\right\vert \left\vert \eta -\eta _{0}\right\vert
^{k}
\end{equation*}%
and
\begin{multline*}
\sum\limits_{\Lambda \in \mathcal{P}_{f}(\mathfrak{L})}\left\Vert \tilde{\tau%
}_{s_{m+1},s}^{(l,\eta )}\left( \left[ \mathbf{W}_{s_{m+1}}^{(l,\eta )}-%
\mathbf{W}_{s_{m+1}}^{(l,\eta _{0})},\mathbf{X}_{s_{m},s_{m+1},s_{m}}^{(l,%
\eta _{0},\eta )},\ldots ,\mathbf{X}_{s_{1},s_{m+1},s_{1}}^{(l,\eta
_{0},\eta )},\tilde{\tau}_{t,s_{m+1}}^{(l,\eta _{0})}(\Phi _{\Lambda })%
\right] ^{(m+2)}\right) \right\Vert _{\mathcal{U}} \\
\leq D\left\vert \Lambda _{l}\right\vert \left\vert \eta -\eta
_{0}\right\vert ^{m+1}
\end{multline*}%
for some constant $D\in \mathbb{R}^{+}$ depending only on $m,d,\mathrm{T}%
,\Psi ,K_{1},\Phi ,\mathbf{F}$. The last assertion also holds for $m=0$.
\end{satz}

\begin{proof}
By Theorems \ref{theorem exp tree decay copy(1)} and \ref{Theorem
Lieb-Robinson copy(1)-bis}, Corollary \ref{theorem exp tree decay} holds in
the non--auto%
\-%
nomous case. Moreover, by Lemma \ref{Lemma Series representation of the
dynamics copy(1)}, Lemma \ref{Lemma Series representation of the dynamics}
is also satisfied in the non--auto%
\-%
nomous case. Therefore, the proof is an easy extension of the proof of
Theorem \ref{Thm Heat production as power series copy(2)}.
\end{proof}

If the interaction has exponential decay, we show that the map $\eta \mapsto
|\Lambda _{l}|^{-1}\mathbf{T}_{t,s}^{(l,\eta )}$ from $\mathbb{R}$ to $%
\mathcal{U}$ is bounded in the sense of Gevrey classes, uniformly w.r.t. $%
l\in \mathbb{R}_{0}^{+}$. This corresponds to Theorem \ref{Thm Heat
production as power series copy(1)} in the non--auto%
\-%
nomous case:

\begin{satz}[Increments as Gevrey maps]
\label{Thm Heat production as power series copy(4)}\mbox{
}\newline
Let $l,\mathrm{T}\in \mathbb{R}_{0}^{+}$, $s,t\in \left[ -\mathrm{T},\mathrm{%
T}\right] $, $\Psi \in \mathcal{W}$, and $\mathbf{V}$ be any potential.
\index{Gevrey map}%
\index{Increment}Assume (\ref{(3.3) NS generalized}) and take an interaction
$\Phi $ satisfying (\ref{assumption boundedness2}) with $\mathbf{v}_{m}=%
\mathrm{e}^{m\varsigma }$. For all $x,y\in \mathfrak{L}$, assume further the
real analyticity of the map $\eta \mapsto \mathbf{w}_{x,y}(\eta ,\cdot )$
from $\mathbb{R}$ to the Banach space $C(\mathbb{R};\mathbb{C)}$, which is
equipped with the supremum norm, as well as the existence of $r\in \mathbb{R}%
^{+}$ such that%
\begin{equation*}
K_{2}\doteq \sup_{x,y\in \mathfrak{L}}\ \sup_{m\in \mathbb{N}}\ \sup_{\eta
,t\in \mathbb{R}}%
\frac{r^{m}\partial _{\eta }^{m}\mathbf{w}_{x,y}(\eta ,t)}{m!}<\infty \ .
\end{equation*}%
\emph{(i)} Smoothness. As a function of $\eta \in \mathbb{R}$, $\mathbf{T}%
_{t,s}^{(l,\eta )}\in C^{\infty }(\mathbb{R};\mathcal{U})$ and for any $m\in
\mathbb{N}$,
\begin{eqnarray*}
\partial _{\eta }^{m}\mathbf{T}_{t,s}^{(l,\eta )}
&=&\sum\limits_{k=1}^{m}\sum\limits_{\Lambda \in \mathcal{P}_{f}(\mathfrak{L}%
)}i^{k}\int_{s}^{t}\mathrm{d}s_{1}\cdots \int_{s}^{s_{k-1}}\mathrm{d}s_{k} \\
&&\qquad \left. \partial _{\varepsilon }^{m}\left[ \mathbf{X}%
_{s_{k},s,s_{k}}^{(l,\eta ,\eta +\varepsilon )},\ldots ,\mathbf{X}%
_{s_{1},s,s_{1}}^{(l,\eta ,\eta +\varepsilon )},\tilde{\tau}_{t,s}^{(l,\eta
)}(\Phi _{\Lambda })\right] ^{(k+1)}\right\vert _{\varepsilon =0}\ .
\end{eqnarray*}%
The above series in $\Lambda $ are absolutely convergent.\newline
\emph{(ii)} Uniform boundedness of the Gevrey norm of density of increments.
There exist $\tilde{r}\equiv \tilde{r}_{d,\mathrm{T},\Psi ,K_{2},\mathbf{F}%
}\in \mathbb{R}^{+}$ and $D\equiv D_{\mathrm{T},\Psi ,K_{2},\Phi }\in
\mathbb{R}^{+}$ such that, for all $l\in \mathbb{R}_{0}^{+}$, $\eta \in
\mathbb{R}$ and $s,t\in \left[ -\mathrm{T},\mathrm{T}\right] $,
\begin{equation*}
\sum_{m\in \mathbb{N}}\frac{\tilde{r}^{m}}{\left( m!\right) ^{d}}\sup_{l\in
\mathbb{R}_{0}^{+}}\left\Vert \left\vert \Lambda _{l}\right\vert
^{-1}\partial _{\eta }^{m}\mathbf{T}_{t,s}^{(l,\eta )}\right\Vert _{\mathcal{%
U}}\leq D\ .
\end{equation*}
\end{satz}

\begin{proof}
Like for Theorem \ref{Thm Heat production as power series copy(3)}, the
assertions are easily proven by extending the proof of Theorem \ref{Thm Heat
production as power series copy(1)} to the non--auto%
\-%
nomous case.
\end{proof}

This theorem has important consequences in terms of increment density limit
\begin{equation*}
\underset{l\rightarrow \infty }{\lim }|\Lambda _{l}|^{-1}\rho (\mathbf{T}%
_{t,s}^{(l,\eta )})
\end{equation*}%
at any fixed $s,t\in \mathbb{R}$ and state $\rho \in \mathcal{U}^{\ast }$.
This limit is to be understood as an accumulation point of the bounded net $%
\{|\Lambda _{l}|^{-1}\rho (\mathbf{T}_{t,s}^{(l,\eta )})\}_{l>0}$:

\begin{koro}[Increment density limit]
\label{Theorem Lieb-Robinson copy(5)}\mbox{
}\newline
Let $\rho \in \mathcal{U}^{\ast }$.%
\index{Increment density} Under the conditions of Theorem \ref{Thm Heat
production as power series copy(4)}, there is a subsequence $\{l_{n}\}_{n\in
\mathbb{N}}\subset \mathbb{R}_{0}^{+}$ such that, for all $s,t\in \left[ -%
\mathrm{T},\mathrm{T}\right] $, the following limit exists%
\begin{equation*}
\eta \mapsto \mathbf{g}_{t,s}\left( \eta \right) \doteq \underset{%
n\rightarrow \infty }{\lim }|\Lambda _{l_{n}}|^{-1}\rho (\mathbf{T}%
_{t,s}^{(l_{n},\eta )})
\end{equation*}%
and defines a smooth function $\mathbf{g}_{t,s}\in C^{\infty }(\mathbb{R})$.
Furthermore, there exist $%
\tilde{r}\equiv \tilde{r}_{d,\mathrm{T},\Psi ,K_{2},\mathbf{F}}\in \mathbb{R}%
^{+}$ and $D\equiv D_{\mathrm{T},\Psi ,K_{2},\Phi }\in \mathbb{R}^{+}$ such
that, for all $\eta \in \mathbb{R}$ and $s,t\in \left[ -\mathrm{T},\mathrm{T}%
\right] $,
\begin{equation*}
\sum_{m\in \mathbb{N}}\frac{\tilde{r}^{m}}{\left( m!\right) ^{d}}\left\vert
\partial _{\eta }^{m}\mathbf{g}_{t,s}\left( \eta \right) \right\vert \leq D\
.
\end{equation*}
\end{koro}

\begin{proof}
Let $\mathrm{T}\in \mathbb{R}_{0}^{+}$. By Theorem \ref{Thm Heat production
as power series copy(3)} (i) for $\eta _{0}=0$ together with (\ref%
{condition1}) and Corollary \ref{Theorem Lieb-Robinson copy(4)} (ii),
\begin{equation}
\sup_{l\in \mathbb{R}_{0}^{+}}\sup_{\eta \in \mathbb{R}}\sup_{s,t\in \left[ -%
\mathrm{T},\mathrm{T}\right] }\left\{ |\Lambda _{l}|^{-1}\rho (\mathbf{T}%
_{t,s}^{(l,\eta )})\right\} <\infty \ .  \label{borne sup1}
\end{equation}%
Furthermore, we infer from Theorem \ref{Thm Heat production as power series
copy(4)} that, for any $m\in \mathbb{N}$,
\begin{equation}
\sup_{l\in \mathbb{R}_{0}^{+}}\sup_{\eta \in \mathbb{R}}\sup_{s,t\in \left[ -%
\mathrm{T},\mathrm{T}\right] }\left\{ |\Lambda _{l}|^{-1}\rho (\partial
_{\eta }^{m}\mathbf{T}_{t,s}^{(l,\eta )})\right\} <\infty \ .
\label{borne sup2}
\end{equation}%
By (\ref{borne sup1}) and (\ref{borne sup2}), the assertions are
consequences of Theorem \ref{Thm Heat production as power series copy(4)}
combined with the mean value theorem and the (Arzel\`{a}--) Ascoli theorem
\cite[Theorem A5]{Rudin}. Indeed, $\{l_{n}\}_{n\in \mathbb{N}}\subset
\mathbb{R}_{0}^{+}$ is taken as a so--called diagonal sequence $%
l_{n}=l_{n}^{(n)}$ of a family $\{l_{n}^{(m)}\}_{n\in \mathbb{N}}$, $m\in
\mathbb{N}_{0}$, of sequences in $\mathbb{R}_{0}^{+}$ such that, for all $%
m\in \mathbb{N}_{0}$, the $m$--th derivative $|\Lambda
_{l_{n}}|^{-1}\partial _{\eta }^{m}\mathbf{T}_{t,s}^{(l_{n}^{(m)},\eta )}$
uniformly converges as $n\rightarrow \infty $. With this choice,%
\begin{equation*}
\partial _{\eta }^{m}\mathbf{g}_{t,s}\left( \eta \right) =\underset{%
n\rightarrow \infty }{\lim }|\Lambda _{l_{n}}|^{-1}\rho (\partial _{\eta
}^{m}\mathbf{T}_{t,s}^{(l_{n},\eta )})\ .
\end{equation*}
\end{proof}

From the above corollary, at dimension $d=1$ and for $s,t$ on compacta, the
increment density limit $\mathbf{g}_{t,s}\in C^{\infty }(\mathbb{R})$
defines a real analytic function. As a consequence, the increment density
limit is never zero for $\eta $ outside a discrete subset of $\mathbb{R}$,
unless $\mathbf{g}_{t,s}$ is \emph{identically} vanishing for \emph{all} $%
\eta \in \mathbb{R}$.

This mathematical property refers to a physical one. It reflects a generic
alternative between either \emph{strictly} positive or \emph{identically}
vanishing heat production density, at macroscopic scale, in presence of
non--vanishing external electric fields. Indeed, by taking $\Phi =\Psi $ in
Theorem \ref{Thm Heat production as power series copy(4)}, $\mathbf{T}%
_{t,s}^{(l,\eta )}$ is related to the heat produced by the presence of an
electromagnetic field, encoded in $\mathbf{W}_{t}^{(l,\eta )}$. If we use
cyclic processes, which means here that $\mathbf{W}_{t}^{(l,\eta )}=0$
outside some compact set $[t_{0},t_{1}]\subset \mathbb{R}$, then the KMS
state $\varrho \in \mathcal{U}^{\ast }$ applied on the energy increment $%
\mathbf{T}_{t_{1},t_{0}}^{(l,\eta )}$ is the total heat production (1st law
of Thermodynamics) with increment density limit equal to $\mathbf{g}%
_{t_{1},t_{0}}(\eta )$. It is non--negative, by the 2nd law of
Thermodynamics. See \cite{OhmV}\ for more details on the 1st and 2nd laws
for the quantum systems considered here. Now, if $\mathbf{g}%
_{t_{1},t_{0}}(\eta )$ is \emph{identically} vanishing for \emph{all} $\eta
\in \mathbb{R}$ then it means that the external perturbation never produces
heat in the system, which is a very strong property. The latter is expected
to be the case, for instance, for superconductors driven by electric
perturbations. This kind of behavior should highlight major features of the
system (like possibly broken symmetry). Hence, if the heat production
density is not \emph{identically} vanishing, generically, it is \emph{%
strictly} positive, at least at dimension $d=1$, because of properties of
real analytic functions mentioned above.

For higher dimensions $d>1$ and $s,t$ on compacta, Corollary \ref{Theorem
Lieb-Robinson copy(5)} implies that the increment density limit $\mathbf{g}%
_{t,s}\in C^{\infty }(\mathbb{R})$ belongs to the Gevrey class%
\begin{equation*}
C_{d}^{\omega }(\mathbb{R})\doteq \left\{ f\in C^{\infty }(\mathbb{R})\ :\
\sup_{\eta \in \mathbb{R}}\left\vert \partial _{\eta }^{m}f\left( \eta
\right) \right\vert \leq D^{m}\left( m!\right) ^{d}\text{ for any }m\in
\mathbb{N}\right\} \ .
\end{equation*}%
If $d>1$, the elements of $C_{d}^{\omega }(\mathbb{R})$ are usually neither
analytic nor quasi--analytic. In particular, functions of $C_{d}^{\omega }(%
\mathbb{R})$ can have arbitrarily small support, while $C_{d}^{\omega }(%
\mathbb{R})\varsubsetneq C_{d^{\prime }}^{\omega }(\mathbb{R})$ whenever $%
d<d^{\prime }$. Thus, the alternative above, which is related to the heat
production density in presence of external electric fields, does not follow
from Corollary \ref{Theorem Lieb-Robinson copy(5)} for higher dimensions $%
d\geq 2$. However, note that, at least for the quasi--free dynamics (also in
the presence of a random potential), the heat production density is a real
analytic function of $\eta $ at any dimension $d\in \mathbb{N}$, at least
for $\eta $ near zero. This follows from \cite[Theorem 3.4]{OhmI}.
Therefore, the above alternative for the heat production density may be true
at any dimension, provided the interaction decays fast enough in space (or
is finite--range, in the extreme case).

Observe finally that if a Gevrey function $f:\mathbb{R\rightarrow R}$ is
invertible on some open interval $I\subset \mathbb{R}$ then the inverse $%
f^{-1}:f(I)\mathbb{\rightarrow R}$ is again a Gevrey function. So, the above
theorem implies that, if the relation between applied field strength $\eta $
and the density of increment at $l\rightarrow \infty $ is injective for some
range of field strengths $\eta $, then the applied field strength in that
range is a Gevrey function of the density of increment. For more details on
Gevrey classes, see, e.g., \cite{HL}.

\section{Applications to Conductivity Measures\label{Section Ohm law}}

\subsection{Charged Transport Properties in Mathematics}

Altogether, the classical theory of linear conductivity (including the
theory of (Landau) Fermi liquids, see, e.g., \cite{brupedrahistoire} for a
historical perspective) is more like a makeshift theoretical construction
than a smooth and complete theory. It is unsatisfactory to use the Drude (or
the Drude--Lorentz) model -- which does not take into account quantum
mechanics -- together with certain ad hoc hypotheses as a proper microscopic
explanation of conductivity. For instance, in \cite{NS1,NS2,SE,Y}, the
(normally fixed) relaxation time of the Drude model has to be taken as an
effective frequency--dependent parameter to fit with experimental data \cite%
{T} on usual metals like gold. In fact, as claimed in the famous paper \cite[%
p. 505]{meanfreepath}, \textquotedblleft \textit{it must be admitted that
there is no entirely rigorous quantum theory of conductivity}%
.\textquotedblright

Concerning AC--conductivity, however, in the last years significant
mathematical progress has been made. See, e.g., \cite%
{Annale,JMP-autre,JMP-autre2,Cornean,OhmI,OhmII,OhmIII,OhmIV,OhmV,OhmVI,W,germinet}
for examples of mathematically rigorous derivations of linear conductivity
from first principles of quantum mechanics in the AC--regime. In particular,
the notion of conductivity measure has been introduced for the first time in
\cite{Annale}, albeit only for non--interacting systems. These results
indicate a physical picture of the microscopic origin of Ohm and Joule's
laws which differs from usual explanations coming from the Drude
(Lorentz--Sommerfeld) model.

As electrical resistance of conductors may result from the presence of
interactions between charge carriers, an important issue is to tackle the
interacting case. This is first\footnote{%
With regard to interacting systems, explicit constructions of KMS states are
obtained in the Ph.D. thesis \cite{W} for a one--dimensional model of
interacting fermions with a finite range pair interaction. But, the author
studies in \cite[Chap. 9]{W} the linear response theory only for
non--interacting fermions, keeping in mind possible generalizations to
interacting systems.} done in \cite{OhmV,OhmVI} for very general systems of
interacting quantum particles on lattices, including many important models
of condensed matter physics like the celebrated Hubbard model. This was out
of scope of \cite%
{Annale,JMP-autre,JMP-autre2,germinet,OhmI,OhmII,OhmIII,OhmIV,W} which
strongly rely on properties of quasi--free dynamics and states.

The central issue in \cite{OhmV,OhmVI} is to get estimates on transport
coefficients related to electric conduction, which are \emph{uniform} w.r.t.
the random parameters and the volume $|\Lambda _{l}|$ of the box $\Lambda
_{l}$ where the electromagnetic field lives. This is crucial to get valuable
information on conductivity in the macroscopic limit $l\rightarrow \infty $
and otherwise the results presented in \cite{OhmV,OhmVI} would loose almost
all their interest. To get such estimates in the non--interacting case \cite%
{OhmI,OhmII,OhmIII,OhmIV}, we applied tree--decay bounds on
multi--commutators in the sense of \cite[Section 4]{OhmI}. The latter are
based on combinatorial results \cite[Theorem 4.1]{OhmI} already used before,
for instance in \cite{FroehlichMerkliUeltschi}, and require the dynamics to
be implemented by Bogoliubov automorphisms. A solution to the issue for the
\emph{interacting} case is made possible by the results of Sections \ref%
{section Energy Increments as Power Series} and \ref{Section existence
dynamics}, which are direct consequences of the Lieb--Robinson bounds for
multi--com%
\-%
mutators. Detailed discussions on the estimates for the interacting case are
found in \cite{OhmV}. See also Corollary \ref{theorem exp tree decay}, which
is an extension of the tree--decay bounds\emph{\ }\cite[Section 4]{OhmI} to
the interacting case.

In \cite{OhmVI} the existence of macroscopic AC--con%
\-%
ductivity measures for interacting systems is derived from the 2nd law of
thermodynamics, explained in Section \ref{sect 2 loi}. The Lieb--Robinson
bound for multi--commutators of order 3 implies that it is always a L\'{e}vy
measure, see \cite[Theorems 7.1 and 5.2]{OhmVI}. We also derive below other
properties of the AC--con%
\-%
ductivity measures from Lieb--Robinson bounds for multi--commutators of\
higher orders. See Sections \ref{Sect Para}--\ref{Sect AC}. In particular,
we study their behavior at high frequencies (Theorems \ref{Thm para
regularite} and \ref{thm moment mu}): in contrast to the prediction of the
Drude (Lorentz--Sommerfeld) model, widely used in physics \cite%
{meanfreepath,Anderson-physics} to describe the phenomenon of electrical
conductivity, the conductivity measure stemming from short--range
interparticle interactions has to decay rapidly at high frequencies.

The proposed mathematical approach to the problem of deriving macroscopic
conductivity properties from the microscopic quantum dynamics of an infinite
system of particles also yield new physical insight, beyond classical
theories of conduction: a notion of current viscosity related to the
interplay of paramagnetic and diamagnetic currents, heat/entropy production
via different types of energy and current increments, existence of (AC--)
conductivity measures from the 2nd law and (possibly) as a spectral
(excitation) measure from current fluctuations are all examples of new
physical concepts derived in the course of the studies performed in \cite%
{OhmII,OhmIV,OhmV,OhmVI} and previously not discussed in the literature.

Note, however, that, by now, our results do not give explicit information on
the conductivity measure for concrete models (like the Hubbard model, for
instance). The latter belongs to \textquotedblleft hard
analysis\textquotedblright , by contrast with our results which are rather
on the side of the \textquotedblleft soft analysis\textquotedblright\
(similar to the difference between knowing the spectrum of a concrete
self--adjoint operator and knowing the spectral theorem). Moreover, our
approach does not directly provide a mathematical understanding from first
principles of Ohm's laws as a bulk property in the DC--regime, which is one
of the most important and difficult problems in mathematical physics for
more than one century. We believe, however, that our results can support
further rigorous developments towards a solution of such a difficult
problem: one could, for instance, try to show, for some class of models,
that the conductivity measure is absolutely continuous w.r.t. to the
Lebesgue measure and that its Radon--Nikodym derivative is continuous at low
frequencies, having a well-defined zero--frequency limit.

We thus present in the following some central results of \cite{OhmV,OhmVI},
with a few complementary studies, as an example of an important application
in mathematical physics of Lieb--Robinson bounds for multi--com%
\-%
mutators.

\subsection{Interacting Fermions in Disordered Media\label{Section Inter dis
media}}

\noindent \underline{(i) Kinetic part:} Let $\Delta _{\mathrm{d}}\in
\mathcal{B}(\ell ^{2}(\mathfrak{L}))$ be (up to a minus sign) the usual $d$%
--dimensional discrete Laplacian defined by%
\index{Kinetic term}%
\index{Discrete Laplacian}%
\begin{equation*}
\lbrack \Delta _{\mathrm{d}}(\psi )](x)\doteq 2d\psi (x)-\sum\limits_{z\in
\mathfrak{L},%
\text{ }|z|=1}\psi (x+z)\ ,\text{\qquad }x\in \mathfrak{L},\ \psi \in \ell
^{2}(\mathfrak{L})\ .
\end{equation*}%
To understand how such terms come about by starting from the usual Laplacian
in the continuum, see for instance \cite[Section II B]{Ne} which derives
effective models on lattices by using so--called Wannier functions in a band
subspace. This defines a short--range interaction%
\index{Interaction!example} $\Psi ^{\left( \mathrm{d}\right) }\in \mathcal{W}
$ by
\begin{equation*}
\Psi _{\Lambda }^{\left( \mathrm{d}\right) }\doteq \langle \mathfrak{e}%
_{x},\Delta _{\mathrm{d}}\mathfrak{e}_{y}\rangle _{\ell ^{2}\left( \mathfrak{%
L}\right) }a_{x}^{\ast }a_{y}+\left( 1-\delta _{x,y}\right) \langle
\mathfrak{e}_{y},\Delta _{\mathrm{d}}\mathfrak{e}_{x}\rangle _{\ell
^{2}\left( \mathfrak{L}\right) }a_{y}^{\ast }a_{x}\in \mathcal{U}^{+}\cap
\mathcal{U}_{\Lambda }
\end{equation*}%
whenever $\Lambda =\left\{ x,y\right\} $ for $x,y\in \mathfrak{L}$, and $%
\Psi _{\Lambda }^{\left( \mathrm{d}\right) }\doteq 0$ otherwise. Recall that
$\left\{ \mathfrak{e}_{x}\right\} _{x\in \Lambda }$ is the (canonical)
orthonormal basis of $\ell ^{2}(\Lambda )$ defined by (\ref{ONB}). \medskip

\noindent
\underline{(ii) Disordered media:}%
\index{Disordered media} Disorder in the crystal is modeled by a random
potential associated with a probability space $(\Omega ,\mathfrak{A}_{\Omega
},\mathfrak{a}_{\Omega })$ defined as follows: Let $\Omega \doteq \lbrack
-1,1]^{\mathfrak{L}}$. I.e., any element of $\Omega $ is a function on
lattice sites with values in $[-1,1]$.\ For $x\in \mathfrak{L}$, let $\Omega
_{x}$ be an arbitrary element of the Borel $\sigma $--algebra of the
interval $[-1,1]$ w.r.t. the usual metric topology. $\mathfrak{A}_{\Omega }$
is the $\sigma $--algebra generated by the cylinder sets $%
\prod\nolimits_{x\in \mathfrak{L}}\Omega _{x}$, where $\Omega _{x}=[-1,1]$
for all but finitely many $x\in \mathfrak{L}$. Then, $\mathfrak{a}_{\Omega }$
is an arbitrary \emph{ergodic}%
\index{Ergodicity} probability measure on the measurable space $(\Omega ,%
\mathfrak{A}_{\Omega })$. This means that the probability measure $\mathfrak{%
a}_{\Omega }$ is invariant under the action%
\begin{equation}
\omega \left( y\right) \mapsto \chi _{x}^{(\Omega )}\left( \omega \right)
\left( y\right) \doteq \omega \left( y+x\right) \ ,\qquad x,y\in \mathbb{Z}%
^{d}\ ,  \label{translation omega}
\end{equation}%
of the group $(\mathbb{Z}^{d},+)$ of lattice translations on $\Omega $ and,
for any $\mathcal{X}\in \mathfrak{A}_{\Omega }$ such that $\chi
_{x}^{(\Omega )}\left( \mathcal{X}\right) =\mathcal{X}$ for all $x\in
\mathbb{Z}^{d}$, one has $\mathfrak{a}_{\Omega }(\mathcal{X})\in \{0,1\}$.
We denote by $\mathbb{E}[\ \cdot \ ]$ the expectation value associated with $%
\mathfrak{a}_{\Omega }$.

Then, any realization $\omega \in \Omega $ and strength $\lambda \in \mathbb{%
R}_{0}^{+}$ of disorder is implemented by the potential%
\index{Potential!example} $\mathbf{V}^{(\omega )}$ defined by
\begin{equation}
\mathbf{V}_{\left\{ x\right\} }^{(\omega )}\doteq \lambda \omega \left(
x\right) a_{x}^{\ast }a_{x}\ ,\qquad x\in \mathfrak{L}\ .
\label{static potential random}
\end{equation}%
\medskip

\noindent
\underline{(iii) Interparticle interactions:}%
\index{Interparticle interactions} They are taken into account by choosing
some short--range interaction $\Psi ^{\mathrm{IP}}\in \mathcal{W}$ such that
$\Psi _{\Lambda }^{\mathrm{IP}}=0$ whenever $\Lambda =\left\{ x,y\right\} $
for $x,y\in \mathfrak{L}$, and
\begin{equation}
\sum\limits_{\Lambda \in \mathcal{P}_{f}(\mathfrak{L})}\left[ \Psi _{\Lambda
}^{\mathrm{IP}},a_{x}^{\ast }a_{x}\right] =0\ ,\qquad \Psi _{\Lambda +x}^{%
\mathrm{IP}}=\chi _{x}\left( \Psi _{\Lambda }^{\mathrm{IP}}\right) \ ,\qquad
\Lambda \in \mathcal{P}_{f}(\mathfrak{L}),\ x\in \mathfrak{L}\ .
\label{static potential0}
\end{equation}%
Here, the family $\{\chi _{x}\}_{x\in \mathfrak{L}}$ of $\ast $--automor%
\-%
phisms of $\mathcal{U}$ implements the action of the group $(\mathbb{Z}%
^{d},+)$ of lattice translations on the CAR $C^{\ast }$--algebra $\mathcal{U}
$, see (\ref{translation}). Observe that this class of interparticle
interactions includes all density--density interactions resulting from the
second quantization of two--body interactions defined via a real--valued and
summable function $v:[0,\infty )\rightarrow \mathbb{R}$ satisfying (\ref%
{sdsdsdsds}).\medskip

\noindent Then, by (i)--(iii), the full interaction
\begin{equation}
\Psi =\Psi ^{\left( \mathrm{d}\right) }+\Psi ^{\mathrm{IP}}\in \mathcal{W}
\label{full interaction}
\end{equation}%
and the potential $\mathbf{V}^{(\omega )}$ uniquely define an
infinite--volume dynamics corresponding to the $C_{0}$--group $\tau
^{(\omega )}\doteq \{\tau _{t}^{(\omega )}\}_{t\in {\mathbb{R}}}$ of $\ast $%
--auto%
\-%
morphisms with generator $\delta ^{(\omega )}$. See Theorem \ref{Theorem
Lieb-Robinson copy(3)}.\medskip

\noindent
\underline{(iv) Space--homogeneous electromagnetic fields:}
\index{Electromagnetic perturbation}Let $l\in \mathbb{R}^{+}$, $\eta \in
\mathbb{R}$, and the compactly supported function $\mathcal{A}\in
C_{0}^{\infty }(\mathbb{R};\mathbb{R}^{d})$ with $\mathcal{A}(t)\doteq 0$
for all $t\leq 0$. Set $E\left( t\right) \doteq -\partial _{t}\mathcal{A}(t)$
for all $t\in \mathbb{R}$. Then, the electric field at time $t\in \mathbb{R}$
equals $\eta E\left( t\right) $ inside the cubic box $\Lambda _{l}$ and $%
(0,0,\ldots ,0)$ outside. Up to negligible terms of order $\mathcal{O}%
(l^{d-1})$, this leads to a perturbation (of the generator of dynamics) of
the form (\ref{bounded pertubationbis0}), (\ref{bounded pertubationbis})
with complex--valued $\{\mathbf{w}_{x,y}\}_{x,y\in \mathfrak{L}}$ functions
of $(\eta ,t)\in \mathbb{R}^{2}$ defined by $\mathbf{w}_{x,x+z}(\eta ,t)=0$
for any $x,z\in \mathfrak{L}$ with $|z|>1$ while%
\begin{equation*}
\mathbf{w}_{x,x\pm e_{q}}(\eta ,t)\doteq \left( \exp \left( \mp i\eta
\int_{0}^{t}E_{q}\left( s\right) \ \mathrm{d}s\right) -1\right) \langle
\mathfrak{e}_{x},\Delta _{\mathrm{d}}\mathfrak{e}_{x\pm e_{q}}\rangle _{\ell
^{2}\left( \mathfrak{L}\right) }=%
\overline{\mathbf{w}_{x\pm e_{q},x}(\eta ,t)}
\end{equation*}%
for any $q\in \{1,\ldots ,d\}$. Here, $E(t)=(E_{1}(t),\ldots ,E_{d}(t))$ and
$\{e_{q}\}_{q=1}^{d}$ is the canonical orthonormal basis of the Euclidian
space $\mathbb{R}^{d}$. These functions clearly satisfy Conditions (\ref%
{condition1})--(\ref{condition2}) and (\ref{condition3}). Note that such
terms can be derived from the usual magnetic Laplacian (minimal coupling) in
the continuum, as explained in \cite[Section III, in particular Corollary 3.1%
]{Ne}. \medskip

\noindent Thus, the system of fermions in disordered medium, the interaction
of which is encoded by (\ref{full interaction}), is perturbed from $t=0$\
onwards by space--homogeneous electromagnetic fields, leading to a
well--defined family $\{\tilde{\tau}_{t,s}^{(\omega ,l,\eta )}\}_{s,t\in {%
\mathbb{R}}}$ of $\ast $--automor%
\-%
phisms, as explained in Theorem \ref{pertubed dynam thm}.

\subsection{Paramagnetic Conductivity}

\noindent \underline{(i) Paramagnetic currents:}%
\index{Paramagnetic current} For any pair $(x,y)\in \mathfrak{L}^{2}$, we
define the current observable by%
\begin{equation}
I_{(x,y)}\doteq i(a_{y}^{\ast }a_{x}-a_{x}^{\ast }a_{y})=I_{(x,y)}^{\ast
}\in \mathcal{U}_{0}\ .  \label{current observable}
\end{equation}%
It is seen as a current because it satisfies a discrete continuity equation.
See, e.g., \cite[Section 3.2]{OhmV}. For any $\mathcal{A}\in C_{0}^{\infty }(%
\mathbb{R};\mathbb{R}^{d})$, $l\in \mathbb{R}^{+}$, $\omega \in \Omega $, $%
\eta \in \mathbb{R}$ and $t\in \mathbb{R}_{0}^{+}$, these observables are
used to define a paramagnetic current increment density observable $\mathbb{J%
}_{\mathrm{p},l}^{(\omega )}\left( t,\eta \right) \in \mathcal{U}^{d}$:%
\begin{equation*}
\left\{ \mathbb{J}_{\mathrm{p},l}^{(\omega )}\left( t,\eta \right) \right\}
_{k}\doteq \left\vert \Lambda _{l}\right\vert ^{-1}\underset{x\in \Lambda
_{l}}{\sum }\left\{
\tilde{\tau}_{t,0}^{(\omega ,l,\eta )}\left( I_{\left( x+e_{k},x\right)
}\right) -\tau _{t}^{(\omega )}\left( I_{\left( x+e_{k},x\right) }\right)
\right\} \ .
\end{equation*}%
Compare with Equation (\ref{increment}).

Note that electric fields accelerate charged particles and induce so--called
diamagnetic currents, which correspond to the ballistic movement of
particles. In fact, as explained in \cite[Sections III and IV]{OhmII}, this
component of the total current creates a kind of \textquotedblleft wave
front\textquotedblright\ that destabilizes the whole system by changing its
state. The presence of diamagnetic currents leads then to the progressive
appearance of paramagnetic currents which are responsible for heat
production and the in--phase AC--conductivity of the system. Diamagnetic
currents are not relevant for the present purpose and are thus not defined
here. For more details, see \cite{OhmII,OhmV,OhmVI}. \medskip

\noindent \underline{(ii) Paramagnetic conductivity:}%
\index{Paramagnetic conductivity} We define the space--averaged paramagnetic
transport coefficient observable $\mathcal{C}_{\mathrm{p},l}^{(\omega )}\in
C^{1}(\mathbb{R};\mathcal{B}(\mathbb{R}^{d};\mathcal{U}^{d}))$, w.r.t. the
canonical orthonormal basis $\{e_{q}\}_{q=1}^{d}$ of the Euclidian space $%
\mathbb{R}^{d}$, by the corresponding matrix entries
\begin{equation}
\left\{ \mathcal{C}_{\mathrm{p},l}^{(\omega )}\left( t\right) \right\}
_{k,q}\doteq
\frac{1}{\left\vert \Lambda _{l}\right\vert }\underset{x,y\in \Lambda _{l}}{%
\sum }\int\nolimits_{0}^{t}i[\tau _{-s}^{(\omega )}(I_{\left(
y+e_{q},y\right) }),I_{\left( x+e_{k},x\right) }]\mathrm{d}s
\label{defininion para coeff observable}
\end{equation}%
for any $l\in \mathbb{R}^{+}$, $\omega \in \Omega $, $t\in \mathbb{R}$ and $%
k,q\in \{1,\ldots ,d\}$. \medskip

\noindent By (i)--(ii), if $\Psi ^{\mathrm{IP}}$ satisfies (\ref{(3.3) NS
generalized0}) with $\varsigma >2d$ (polynomial decay) then we infer from
Theorem \ref{Thm Heat production as power series copy(3)} that, for any $%
\mathcal{A}\in C_{0}^{\infty }(\mathbb{R};\mathbb{R}^{d})$,
\begin{equation}
\mathbb{J}_{\mathrm{p},l}^{(\omega )}\left( t,\eta \right) =\eta \mathbf{J}_{%
\mathrm{p},l}^{(\omega )}(t)+\mathcal{O}\left( \eta ^{2}\right) \ .
\label{linear response current1}
\end{equation}%
The correction terms of order $\mathcal{O}(\eta ^{2})$ are uniformly bounded
in $l\in \mathbb{R}^{+}$, $\omega \in \Omega $ and $\lambda ,t\in \mathbb{R}%
_{0}^{+}$. By explicit computations, one checks that
\begin{equation}
\mathbf{J}_{\mathrm{p},l}^{(\omega )}(t)=\int_{0}^{t}\tau _{t}^{(\omega
)}\left( \mathcal{C}_{\mathrm{p},l}^{(\omega )}\left( t-s\right) \right)
E\left( s\right) \mathrm{d}s  \label{linear response current2}
\end{equation}%
for any $\mathcal{A}\in C_{0}^{\infty }(\mathbb{R};\mathbb{R}^{d})$, $l\in
\mathbb{R}^{+}$, $\omega \in \Omega $ and $t\in \mathbb{R}_{0}^{+}$. The
latter is the paramagnetic \emph{linear response current}.%
\index{Paramagnetic current!linear response} For more details, see also \cite%
[Theorem 3.7]{OhmV}. Here, for any $\mathbf{D}\in \mathcal{B}(\mathbb{R}^{d};%
\mathcal{U}^{d})$, $\tau _{t}^{(\omega )}\left( \mathbf{D}\right) \in
\mathcal{B}(\mathbb{R}^{d};\mathcal{U}^{d})$ is, by definition, the linear
operator on $\mathbb{R}^{d}$ defined, w.r.t. the canonical orthonormal basis
$\{e_{q}\}_{q=1}^{d}$ of the Euclidian space $\mathbb{R}^{d}$, by the matrix
entries%
\begin{equation*}
\left\{ \tau _{t}^{(\omega )}\left( \mathbf{D}\right) \right\} _{k,q}\doteq
\tau _{t}^{(\omega )}\left( \left\{ \mathbf{D}\right\} _{k,q}\right) \
,\qquad k,q\in \{1,\ldots ,d\}\ .
\end{equation*}

\subsection{2nd law of Thermodynamics and Equilibrium States\label{sect 2
loi}}

\noindent
\underline{(i) States:}%
\index{States} $\rho \in \mathcal{U}^{\ast }$ is a state if $\rho \geq 0$,
that is, $\rho (B^{\ast }B)\geq 0$ for all $B\in \mathcal{U}$, and $\rho (%
\mathbf{1})=1$. States encode the statistical distribution of all physical
quantities associated with observables $B=B^{\ast }\in \mathcal{U}$. See
Section \ref{sect Algebraic Formulation of Quantum Mechanics}.

For any $\mathbf{D}\in \mathcal{B}(\mathbb{R}^{d};\mathcal{U}^{d})$, $\rho
\left( \mathbf{D}\right) \in \mathcal{B}(\mathbb{R}^{d})$ is, by definition,
the linear operator defined, w.r.t. the canonical orthonormal basis $%
\{e_{q}\}_{q=1}^{d}$ of $\mathbb{R}^{d}$, by
\begin{equation*}
\left\{ \rho \left( \mathbf{D}\right) \right\} _{k,q}\doteq \rho \left(
\left\{ \mathbf{D}\right\} _{k,q}\right) \ ,\qquad k,q\in \{1,\ldots ,d\}\ .
\end{equation*}%
\medskip \noindent
\underline{(ii) 2nd law of thermodynamics:}%
\index{2nd law} As explained in \cite{lieb-yngvasonPhysReport,lieb-yngvason}%
, different formulations of the same principle have been stated by Clausius,
Kelvin (and Planck), and Cara%
\-%
th\'{e}odory. Our study is based on the Kelvin--Planck statement while
avoiding the concept of \textquotedblleft cooling\textquotedblright\ \cite[%
p. 49]{lieb-yngvasonPhysReport}. It can be expressed as follows \cite[p. 276]%
{PW}: \medskip

\noindent \textit{Systems in the equilibrium are unable to perform
mechanical work in cyclic processes.}\medskip

\noindent
\underline{(iii) Passive states:}%
\index{Passivity}%
\index{States!passive} To define equilibrium states, the 2nd law, as
expressed in \cite{PW}, is pivotal because it leads to a clear mathematical
formulation of the Kelvin--Planck notion of equilibrium: For any strongly
continuous one--parameter group $\tau \equiv \{\tau _{t}\}_{t\in {\mathbb{R}}%
}$ of $\ast $--automorphisms of $\mathcal{U}$, one obtains a well--defined
strongly continuous two--parameter family $\{\tau _{t,t_{0}}^{(\mathbf{W}%
)}\}_{t\geq t_{0}}$ of $\ast $--automorphisms of $\mathcal{U}$\ by
perturbing the generator of dynamics with bounded time--dependent symmetric
derivations%
\begin{equation*}
B\mapsto i\left[ \mathbf{W}_{t},B\right] \ ,\qquad B\in \mathcal{U}\ ,\ t\in
\mathbb{R}\ ,
\end{equation*}%
for any arbitrary
\index{Cyclic process}\emph{cyclic process} $\{\mathbf{W}_{t}\}_{t\geq
t_{0}} $ of time length $T\geq 0$, that is, a differentiable family $\{%
\mathbf{W}_{t}\}_{t\geq t_{0}}\subset \mathcal{U}$ of self--adjoint elements
of $\mathcal{X}$ such that $\mathbf{W}_{t}=0$ for all real times $t\notin %
\left[ t_{0},T+t_{0}\right] $. Then, a state $\varrho \in \mathcal{U}^{\ast
} $ is \emph{passive} (cf. \cite[Definition 1.1]{PW}) iff the work%
\begin{equation*}
\int_{t_{0}}^{t}\varrho \circ \tau _{t,t_{0}}^{(\mathbf{W})}\left( \partial
_{t}\mathbf{W}_{t}\right) \mathrm{d}t
\end{equation*}%
performed on the system is non--negative for all cyclic processes $\{\mathbf{%
W}_{t}\}_{t\geq t_{0}}$ of any time length $T\geq 0$. By \cite[Theorem 1.1]%
{PW}, such states are invariant w.r.t. the unperturbed dynamics: $\varrho
=\varrho \circ \tau _{t}$\ for any $t\in {\mathbb{R}}$.

If $\tau =\tau ^{(\omega )}$ with $\omega \in \Omega $ then, as explained in
\cite[Section 2.6]{OhmV}, at least one passive state $\varrho ^{(\omega )}$
exists. It represents an equilibrium state of the system (in a broad sense),
the mathematical definition of which encodes the 2nd law.\medskip

\noindent
\underline{(iv) Random invariant passive states:}%
\index{States!random invariant passive} We impose two natural conditions on
the map $\omega \mapsto \varrho ^{(\omega )}$ from the set $\Omega $ to the
dual space $\mathcal{U}^{\ast }$:

\begin{itemize}
\item Translation invariance. Using definitions (\ref{translation}) and (\ref%
{translation omega}), we assume that%
\begin{equation}
\varrho ^{(\chi _{x}^{(\Omega )}(\omega ))}=\varrho ^{(\omega )}\circ \chi
_{x}\ ,\qquad x\in \mathfrak{L}=\mathbb{Z}^{d}\ .
\label{translation invariance random state}
\end{equation}

\item Measurability. The map $\omega \mapsto \varrho ^{(\omega )}$ is
measurable w.r.t. to the $\sigma $--algebra $\mathfrak{A}_{\Omega }$ on $%
\Omega $ and the Borel $\sigma $--algebra $\mathfrak{A}_{\mathcal{U}^{\ast
}} $ of $\mathcal{U}^{\ast }$ generated by the weak$^{\ast }$--topology.
Note that a similar assumption is also used to define equilibrium for
classical systems in disordered media, see, e.g., \cite{bovier}.
\end{itemize}

\noindent A map satisfying such properties is named here \emph{a random
invariant state} \cite[Definition 3.1]{OhmVI}. Such maps always exist in the
one--dimension case if the norm $\left\Vert \Psi ^{\mathrm{IP}}\right\Vert _{%
\mathcal{W}}$ of the interparticle interaction is finite. The same is true
in any dimension if the inverse temperature $\beta \in \mathbb{R}^{+}$ is
small enough. This is a consequece of the uniqueness of KMS, which is
implied by the mentioned conditions. By using methods of constructive
quantum field theory, one can also verify the existence of such random
invariant passive states $\varrho ^{(\omega )}$, $\omega \in \Omega $, at
arbitrary dimension and any fixed $\beta \in \mathbb{R}^{+}$, if the
interparticle interaction $\left\Vert \Psi ^{\mathrm{IP}}\right\Vert _{%
\mathcal{W}}$ is small enough and (\ref{static potential0}) holds. See, for
instance, \cite[Theorem 2.1]{FU} (together with \cite[Theorem 1.4]{PW}) for
the small $\beta $ case in quantum spin systems. See also \cite[Section 3.3]%
{OhmVI} for further discussions on this topic.

\subsection{Macroscopic Paramagnetic Conductivity\label{Sect Para}}

For any short--range interaction $\Psi ^{\mathrm{IP}}\in \mathcal{W}$, the
limit%
\index{Paramagnetic conductivity!macroscopic}
\begin{equation}
\mathbf{\Xi }_{\mathrm{p}}\left( t\right) \doteq \underset{l\rightarrow
\infty }{\lim }\mathbb{E}\left[ \varrho ^{(\omega )}(\mathcal{C}_{\mathrm{p}%
,l}^{(\omega )}\left( t\right) )\right] \in \mathcal{B}(\mathbb{R}^{d})
\label{paramagnetic conductivity}
\end{equation}%
exists and is uniform for $t$ on compacta. To see this, use the usual
Lieb--Robinson bounds (Theorem \ref{Theorem Lieb-Robinson copy(3)} (iv)) to
estimate (\ref{defininion para coeff observable}) in the limit $l\rightarrow
\infty $. Here, for any measurable $\mathbf{D}^{(\omega )}\in \mathcal{B}(%
\mathbb{R}^{d})$, the expectation value $\mathbb{E}[\mathbf{D}^{(\omega
)}]\in \mathcal{B}(\mathbb{R}^{d})$ (associated with $\mathfrak{a}_{\Omega }$%
) is defined, w.r.t. the canonical orthonormal basis $\{e_{q}\}_{q=1}^{d}$
of $\mathbb{R}^{d}$, by the matrix entries
\begin{equation*}
\left\{ \mathbb{E}%
\big[%
\mathbf{D}^{(\omega )}%
\big]%
\right\} _{k,q}\doteq \mathbb{E}%
\big[%
\left\{ \mathbf{D}\right\} _{k,q}%
\big]%
\ ,\qquad k,q\in \{1,\ldots ,d\}\ .
\end{equation*}%
The function $\mathbf{\Xi }_{\mathrm{p}}\in C^{1}(\mathbb{R};\mathcal{B}(%
\mathbb{R}^{d}))$ can be directly related to a linear response current, as
suggested by (\ref{linear response current1})--(\ref{linear response
current2}). See \cite[Theorem 4.2 (p)]{OhmVI} for more details. [If one does
not take expectation values of currents, one can also show that the limit $%
l\rightarrow \infty $ of $\varrho ^{(\omega )}(\mathbf{J}_{\mathrm{p}%
,l}^{(\omega )})$ almost everywhere exists and equals the expectation value,
in the same limit, by using the Akcoglu--Krengel ergodic theorem, see \cite%
{OhmIII,OhmVI}.]

\cite[Theorem 7.1]{OhmVI} asserts that%
\begin{equation*}
\mathbf{\Xi }_{\mathrm{p}}\in C^{2}(\mathbb{R};\mathcal{B}(\mathbb{R}^{d}))%
\text{ }
\end{equation*}%
if $\Psi ^{\mathrm{IP}}\in \mathcal{W}$ and (\ref{(3.3) NS generalized0})
holds with $\varsigma >2d$. Now, we give a stronger version of this result
which is an application of Lieb--Robinson bounds for multi--com%
\-%
mutators (Theorems \ref{Theorem Lieb-Robinson copy(1)}--\ref{theorem exp
tree decay copy(1)}) of high orders. This new result on the regularity of
the function $\mathbf{\Xi }_{\mathrm{p}}$ of time has important consequences
on the asymptotics of AC--Conductivity measures at high frequencies, see
Theorem \ref{thm moment mu}.

\begin{satz}[Regularity of the paramagnetic conductivity]
\label{Thm para regularite}\mbox{
}\newline
Let $\lambda \in \mathbb{R}_{0}^{+}$ and assume that the map $\omega \mapsto
\varrho ^{(\omega )}$ is a random invariant passive state and $\Psi ^{%
\mathrm{IP}}\in \mathcal{W}$ satisfies (\ref{static potential0}).%
\index{Paramagnetic conductivity!macroscopic} \newline
\emph{(i)} Polynomial decay: Assume $\Psi ^{\mathrm{IP}}$ satisfies (\ref%
{(3.3) NS generalized0}). Then, for any $m\in \mathbb{N}$ satisfying $%
d(m+1)<\varsigma $, $\mathbf{\Xi }_{\mathrm{p}}\in C^{m+1}(\mathbb{R};%
\mathcal{B}(\mathbb{R}^{d}))$ and, uniformly for $t$ on compacta,
\begin{equation}
\partial _{t}^{m+1}\mathbf{\Xi }_{\mathrm{p}}\left( t\right) =\underset{%
l\rightarrow \infty }{\lim }\partial _{t}^{m+1}\mathbb{E}\left[ \varrho
^{(\omega )}(\mathcal{C}_{\mathrm{p},l}^{(\omega )}\left( t\right) )\right]
\ .  \label{regularity para1}
\end{equation}%
\emph{(ii)} Exponential decay: Assume $\Psi ^{\mathrm{IP}}$ satisfies (\ref%
{(3.3) NS generalized}). Then, for all $m\in \mathbb{N}$, $\mathbf{\Xi }_{%
\mathrm{p}}\in C^{\infty }(\mathbb{R};\mathcal{B}(\mathbb{R}^{d}))$ and (\ref%
{regularity para1}) holds true with the limit being uniform for $t$ on
compacta.
\end{satz}

\begin{bemerkung}[Fermion systems with random Laplacians]
\label{Fermion systems with random Laplacians}\mbox{
}\newline
\index{Discrete Laplacian!random}The same assertion holds for the random
models treated in \cite{OhmVI}, i.e., for fermions on the lattice with
short--range and translation invariant (cf. (\ref{static potential0}))
interaction $\Psi ^{\mathrm{IP}}\in \mathcal{W}$, random potentials (cf. (%
\ref{static potential random})) and, additionally, random next neighbor
hopping amplitudes. [So, $\Delta _{\mathrm{d}}$ is replaced in \cite{OhmVI}
with a random Laplacian $\Delta _{\omega ,\vartheta }$.] Similar to what is
done here, disorder is defined in \cite{OhmVI} via ergodic distributions of
random potentials and hopping amplitudes.
\end{bemerkung}

\noindent The proof of this statement is a consequence of the following
general lemma:

\begin{lemma}
\label{Lemma utile para}\mbox{ }\newline
Let $\Psi \in \mathcal{W}$ and $\mathbf{V}$ be any potential such that
\begin{equation}
\sup_{x\in \mathfrak{L}}\left\Vert \mathbf{V}_{\left\{ x\right\}
}\right\Vert _{\mathcal{U}}<\infty \ .  \label{condition bounded}
\end{equation}%
Take $T\in \mathbb{R}_{0}^{+}$ and $B_{0},B_{1}\in \mathcal{U}_{0}$. \newline
\emph{(i)} Polynomial decay: Assume (\ref{(3.3) NS generalized0}). Then, for
any $m\in \mathbb{N}$ satisfying $dm<\varsigma $, $\mathcal{U}_{0}\subseteq
\mathrm{Dom}(\delta ^{m})$. Moreover, if $d(m+1)<\varsigma $,%
\begin{equation}
\sum_{y\in \mathfrak{L}}\sup_{t\in \left[ -T,T\right] }\sup_{x\in \mathfrak{L%
}}\left\Vert \left[ \tau _{t}\circ \chi _{x}(B_{1}),\delta ^{m}\circ \chi
_{y}\left( B_{0}\right) )\right] \right\Vert _{\mathcal{U}}<\infty \ .
\label{ineq sup macro cond0}
\end{equation}%
\emph{(ii)} Exponential decay: Assume (\ref{(3.3) NS generalized}). Then,
\begin{equation*}
\mathcal{U}_{0}\subseteq \underset{m\in \mathbb{N}}{\bigcap }\mathrm{Dom}%
\left( \delta ^{m}\right) \subset \mathcal{U}
\end{equation*}%
and (\ref{ineq sup macro cond0}) holds true for all $m\in \mathbb{N}$.
\end{lemma}

\begin{proof}
(i) Because of (\ref{condition bounded}), assume w.l.o.g. that $\mathbf{V}=0$%
. Take $t\in \mathbb{R}$, $n_{0},n_{1}\in \mathbb{N}$ and local elements $%
B_{0}\in \mathcal{U}_{\Lambda _{n_{0}}}$ and $B_{1}\in \mathcal{U}_{\Lambda
_{n_{1}}}$. Then, we infer from Theorem \ref{Theorem Lieb-Robinson copy(3)}
(ii) and (\ref{definition D1})--(\ref{set eq}) that, for any $x,y\in
\mathfrak{L}$ and $n\in \mathbb{N}$,
\begin{eqnarray}
&&\left\Vert \left[ \tau _{t}\circ \chi _{x}(B_{1}),\delta ^{n}\circ \chi
_{y}\left( B_{0}\right) )\right] \right\Vert _{\mathcal{U}}  \notag \\
&\leq &\sum_{x_{n}\in \mathfrak{L}}\sum_{m_{n}\in \mathbb{N}%
_{0}}\sum\limits_{\mathcal{Z}_{n}\in \mathcal{D}(x_{n},m_{n})}\cdots
\sum_{x_{1}\in \mathfrak{L}}\sum_{m_{1}\in \mathbb{N}_{0}}\sum\limits_{%
\mathcal{Z}_{1}\in \mathcal{D}(x_{1},m_{1})}  \label{ineq sup macro cond1} \\
&&\qquad \qquad \left\Vert \left[ \tau _{t}\circ \chi _{x}(B_{1}),\Psi _{%
\mathcal{Z}_{n}},\ldots ,\Psi _{\mathcal{Z}_{1}},\chi _{y}\left(
B_{0}\right) \right] ^{(n+2)}\right\Vert _{\mathcal{U}}\ .  \notag
\end{eqnarray}%
Therefore, we can directly use Lieb--Robinson bounds for multi--com%
\-%
mutators of order $n+2$ to bound (\ref{ineq sup macro cond1}): We combine
Theorems \ref{Theorem Lieb-Robinson copy(1)} and \ref{theorem exp tree decay
copy(1)} (i) with Equation (\ref{ki trivial estimate1}) to deduce from (\ref%
{ineq sup macro cond1}) that, for any $x,y\in \mathfrak{L}$ and $n\in
\mathbb{N}$,%
\begin{eqnarray}
&&\left\Vert \left[ \tau _{t}\circ \chi _{x}(B_{1}),\delta ^{n}\circ \chi
_{y}\left( B_{0}\right) )\right] \right\Vert _{\mathcal{U}}  \notag \\
&\leq &2^{n+1}d^{%
\frac{\varsigma \left( n+1\right) }{2}}(1+n_{0})^{\varsigma }\left\Vert
B_{1}\right\Vert _{\mathcal{U}}\left\Vert B_{0}\right\Vert _{\mathcal{U}}
\label{ineq sup macro cond2} \\
&&\times \left( 2\Vert \Psi \Vert _{\mathcal{W}}\left\vert t\right\vert
\mathrm{e}^{4\mathbf{D}\left\vert t\right\vert \Vert \Psi \Vert _{\mathcal{W}%
}}\left\Vert \mathbf{u}_{\cdot ,n_{1}}\right\Vert _{\ell ^{1}(\mathbb{N}%
)}+(1+n_{1})^{\varsigma }\right)  \notag \\
&&\times \left( \underset{x\in \mathfrak{L}}{\sup }\left( \sum_{m\in \mathbb{%
N}_{0}}\left( 1+m\right) ^{\varsigma }\sum\limits_{\mathcal{Z}\in \mathcal{D}%
(x,m)}\left\Vert \Psi _{\mathcal{Z}}\right\Vert _{\mathcal{U}}\right)
\right) ^{n}  \notag \\
&&\times \sum_{x_{n}\in \mathfrak{L}}\cdots \sum_{x_{1}\in \mathfrak{L}%
}\left( \sum_{T\in \mathcal{T}_{n+2}}\prod\limits_{\{j,l\}\in T}\frac{1}{%
(1+\left\vert x_{j}-x_{l}\right\vert )^{\varsigma \left( \max \{\mathfrak{d}%
_{T}(j),\mathfrak{d}_{T}(l)\}\right) ^{-1}}}\right)  \notag
\end{eqnarray}%
with $x_{0}\doteq y\in \mathfrak{L}$ and $x_{n+1}\doteq x\in \mathfrak{L}$.
If $\Psi \in \mathcal{W}$ and Condition (\ref{(3.3) NS generalized0}) holds
true, then one easily verifies (\ref{assumption boundedness2}) with $\mathbf{%
v}_{m}=\left( 1+m\right) ^{\varsigma }$. Recall also that the condition $%
\varsigma >\left( n+1\right) d$ yields (\ref{assumption boundedness3bis})
with $k=n+1$. Using these observations, one directly arrives at (\ref{ineq
sup macro cond0}), starting from (\ref{ineq sup macro cond2}).

Remark that $\mathcal{U}_{0}\subseteq \mathrm{Dom}(\delta ^{n})$ is proven
exactly in the same way. In fact, it is easier to prove and only requires
the condition $\varsigma >nd$ because we have in this case
multi--commutators of only order $n+1$.

\noindent (ii) The proof is very similar to the polynomial case. We omit the
details. See Theorem \ref{theorem exp tree decay copy(1)} (ii) and (\ref{ki
trivial estimate2}), and in the case (\ref{(3.3) NS generalized}) holds and $%
\Psi \in \mathcal{W}$, note again that Condition (\ref{assumption
boundedness2}) is satisfied with $\mathbf{v}_{m}=\mathrm{e}^{m\varsigma }$.
\end{proof}

\noindent We are now in position to prove\ Theorem \ref{Thm para regularite}%
. \medskip

\begin{proof}
Fix $k,q\in \{1,\ldots ,d\}$, $t\in \mathbb{R}$ and $m\in \mathbb{N}$. By
Theorem \ref{Theorem Lieb-Robinson copy(3)} (i), $\tau ^{(\omega )}\doteq
\{\tau _{t}^{(\omega )}\}_{t\in {\mathbb{R}}}$ is a $C_{0}$--group of $\ast $%
--auto%
\-%
morphisms with generator $\delta ^{(\omega )}$. It is, indeed, associated
with the interaction (\ref{full interaction}) and the potential defined by (%
\ref{static potential random}). If $\Psi ^{\mathrm{IP}}$ satisfies (\ref%
{(3.3) NS generalized0}), then Condition (\ref{(3.3) NS generalized0}) also
holds true for the full interaction (\ref{full interaction}). A similar
observation can be made when $\Psi ^{\mathrm{IP}}$ satisfies (\ref{(3.3) NS
generalized}).

Paramagnetic current observables (\ref{current observable}) are local
elements, i.e.,$\ I_{(x,y)}\in \mathcal{U}_{0}$ for any $(x,y)\in \mathfrak{L%
}^{2}$. Then, by Lemma \ref{Lemma utile para}, we thus compute from (\ref%
{defininion para coeff observable}) that, for any $m\in \mathbb{N}$ such
that $\mathcal{U}_{0}\subseteq \mathrm{Dom}(\delta ^{m})$,%
\begin{eqnarray}
&&\partial _{t}^{m+1}\left\{ \mathbb{E}\left[ \varrho ^{(\omega )}(\mathcal{C%
}_{\mathrm{p},l}^{(\omega )}\left( t\right) )\right] \right\} _{k,q}
\label{derivee m fois} \\
&=&-\frac{1}{\left\vert \Lambda _{l}\right\vert }\underset{x,y\in \Lambda
_{l}}{\sum }\mathbb{E}\left[ \varrho ^{(\omega )}\left( i[\tau
_{-t}^{(\omega )}\circ (\delta ^{(\omega )})^{m}(I_{\left( y+e_{q},y\right)
}),I_{\left( x+e_{k},x\right) }]\right) \right] \ .  \notag
\end{eqnarray}%
The last function of $\omega \in \Omega $ in the expectation value $\mathbb{E%
}[\ \cdot \ ]$ (associated with $\mathfrak{a}_{\Omega }$) is measurable,
because $\omega \mapsto \varrho ^{(\omega )}$ is, by definition, a random
invariant state while one can check that the map
\begin{equation*}
\omega \mapsto i[\tau _{-t}^{(\omega )}\circ (\delta ^{(\omega
)})^{m}(I_{\left( y+e_{q},y\right) }),I_{\left( x+e_{k},x\right) }]
\end{equation*}%
from $\Omega $ to $\mathcal{U}$ is continuous, using Theorem \ref{Theorem
Lieb-Robinson copy(3)} and the second Trotter--Kato approximation theorem
\cite[Chap. III, Sect. 4.9]{EngelNagel}. Additionally, if $\varrho ^{(\omega
)}$ is a passive state w.r.t. to $\tau ^{(\omega )}$ for any $\omega \in
\Omega $ then $\varrho ^{(\omega )}=\varrho ^{(\omega )}\circ \tau
_{t}^{(\omega )}$, see \cite[Theorem 1.1]{PW}. Therefore, it follows from (%
\ref{derivee m fois}) that
\begin{eqnarray}
&&\partial _{t}^{m+1}\left\{ \bar{\varrho}\left( \mathcal{C}_{\mathrm{p}%
,l}^{(\omega )}\left( t\right) \right) \right\} _{k,q}
\label{derivee m fois1} \\
&=&\frac{1}{\left\vert \Lambda _{l}\right\vert }\underset{x,y\in \Lambda _{l}%
}{\sum }\mathbb{E}\left[ \varrho ^{(\omega )}\left( i[\tau _{t}^{(\omega
)}\left( I_{\left( x+e_{k},x\right) }\right) ,(\delta ^{(\omega
)})^{m}(I_{\left( y+e_{q},y\right) })]\right) \right] \ .  \notag
\end{eqnarray}%
Now, if (\ref{static potential0}) and (\ref{translation invariance random
state}) hold true, then, by using the fact that $\mathfrak{a}_{\Omega }$ is
also a translation invariant probability measure (it is even ergodic), we
obtain from (\ref{derivee m fois1}) that, for any $m\in \mathbb{N}$ such
that $\mathcal{U}_{0}\subseteq \mathrm{Dom}(\delta ^{m})$,%
\begin{eqnarray}
&&\partial _{t}^{m+1}\left\{ \bar{\varrho}\left( \mathcal{C}_{\mathrm{p}%
,l}^{(\omega )}\left( t\right) \right) \right\} _{k,q}
\label{derivee seconde1derivee seconde1} \\
&=&\underset{y\in \mathfrak{L}}{\sum }\xi _{l}\left( y\right) \mathbb{E}%
\left[ \varrho ^{(\omega )}\left( i[\tau _{t}^{(\omega )}\left( I_{\left(
e_{k},0\right) }\right) ,(\delta ^{(\omega )})^{m}\circ \chi _{y}(I_{\left(
e_{q},0\right) })]\right) \right]  \notag
\end{eqnarray}%
with%
\begin{equation*}
\xi _{l}\left( y\right) \doteq \frac{1}{\left\vert \Lambda _{l}\right\vert }%
\underset{x\in \Lambda _{l}}{\sum }\mathbf{1}_{\left\{ y\in \Lambda
_{l}-x\right\} }\in \left[ 0,1\right] \ ,\qquad y\in \mathfrak{L}\ ,\ l\in
\mathbb{R}^{+}\ .
\end{equation*}%
For any $l\in \mathbb{R}^{+}$, the map $y\mapsto \xi _{l}\left( y\right) $
on $\mathfrak{L}$ has finite support and, for any $y\in \mathfrak{L}$,
\begin{equation}
\underset{l\rightarrow \infty }{\lim }\xi _{l}\left( y\right) =1\ .
\label{limit xi}
\end{equation}%
As a consequence, if (i) $\Psi ^{\mathrm{IP}}$ satisfies (\ref{(3.3) NS
generalized0}) and $d(m+1)<\varsigma $ or (ii) $\Psi ^{\mathrm{IP}}$
satisfies (\ref{(3.3) NS generalized}), then, by combining Lemma \ref{Lemma
utile para} with Lebesgue's dominated convergence theorem, one gets from (%
\ref{paramagnetic conductivity}) and (\ref{derivee seconde1derivee seconde1}%
)--(\ref{limit xi}) that the map
\begin{equation*}
t\mapsto \partial _{t}^{m+1}\left\{ \mathbb{E}\left[ \varrho ^{(\omega )}(%
\mathcal{C}_{\mathrm{p},l}^{(\omega )}\left( t\right) )\right] \right\} =%
\mathbb{E}\left[ \partial _{t}^{m+1}\varrho ^{(\omega )}(\mathcal{C}_{%
\mathrm{p},l}^{(\omega )}\left( t\right) )\right]
\end{equation*}%
converges uniformly on compacta, as $l\rightarrow \infty $, to the
continuous function $\partial _{t}^{m+1}\mathbf{\Xi }_{\mathrm{p}}\in C(%
\mathbb{R};\mathcal{B}(\mathbb{R}^{d}))$.
\end{proof}

\subsection{AC--Conductivity Measure\label{Sect AC}}

By applying \cite[Theorems 5.2 and 5.6 (p), Remark 5.3]{OhmVI} to the
interacting fermion system under consideration we get a \emph{L\'{e}%
vy--Khintchine representation} of the paramagnetic (in--phase) conductivity $%
\mathbf{\Xi }_{\mathrm{p}}$: Assume $\Psi ^{\mathrm{IP}}$ satisfies (\ref%
{(3.3) NS generalized0}) with $\varsigma >2d$ (polynomial decay). Then,
there is a unique finite and symmetric $\mathcal{B}_{+}(\mathbb{R}^{d})$%
--valued measure $\mathbf{\mu }$ on $\mathbb{R}$ such that, for any $t\in
\mathbb{R}$,%
\index{Paramagnetic conductivity!L\'{e}vy--Khintchine representation}
\begin{equation}
\mathbf{\Xi }_{\mathrm{p}}\left( t\right) =-%
\frac{t^{2}}{2}\mathbf{\mu }\left( \left\{ 0\right\} \right) +\int\nolimits_{%
\mathbb{R}\backslash \left\{ 0\right\} }\left( \cos \left( t\nu \right)
-1\right) \nu ^{-2}\mathbf{\mu }\left( \mathrm{d}\nu \right) \ .
\label{Levy--Khintchine}
\end{equation}%
Here, $\mathcal{B}_{+}(\mathbb{R}^{d})\subset \mathcal{B}(\mathbb{R}^{d})$
stands for the set of positive linear operators on $\mathbb{R}^{d}$, i.e.,
symmetric operators w.r.t. to the canonical scalar product of $\mathbb{R}%
^{d} $ with positive eigenvalues. The (in--phase)\ \emph{AC--cond%
\-%
uctivity measure} is defined from the measure $\mathbf{\mu }$ as follows:

\begin{definition}[AC--conductivity measure]
\label{def second law2 copy(2)}\mbox{ }\newline
We name the L\'{e}vy measure $\mu _{\mathrm{AC}}$, the restriction of $\nu
^{-2}\mathbf{\mu }\left( \mathrm{d}\nu \right) $ to $\mathbb{R}\backslash
\{0\}$, the (in--phase)\emph{\ }AC--cond%
\-%
uctivity measure.%
\index{AC--conductivity measure}
\end{definition}

Indeed, by \cite[Theorems 5.1 and 5.6 (p)]{OhmVI}, one checks that $\mu _{%
\mathrm{AC}}$ quantifies the energy (or heat) production $Q$ per unit volume
due to the component of frequency $\nu \in \mathbb{R}\backslash \{0\}$ of
the electric field, in accordance with Joule's law in the AC--regime:
Indeed, for any smooth electric field $E\left( t\right) =\mathcal{E}\left(
t\right)
\vec{w}$ with $\vec{w}\in \mathbb{R}^{d}$, $\mathcal{E}\doteq -\partial _{t}%
\mathcal{A}(t)$ and $\mathcal{A}\in C_{0}^{\infty }(\mathbb{R};\mathbb{R})$,
the total heat per unit volume produced by the electric field (after being
switch off) is equal to%
\begin{equation*}
Q=\frac{1}{2}\int\nolimits_{\mathbb{R}}\mathrm{d}s_{1}\int\nolimits_{\mathbb{%
R}}\mathrm{d}s_{2}\mathcal{E}_{s_{2}}\mathcal{E}_{s_{1}}\langle \vec{w},%
\mathbf{\Xi }_{\mathrm{p}}\left( s_{1}-s_{2}\right) \vec{w}\rangle _{\mathbb{%
R}^{d}}\ .
\end{equation*}%
If the Fourier transform $\mathcal{\hat{E}}$ of $\mathcal{E}\in
C_{0}^{\infty }(\mathbb{R};\mathbb{R})$ has support away from $\nu =0$, then
\begin{equation*}
Q=\frac{1}{2}\int_{\mathbb{R}\backslash \{0\}}|\mathcal{\hat{E}}\left( \nu
\right) |^{2}\ \langle \vec{w},\mu _{\mathrm{AC}}\left( \mathrm{d}\nu
\right) \vec{w}\rangle _{\mathbb{R}^{d}}\ .
\end{equation*}%
Moreover, by\ using \cite[Theorems 4.2 and 5.6 (p)]{OhmVI} together with
simple computations, one checks that the in--phase linear response currents $%
J_{\mathrm{in}}$, which is the component of the total current producing
heat, also called active current, is equal in this case to%
\begin{equation*}
J_{\mathrm{in}}(t)=\int\nolimits_{\mathbb{R}\backslash \{0\}}\mathcal{\hat{E}%
}\left( \nu \right) \ \mathrm{e}^{i\nu t}\ \mu _{\mathrm{AC}}\left( \mathrm{d%
}\nu \right) \vec{w}\ .
\end{equation*}

By (\ref{Levy--Khintchine}) and Definition \ref{def second law2 copy(2)},
observe that the AC--conducti%
\-%
vity measure $\mu _{\mathrm{AC}}$ of the system under consideration is a L%
\'{e}vy measure. This is reminiscent of experimental observations of other
quantum phenomena like (subrecoil) laser cooling \cite{9780511755668}. In
fact, an alternative effective description of the pheno%
\-%
menon of linear conductivity by using L\'{e}vy processes in Fourier space is
discussed in \cite[Section 6]{OhmVI}.

The explicit form of the conductivity measure for concrete models (like the
Hubbard model, for instance) is still an open problem. However, in \cite[%
Section 5.3]{OhmVI}, we were able to qualitatively compare the
AC--conductivity measure associated with the celebrated Drude model with the
L\'{e}vy measure $\mu _{\mathrm{AC}}$ given by Definition \ref{def second
law2 copy(2)}. Indeed, the (in--phase) AC--conductivity measure obtained
from the Drude model is absolutely continuous w.r.t. the Lebesgue measure
with the function%
\index{AC--conductivity measure!Drude model}
\begin{equation}
\nu \mapsto \vartheta _{\mathrm{T}}\left( \nu \right) \sim
\frac{\mathrm{T}}{1+\mathrm{T}^{2}\nu ^{2}}  \label{drude function}
\end{equation}%
being the corresponding Radon--Nikodym derivative. Here, the \emph{%
relaxation time} $\mathrm{T}>0$ is related to the mean time interval between
two collisions of a charged carrier with defects in the crystal. See for
instance \cite[Section 1]{OhmIV} for more discussions. This measure \emph{%
heavily overestimates} $\mu _{\mathrm{AC}}$ at high frequencies. Indeed, as
explained in \cite[Section 5.3]{OhmVI}, by finiteness of the positive
measure $\mathbf{\mu }$, the AC--conductivity measure satisfies
\begin{equation}
\mu _{\mathrm{AC}}\left( \left[ \nu ,\infty \right) \right) \leq \nu ^{-2}%
\mathbf{\mu }\left( \left[ \nu ,\infty \right) \right) \leq \nu ^{-2}\mathbf{%
\mu }\left( \mathbb{R}\right) \ ,\qquad \nu \in \mathbb{R}^{+}\ ,
\label{AC behavior}
\end{equation}%
provided $\Psi ^{\mathrm{IP}}$ satisfies (\ref{(3.3) NS generalized0}) with $%
\varsigma >2d$. The same property of course holds for negative frequencies,
by symmetry of $\mathbf{\mu }$ (w.r.t. $\nu $). Compare (\ref{AC behavior})
with (\ref{drude function}). From Theorem \ref{Thm para regularite}, much
stronger results on the frequency decay of $\mu _{\mathrm{AC}}$ can be
obtained if the interaction $\Psi ^{\mathrm{IP}}$ is fast decaying in space:

\begin{satz}[Moments of AC--conductivity measures]
\label{thm moment mu}\mbox{
}\newline
Let $\lambda \in \mathbb{R}_{0}^{+}$, $\Psi ^{\mathrm{IP}}\in \mathcal{W}$
satisfying (\ref{static potential0}), and assume that the map $\omega
\mapsto \varrho ^{(\omega )}$ is a random invariant passive state.
\index{AC--conductivity measure}\newline
\emph{(i)} Polynomial decay: Assume $\Psi ^{\mathrm{IP}}$ satisfies (\ref%
{(3.3) NS generalized0}) with $\varsigma >2d$. Then, for any $m\in \mathbb{N}
$ satisfying $d(m+1)<\varsigma $,
\begin{equation}
\int\nolimits_{\mathbb{R}\backslash \{0\}}\nu ^{m+1}\mu _{\mathrm{AC}}\left(
\mathrm{d}\nu \right) \in \mathcal{B}_{+}(\mathbb{R}^{d})\ ,  \label{moemnt}
\end{equation}%
i.e., the $(m+1)$-th moment of the measure $\mu _{\mathrm{AC}}$ exists.%
\newline
\emph{(ii)} Exponential decay: Assume $\Psi ^{\mathrm{IP}}$ satisfies (\ref%
{(3.3) NS generalized}). Then, (\ref{moemnt}) holds true for all $m\in
\mathbb{N}$.
\end{satz}

\begin{proof}
By (\ref{Levy--Khintchine}) and Lebesgue's dominated convergence theorem,
for any $t\in \mathbb{R}$,
\begin{equation*}
\partial _{t}^{2}\mathbf{\Xi }_{\mathrm{p}}\left( t\right) =-\int\nolimits_{%
\mathbb{R}}\cos \left( t\nu \right) \mathbf{\mu }\left( \mathrm{d}\nu
\right) =-\ \int\nolimits_{\mathbb{R}}\mathrm{e}^{it\nu }\mathbf{\mu }\left(
\mathrm{d}\nu \right)
\text{ },
\end{equation*}%
provided $\varsigma >2d$ in (\ref{(3.3) NS generalized0}) (with $\Psi =\Psi
^{\mathrm{IP}}$). In other words, the finite and symmetric $\mathcal{B}_{+}(%
\mathbb{R}^{d})$--valued measure $\mathbf{\mu }$ on $\mathbb{R}$ can be seen
as the Fourier transform of $-\partial _{t}^{2}\mathbf{\Xi }_{\mathrm{p}%
}\left( t\right) $ or, that is, as the characteristic function of $\mathbf{%
\mu }$. Therefore, by well--known properties of characteristic functions
(see, e.g., \cite[Theorem 3.3.9.]{Durrett} for the special case $n=2$ and
\cite[Theorem 15.34]{Klenke} for the general case $n\in 2\mathbb{N}_{0}$),
for any even $n\in \mathbb{N}_{0}$, $\partial _{t}^{2}\mathbf{\Xi }_{\mathrm{%
p}}\in C^{n}(\mathbb{R};\mathcal{B}(\mathbb{R}^{d}))$ implies that%
\begin{equation*}
\int\nolimits_{\mathbb{R}}\nu ^{n}\mathbf{\mu }\left( \mathrm{d}\nu \right)
\in \mathcal{B}_{+}(\mathbb{R}^{d})\ .
\end{equation*}%
If $m\in \mathbb{N}_{0}$ is odd, then, by the above assertion for $n<m$ and
the symmetry of the measure $\mathbf{\mu }$ (which follows from the symmetry
of $\mu _{\mathrm{AC}}$), we conclude that%
\begin{equation*}
\int\nolimits_{\mathbb{R}}\nu ^{m}\mathbf{\mu }\left( \mathrm{d}\nu \right)
=0\in \mathcal{B}_{+}(\mathbb{R}^{d})\ .
\end{equation*}%
This observation combined with Theorem \ref{Thm para regularite} and\
Definition \ref{def second law2 copy(2)} yields Assertions (i)--(ii).
\end{proof}

\begin{bemerkung}[Fermion systems with random Laplacians]
\label{Fermion systems with random Laplacians copy(1)}\mbox{
}\newline
The same assertion holds for the random models treated in \cite{OhmVI}. See
also Remark \ref{Fermion systems with random Laplacians}.
\end{bemerkung}

\noindent This last theorem is a significant improvement of the asymptotics (%
\ref{AC behavior}) of \cite{OhmVI} and is a straightforward application of
Lieb--Robinson bounds for multi--com%
\-%
mutators of high orders (Theorems \ref{Theorem Lieb-Robinson copy(1)}--\ref%
{theorem exp tree decay copy(1)}), see Lemma \ref{Lemma utile para}.\bigskip

\noindent \textit{Acknowledgments:} This research is supported by the agency
FAPESP under Grant 2013/13215-5 as well as by the Basque Government through
the grant IT641-13 and the BERC 2014-2017 program and by the Spanish
Ministry of Economy and Competitiveness MINECO: BCAM Severo Ochoa
accreditation SEV-2013-0323, MTM2014-53850.

\printindex%

\end{document}